\documentclass[11pt]{article}

\usepackage[english]{babel}
\usepackage{a4wide}
\usepackage{amsfonts}
\usepackage{amsmath} \usepackage{tikz,graphicx,subfigure,overpic,verbatim}
\usepackage{color} \usepackage{amssymb} \usepackage{amsthm}
\usepackage{hyperref}
\usepackage[fixlanguage]{babelbib}
\usetikzlibrary{arrows,decorations.markings}\def\R{\mathbb{R}}
\def\C{\mathbb{C}}
\def\N{\mathbb{N}}

\def\EE{\mathbb{E}}

\def\Var{\mathop{\mathrm{Var}}\nolimits}

\def\Tr{\mathop{\mathrm{Tr}}\nolimits}
\def\Re{\mathop{\mathrm{Re}}\nolimits}
\def\Im{\mathop{\mathrm{Im}}\nolimits}

\def\supp{\mathop{\mathrm{supp}}\nolimits}
\def\dist{\mathop{\mathrm{dist}}\nolimits}
\def\diag{\mathop{\mathrm{diag}}\nolimits}

\def\eps{\varepsilon}

\numberwithin{equation}{section}

\newtheorem{theorem}{Theorem}[section]
\newtheorem{lemma}[theorem]{Lemma}
\newtheorem{proposition}[theorem]{Proposition}

\newtheorem{corollary}[theorem]{Corollary}

\theoremstyle{definition}

\theoremstyle{remark}

\newtheorem{remark}[theorem]{Remark}

\title{On mesoscopic equilibrium for linear statistics in Dyson's Brownian motion}

\author{Maurice Duits\footnote{Department of Mathematics, Royal Institute of Technology (KTH), Lindstedtsv\"agen 25, SE-10044 Stockholm, Sweden. Supported in part by the grant KAW 2010.0063 from the Knut and Alice Wallenberg Foundation and by the Swedish Research Council (VR) Grant no. 2012-3128.}   \and Kurt Johansson\footnote{Department of Mathematics, Royal Institute of Technology (KTH), Lindstedtsv\"agen 25, SE-10044 Stockholm, Sweden. Supported in part by the grant KAW 2010.0063 from the Knut and Alice Wallenberg Foundation and by the Swedish Research Council (VR). }}
\date{}

\begin{document}

\maketitle
\begin{abstract}
In this paper we study mesoscopic fluctuations for Dyson's Brownian motion with  $\beta=2$. Dyson showed  that the Gaussian Unitary Ensemble (GUE) is the invariant measure for this stochastic evolution and  conjectured that, when starting  from a  generic configuration of initial points, the time that is needed for the GUE statistics to become dominant  depends on the scale we look at: The microscopic correlations arrive at the equilibrium regime sooner than the macrosopic correlations. In this paper we   investigate the transition on  the intermediate, i.e. mesoscopic, scales. The time scales that we consider are such that the system is already in microscopic equilibrium (sine-universality for the local correlations), but we have not yet  reached equilibrium at the macrosopic scale. We describe the transition to equilibrium on all  mesoscopic scales  by means of Central Limit Theorems  for linear statistics with sufficiently smooth test functions. We consider two situations: deterministic initial points and randomly chosen initial points. In the random situation, we obtain a transition from  the classical Central Limit Theorem for independent random variables to the one for the GUE.
\end{abstract}
\tableofcontents

\section{Introduction}
In \cite{Dyson} Dyson introduced a stochastic evolution on $n$ interacting particles, that we nowadays refer to as Dyson's Brownian Motion. Dyson observed that this dynamics is a natural evolution since the invariant measures are given by the $\beta$-ensembles in random matrix theory. The evolution of the particles is given by the following system of stochastic differential equations,
\begin{equation}\label{eq:DysonSDE}
{\rm d}x_i=\sqrt{\frac 2{n\beta}}\,dB_i-x_i\,{\rm d}t+\frac 1n\sum_{j\neq i}\frac{{\rm d}t}{x_i-x_j},
\end{equation}
for $1\le i\le n$, where $x_i(t)$ are the positions of the particles, $\beta>0$, and the  $B_i$'s are independent standard Brownian motions. For $\beta=2$ this gives the time evolution of the eigenvalues
of an $n\times n$ random Hermitian matrix $M_n(t)$ where the elements of the matrix evolve according to independent Ornstein-Uhlenbeck processes. There is a similar interpretation
for $\beta=1,4$ but we will only consider the case $\beta=2$ in this paper. For the matrices the transition function is given by
\begin{equation}\label{eq:OUtransition}
P_t(H,H')=C_{n,t}\exp\left(-n\frac{\Tr(H-{\rm e}^{-t}H')^2}{1-{\rm e}^{-2t}}\right).
\end{equation}
From this we see that the matrix $M_n(t)$ for a fixed $t\ge 0$ can be written
\begin{equation}\label{eq:interpolatingmodel}
M_n(t)={\rm e}^{-t} \Xi_n+\sqrt{1-{\rm e}^{-2t}} X_n,
\end{equation}
where $\Xi_n=M_n(0)$ is the initial condition, which, by unitary invariance, we can take to be a diagonal matrix,
 $$\Xi_n=\diag(\xi_1^{(n)},\ldots,\xi_n^{(n)}),$$ for vectors  $\xi^{(n)}=(\xi^{(n)}_1,\ldots,\xi_n^{(n)} )\in \R^n.$
The matrix $X_n$  is a GUE matrix of size $n$, that is it is distributed according to 
\begin{equation}\label{eq:GUEdensity}
\frac 1{Z_n}{\rm e}^{-n\Tr X^2}\,dX.
\end{equation}
In other words, $X_n$ is an  $n\times n$ Hermitian matrix where the  upper triangular entries $(X_n)_{ij}$ are complex (for $i<j$) or real (for $i=j$) independent Gaussian random variables such that 
\[ \EE[(X_n)_{ij}]=0, \qquad \EE[|(X_n)_{ij}|^2]=\frac{1}{2n}. \]
The random matrix model \eqref{eq:interpolatingmodel} is also referred to in the literature as deformed GUE  or GUE with external source. 

If we denote the eigenvalues of $M_n(t)$ by $x_1(t),\ldots,x_n(t)$ they have the same distribution as the solution of (\ref{eq:DysonSDE}) with initial condition $\xi^{(n)}$.
  Clearly, $M_n(t)$ interpolates between $M_n(0)= \Xi_n$ and the GUE matrix $M_n(\infty)=X_n$.
We see that the equilibrium density
for (\ref{eq:DysonSDE}) with $\beta=2$ is the eigenvalue density for (\ref{eq:GUEdensity}), i.e.
\begin{equation}\label{eq:equilibriumdensity}
\frac 1{Z_n}\prod_{1\le i<j\le n}(x_i-x_j)^2\prod_{j=1}^n e^{-nx_j^2}.
\end{equation}
The global asymptotic distribution of the eigenvalues is given by the semi-circle law, $\frac 1{\pi}\sqrt{2-x^2}\,dx$.

Based on formal computations, Dyson \cite{Dyson} conjectured that for a generic configuration of the initial points, the time that is needed for the GUE statistics to become dominant depends on the particular scale we are looking at. In particular to see equilibrium at the microscopic scale, i.e. to get a sine-kernel determinantal point process in the limit, we have to look at the limit $n\to \infty$ and take $t=t_n$ such that $nt \to \infty$. This case has been investigated in several papers, e.g. \cite{FN1,FN2,GG,G}. More recently it has played a
prominent role in proving universality for Wigner matrices, see \cite{EYY} for a discussion in relation to the Dyson's conjecture and \cite{EY} for a general review of the recent universality results. In this paper we investigate the approach to equilibrium at longer, so called mesoscopic length scales by investigating the asymptotics of a mesoscopic linear statistic of the eigenvalues. It seems reasonable that if we look at longer length scales it will take a longer time to reach equilibrium since it takes some time for the effect of the repulsion in (\ref{eq:DysonSDE}) and the new randomness coming from the Brownian term to propagate in the particle system.

Consider the particles $x_j(t)$, $1\le j\le n$ in (\ref{eq:DysonSDE}) or equivalently the eigenvalues of $M_n(t)$ in
(\ref{eq:interpolatingmodel}) which have the same distribution. Fix $0<\alpha, \gamma<1$ and let $f$ be a real-valued
compactly supported function on the real line. Consider the time scale $t_n=\tau n^{-\gamma}/\sqrt{2-{x^{\ast}}^2}$. Note that we
exclude very short times and also times of order $\sim 1$ and times going to infinity with $n$.
For $x_\ast\in (-\sqrt{2},\sqrt{2})$ we define the mesoscopic linear statistic at time $t_n$ and at scale $n^{-\alpha}$ around $x_\ast$ by
\begin{equation}\label{eq:deflinearstatistic}
Y_n(f)=\sum_{j=1}^nf(n^\alpha(x_j(t_n)-x_\ast)).
\end{equation}
The case $\alpha =0$ corresponds to the global or macroscopic scale, i.e. we are sampling all the eigenvalues, whereas the
case $\alpha =1$ corresponds to sampling the points locally or microscopically, i.e. as a point process. For this reason the
range $0<\alpha <1$ is called the mesoscopic scale.

In the equilibrium case, i.e. for GUE (\ref{eq:GUEdensity}), we can form the mesoscopic linear statistic
\begin{equation}\label{eq:GUElinearstatistic}
S_n(f)=\sum_{j=1}^nf(n^\alpha(\lambda_j-x_\ast)),
\end{equation}
where $\lambda_j$ are the eigenvalues of $M$. In that case, the results in \cite{BdMK1,BdMK2,Sosh3} suggest that we have the following Central Limit Theorem
\begin{align}\label{eq:NewCLTGUE}
S_n(f)-\EE S_n(f) \to N(0,\sigma_{\infty,f} ^2),
\end{align} 
in distribution as $n\to \infty$, where $\sigma_{\infty}(f)^2$ is given by 
\begin{align} \nonumber
\sigma_{\infty}(f)^2=\frac{1}{4 \pi^2}\iint  \left(\frac{f(u) -f(v)}{u-v}\right)^2 {\rm d} u {\rm dÊ} v.
\end{align}
Note that $\sigma_{\infty}(f)^2$ does not depend on $\alpha$.  Indeed,  similar statements have been derived for random matrix models from the classical compact groups \cite{Sosh3}, and for smoothened statistics for the GOE  \cite{BdMK1}  and for Wigner matrices \cite{BdMK2} (with $0<\alpha<1/8$). In fact, we believe it is likely that the techniques in \cite{BdMK1} can be used to prove \eqref{eq:NewCLTGUE}  for GUE for some class of functions $f$. It seems reasonable to expect that this Central Limit Theorem is universal and holds for many ensembles from random matrix theory. It is therefore natural to ask the question whether a similar result holds for the
mesoscopic linear statistic (\ref{eq:deflinearstatistic}), i.e.  a Central Limit Theorem for the appropriately normalized statistic
\begin{align}\label{eq:CLTMesoscopic}
Y_n(f)-\EE Y_n(f).
\end{align} 
We will consider this question for both deterministic and random initial points.
Consider the case of random initial points, i.e. $\xi_1^{(n)},\ldots,\xi_n^{(n)}$ are sampled independently from some distribution; we will take
the semi-circle law for simplicity. For short times, we can expect that we should have the standard  Central Limit Theorem with $\sqrt{n^{1-\alpha}}$ normalization,
and for long times we should get the same, GUE-type limit as in (\ref{eq:NewCLTGUE}) with no normalization. What happens in between? Is there a transition for a certain relation between $\alpha$ and $\gamma$? We can also ask the same question for a deterministic sequence of initial points. These are the questions we will address in this paper.   We note that similar mesoscopic fluctuations questions have been answered  in the recent papers \cite{EKN1,EKN2} for a class of band-matrix type models with varying bandwidth.

In a nutshell, our main results are as follows. In the case of deterministic initial points  we show, under a  regularity condition on the points $\xi^{(n)}_j$, that there are two regimes: $\alpha<\gamma$ and  $\alpha>\gamma$. In the regime $\alpha<\gamma$, the time scales are relatively short and the system has not yet reached equilibrium. We prove for that regime  that the the variance of the linear statistic is of lower order compared to GUE, see Theorem \ref{th:variance}.  For $\alpha>\gamma$, when we consider long time scales,  the system has reached equilibrium in the sense that we retrieve, in Theorem \ref{th:CLTfixed}, the Central Limit Theorem in \eqref{eq:NewCLTGUE}.  For $\alpha= \gamma$ there is a transition to equilibrium that we also describe in terms of  a Central Limit Theorem, see also Theorem \ref{th:CLTfixed}. In the latter,  the  limiting variance depends on a transition parameter. For random initial points, we prove that  there are three regimes: one regime where the leading asymptotic term for the linear statistic  resembles that of  independent random variables, a regime where we observe random matrix statistics  and a third intermediate regime in which the size of the fluctuations gradually changes from  $\sim \sqrt{ n^{1-\alpha}}$ to $\sim 1$. In Theorems \ref{th:random1} and \ref{th:random2} we  provide Central Limit Theorems for the fluctuations in all three regimes and for the transitions. A remarkable feature is that the intermediate regime and its associated Central Limit Theorem depend on $f$ only through the lowest vanishing moment.

\section{Statement of results} 

In this section we will state our main results. The proofs are postponed to later sections.

\subsection{Assumptions on $\xi_j^{(n)}$}

In the analysis we will pose some assumptions on the initial points $\xi^{(j)}_n$, that we will now discuss.

First of all, we will assume that normalized counting measure on the components of the vector $\xi^{(n)}$ converges weakly to the semi-circle law
\begin{align}\label{eq:weaklimit}
\frac{1}{n} \sum_{j=1}^n \delta_{\xi_j^{(n)}} \to \frac{1}{\pi } \sqrt {2-x^2} {\rm d} x,  \text { on } [-\sqrt 2, \sqrt 2]. 
\end{align}
This assumption leads to the convenient property that for any $t>0$ we have that the empirical measure on the eigenvalues $x_j(t)$ converges to the semi-circle law, i.e. $\frac{1}{n} \sum_{j=1}^n \delta_{x_j(t)} \to \frac{1}{\pi } \sqrt {2-x^2} {\rm d} x$ for any $t\geq 0$. We believe  the semi-circle is not  essential and  for a different limiting measure we one can still formulate the analogues of our results. However, with a different limiting measure, there is an additional phenomenon that the limiting eigenvalue distribution evolves towards the semi-circle as $n,t\to \infty$. Since in our opinion this is a mere distraction from the actual behavior we want to capture, we choose to work with the semi-circle as the limiting measure for the initial points.

We will also require a  (mild) regularity in  the points $\xi_j^{(n)}$.  To this end, we define for $n\in \N$,  $a>0$ and $U\subset \R$ the set
\begin{equation} \label{eq:defCn}
\mathcal C_n(U, a) = \left\{ \xi^{(n)} \in \R^n \mid \sup_{w:{ \Im w \geq 1/n},\ { \Re w\in U}} \sqrt{\frac{\Im w}{n}} \left|\sum_{j=1}^n \left(\frac{1}{w-\xi_j^{(n)}}-\frac{1}{\pi}\int \frac{\sqrt{2-\xi^2}}{w-\xi} \ {\rm d}\xi\right)\right| \leq a\right\}.
\end{equation}
Note that we have  $\mathcal C_n(U, a_1) \subset \mathcal C_n(U, a_2) $ if $a_1< a_2$ and $\mathcal C_n(U_1, a) \subset \mathcal C_n(U_2, a) $ if $U_2 \subset U_1$. 
\begin{equation} \label{eq:defC}
\mathcal C(U,A,\delta)=\{(\xi^{(n)})_{n\in \mathbb N} \mid \forall n: \xi^{(n)}\in \mathcal C_n(U,An^\delta)\}
\end{equation}
The regularity condition that we will pose on the initial points is the following. Suppose we look around  a point $x_*\in  (-\sqrt 2, \sqrt 2)$ at a time $t \sim n^{-\gamma}$ for some $0 <\gamma <1$. We then allow for sequences $\xi=(\xi^{(n)})_{n\in \N}$ such that for some $$0<\delta <\frac{1}{2}\frac{1-\gamma}{1+\gamma}, A>0 \textrm{  and an open interval } U \ni x_* \textrm{ we have  }\xi^{(n)} \in \mathcal C_n(U, A n^\delta),$$ for sufficiently large $n\in \N$. 
The defining inequality for the set in \eqref{eq:defCn} is trivially satisfied for large $w$. It is most restrictive for $w$ close to the imaginary axis. The condition poses a regularity condition on the local distribution of points $\xi^{(n)}_j$ near each point of $U$. Although we believe some regularity is necessary, we emphasize that the above condition is rather weak. Indeed, we will show later that if $\xi^{(n)}$ consists of independent samples of the semi-circle law, then the regularity condition is satisfied with probability one for $U=\R$. Note that  a straightforward  computation shows that  variance of the left-hand side of the defining inequality in \eqref{eq:defCn} for $w$ in the allowed region is of finite order.

\subsection{Deterministic initial points}

We now discuss our main results, Theorem \ref{th:variance} and \ref{th:CLTfixed},  in case of deterministic initial points $\xi_j^{(n)}$. The setup for both theorems is the following  We let $f\in C^1_c(\R)$,  where $C^1_c(\R)$ stands for the class of continuously differentiable functions with compact support, and consider the linear statistic $Y_n(f)$ as defined in \eqref{eq:deflinearstatistic} with parameters
\begin{equation} \label{eq:para1}
0 < \alpha,\gamma <1,  \quad x^* \in (-\sqrt 2, \sqrt 2),
\end{equation}
and 
\begin{equation}\label{eq:para2}
\tau>0,  \quad \text{ and }  \quad   t=  \frac{\tau}{n^{\gamma} \sqrt{2-x_*^2}}.
\end{equation}
Moreover, we allow the following set of initial points $\xi=(\xi^{(n)})_{n\in \N}$.  Let  
\begin{equation}\label{eq:para3}
0<\delta<\frac{1}{2} \frac{1-\gamma}{1+\gamma},\  A>0 \text{  and  } U \text{ an open interval containing } x_*.
\end{equation} Then we consider initial points $\xi= (\xi^{(n)})_{n \in \N}\in \mathcal C(U,A,\delta)$ with $\mathcal C(U,A,\delta)$ as defined in \eqref{eq:defC}. 

We  will use the notations
\begin{align}
\sigma_\infty(f)^2&:= \frac{1}{4 \pi^2}\iint  \left(\frac{f(u) -f(v)}{u-v}\right)^2 {\rm d} u {\rm dÊ} v,\label{eq:defsigmafinfty} \\
\sigma_\tau(f)^2&:=\frac{\tau}{2 \pi^2}\iint  \left(\frac{f(u) -f(v)}{u-v}\right)^2 \frac{2 \tau}{(u-v)^2+4\tau^2}{\rm d} u {\rm dÊ} v.\end{align}
The first theorem is on the variance of $Y_n(f)$.

\begin{theorem}\label{th:variance}   Let $ f \in C_c^1(\R)$. Then, as $n\to \infty$,  the limiting behavior of variance of the linear statistic  $Y_n(f)$, with parameters  as in \eqref{eq:para1}--\eqref{eq:para3}, is given by 
\begin{align}\nonumber
 \Var Y_n(f)=
\begin{cases}
\sigma_\infty(f)^2+o(1)& \alpha >\gamma,\\
\sigma_\tau(f)^2+o(1), &\alpha=\gamma,\\
o(1), & \alpha <\gamma.
\end{cases}\end{align}
 uniformly for $\xi\in \mathcal C (U, A, \delta)$.
\end{theorem}
This result shows how the transition to the GUE regime comes about. As long as $\alpha <\gamma$ we are  still too close to the initial situation to have GUE fluctuations. The transition occurs at $\alpha=\gamma$ where we have a variance that is a weighted Sobolev norm with a weight that depends on a parameter $\tau >0$. When $\tau \to \infty$ we enter the region $\alpha>\gamma$ and the variance is as in the GUE. 
\begin{remark}
The variance $\sigma_\infty^2(f)$ is a half-order Sobolev norm and can alternatively be rewritten in a nice form in terms of the Fourier transform. Denote the Fourier transform of $f$ by
 \begin{equation}\label{eq:ff1} \hat f(\omega)= \frac{1}{\sqrt{ 2 \pi}} \int_{-{\rm i} \infty}^\infty f(x) {\rm e}^{-{\rm i} x \omega} {\rm d} x.
 \end{equation}
Then we have
$$\sigma_\infty(f)^2= \frac{1}{2\pi} \int_{-\infty}^\infty|\hat f(\omega)|^2  |\omega|  {\rm d} \omega.$$
Similarly, $\sigma_\tau^2(f)$ has the representation
$$\sigma_\tau(f)^2= \frac{1}{4\pi} \int_{-\infty}^\infty |\hat f(\omega)|^2 \frac{{\rm e}^{-2 \tau |\omega|}-1+2 \tau|\omega|}{\tau}  {\rm d} \omega.$$
We see, again, that for $\tau \to \infty$ we obtain $\sigma_{\tau}(f)^2\to \sigma_\infty(f)^2$. Moreover, $$\sigma_\tau(f)^2= \frac{\tau }{2 \pi} \int (f'(x))^2 {\rm d} x+\mathcal O(\tau^2),$$
as $\tau \downarrow 0$ for sufficiently smooth $f$. The leading term in the latter expansion resembles the one in the situation where we consider the same linear statistic but now with $x_j(t)$ replaced by $n$ independent Ornstein-Uhlenbeck processes (see also Appendix A) leaving from $\xi_j^{(n)}$.  This suggests that in the regime $\alpha<\gamma$ there is no difference between the two situations in the leading order term in the asymptotic expansion for the variance.  A full  rigorous analysis for the linear statistic in this regime will be left open and is not part of the present paper.
\end{remark}

The proof of Theorem \ref{th:variance} is postponed to Section \ref{sec:variance}. It is based on the fact that for fixed  initial points $\xi_j^{(n)}$ and the eigenvalues $x_1(t), \ldots,x_n(t)$  at time $t>0$ form a determinantal point process with a kernel  $K_n$ given by 
 kernel $K_n$ given by 
\begin{align}\label{eq:defKn}
K_n(x,y;t) =\frac{ n}{\sinh t (2 \pi {\rm i})^2} \oint_\Sigma {\rm d}z \int_\Gamma {\rm d}w  \frac{ {\rm e}^{\frac{n} {1-{\rm e}^{-2t}}({\rm e}^{-t} w-x)^2}}{ {\rm e}^{\frac{n} {1-{\rm e}^{-2t} }({\rm e}^{-t} z-y)^2}} \prod_{j=1}^n \frac{w-\xi_j^{(n)}}{z-\xi_j^{(n)}} \frac{1}{w-z},
\end{align}
where the contour $\Sigma$ is a counter clockwise oriented simple contour surrounding the poles $\xi_1^{(n)},\ldots, \xi_n^{(n)}$, and $\Sigma$ is a contour that connects $-{\rm i} \infty $ to ${\rm i} \infty$ and lies at the right of $\Sigma$. This fact was first proved in \cite{J1}. For completeness we will provide a short discussion of the proof in the appendix.

The fact that the points form a determinantal point process implies that the probability measure on the the eigenvalues $x_j(t)$ of \eqref{eq:interpolatingmodel} is given by 
\begin{align}\nonumber
\frac{1}{n!} \det \left(K_n(x_1,x_j;t\right)_{i,j=1}^n {\rm d} x_1 \cdots {\rm d} x_n,
\end{align}
and similarly for the $k$-point correlations, cf. \eqref{eq:correlations}.  In particular, the variance for any linear statistic for a determinantal point process can be rewritten in terms of its kernel, cf. \eqref{eq:variancedeterminantalpointprocess}.  An important ingredient in the proof of Theorem \ref{th:variance} is an asymptotic analysis of the kernel $K_n$ as $n\to \infty$ for the parameters that we are interested in. This we do in Section \ref{sec:steepest} using methods of classical steepest descent. Although the method is standard and used often in the recent literature to deal with double integral formulas, the analysis here is rather delicate which is due to the fact  that we are looking at  regimes close to the initial points.

In the next theorem, our main result for deterministic initial points, we provide Central Limit Theorems for the fluctuations. 

\begin{theorem}\label{th:CLTfixed}
Let $f\in C_c^1(\R)$. Then, as $n\to \infty$, the limiting behavior of the moment generating function for the linear statistic  $Y_n(f)$, with parameters as in  \eqref{eq:para1}--\eqref{eq:para3}, is given by 
\begin{equation}\nonumber
\EE \left[\exp \lambda (Y_n(f)-\EE Y_n(f) )\right]=\begin{cases}{\rm e}^ {\frac{1}{2} \lambda^2 \sigma_{\infty}(f)^2} (1+o(1)), &\alpha>\gamma\\
{\rm e}^{\frac{1}{2}  \lambda^2 \sigma_{\tau}(f)^2} (1+o(1)), &\alpha=\gamma\end{cases}\end{equation}
uniformly for $\lambda $ in a small neighborhood of the origin and $\xi\in \mathcal C(U,A,\delta)$.

Hence, for any $(\xi^{(n)})_{n \in \N}\in  \mathcal C_n(U,A,\delta)$ and $f\in C_c^1(\R)$ we have   \begin{equation}\nonumber
Y_n(f)-\EE Y_n(f) \overset{\mathcal D}{\longrightarrow} 
\begin{cases}
N(0,\sigma_{\infty}(f)^2), &\alpha >\gamma,\\
N(0,\sigma_{\tau}(f)^2), & \alpha =\gamma.
\end{cases},
\end{equation}

\end{theorem}

The proof of this theorem will be postponed to Section \ref{sec:loopeqn}.

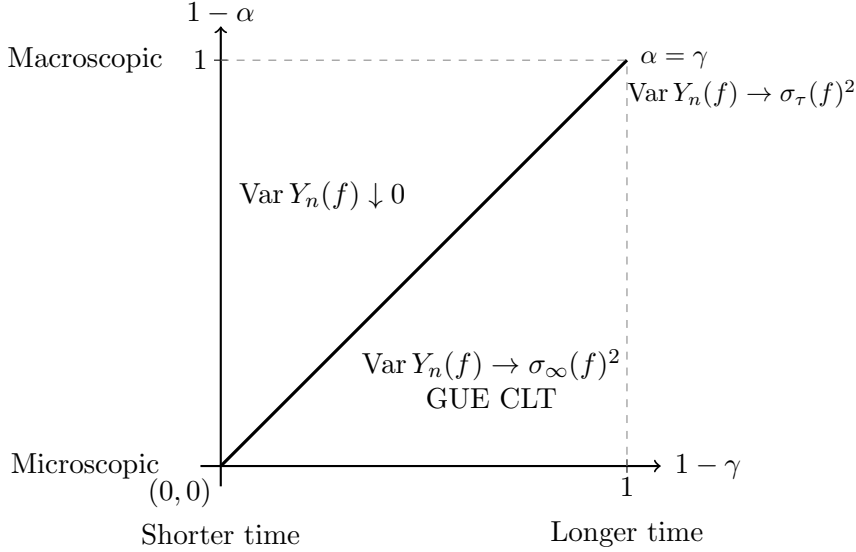
\begin{figure}[t]
\begin{center}
\begin{tikzpicture}[scale=0.9]
\draw[->,thick] (-.3,0)--(6.5,0);
\draw[->,thick] (0,-.3)--(0,6.5);
\draw[-,very thick] (0,0) --(6,6);
\draw (-0.6,-0.4) node {$(0,0)$};
\draw (-.3,6) node {$1$};
\draw (6,-.3) node {$1$};
\draw (0,6.7) node {$1-\alpha$};
\draw (6.7,6) node {\small{$\alpha=\gamma $}};
\draw (7.7,5.5) node {\small{$\Var Y_n(f) \to \sigma_\tau(f)^2$}};
\draw (7.2,0) node {$1-\gamma$};
\draw[-] (6,-.1)--(6,.1);
\draw[-] (-.1,6)--(.1,6);
\draw[help lines,dashed] (0,6)--(6,6)--(6,0);
\draw (-2,0) node{Microscopic};
\draw (-2,6) node{Macroscopic};
\draw (0,-1) node{Shorter time};
\draw (6,-1) node{Longer time};
\draw (1.5,4) node{$ \Var Y_n(f)  \downarrow 0$} ;
\draw (4,1.5) node{$ \Var Y_n(f) \to  \sigma_\infty(f)^2$} ;
\draw (4,1) node{GUE CLT} ;

\end{tikzpicture}
\caption{Diagram representing the different regimes for deterministic initial points. In the regimes we have indicated the limiting behavior of the variance of  $Y_n(f)-\EE Y_n(f)$ in Theorem \ref{th:variance}. In the regime $\alpha> \gamma$ we retrieve the Central Limit Theorem of GUE type. At the line there is also  a Central Limit Theorem but with a transition in the variance.}
\label{fig:phase}
\end{center}
\end{figure}

One approach to prove  \eqref{eq:CLTsuff} is to use the determinantal structure  and compute the asymptotic behavior  of the cumulants. For example, CLT's on the mesoscopic scale were proved for the classical compact groups in \cite{Sosh3} in this way. The cumulant approach was also  a starting point in \cite{BD2}. However, the arguments \cite{Sosh3}  depend strongly on the special structure of the (self-ajoint) kernel for the determinantal point process.  It appears to be difficult task to use these methods in our situation and we have not been able to do so.    The proof that we will present here is based on the loop equations.  Rigorous proofs for Central Limit Theorem for linear statistics using loop equations have been given for Unitary Ensembles \cite{Jduke} and  for the fluctuations in normal matrix models in  \cite{AHM} leading to the Gaussian Free Field. However, in these cases they are used to compute the fluctuations on the macroscopic scale, whereas we will use them here on the mesoscopic scale. 
 
\subsection{Concentration inequalities}

In the proof of Theorem \ref{th:CLTfixed} we will need  an a priori bound as an input into the   loop equations. To this end we prove a concentration inequality for linear statistics that is an interesting result in its own right. Moreover, this inequality allows us to extend Theorems \ref{th:variance} and  \ref{th:CLTfixed} to a larger class of functions than $C_c^1(\R)$. 

We define the 
\[\mathcal L_w^{{\rm Pre}}:= \left\{f:\R\to \R \mid  |f|_{\mathcal L_w} := \sup_{x,y\in \R} \sqrt{1+x^2}\sqrt{1+y^2}\left|\frac{f(x)-f(y)}{x-y}\right|<\infty\right\}\]
Note that $|f|_{\mathcal L_w}$ is weighted Lipschitz constant.   Since  $|f|_{\mathcal L_w}=0$ iff $f$ is a constant, it is a pre-norm on $\mathcal L_w^{{\rm Pre}}$. By fixing the constant we obtain the normed space 
\[\mathcal L_w:= \{f\in \mathcal L_w^{{\rm Pre}} \mid \lim_{x\to \infty} f(x)=0.\}\]
Note that  $C_c^1(\R)\subset \mathcal L_w$. Then for any function $f\in \mathcal L_w$ we have the following bound on the moment-generating function of a linear statistic.

\begin{proposition} \label{prop:localconcenpre}
Consider the linear statistic $Y_n(f)$ in \eqref{eq:deflinearstatistic} with parameters as in \eqref{eq:para1}--\eqref{eq:para3}. Then there exists  constants $d_1,d_2>0$ such that for $ n$ sufficiently large we have
\begin{equation}\label{eq:localconcent}
\left|\log \EE \left[\exp \lambda \left( Y_n(f)(t)-\EE[Y_n(f)(t)]\right)\right]\right|\\
\leq  d_2 |f|_{\mathcal L_w}^2 |\lambda|^2,
\end{equation}
for $f\in \mathcal L_w$, complex $\lambda$ such that $|\lambda|\leq 1/(d_1 \|f\|_\infty)$ and $(\xi^{(n)})_{n \in \N}  \in \mathcal C(U,A,\delta)$.
\end{proposition} 

Note the left-hand side of \eqref{eq:localconcent}  depends on $\alpha$ by $Y_n(f)$ in \eqref{eq:deflinearstatistic}, but the right-hand side does not. A particular consequence of this proposition is that, under the same conditions and assumptions, there exists a constant $d>0$ such that for $n$ sufficiently large we have
\begin{align} \label{eq:variancegenb}
\Var Y_n(f) \leq d |f|_{\mathcal L_w}^2,
\end{align}
for any $f\in \mathcal L_w$. Moreover, by applying the exponential version of the Chebyshev inequality follows that any $\eps>0$ the probability that $|Y_n(f)(t)-\EE[Y_n(f)(t)]|>n^\eps$ is exponentially small. In other words, starting from initial points $\xi \in \mathcal C(U,A,\delta)$, for times $t>>n^{-1}$ the eigenvalue distribution is very close to the semi-circle law  at all mesoscopic scales. In this sense, we can think of Proposition \ref{prop:localconcenpre} as a local semi-circle law analogous to the recent results for Wigner matrices, e.g. \cite{EY,EYY}. 

The proof of Proposition \ref{prop:localconcenpre} is postponed to Section \ref{sec:concen}. The Proposition is the  result of combining two different concentration inequalities. Firstly, using the determinantal structure of the correlations we prove a  concentration inequality for $f\in C_c^1(\R)$ (cf. Proposition \ref{th:boundonlinstat}). This concentration inequality is inspired on an inequality for   determinantal point processes with a self-adjoint projection operators as a kernel that was introduced in \cite{BD}. The idea behind this inequality was also used in \cite{BD2} in the context of Central Limit Theorms for linear statistics.  In the present setting, the kernel is no longer self-adjoint and this leads to important complication that we need to deal with. Still, the concentration inequality alone is not sufficient for our purposes because it only holds for function with compact support. In the loop equations we need a concentration inequality for functions with unbounded support, as in Proposition \ref{prop:localconcenpre}. Therefore we combine Proposition with \ref{th:boundonlinstat} with  a concentration inequality due to Herbst (cf. Propostion \ref{prop:herbst}). This strong result is valid for probability  measures satisfying the Logarithmic Sobolev inequality. The benefit of this result is that it holds for any Lipschitz function. The downside is that  the bound depends on the Lipschitz norm and when we rescale the function this Lipschitz norm blows up. In other words, it is not designed for the mesoscopic scales that we are interested in. However, a combination of the two  inequalities leads to Proposition \ref{prop:localconcenpre} which contains all that we need for our purposes.

\subsection{Random initial points}
Theorems \ref{th:variance} and \ref{th:CLTfixed} deal with the random point process for fixed $\xi^{(n)}\in \R^n$. We now study the process when we also let the $\xi^{(n)}_j$ be random. More precisely, we  let the initial points $\xi^{(n)}_j$ be independent random variables each having the semi-circle as their distribution. Now, instead of an increasing magnitude of the variance as in the case of fixed initial point, we will have a stronger disorder for short time scales. The statistics for these short times scales coincide with those for  independent random variables. As time evolves, the  Dyson's Brownian starts regularizing the distribution of points, until eventually, the statistics obey the laws of a the eigenvalues of a  GUE ensemble.

Before we pose our main results and see how the transition comes about,  we will first introduce some new notation. We recall that when fixing the initial points, the $x_j(t)$ form a determinantal point process with kernel $K_n$ in \eqref{eq:defKn}.  When dealing with random initial points, we consider the probability measure on $\R^n\times\R^n$ defined by 
\begin{equation}\label{eq:randommeasure}
\frac{1}{n!\pi^n} \det \left(K_{n}(x_i,x_j;t)\right)_{i,j=1}^n \prod_{j=1}^n  \sqrt{2-\xi^2_j}  {\rm d}x_1 \cdots {\rm d} x_n {\rm d}  \xi_1 \cdots{\rm d}\xi_n.
\end{equation}
 Let $\EE_{K_n}$ denote the expectation with respect to the determinantal point process with kernel $K_n$ for given $\xi^{(n)}$. Let $\EE_{\xi}$ denote the expectation with respect to the random initial points. Hence the expectation with respect to the full measure \eqref{eq:randommeasure} is given by $\EE_{\xi} \EE_{K_{n}}$. We will derive Central Limit Theorems for the centered linear statistic 
\begin{align}\nonumber
Y_n(f)-\EE_{\xi} \EE_{K_{n}} Y_n(f),
\end{align}
under the same choice of parameters as in \eqref{eq:para1}--\eqref{eq:para3}. Of course, we can not force $\xi \in C(U,A,\delta) $, since the $\xi_j^{(n)}$ are random,  but we will prove that $\xi \in C(\R,A,\delta) $ (hence $U=\R$) with high probability (cf. Lemma  \ref{lem:highprob}).  For technical reasons we will work with $f\in C^\infty_c(\R),$ (instead of $C_c^1(\R)$),
where $C^\infty_c(\R)$ stands for the class of infinitely differentiable functions with compact support. 

We also  need some more notation. Define the operator $\mathcal P_\tau$ as 
 \begin{align}\label{eq:defPeps}
 \mathcal P_\tau g(x)= P_\tau* g(x)= \int g(y) P_\tau(x-y) {\rm d} y,
 \end{align}
 where $P_\tau$ is the Poisson kernel
 \begin{equation}\nonumber
 P_\tau(x)= \frac{1}{\pi } \Im \frac{1}{x-{\rm i} \tau} =\frac{1}{\pi} \frac{\tau}{x^2+\tau^2}. 
 \end{equation}
The following is our first result on random initial points.

\begin{theorem}\label{th:random1}
Let $f\in C_c^\infty(\R)$. Then, as $n\to \infty$,
\begin{equation}\nonumber
\frac{1}{n^{(1-\alpha)/2}} \left(Y_n(f)-\EE_{\xi} \EE_{K_{n}} Y_n(f) \right)\to 
\begin{cases}
N(0,\pi^{-1} \sqrt{2-x_*^2}\|f\|^2_2),  & \text{if } \alpha < \gamma,\\
N(0,\pi^{-1} \sqrt{2-x_*^2}\|\mathcal P_\tau f\|^2_2),  & \text{if } \alpha = \gamma,\\
\end{cases}
\end{equation}
in distribution. 
\end{theorem}

Note that the results for $\alpha < \gamma$ matches the classical Central Limit Theorem when sampling  a scaled function with independent random variables taken from the semi-circle law.  When $\alpha= \gamma$ a first transition happens and  the variance takes a different form.  However,  the normalization is still the same as in the case of independent random variables and we are still away from  GUE the regime.
\begin{figure}[t]
\begin{center}
\begin{tikzpicture}[scale=0.95]
\draw[->,thick] (-.3,0)--(6.5,0);
\draw[->,thick] (0,-.3)--(0,6.5);
\draw[-,very thick] (0,0) --(6,6);
\draw[-,very thick] (0,0)--(6,4);
\draw (-0.6,-0.4) node {$(0,0)$};
\draw (-.3,6) node {$1$};
\draw (6,-.3) node {$1$};
\draw (0,6.7) node {$1-\alpha$};
\draw (7.2,4) node {\small{$\alpha= \frac{1+(2p+1) \gamma}{2p+2}$}};
\draw (6.7,6) node {\small{$\alpha=\gamma$}};
\draw (7.2,0) node {$1-\gamma$};
\draw[-] (6,-.1)--(6,.1);
\draw[-] (-.1,6)--(.1,6);
\draw[help lines,dashed] (0,6)--(6,6)--(6,0);
\draw (-2,0) node{Microscopic};
\draw (-2,6) node{Macroscopic};
\draw (0,-1) node{Shorter time};
\draw (6,-1) node{Longer time};
\draw (1.5,4) node{$ \sim n^{(1-\alpha)}$} ;
\draw (1.5,3.5) node{classical CLT} ;
\draw (4,1) node{$ \sim 1$} ;
\draw (4,0.5) node{GUE CLT} ;
\draw (4,3.5)   node [rotate=37] {$ \sim n^{1-\alpha+(2p+1)(\gamma-\alpha)}$};
\draw (4,3) node  [rotate=37]  {$\mu_p(f)$ CLT} ;

\end{tikzpicture}
\caption{Diagram representing the different regimes for random initial points. In the regimes we have plotted the magnitude of the variance of  $Y_n(f)-\EE _{\xi^{(n)}} \EE _{K_{n}} Y_n(f)$ and also the type of CLT that holds corresponding to Theorems \ref{th:random1} and \ref{th:random2}.}
\label{fig:phase2}
\end{center}
\end{figure}
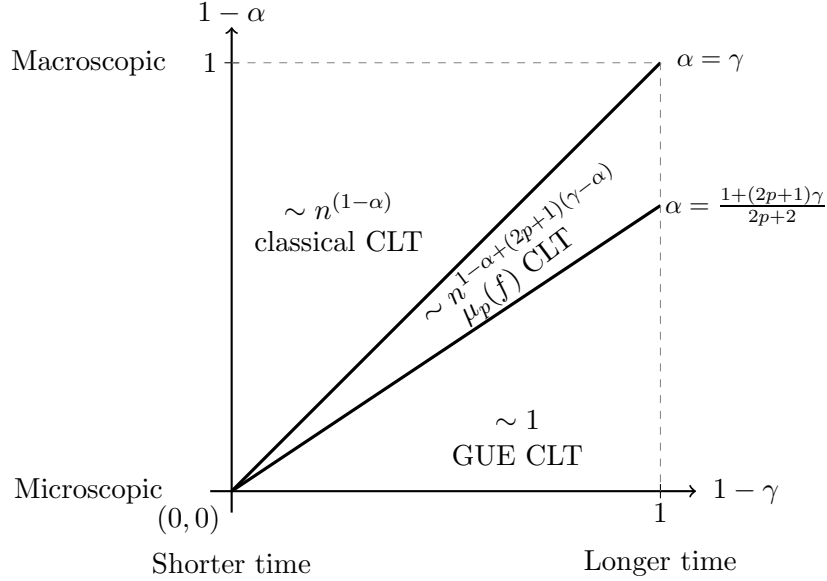

Before we come to the next results, we introduce some more notation. For $k\in \N$ the $k$-th moment of $f$ is denoted by $\mu_k(f)$, i.e.
\begin{align}\nonumber
\mu_k(f)= \int _\R x^k f(x) {\rm d} x.
\end{align}
We also recall the definition of $\sigma_\infty(f)^2$ in \eqref{eq:defsigmafinfty}. 
In the following result we show how the linear statistic $Y_n(f)$ behaves in the regime  $\alpha>\gamma$. 
\begin{theorem}\label{th:random2}
Let $f\in C_c^\infty(\R)$ and  $p\in \N$ be such that $\mu_0(f)=\ldots=\mu_{p-1}(f)=0$ and $\mu_p(f)\neq 0$. Set \begin{equation}\label{eq:defSp}
S_p(f)=\frac{\sqrt{2-x_*^2}(2p)! \mu_p(f)^2}{  2 \pi^2  (p!)^2 4^{2p} \tau ^{2p+1} }
\end{equation}
Then we have the following Central Limit Theorems :\begin{enumerate}
\item If $\gamma<\alpha < \frac{1}{2p+2} \left((2p+1)\gamma+1\right)$, then, as $n\to \infty$,
\begin{equation}\nonumber
\frac{1}{n^{\left(1-\alpha+(2p+1)(\gamma-\alpha)\right)/2}} \left(Y_n(f)-\EE_{\xi^{(n)}} \EE_{K_{n}} Y_n(f) \right)\to 
N(0,S_p(f)),  
\end{equation}
in distribution. 
\item  If $\alpha =  \frac{1}{2p+2} \left((2p+1)\gamma+1\right)$, then, as $n\to \infty$,
\begin{equation}\nonumber
Y_n(f)-\EE_{\xi^{(n)}} \EE_{K_{n}} Y_n(f) \to 
N(0,S_p(f)+\sigma_\infty(f)^2),  
\end{equation}
in distribution. 
\item  If $1>\alpha > \frac{1}{2p+2} \left((2p+1)\gamma+1\right),$ then, as $n\to \infty$,
\begin{equation}\nonumber
Y_n(f)-\EE_{\xi} \EE_{K_{n}} Y_n(f) \to 
N(0,\sigma_\infty(f)^2(t)),  
\end{equation}
in distribution. 
\end{enumerate}
\end{theorem}
In  Figure \ref{fig:phase2} we plotted a $\alpha\gamma$-diagram summarizing the different results in Theorem \ref{th:random1} and \ref{th:random2}. We see that for shorter times scale the leading order term  in  the linear statistic comes form the independently chosen initial points. Then at the line $\alpha=\gamma$ we find a transition due to the   regularizing effect from the Dyson's Brownian Motion and the order of the variance starts to decrease. Initially, the order of magnitude depends on the first vanishing moment of $f$ and decreases linearly with $\gamma$, until we reach final transition where the randomness of the Dyson's Brownian Motion starts to take over.

The fact that in the intermediate regime we have a Central Limit Theorem where only the lowest non-vanishing moment of $f$ play a role is somewhat surprising, and we do not know of a convincing intuitive explanation. However, it is easy to verify that it fits with the transition on the line $\alpha=\gamma$. We denote the Fourier transform of $f$ by $\hat f$, see \eqref{eq:ff1}.
Then, with $p$ such that $\mu_0(f)=\ldots=\mu_{p-1}(f)=0$ and $\mu_p(f) \neq 0$, we have 
\begin{align*}
\frac{\sqrt {2-x_*^2}}{\pi} \|P_\tau f\|_{\mathbb L_2(\R)}^2&= \frac{\sqrt {2-x_*^2}}{\pi}\int_{-\infty}^\infty |\hat f(\omega)|^2 {\rm e}^{-2 \tau |\omega|}{\rm d} {\omega}
\\
&=  \frac{\sqrt {2-x_*^2}}{\pi}\frac{1}{\tau} \int_{-\infty}^\infty |\hat f(\omega/\tau )|^2 {\rm e}^{-2 |\omega|}{\rm d} {\omega}\\
&= \frac{\sqrt {2-x_*^2}}{\pi}\frac{1}{\tau}  \int_{-\infty}^\infty \left|\sum_{j=0}^\infty(-{\rm i})^j \hat f^{(j)} (0) \frac{\omega^j}{ j! \tau^j} \right|^2 {\rm e}^{-2 |\omega|}{\rm d} {\omega}\\
&= \frac{\sqrt {2-x_*^2}}{\pi} \frac{|\hat f^{(p)} (0) |^{2}}{(p!)^2\tau^{2p+1}} \int_{-\infty}^\infty |\omega|^{2p} {\rm e}^{-2 |\omega|}{\rm d} {\omega}(1+\mathcal O (1/\tau))\\
&= \frac{\sqrt {2-x_*^2}}{\pi}\frac{|\hat f^{(p)} (0) |^{2}}{\tau^{2p+1}} \frac{(2p)!}{4^p(p!)^2} (1+\mathcal O (1/\tau)),
\end{align*}
as $\tau \to \infty$. By replacing $\tau$ with $\tau n^{\alpha-\gamma}$, and observing $\mu_j(f)= \sqrt{2 \pi} (-{\rm i})^j \hat f^{(j)}(0)$,  we retrieve the variance in the intermediate regime.  Finally,  one may argue that   a function for which all the low moments vanish, depends stronger on the higher frequencies than on the lower and hence Theorem \ref{th:random2} tells us that, in some sense, the higher frequencies  decay  faster than the lower frequencies.

\subsection{Further remarks}
\begin{remark}
In this paper, we always assume that $\gamma>0$ which means we only look at short time scales. The main reason for this  is that in some of the technical parts of the proofs certain aspect simplify for these values of $\gamma$. Nevertheless, we strongly believe that the arguments in the present paper can be extended for the case $\gamma \leq 0$ which represent longer time scales. In fact, it would be interesting to investigate how the lines in Figure \ref{fig:phase2}, that separate the different behaviors in the case of random in initial points, continue to the right of $1-\gamma=1$.    On the macrosopic scale $\alpha=0$ the transition has been  discussed in  \cite{Bender,Ca,Israelsson}.  For the mesososcopic scales we leave this as an open problem.
\end{remark}
\begin{remark} Since the questions that we answer in the present paper can also be posed for  general $\beta$-Dysons Brownian Motion, the natural question arises, if and how the results of the paper extend to the situation $\beta \neq 2$. It is reasonable to expect that similar phenomena occur, including the diagrams in Figure \ref{fig:phase} and \ref{fig:phase2}.  But the precise statement may involve new parameters depending on the precise value of $\beta$. It should be noted from a technical perspective, for $\beta \neq 2$ we loose the property that the $x_j(t)$ form a determinantal point point process. This property is heavily used in part of our arguments (see Sections 3, 5 and 6) and hence the proof for the general case requires new  ideas. On the other hand, for  $\beta=1,4$ we also  have  loop equations and we believe that this part of the analysis (see Sections 4, 7, 8 and 9)  can be extended to the cases $\beta=1,4$.
\end{remark}
\begin{remark}
In the same spirit as the previous remark, it is interesting to consider the model of random initial points and $\beta=\infty$. In that case, the $x_j(t)$ will freeze as $t \to \infty$ and converge to the zeros of the Hermite polynomial of degree $n$. In this case, the trajectories of the particles are deterministic and we capture only the regularizing effect of Dyson's Brownian Motion. It is reasonable to expect that we would have the same dependence on the highest non-vanishing moment of the function $f$. 
\end{remark}

 \begin{remark}\label{rem:proper}
By Proposition \ref{prop:localconcenpre}, the moment-generating function for the centered linear statistic  (and hence the moments)  are continuous with respect to $|\cdot|_{\mathcal L_w}$ with constant that are uniform in $n$. This allows us to extend  Theorems \ref{th:variance} and \ref{th:CLTfixed} for a more general class of functions $f$. We define the space $\overline{\mathcal L}_{w,c}$ as the closure $\overline{\mathcal L}_{w,c}={\rm cl} (C_c^1(\R)),$ where the closure is taken with respect to the norm $|f|_{\mathcal L_w}$.  Then, as a consequence to Proposition \ref{prop:localconcenpre},  we see that Theorems \ref{th:variance} and \ref{th:CLTfixed} also hold with $f\in C^1_c(\R)$ replaced by $f\in \overline{\mathcal L}_{w,c}$. However, we believe that this extension is still not optimal and that it is possible to prove the statement for more general classes of functions. We also wish to stress that $\overline{\mathcal L}_{w,c}$ is a proper subspace of $\mathcal L_w$. To see this,  let us consider the function $f$ defined by $f(x)=\Im (z-x)^{-1}$. Then it is not hard to check that $f\in \mathcal L_w$. However, $f\notin \overline{\mathcal L}_{w,c}$. Indeed, for any function $g\in C_c^1(\R)$, we have by compactness of support that 
\begin{equation*}
\begin{split}
|f-g|_{\mathcal L_w} \geq \sup_x \lim_{y\to \infty}   \sqrt{1+x^2}\sqrt{1+y^2}\left|\frac{f(x)-g(x)-(f(y)-g(y)}{x-y}\right|\\
=\sup_x    \sqrt{1+x^2}\left|{f(x)-g(x)}\right|\geq \lim_{x\to \infty} \sqrt{1+x^2}\left|{f(x)-g(x)}\right|=1.
\end{split}
\end{equation*}
Hence $f$ can never be approximated by $C^1_c(\R)$ functions in $\mathcal L_w$. The reason for this, is that $f$ has too fat tails. Indeed, the same argument shows that for every $h\in \overline{\mathcal L}_{w,c}$ we must have $\lim_{x\to \infty} x h(x)=0$.
\end{remark}

\begin{remark}
Our results on the asymptotic behavior of the $K_n$ in Section 5 can also be used to prove sine universality on the microscopic scale for any $0<\gamma<1$ which is part of Dyson's conjecture.
\end{remark}
\subsection{Overview of the rest of the paper}

The rest of the paper is organized as follows. Sections 3--7 are devoted to the proofs of Theorems \ref{th:variance}, \ref{th:CLTfixed} and Proposition \ref{prop:localconcenpre} , dealing with deterministic initial points. Then we prove  Theorems \ref{th:random1} and \ref{th:random2} on random initial points in Sections 8 and 9 respectively.  Sections 3 and 4 are devoted to the proofs of Theorems \ref{th:variance} and \ref{th:CLTfixed} respectively, but we will postpone some of the more technical arguments. More precisely, the asymptotic analysis for the kernel $K_n$, which we will need in Section 3, will be done thoroughly in Section 5. Similarly, the more technical arguments in the loop equations are no included in Section 4 but postponed to Section 7. Finally, Proposition \ref{prop:localconcenpre} is proved in Section 6, which also depends on some of the results in Section 5.  This concludes the arguments for the case of deterministic initial points. As for the case of random initial point, we prove Theorem \ref{th:random1} in Section 8. That section also contains a proof of Theorem \eqref{th:random2} in the special case that $p=0$. The case $p \neq 0$ requires  subtle administrative work. We will do this separately in Section 9. Finally, we included an Appendix with a short proof of \eqref{eq:defKn}. 
\section{Proof of Theorem \ref{th:variance} }\label{sec:variance}

In this section we prove Theorem \ref{th:variance}.  The proof of some of the results that we will need, require some work and we will postpone them to later sections for clarity reason. Throughout this section, we will always consider the linear statistic $Y_n(f)$ with $f\in C_c^1(\R)$ and parameters in \eqref{eq:para1}--\eqref{eq:para3}. 

\subsection{Determinantal strucure }

For fixed initial points $\xi_1^{(n)},\ldots \xi_n^{(n)}$,  the eigenvalues  $x_1(t),\ldots,x_n(t)$ of the random matrix $M_n(t)$ in \eqref{eq:interpolatingmodel} form a determinantal point process with kernel $K_n$ as in \eqref{eq:defKn}. 
This means that the $k$-point correlation function $\rho_k$ for the point process is given by 
\begin{multline}\label{eq:correlations}
\rho_k(x_1,\ldots,x_k)=\frac{1}{(n-k)!} \int \cdots \int\det\left(K_n(x_i,x_j;t)\right)_{i,j=1}^n {\rm d} x_{k+1} \cdots {\rm d} x_n\\=
\det  \left(K_n(x_i,x_j;t)\right)_{i,j=1}^k.
\end{multline}
The proof of this fact and the double integral representation \eqref{eq:defKn}  for the kernel has been computed before \cite{J1}, but for completeness we will provide a proof in the Appendix.  There we will also indicate that   $K_n$  satisfies the reproducing property
\begin{align} \label{eq:reproducing}
K_n(x,y;t)=\int_\R K_n(x,z;t) K_n(z,y;t) {\rm d} z,
\end{align}
This can be directly verified by using \eqref{eq:defKn}  but it is also  a particular consequence of the fact that our determinantal point process is a biorthogonal ensemble \cite{Bor}. For more details on general determinantal point process we refer to the discussions in \cite{BorDet,J4,K,L,Sosh} and the general reference on random matrix theory \cite{AGZ,F}. 

Because of the determinantal structure, we have the following useful identities for the expectation and variance of a linear statistic 
\begin{align} \label{eq:variancedeterminantalpointprocess}
\EE \sum_{j=1}^n g(x_j(t))  &=\int g(x) K_n(x,x;t) {\rm d} x\\
\Var \sum_{j=1}^n g(x_j(t)) & = \int g(x)^2 K_n(x,x;t) {\rm d} x\\
& \nonumber \qquad  -\iint g(x) g(y) K_n(x,y;t)K_n(y,x;t) {\rm d} x {\rm d} y.
\end{align}
By using \eqref{eq:defKn}  it is standard that one can rewrite the variance as 
\begin{align} \label{eq:variancebio}
\Var \sum_{j=1}^n g(x_j(t)) & =\frac12  \iint \left(g(x)-g(y)\right)^2 K_n(x,y;t)K_n(y,x;t) {\rm d} x {\rm d} y\\
&\nonumber =\frac12  \iint \left(\frac{g(x)-g(y)}{x-y}\right)^2 (x-y)^2K_n(x,y;t)K_n(y,x;t) {\rm d} x {\rm d} y.
\end{align}

By the latter formula we see that in order to prove Theorem \ref{th:variance} it is sufficient to compute the asymptotic behavior of $(x-y) K_n(x,y;t)$. The downside of \eqref{eq:variancebio} is that we need the asymptotic behavior of $(x-y) K_n(x,y)$ for every $x,y\in \R$, even in case $g$ is a local function. Apart from this conceptual flaw, it is also problematic from a technical point of view, since the computation of the asymptotic behavior requires a significant effort. 

We will get around this issue  by using the following idea. Suppose $g$ is a function that has support inside an interval $I$. We  introduce the function $R^I_n$ defined by
\begin{align}\nonumber
R^I_n(x,y;t)=\int_I K_n(x,z;t) K_n (z,y;t) {\rm d} z-K_n(x,y;t). 
\end{align}
Then, if $g$ has support inside $I$, we have 
\begin{align}\label{eq:varianceginI}
\Var \sum_{j=1}^n g(x_j(t)) & =\frac12  \iint_{I \times I} \left(\frac{g(x)-g(y)}{x-y}\right)^2 (x-y)^2K_n (x,y;t)K_n(y,x;t) {\rm d} x {\rm d} y\\
&\qquad - \int_I g(x)^2 R_n^I(x,x;t)  {\rm d} x.
\end{align}
The benefit of this expression is that it now suffices to compute the asymptotic behavior of $(x-y) K_n(x,y;t)$ only for $x,y \in I$ and this is an easier task. Of course, we need to find the asymptotic behavior of $R^I_n(x,x;t)$ also, but this can be done using the same principles that we use for $(x-y) K_n(x,y;t)$.

\subsection{Asymptotic results for $K_n$ and $R_n^I$}
In order to describe the asymptotic asymptotic behavior  of $K_n(x,y)(x-y)$ we introduce the function $F_n$ which is analytic $\C\setminus (-\infty,\max_j \xi_j^{(n)}]$ by 
\begin{align}\label{eq:defFnpre}
F_n(w;x)=(w{\rm e}^{-t}-x)^2 +\frac{1-{\rm e}^{-2t}}{n}\sum_{j=1}^n \log \left (w-\xi_j^{(n)}\right),
\end{align} 
where the logarithmic $\log (w-\xi_j^{(n)})$ is chosen such that it is analytic in $\C\setminus (-\infty, \xi_j^{(n)}]$ and takes real values on $(\xi_j^{(n)},\infty)$.  Note $F_n(w;x)$ also depends on $t$, but to avoid cumbersome notation we will not indicate this in the notation. We can now rewrite \eqref{eq:defKn} as 
\begin{align}\label{eq:defKn2a}
K_n(x,y;t) =\frac{ n}{\sinh t (2 \pi {\rm i})^2} \oint_\Sigma {\rm d}z \int_\Gamma {\rm d}w  \frac{ {\rm e}^{\frac{n} {1-{\rm e}^{-2t} }F_n(w;x)}}{ {\rm e}^{\frac{n} {1-{\rm e}^{-2t}}F_n(z;y)}} \frac{1}{w-z}.
\end{align}
We compute the asymptotic behavior by a saddle point analysis on the latter expression.  In the saddle point method,  we look for a point $\Omega_n(x)$ in the upper half plane such that $F_n'(\Omega_n(x);x)=0$. By as standard steepest descent analysis we can show that the leading term in the asymptotic expansion come from the neighborhood of these points.  However,  it is not a priori clear that the saddle point $\Omega_n(x)$ exists and this requires a proof. Let us first define for $x\in (-\sqrt 2, \sqrt 2)$ the point
\begin{equation}\label{eq:limitsaddle}
\Omega(x) = x \cosh t +{\rm i} \sqrt{2-x^2} \sinh t.
\end{equation}
Then the following lemma shows that if the sequence $(\xi^{(n)})_n$ is sufficiently regular in a neighborhood of an  interval $I$, then for large enough $n$ we have that  for each $x\in I$ there exists an  $\Omega_n(x)$ close to $\Omega(x)$. We recall that we work with the choice of parameters in \eqref{eq:para1}--\eqref{eq:para3}.
\begin{lemma}\label{lem:saddlepre}  Let $I\subset U$ be compact. Then for $n$ sufficiently large we have that for   $x\in I$  and $\xi \in \mathcal C(U,A,\delta)$  there exists a unique $\Omega_n(x;t)$  in the ball 
\begin{align}\label{eq:approxballpre}\Omega(x)+ \left(\tfrac{1-{\rm e}^{-2t}}{n}\right)^{\tfrac12-\delta} B_{0,1}\end{align}
such that  $F_n'(\Omega_n(x);x)=0$.
\end{lemma}
The proof of this lemma will be postponed to Section 3. There we will also derive some additional properties that will be useful. In particular, under the same conditions as Lemma \eqref{lem:saddlepre} we have  \begin{align}\label{eq:Omegalipschitzpre}
\Omega_n(x)-\Omega_n(y)=(x-y) (1+ o(1)),\end{align}
uniformly for $x,y\in I$ as $n \to \infty$ (cf. Lemma \ref{lem:Omegalipschitz}).

With the existence established we can now formulate the asymptotic behavior of $(x-y)K_n(x,y;t)$ and $R^I_n(x,x;t)$, with the choice of parameters in  \eqref{eq:para1} and \eqref{eq:para3}. 
\begin{lemma}\label{lem:asymKnpre}
 Let $I\subset U$ be compact. Then,  as $n\to \infty$, 
 \begin{multline}\label{eq:asymptoticsKnmainpre}
 (x-y) K_n(x,y;t)=-\frac{x-y}{2\pi {\rm i}} \left(\frac{{\rm e}^{\frac{n}{1-{\rm e}^{-2t} }\left( F_n(\Omega_n(x);x)- F_n(\Omega_n(y);y)\right) }}{\Omega_n(x)-\Omega_n(y)}\right.\\
 \left.+\frac{{\rm e}^{\frac{n}{1-{\rm e}^{-2t}}\left( F_n(\overline{\Omega_n(x)};x)- F_n(\Omega_n(y);y)\right) }}{\overline{\Omega_n(x)}-\Omega_n(y)}-\frac{ {\rm e}^{\frac{n}{1-{\rm e}^{-2t}}\left( F_n(\Omega_n(x);x)- F_n(\overline{\Omega_n(y);y)}\right) }}{\Omega_n(x)-\overline{\Omega_n(y)}}
\right.\\
 \left.-\frac{ {\rm e}^{\frac{n}{1-{\rm e}^{-2t}}\left( F_n(\overline{\Omega_n(x)};x)- F_n(\overline{\Omega_n(y)};y)\right) }}{\overline{\Omega_n(x)}-\overline{\Omega_n(y)}}
\right)\left(1+o(1)\right),
 \end{multline}
uniformly for $x,y \in I$ and $\xi \in \mathcal C(U,A,\delta)$. 
\end{lemma}
\begin{lemma} \label{lem:presesultonR}  Let $I\subset U$ be compact. Then there exists an $r>0$ such that  $|R ^I_n(x,x;t)|<r$  for $n$ sufficiently large, $x\in I$ and $\xi \in \mathcal C(U,A,\delta)$. 
\end{lemma}
The proofs of these Lemma's will be postponed and can be found in Section 3, together with further asymptotic results on $K_n$ and $R^I_n$. We are now ready to prove Theorem \ref{th:variance}. 
\subsection{Proof of Theorem \ref{th:variance}}

\begin{proof}[Proof of Theorem \ref{th:variance}]
The starting point is that by \eqref{eq:varianceginI} and with $g(x)=f(n^\alpha (x-x_*))$ we have
\begin{multline*}
\Var Y_n(f)= \frac{1}{2} \iint_{I \times I} \left(f(n^\alpha (x-x_*)) -f(n^\alpha (y-x_*))\right)^2 K_n (x,y;t) K_n(y,x;t) {\rm d} x {\rm dÊ} y\\-\int_I f(n^\alpha(x-x_*))^2 R_n^I(x,x;t) {\rm d}x.
\end{multline*}
We rescale the variables as
\begin{align}\label{eq:xytouv}
\begin{cases}
x=x_*+n^{-\alpha} u,\\
y=x_*+n^{-\alpha}v.
\end{cases}
\end{align}
Then we can rewrite the variance as
\begin{multline}\label{eq:varianceb}
\Var Y_n(f)= \frac{1}{2} \iint_{I\times I} \left(\frac{f(u) -f(v)}{u-v}\right)^2 (x-y)^2K_n (x,y) K_n(y,x;t) {\rm d} u {\rm dÊ} v\\-n^{-\alpha} \int_I  f(u)^2 R_n^I(x,x;t) {\rm d}u,
\end{multline}
with $x,y$ as in \eqref{eq:xytouv}. 
By Lemma \ref{lem:presesultonR} we see that the integral concerning $R^I_n$ is of order $\mathcal O(n^{-\alpha})$ and hence it can be ignored. We compute the variance by inserting the leading order terms of \eqref{eq:asymptoticsKnmainpre} into the first double integral on the right-hand side of \eqref{eq:varianceb}.  First, let us note that because of $F'_n(\Omega_n(x);x)=0=F'_n(\Omega_n(y);y)$ we have 
\begin{align*}
\begin{split}
\frac{n}{1-{\rm e}^{-2t}} \frac{{\rm d}}{{\rm d}u } \Im F_n(\Omega_n(x);x)&=\frac{2 n^{1-\alpha}}{1-{\rm e}^{-2t}} \Im \Omega_n(x) = (\sqrt{2-{x_*^2}}+o(1)) n^{1-\alpha} ,\\
\frac{n}{1-{\rm e}^{-2t}} \frac{{\rm d}}{{\rm d}v } \Im F_n(\Omega_n(y))&=\frac{2  n^{1-\alpha}}{1-{\rm e}^{-2t}} \Im \Omega_n(y)= (\sqrt{2-{x_*^2}} +o(1))n^{1-\alpha},
\end{split}
 \end{align*}
as $n\to \infty$ uniformly for $x,y\in I$.  Since also $0<\alpha<1$, we see when computing plugging the asymptotics of \eqref{eq:asymptoticsKnmainpre} into \eqref{eq:varianceb}   that the terms that still contain $${\rm e}^{\pm \frac{n}{1-{\rm e}^{-2t}} \Im F_n(\Omega_n(x))}, \qquad \text{      and        } \qquad  {\rm e}^{\pm \frac{n}{1-{\rm e}^{-2t}} \Im F_n(\Omega_n(y))}$$  are highly oscillating.  Therefore, we write
 \begin{multline}\label{eq:longproduct}
 (x-y)^2 K_n(x,y;t) K_n(y,x;t)\\
=\frac{1}{4\pi^2}
\left(\frac{x-y}{\Omega(x)-\Omega(y)}\right)^2
+\frac{1}{4\pi^2}
\left(\frac{x-y}{\overline{\Omega(x)}-\overline{\Omega(y)}}\right)^2
\\-\frac{1}{4\pi^2} 
\left(\frac{x-y}{{\overline{\Omega(x)}-\Omega(y)}}\right)^2
-\frac{1}{4\pi^2}
\left(\frac{x-y}{\Omega(x)-\overline{\Omega(y)}}\right)^2\\
+\text{ highly osscillating terms }+o(1)
\end{multline}
as $n\to \infty$. Now note that, by the Riemann-Lebesgue lemma  we can discard all highly oscillating terms in the integral for the variance and we can write 
\begin{multline}\label{eq:varianced}
\Var Y_n(f)= \frac{1}{2} \iint_{I\times I} \left(\frac{f(u) -f(v)}{u-v}\right)^2  \left(\frac{1}{4\pi^2}
\left(\frac{x-y}{\Omega(x)-\Omega(y)}\right)^2
\right.\\\left.
+\frac{1}{4\pi^2}
\left(\frac{x-y}{\overline{\Omega(x)}-\overline{\Omega(y)}}\right)^2
-\frac{1}{4\pi^2} 
\left(\frac{x-y}{{\overline{\Omega(x)}-\Omega(y)}}\right)^2
\right.\\\left.-\frac{1}{4\pi^2}
\left(\frac{x-y}{\Omega(x)-\overline{\Omega(y)}}\right)^2 \right){\rm d} u {\rm dÊ} v +o(1),
\end{multline}
as $n\to \infty$. 
 Hence it remains, to simplify the four terms.  We will treat the the cases $\alpha>\gamma$, $/\alpha <\gamma$ and $\alpha=\gamma$ differently.

Let us first consider the case $\alpha>\gamma$. First note that by  \eqref{eq:Omegalipschitzpre} and \eqref{eq:xytouv}  we have
\begin{align}\label{eq:varianceAA1}
\begin{split}
\frac{x-y}{\Omega_n(x)-\Omega_n(y)}=1+o(1),\qquad 
\frac{x-y}{\overline{\Omega_n(x)}-\overline{\Omega_n(y)}}=1+o(1),
\end{split}
\end{align}
as $n\to \infty$.
Moreover, by writing
$$
\frac{x-y}{\overline{\Omega_n(x)}-\Omega_n(y)}=\frac{x-y}{\Omega_n(x)-\Omega_n(y)-2{\rm i}\Im \Omega_n(x)}
$$
and by  $\alpha>\gamma$, Lemma \ref{lem:saddlepre} and \eqref{eq:limitsaddle}  we have that $\Im \Omega_n(x)>> x-y$ and 
\begin{align}\label{eq:varianceAA2}
\frac{x-y}{\overline{\Omega_n(x)}-\Omega_n(y)}=\mathcal O(n^{\gamma-\alpha}), \qquad n\to \infty.
\end{align}
By inserting \eqref{eq:varianceAA1} and \eqref{eq:varianceAA2} (and its conjugate) into \eqref{eq:varianced}  we arrive at  the statement for $\alpha>\gamma$. 

 Now let us treat the case $\alpha=\gamma$. In that case,we still have \eqref{eq:varianceAA1} but instead of \eqref{eq:varianceAA2} we have
 $$
\frac{x-y}{\overline{\Omega_n(x)}-\Omega_n(y)}=\frac{x-y}{\Omega_n(x)-\Omega_n(y)-2{\rm i}\Im \Omega_n(x)}=\frac{u-v}{u-v-2{\rm i} \tau }(1+o(1)).
$$
 and 
 $$
\frac{x-y}{\Omega_n(x)-\overline{\Omega_n(y))}}=\frac{x-y}{\Omega_n(x)-\Omega_n(y)+2{\rm i}\Im \Omega_n(y)}=\frac{u-v}{u-v+2{\rm i} \tau }(1+o(1)).
$$
as $n\to \infty$. The statement follows after realizing that 
\begin{align}
2-\left(\frac{u-v}{u-v+2{\rm i} \tau }\right)^2-\left(\frac{u-v}{u-v-2 {\rm i} \tau }\right)^2=\frac{8\tau^2}{(u-v)^2+4\tau^2}
\end{align}
and inserting this into \eqref{eq:varianceAA2}.

The remaining case $\alpha <\gamma$ can be done in the same way as the case $\alpha=\gamma$ and is left to the reader. It can also be obtain from the case $\alpha=\gamma$, by taking $\tau\to 0$. 
   \end{proof}

\section{Proof of Theorem \ref{th:CLTfixed}}\label{sec:loopeqn}

 In this section we prove Theorem \ref{th:CLTfixed}.  We will always assume \eqref{eq:para1}--\eqref{eq:para3} and  $\xi \in \mathcal C(U,A,\delta)$ where we recall \eqref{eq:defC}. In particular, the estimates and constants in the order terms that we derive hold uniformly  for $\xi \in \mathcal C(U,A,\delta)$.
 
 \subsection{Overview of the proof}

Fix $f\in C_c^1(\R)$, a continuously differentiable function with compact support. We prove Theorem \ref{th:CLTfixed} as follows. First we smoothen $f$ and consider $f^\eps;= P_\eps*f$ for $0<\eps<1$. Here $ P_\eps(x)= \pi^{-1} \frac{\eps}{x^2+\eps^2}$ (see also \eqref{eq:defPeps}) and hence
\begin{equation}\label{eq:deffeps}
f^\eps(x)= \int_\R  f(y) P_\eps(x-y) {\rm d} y= \frac{1}{\pi} \int  f(y) \Im  \frac{1}{y-x-{\rm i} \eps} {\rm d} y.
 \end{equation}
 We then prove that, for some $\sigma_{n,\eps}$, we have
 \begin{equation}\label{eq:CLTsuff}
\frac{\partial }{\partial \lambda } \log \EE \left[\exp \lambda \left( Y_n(f^\eps)- \EE Y_n(f^\eps)\right)  \right]-\sigma^2_{n,\eps} \lambda=o(1),
\end{equation}
as $n \to \infty$, 
 uniformly for $\lambda$ in a sufficiently small neighborhood of the origin and $0<\eps<1$.  An important ingredient in the proof of \eqref{eq:CLTsuff} are the so-called loop equations.   Moreover,   Proposition \ref{prop:localconcenpre} will be used as an important input for these equations.   After we have established \eqref{eq:CLTsuff}, we continue the proof of Theorem \ref{th:CLTfixed}  by  showing that $|f^\eps-f |_{\mathcal L_w} \to 0$. We can then use Proposition \ref{prop:localconcenpre} and an approximation argument to show that we have \eqref{eq:CLTsuff} also for $f$. Since we already computed the limiting value of the variance for any $f\in \C^1_c(\R)$ in Theorem \ref{th:variance}, this finishes the proof Theorem \ref{th:CLTfixed}.

To show the benefit of working with $f^\eps$ we first introduce some notation.  We define	
 \begin{equation}\label{eq:defEh}\EE^{h}[\cdot]=\frac{\EE\left[[\cdot] \exp  h(M)\right]}{\EE[\exp h(M)]},\end{equation}
for any functional $h$ on the space of hermitian matrices.  Note that for $h=0$  we just obtain $\EE^0=\EE$.   In this section we  choose to write $\EE^0$ to emphasize the difference with $\EE^h$. In case of a linear statistic $h(M)=\Tr f(M)$,  we will also write $\EE^f:=\EE^h$.  

Now for $0<\alpha<1$ and any function $g$ we define
$$g_\alpha(x)=g(n^\alpha (x-x^*)).$$ Then
\begin{multline}\label{eq:functionsT} 
f_\alpha^\eps(M)=f^\eps(n^{\alpha} (M-x^*))\\=  \frac{1}{2\pi {\rm i}} \int f(s) \left(\frac{1}{x-{\rm i}\eps -n^\alpha (M-x_*)}-\frac{1}{x+{\rm i}\eps -n^{\alpha} (M-x_*)}\right) {\rm d} x\\
=\frac{1}{2 \pi{\rm i} n^\alpha} \int    \frac{ f(x)}{(x-{\rm i} \eps)/n^\alpha+x_*- M}{\rm d}x-\frac{1}{2\pi{\rm i} n^\alpha} \int \    \frac{f(x)}{(x+{\rm i} \eps)/n^\alpha+x_*- M}{\rm d}x.
\end{multline}
With these notations the left-hand side of \eqref{eq:CLTsuff} can be rewritten and we are left with proving 
 \begin{equation}\label{eq:goallinearity}
 \EE^{\lambda f_\alpha^\eps}[\Tr f^\eps_\alpha(M)]-\EE^0[\Tr f^\eps_\alpha(M)]- \sigma^2_{n,\eps} \lambda, \qquad \textrm{ as } n \to \infty,\end{equation}
uniformly for $\lambda$ in some neighborhood of the origin. In view of \eqref{eq:functionsT}  we will start by analyzing
 \begin{align}\label{eq:startloopeqn}
 \frac{1}{n^\alpha} \left(\EE^{\lambda f_\alpha^\eps}\left[\Tr \frac{1}{z/n^{\alpha}+x_*-M}\right] -\EE^0\left[\Tr \frac{1}{z/n^{\alpha}+x_*-M}\right] \right),
 \end{align}
 asymptotically as $n\to \infty$ for $z$ in compact subsets of $\C \setminus \R$. 

\subsection{The loop equations}

The starting point of the analysis, is the following equation. 

\begin{lemma}  Let $h$ be bounded and differentiable function and $J$ a constant $n\times n$ matrix. Then the following equation holds
	\begin{multline}\label{eq:loopeqn}
		\EE^{\lambda h} \left[\Tr  \frac{1}{z-M} J \right]= \Tr \frac{1}{\zeta^{\lambda h}(z)-q\Xi_n} J\\
			+ \lambda \tfrac{1-q^2}{2n}\EE^{\lambda h} \left[\Tr \frac{h'(M)}{z-M} J \frac{1}{\zeta^{\lambda h}(z)-q \Xi_n}\right]
			+ \tfrac{1-q^2}{ 2n}K^{\lambda h,J}(z),	\end{multline}		
			where
		\begin{align}
		\zeta^{\lambda h}(z)&=z-\tfrac{1-q^2}{2n}\EE^{\lambda h} \left[\Tr \frac{1}{z-M}  \right], \label{eq:defzeta}\\
K_n^{\lambda h,J}(z)&=\EE^{\lambda h} \left[\Tr \frac{1}{z-M}\Tr \frac{1}{z-M} J \frac{1}{\zeta^{\lambda h}(z)-q \Xi_n}\right] 
\nonumber \\
&\qquad -\EE^{\lambda h} \left[\Tr \frac{1}{z-M}\right]\EE^{\lambda h }\left[\Tr \frac{1}{z-M}J \frac{1}{\zeta^{\lambda h}(z)-q \Xi_n}\right], \label{eq:defK} 
\end{align}
and $z\in \C\setminus \R$, $\lambda \in \C$ and $n\in \N$.
\end{lemma}
\begin{proof}
The proof is standard  and follows by performing the following change of variable 
$$M\mapsto M+\frac{s}{z-M} J\frac{1}{\zeta -q Y},$$
in the matrix integral
$$\int {\rm e}^{-\frac{n}{1-q^2} (M-q \Xi_n)^2 +\Tr h(M)} {\rm d} M.$$
For the moment,  $\zeta$  is still arbitrary in $\C \setminus \R$ when performing this change of variables. Naturally, the value of the integral does not change.  Hence  by taking the derivative with respect to $s$ (after the change of variable) and setting $s=0$ we obtain the following 
\begin{multline}\label{eq:loopeqnA}
-\frac{2 n}{1-q^2}\EE^{\lambda h} \left[ \Tr (M-q\Xi_n) \frac{1}{z-M}J\frac{1}{\zeta-q\Xi_n}\right] \\+\lambda \EE^{\lambda h}\left[\Tr h'(M) \frac{1}{z-M} J \frac{1}{\zeta-q \Xi_n}\right]\\
+ \EE^{\lambda h} \left[\Tr \frac{1}{z-M} \Tr \frac{1}{z-M} J\frac{1}{\zeta -q \Xi_n}\right]=0,
\end{multline}
where the last term comes from the Jacobian for the change of variables. By some simple calculus we see
\begin{multline*}
\Tr (M-q\Xi_n) \frac{1}{z-M}J\frac{1}{\zeta-q\Xi_n}\\
= \Tr\frac{1}{z-M} J -\Tr  J \frac{1}{\zeta-q Y}-(\zeta-z) \Tr \frac{1}{z-M} J \frac{1}{\zeta-q\Xi_n}
\end{multline*}
Now the statement follows by choosing $\zeta$ as in \eqref{eq:defzeta} and reorganizing \eqref{eq:loopeqnA}. 
 \end{proof}
We will always apply the last result with $h=f^\eps$  for $f^\eps=P_\eps*f$ and $f\in C_c^1(\R)$. In that case  $h'(M)={f^\eps}'(M)$ is defined as
 \begin{align} \label{eq:interphprimeM}
 {f^\eps}'(M)=\frac{1}{\pi} \int  f(s) \Im  \left(\frac{1}{s-{\rm i} \eps-M}\right)^2 {\rm d} s.
 \end{align}
 By applying \eqref{eq:loopeqn} twice, once with $\lambda=0$, and taking the difference we get an equation for the right-hand side of \eqref{eq:startloopeqn}. First, we define 
 \begin{align}
 \label{eq:defDn}
 D^{\lambda h}_n(z)=\EE^{\lambda h}\left[\Tr \frac{1}{z-M} \right]-\EE^0 \left[\frac{1}{z-M} \right].
 \end{align}
Then we have that $D_n^{\lambda h}$ is the solution to a quadratic equation.
\begin{corollary}
With $D_n^{\lambda h}$ as \eqref{eq:defDn} we have 
\begin{align}\label{eq:Dn}
A^{\lambda h}_n(z) D^{\lambda h}_n(z)^2+B^{\lambda h}_n(z) D_n^{\lambda h}(z)+C^{\lambda h}_n(z)=0,
\end{align}
where 
\begin{align}\label{eq:defABC}
\begin{split}
A_n^{\lambda h}(z)&=-\left(\frac{1-q^2}{2n}\right)^2 \Tr \frac{1}{\zeta^{\lambda h}(z)-q \Xi_n}\frac{1}{(\zeta^{0}(z)-q \Xi_n)^2} \\
B_n^{\lambda h}(z)&=1+\frac{1-q^2}{2n} \Tr \frac{1}{(\zeta^0(z)-q \Xi_n)^2}\\
&\qquad +\left(\frac{1-q^2}{2n}\right)^2  \EE^{\lambda h}\left[\Tr\frac{h'(M)}{z-M} \frac{1}{\zeta^0(z)-q \Xi_n} \frac{1}{\zeta^{\lambda h}(z)-q \Xi_n}\right]\\
C_n^{\lambda h}(z)&=-\lambda  \frac{1-q^2}{2n}  \EE^{\lambda h}\left[\Tr \frac{h'(M)}{z-M} \frac{1}{\zeta^0(z)-q \Xi_n}\right]
  -\frac{1-q^2}{2n}\left(K_n^{\lambda h}(z)-K^0(z) \right)
\end{split}
\end{align}
\end{corollary}
\begin{proof}
First note that by taking the difference of \eqref{eq:loopeqn} for general and $\lambda $ and \eqref{eq:loopeqn} with $\lambda =0$ and $J=I$,  we obtain 
	\begin{multline}\label{eq:Dna}
		\left(\EE^{\lambda h}\left[\Tr \frac{1}{z-M}\right]-\EE^0\left[\Tr \frac{1}{z-M}\right]\right)\\
		= \Tr \left( \frac{1}{\zeta^{\lambda h}(z)-q \Xi_n}\right)- \Tr \left(  \frac{1}{\zeta^{0}(z)-q \Xi_n}\right)\\
		+\lambda \tfrac{1-q^2}{2n} \EE^{\lambda h}\left[\Tr \frac{h'(M)}{z-M} \frac{1}{\zeta^{\lambda h}(z)-q \Xi_n}\right]
		+\tfrac{1-q^2}{2n}\left(K_n^{\lambda h}(z)-K^0(z) \right)
	\end{multline}
	for $z\in \C \setminus \R$, $\lambda\in\C$ and $n\in \N$. Then note that 
	\begin{multline*}
\Tr \frac{1}{\zeta^{\lambda h}(z)-q \Xi_n}-\Tr \frac{1}{\zeta^{0}(z)-q \Xi_n}
\\
=\left(\zeta^0(z)-\zeta^{\lambda h} (z)\right)\left(\Tr \frac{1}{(\zeta^{0}(z)-q \Xi_n)^2}+(\zeta^0(z)-\zeta^{\lambda h}(z)) \Tr  \frac{1}{(\zeta^{0}(z)-q \Xi_n)^2} \frac{1}{\zeta^{\lambda h}(z)-q \Xi_n}\right)
\end{multline*}
and that
\begin{equation*}
\zeta^{\lambda h}(z)-\zeta^0(z)=-\frac{1-q^2}{2n} \left(\EE^{\lambda h} \left[\Tr \frac{1}{z-M}\right]-\EE^{0} \left[\Tr \frac{1}{z-M}\right]\right)
=-\frac{1-q^2}{2n}  D_n^{\lambda h}(z). 
\end{equation*}
and hence 
\begin{multline}\label{eq:Dnb}
\Tr \frac{1}{\zeta^{\lambda h}(z)-q \Xi_n}-\Tr \frac{1}{\zeta^{0}(z)-q \Xi_n}
\\
=-\frac{1-q^2}{2n}  D_n^{\lambda h}(z)\left(\Tr \frac{1}{(\zeta^{0}(z)-q \Xi_n)^2} -\frac{1-q^2}{2n}  D_n^{\lambda h}(z)\Tr  \frac{1}{(\zeta^{0}(z)-q \Xi_n)^2} \frac{1}{\zeta^{\lambda h}(z)-q \Xi_n}\right).
\end{multline}
Moreover, 
\begin{equation}
\begin{split}\label{eq:Dnc}
\EE^{\lambda h}\left[\Tr \frac{h'(M)}{z-M} \frac{1}{\zeta^{\lambda h} (z)-q \Xi_n}\right]
-\EE^{\lambda h}\left[\Tr \frac{h'(M)}{z-M} \frac{1}{\zeta^0(z)-q \Xi_n}\right]\\=
(\zeta^0(z)-\zeta^{\lambda h}(z)) \EE^{\lambda h}\left[\Tr \frac{h'(M)}{z-M} \frac{1}{\zeta^0(z)-q \Xi_n}\frac{1}{\zeta^{\lambda h}(z)-q \Xi_n}\right]\\
=-\frac{1}{2} D_n^{\lambda h} (z)\EE^{\lambda h}\left[\Tr \frac{h'(M)}{z-M} \frac{1}{\zeta^0(z)-q \Xi_n}\frac{1}{\zeta^{\lambda h}(z)-q \Xi_n}\right]
\end{split}\end{equation}
By substituting \eqref{eq:Dnb} and \eqref{eq:Dnc}  back into \eqref{eq:Dna} and rearranging terms gives the statement.
\end{proof}

\subsection{Loop equations on the mesoscopic scale}
To use the loop equations for the mesocopic scales that we are interested  in we do the following. By replacing $z$ by $z/n^\alpha+x_*$, we have any $\alpha>0$ that
\begin{multline}\label{eq:Dnalpha}
n^{\alpha} A^{\lambda h_\alpha}_n(z/n^{\alpha}+x_*)\left(\frac{1}{n^\alpha} D^{\lambda h_\alpha}_n(z/n^{\alpha}+x_*)\right)^2\\+B^{\lambda h_\alpha}_n(z/n^{\alpha}+x_*) \frac{1}{n^\alpha} D_n^{\lambda h_\alpha}(z/n^{\alpha}+x_*)\\+\frac{1}{n^\alpha} C^{\lambda h_\alpha}_n(z/n^{\alpha}+x_*)=0,
\end{multline}
where $A_n^{\lambda h_\alpha}, B_n^{\lambda h_\alpha}$ and $C_n^{\lambda h_\alpha}$ as in \eqref{eq:defABC}.
 
By using the self-improving mechanism  behind the loop equation and using Proposition \ref{prop:localconcenpre} we have the following estimates.
\begin{lemma} \label{lem:estimateA}
Let $f\in C^1_c(\R)$. Then, as $n\to \infty$,
\begin{align}
n^\alpha A_n^{\lambda f_\alpha^\eps}(z/n^\alpha+x_*)&=\mathcal O(n^{\alpha-1}), \label{eq:lemestimateA}\\
B_n^{\lambda f_\alpha^\eps}(z/n^\alpha+x_*)&=1+\mathcal O(n^{-\gamma}),\label{eq:lemestimateB}\\
\label{eq:DNlinearalphahalf2}
\frac{1}{n^\alpha} C_n^{\lambda f_\alpha^\eps}(z/n^\alpha+x_*) &=\\
&\hspace*{-2cm}-\lambda \frac{1-q^2}{2n^{1+2 \alpha}}\EE^0\left[\Tr \frac{{f_\alpha^\eps}'(M)}{z/n^\alpha+x_*-M} \frac{1}{\zeta^0(z/n^\alpha+x_*)-q\Xi_n}\right]
\nonumber+\mathcal O(n^{ \alpha-1}),
\end{align}
 uniformly for $\lambda$  in a sufficiently small neighborhood of the origin, $z$ in compact subsets of $\C\setminus \R$, $0<\eps<1$ and $\xi\in C(U,A,\delta) $.
\end{lemma}

The proof of this lemma is postponed to Section \ref{sec:estimatesloopeq}. Combined with \eqref{eq:Dnalpha}  we find the following consequence which proves \eqref{eq:goallinearity}.
\begin{corollary}\label{cor:Dnalpha}
Let $f\in C_c^1(\R)$. We have, as $n\to \infty$, 
$$\frac{1}{n^\alpha} D_n^{\lambda f_\alpha^\eps} (z/n^\alpha+x_*)=\lambda \frac{1-q^2}{2n^{1+2 \alpha}}\EE^0\left[\Tr \frac{{f_\alpha^\eps}'(M)}{z/n^\alpha+x_*-M} \frac{1}{\zeta^0(z/n^\alpha+x_*)-q\Xi_n}\right]+o(1),$$
 uniformly for $\lambda$ in a sufficiently small neighborhood of the origin, $z$ in compact subsets of $\C\setminus \R$, $0<\eps<1$ and $\xi\in \mathcal C(U, A,\delta)$. 
\end{corollary}

\subsection{Proof of Theorem \ref{th:CLTfixed}}\label{eq:proofoflinear}
We now come to the proof of Theorem \ref{th:CLTfixed}. We will need the following approximation result. 
\begin{lemma}\label{lem:approxeps}
Let $f\in \C^1_c(\R)$, i.e a continuously differentiable function with compact support. Then $|f-P_\eps*f|_{\mathcal L_w} \to 0$ as $\eps\downarrow 0$. 
\end{lemma}
\begin{proof} It is standard that $P_\eps(s)=\frac{1}{\pi} \frac{\eps}{s^2+\eps^2}$ converges to the delta-functional at the origin as $\eps \downarrow 0$. Moreover, since $f$ is continuous and has compact support it follows that $\|f-P_\eps* f\|_\infty\to 0$ as $\eps \downarrow 0$.  

Now split $f^\eps=f_1^\eps +f_2^\eps$ where 
\begin{align}\nonumber
f_1^\eps(x) &=\frac{1}{\pi} \int_{|s|\leq 1 } f(x-s) \frac{\eps}{s^2+\eps^2} {\rm d}s,\\ f_2^\eps(x) &=\frac{1}{\pi} \int_{|s|> 1 } f(x-s) \frac{\eps}{s^2+\eps^2} {\rm d}s. 
\end{align}
The key idea behind the splitting is that $f_1^\eps$ has compact support and that, by standard arguments, we have $\|f-f_1^\eps\|_\infty\to 0$ as $\eps \downarrow 0$.  We claim this holds also when $\|\cdot\|_\infty$ is replaced by $|\cdot|_{\mathcal L_w}$.   To this end, let $I$ be a closed interval containing the origin and  the supports of $f$ and $f_1^\eps$ and split
\begin{align}\label{eq:assist1}
\begin{split}
|f-f_1^\eps|_{\mathcal L_w}= \sup_{x,y} \sqrt{1+x^2} \sqrt{1+y^2}\left|\frac{(f-f_1^\eps)(x)-(f-f_1^\eps)(y)}{x-y}\right|\\\leq
\sup_{x\in I,y\in 2 I} \sqrt{1+x^2} \sqrt{1+y^2}\left|\frac{(f-f_1^\eps)(x)-(f-f_1^\eps)(y)}{x-y}\right|\\+\sup_{x\in I,y \notin 2 I} \sqrt{1+x^2} \sqrt{1+y^2}\left|\frac{(f-f_1^\eps)(x)-(f-f_1^\eps)(y)}{x-y}\right|
 \end{split}
\end{align}
Then first we have
\begin{align}\label{eq:assist2}
\begin{split}
 \sup_{x\in I,y\notin 2I } \sqrt{1+x^2} \sqrt{1+y^2}\left|\frac{(f-f_1^\eps)(x)-(f-f_1^\eps)(y)}{x-y}\right|\\\leq
 2(\|f-f_1^\eps\|_\infty)\sup_{x \in I ,y\notin 2 I} \left|\frac{\sqrt{1+x^2} \sqrt{1+y^2}}{x-y}\right|=d \|f-f_1^\eps\|_\infty ,
 \end{split}
\end{align}
for some constant $d>0$. Moreover,
Since both $f$ and $f_1^\eps$ have compact support and are continuously differentiable, there exists a constant $c>0$ such that
\begin{align}\label{eq:assist3}
\begin{split}
 \sup_{x\in I,y\in 2I } \sqrt{1+x^2} \sqrt{1+y^2}\left|\frac{(f-f_1^\eps)(x)-(f-f_1^\eps)(y)}{x-y}\right|\\\leq
 c\sup_{x \in I ,y\in 2 I} \left|\frac{(f-f_1^\eps)(x)-(f-f_1^\eps)(y)}{x-y}\right|=c \|f'-{f_1^\eps}'\|_\infty.
 \end{split}
\end{align}  Note that we can change the order of integration and differentiation and obtain $\|f'-{f_1^\eps}' \|_\infty\to 0$ as $\eps \downarrow 0$ and hence, by inserting \eqref{eq:assist2} and \eqref{eq:assist3} into \eqref{eq:assist1}, we obtain $|f-f_1^\eps|_{\mathcal L_w}\to 0$ as $\eps \downarrow 0$.

By the triangular inequality, it remains to show that $|f_2^\eps|_{\mathcal L_w} \to 0$ as $\eps\downarrow 0$.  To this end, we note that
\begin{align}\nonumber
\begin{split}
|f_2^\eps|_{\mathcal L_w} = \sup_{x,y} \frac{\sqrt{1+x^2}\sqrt{1+y^2}}{\pi} \left|\int_{|s|\geq 1} \frac{f(x-s)-f(y-s)}{x-y} \frac{\eps}{s^2+\eps^2} {\rm d} s\right|\\
= \sup_{x,y} \frac{\sqrt{1+x^2}\sqrt{1+y^2}}{\pi} \left|\int_{|s|\geq 1: s\in (x-\supp f)\cup (y-\supp f)} \frac{f(x-s)-f(y-s)}{x-y} \frac{\eps}{s^2+\eps^2} {\rm d} s\right|\\
\leq |f|_{\mathcal L_w}\sup_{x,y} \frac{1}{\pi} \int_{|s|\geq 1: s\in (x-\supp f)\cup (y-\supp f)} \frac{\sqrt{1+x^2}\sqrt{1+y^2}}{\sqrt{1+(s-x)^2} \sqrt {1+(s-y)^2} } \frac{\eps}{s^2+\eps^2} {\rm d} s.
\end{split}
\end{align}
By combining the latter with the inequalities $(s^2+\eps^2)\geq (s^2+1)/2$ and
\begin{align}\nonumber
(1+(s-x)^2) (1+s^2) \geq \frac{1}{2} (1+x^2), \qquad \text{for } s,x\in \R,\end{align}
we obtain 
\begin{align}\nonumber
\begin{split}
|f_2^\eps|_{\mathcal L_w}  \leq \frac{8 |f|_{\mathcal L_w} |\supp(f)|\eps }{\pi}.
\end{split}
\end{align}
Hence $|f_2^\eps|_{\mathcal L_w} \to 0$ as $\eps \downarrow 0$ and this finishes the proof. \end{proof}
We are now ready for the proof of Theorem \ref{th:CLTfixed}.
\begin{proof}[Proof of Theorem \ref{th:CLTfixed}] 
By linearity of a linear statistic, we have 
\begin{multline}\nonumber
\EE^0 \left[\exp \lambda  \left(Y_n(f)-\EE Y_n(f)\right) \right]-\EE^0 \left[\exp \lambda  \left(Y_n({f^\eps})-\EE Y_n({f^\eps})\right) \right]\\
=\int_0 ^\lambda \EE^0 \Big[\left(Y_n({f-f^\eps})-\EE Y_n({f-f^\eps})\right) \exp \mu  \left(Y_n({f-f^\eps})-\EE^0 Y_n({f-f^\eps})\right)  \\ \times  \exp \lambda  \left(Y_n({f^\eps})-\EE^0 Y_n({f^\eps})\right) \Big]{\rm d}\mu.
\end{multline}
Hence, by taking the absolute value in the integral, taking the supremum and then using Cauchy-Schwarz twice we have 
\begin{multline}\nonumber
\left|\EE^0 \left[\exp \lambda  \left(Y_n(f)-\EE Y_n(f)\right) \right]-\EE^0 \left[\exp \lambda  \left(Y_n({f^\eps})-\EE^0 Y_n({f^\eps})\right) \right]\right|\leq \\
|\lambda| \left( \Var Y_n({f-f^\eps})\right)^{1/2}   
\sup_{t \in [0,\Re \lambda ]} \left|\EE^0 \left[ \exp 2 t  \left(Y_n({f-f^\eps})-\EE Y_n({f-f^\eps})\right) \right.\right.\\
\times \left.  \left.\exp 2\Re \lambda  \left(Y_n({f^\eps})-\EE^0 Y_n({f^\eps})\right) \right]\right|^{1/2}\\
\leq |\lambda| \left( \Var Y_n({f-f^\eps})\right)^{1/2}   
\sup_{t \in [0,\Re \lambda ]} \left|\EE^0 \left[ \exp 4 t  \left(Y_n({f-f^\eps})-\EE^0 Y_n({f-f^\eps})\right) \right]\right| ^{1/4}\\
\times \left|\EE^0 \left[\exp 4\Re \lambda  \left(Y_n({f^\eps})-\EE^0 Y_n({f^\eps})\right) \right]\right|^{1/4}.
\end{multline}
By combining this with Proposition \ref{prop:localconcenpre} and \eqref{eq:variancegenb} we see that there are constants $d_1,d_2>0$ such that for $n$ sufficiently large we have 
\begin{multline}\label{eq:bijna}
\left|\EE^0 \left[\exp \lambda  \left(Y_n(f)-\EE Y_n(f)\right) \right]-\EE^0 \left[\exp \lambda  \left(Y_n({f^\eps})-\EE Y_n({f^\eps})\right) \right]\right| \\
\leq d_2 |\lambda||f-f^\eps |_{\mathcal L_w}
\exp d_1 (\Re \lambda)^2\left(  |f-f^\eps|_{\mathcal L_w} ^2+  |f^\eps|_{\mathcal L_w}^2 \right),
\end{multline}
for $\lambda $ in a sufficiently small neighborhood of the origin, $\eps>0$  and  $\xi^{(n)} \in  \mathcal C_n(U,n^\delta)$. Since $|f-f^\eps|_{\mathcal L_w}\to 0$ and  $\|f^\eps \|_\infty\to \|f\|_\infty$ as $\eps\downarrow 0$,  we then have 
\begin{multline}\label{eq:bijna2} \left|\EE^0 \left[\exp \lambda  \left(Y_n(f)-\EE Y_n(f)\right) \right]-\EE^0 \left[\exp \lambda  \left(Y_n({f^\eps})-\EE Y_n({f^\eps})\right) \right]\right|\\= \mathcal O(|f-f^\eps|_{\mathcal L_w}),
\end{multline}
as $\eps \downarrow 0$, where the constant can be chosen 
uniformly for $\lambda$  in a sufficiently small neighborhood of the origin and $n$ sufficiently large. 

Let us for now use the notation 
\begin{align*}
F_n(\lambda)&= \EE^0 \left[\exp \lambda  \left(Y_n(f)-\EE Y_n(f)\right) \right]\\
F_n^\eps(\lambda)&=\EE^0 \left[\exp \lambda  \left(Y_n(f^\eps)-\EE Y_n(f^\eps)\right) \right].
\end{align*}
Then $F_n(0)=F_n^\eps(0)=0=F_n'(0)={F_n^\eps}' (0)$. Note that by \eqref{eq:boundonlinstat} we also have that $|F_n(\lambda)|$ is uniformly bounded from above and below for $\lambda$ in a sufficiently small neighborhood of the origin.

Now take $\eps=\eps_n\to 0$ as $n\to \infty$. Then by Corollary \ref{cor:Dnalpha} and \eqref{eq:defDn}  we obtain \eqref{eq:CLTsuff}. By combining this with \eqref{eq:bijna2} and the fact that $|F_n|$ is  bounded from below we have
\begin{equation}\label{eq:logratioF}
\log F_n(\lambda) = \log F_n^{\eps_n}(\lambda) - \log\left(1+\frac{F_n^{\eps_n}(\lambda)-F_n(\lambda)}{F_n(\lambda)}\right)= \frac{1}{2} \sigma_{n,\eps_n}^2 \lambda^2+ o(1),
\end{equation}
as $n\to \infty$, uniformly for $\lambda $ in a sufficiently small neighborhood of the origin (not depending on $n$). Here $\sigma_{n,\eps}$ is some sequence of numbers related to the expression in Corollary \ref{cor:Dnalpha}, which is not very explicit. However, by analyticity and \eqref{eq:logratioF}, we have
\begin{equation}\nonumber
\Var X_{n}(f) =2 \oint \log F_n(\lambda)\frac{ {\rm d} \lambda}{\lambda^2}  =\sigma_{n,\eps_n}^2+ o(1),
\end{equation}
as $n\to \infty$, where the integral is over a sufficiently  small and counterclockwise oriented circle around the origin.  Hence we  see that $\sigma_{n,\eps_n}^2$ converges to the limiting variance in Theorem \ref{th:CLTfixed}. By inserting this back into \eqref{eq:logratioF} we prove the statement.
 \end{proof}

\section{Asymptotic analysis of $K_n$ and $R_n ^I$}\label{sec:steepest}

In this section we will apply steepest descent techniques to obtain various asymptotic properties for $K_n$ and $R^I_n$ as $n\to \infty$. In particular, we will prove Lemma's \ref{lem:saddlepre} and \ref{lem:asymKnpre}. Moreover, in the upcoming sections we will also need some more asymptotic properties that we will also derive in this Section. Throughout this section we will  always choose the parameters as in \eqref{eq:para1}--\eqref{eq:para3} and also set
  \begin{equation}\label{eq:qto1} 
  \quad q=q_t={\rm e}^{-t}.
  \end{equation} 
 
  \subsection{Integrable form of $K_n$}

It will be useful to write the kernel in a so-called integrable form. To this end, 
we rewrite \eqref{eq:defFnpre} as
\begin{equation}\label{eq:defFn}
F_n(w;x)=(qw-x)^2+\frac{1-q^2}{n} \sum_{j=1}^n \log(w-\xi^{(n)}_j).
\end{equation}
When $x$ is clear form the context then we will sometimes use the short notation $F_n(w)$. 
Then we can also  rewrite the kernel $K_n$ in \eqref{eq:defKn2a} as
\begin{align}\label{eq:defKn2}
K_n(x,y) =\frac{2 q n}{(1-q^2)(2 \pi {\rm i})^2} \oint_\Sigma {\rm d}z \int_\Gamma {\rm d}w  \frac{ {\rm e}^{\frac{n} {1-q^2}F_n(w;x)}}{ {\rm e}^{\frac{n} {1-q^2}F_n(z;y)}} \frac{1}{w-z}.
\end{align}

\begin{lemma} The kernel $K_n$ as defined in \eqref{eq:defKn2a} can be written as 
 \begin{align}\label{eq:kernelintegrable}
 K_n(x,y)= \frac{\sum_{j=0}^n \phi_j(x)\psi_j(y)}{x-y},
 \end{align}
 where 
 \begin{align}\label{eq:defphi0}
 \begin{split}
 \phi_0(x)&= \sqrt{\frac{2 n q^2}{1-q^2}} \frac{1}{2 \pi {\rm i}} \int_\Gamma {\rm d} w{\rm e}^{\frac{n}{1-q^2} F_n(w;x)}\\
 \psi_0(y)&= \sqrt{\frac{2 n q^2}{1-q^2}} \frac{1}{2 \pi {\rm i}}  \oint_\Sigma {\rm d} z{\rm e}^{-\frac{n}{1-q^2}F_n(z;y)}\\
 \phi_j(x)&=  \frac{1}{2 \pi {\rm i}}  \int_\Gamma  {\rm d} w{\rm e}^{\frac{n}{1-q^2} F_n(w;x)}\frac{1}{w-\xi_j^{(n)}}, \qquad j=1,\ldots,n,\\
 \psi_j(y)&= - \frac{1}{2 \pi {\rm i}} \oint_\Sigma {\rm d} z {\rm e}^{-\frac{n}{1-q^2}F_n(z;y)}\frac{1}{z-\xi_j^{(n)}}\qquad j=1,\ldots,n.
 \end{split}
 \end{align}
\end{lemma}
\begin{proof}
Let us write
\begin{align}\label{eq:defGn} G_n(w;x)=F_n(w;x)+2q w x=q^2 w^2-x^2 + \frac{1-q^2}{n} \sum_{j=1}^n \frac{1}{w-\xi_j^{(n)}}.
\end{align}
The proof is based on the following trick
$$
(x-y) K_n(x,y)=\frac{1}{(2\pi {\rm i})^2} \oint_\Sigma {\rm d} z \int_\Gamma {\rm d} w \frac{{\rm e}^{\frac{n}{1-q^2} G_n(w)}}{{\rm e}^{\frac{n}{1-q^2} G_n(z)}} \frac{1}{w-z} \\
 \left.  \frac{\partial}{\partial s} \frac{{\rm e}^{-\frac{2n}{1-q^2} q(w-s)x}}{{\rm e}^{-\frac{2n}{1-q^2}  q(z-s)y)}}\right|_{s=0}. 
$$
By taking the derivative outside the integral  and the change of variables $w\mapsto w+s$ and $z\mapsto z+s$ we obtain
\begin{align*}
(x-y) K_n(x,y)=\frac{n}{(1-q^2)(2\pi {\rm i})^2} \oint_\Sigma {\rm d} z \int_\Gamma {\rm d} w \frac{{\rm e}^{\frac{n}{1-q^2} G_n(w)-q w x}}{{\rm e}^{\frac{n}{1-q^2} G_n(z)-q z y}} \frac{G_n'(w)-G_n'(z)}{w-z}.
\end{align*}
Now note that from \eqref{eq:defGn} we can write \begin{align*}
 \frac{G_n'(w)-G_n'(z)}{w-z}=2q^2-\frac{1-q^2}{n} \sum_{j=1}^n \frac{1}{w-\xi_j^{(n)}}\frac{1}{z-\xi_j^{(n)}}.
\end{align*}
From here the statement easily follows.
\end{proof}
  In the upcoming analysis, we first derive the asymptotics for $\phi_j$ and $\psi_j$. This we will do by using classical steepest descent argument on their contour integral representations. For this analysis we need to find the relevant critical point of $F_n$ which is done in Section \ref{subsec:saddle}. Then we   deform the contours $\Sigma$ and $\Gamma$ to contours of steep descent and ascent in Section \ref{subsec:contours}. From these ingredients,  we compute in Section \ref{subsec:asymphi}  the asymptotic for $\phi_j$ and $\psi_j$  and then the asymptotics of $K_n$ in Section \ref{subsec:asymKn}.
  
When $n$ grows large it is natural to expect that the function $F_n$ in \eqref{eq:defFn}  approaches $F$ given by 
\begin{align*}
F(w)=F(w;x)=(qw-x)^2 +(1-q^2) \int_{- \sqrt 2}^{\sqrt 2} \log (w-y)Ê{\rm d} \mu_{sc}(y). 
\end{align*}
By a rather straightforward computation it follows that 
\begin{align*}
F'(w)=2q (qw-x) +(1-q^2)\left(w-(w^2-2)^{1/2}\right)\end{align*}
with a branch cut along $(-\sqrt 2, \sqrt 2)$ and the branch for the square root is taken such that $(w^2-2)^{1/2}=w+\mathcal O(w^{-1})$ for $w\to \infty$.   From here it is easy to compute that $F$ has two critical points 
\begin{align}\label{eq:defOmegax}
\Omega(x)=x \tfrac12 (q+1/q) +{\rm i} \sqrt{2-x^2} \tfrac12Ê(1/q-q),
\end{align}
and $\overline \Omega(x)$, which both are saddle points (i.e. $F'(\Omega(x))=0$ and $F''(\Omega(x))\neq 0$ and similarly for $\overline \Omega(x)$).  We expect that when $n$ grows large, the function $F_n$ has a saddle point $\Omega_n(x)$ that is close to $\Omega(x)$. However, for this to be the true at the scale that we are interested in,  we need our sequence $\xi=(\xi^{(n)})_n$ to be sufficiently regular and this is why we restrict to $\xi  \in \mathcal C(U,A,\delta)$. 
\subsection{The functions $\mathcal E_j$}

Before we start with the steepest descent analysis, let us first  define the following auxiliary functions
\begin{align}\label{eq:defE1}
\mathcal E_1(x;\xi^{(n)})&= \left(\frac{1-q^2}{n}\right)^{1/2} \left| \sum_{j=1}^n  \left(\frac{1}{\Omega(x)-\xi_j^{(n)}}-\int \frac{1}{\Omega(x)-\xi} {\rm d} \mu_{s.c.}(\xi)\right)\right| \\
\mathcal E_2(x;\xi^{(n)})&= \frac{1-q^2}{n} \left| \sum_{j=1}^n  \left(\frac{1}{\left(\Omega(x)-\xi_j^{(n)}\right)^2}-\int \frac{1}{\left(\Omega(x)-\xi\right)^2} {\rm d} \mu_{s.c.}(\xi)\right)\right| \label{eq:defE2}\\
\mathcal E_3(x;\xi^{(n)})&= \left(\frac{1-q^2}{n}\right)^{3/2}\sum _{j=1}^n\frac{1}{| \Omega(x)-\xi_j^{(n)}|^{3}} \label{eq:defE3}
\end{align}
We will also use the short-hand notation $\mathcal E_j(x;\xi^{(n)})=\mathcal E_j(x)$.

The point of introducing these function is that in the saddle point analysis, all the error terms can be expressed in term of $\mathcal E_j$. There we need to   estimate the functions $\mathcal E_j$, which we do in the following lemma.  The main observation is that if $\xi  \in \mathcal C(U, A,\delta)$ (we recall \eqref{eq:defC}), then we can bound the asymptotic behavior of $\mathcal E_j$ for large $n$. 

We recall that we choose $U\subset(-\sqrt 2,\sqrt 2)$ to be an open interval containing $x_*.$

\begin{lemma}\label{lem:condone}
 Let $I \subset U$ be compact. Then there exists constant $a_1,a_2,a_3>0$ (independent of $\delta$) such that 
\begin{align}\label{eq:condone}
\begin{split}
 \mathcal E_1(x;\xi^{(n)})&\leq a_1 n^\delta, \\
  \mathcal E_2(x;\xi^{(n)})&\leq a_2 n^{-\gamma(1-\gamma))/2}, \\
 \mathcal E_3(x;\xi^{(n)})&\leq a_3 n^{-(1-\gamma)/2},
\end{split}\end{align}
for $x\in I$, $\xi \in \mathcal C(U,A n^\delta)$ and $n \in \N$.  
\end{lemma}
\begin{proof}
The estimate for $\mathcal E_1$ directly follows from the definition of $\mathcal C_n$ and the fact that $\Im \Omega(x)\geq 1/n$. 

For the estimate on $\mathcal E_2$ we use the following notation. We define 
\begin{align*}
\Delta F_n(w)=\sum_{j=1}^n \left(\frac{1}{w-\xi_j^{(n)}} -\int \frac{1}{w-\xi}{\rm d} \mu_{s.c.} (\xi)\right).
\end{align*}
  To estimate $\mathcal E_2(x;\xi^{(n)})$ we need to estimate $\Delta F_n'(\Omega(x))$. We do this by using the Cauchy integral
  $$\Delta F_n'(\Omega(x))=\frac{1}{2\pi {\rm i}} \oint_C \frac{\Delta F_n(w)}{(w-\Omega(x))^2}\ {\rm d} w.$$ Here $C$ is a counter clockwise orient simple contour around $\Omega(x)$. More precisely,  we choose $C$ to be a the boundary of a rectangle where the horizontal sides are 
  $$ U+{\rm i} a n^{-\gamma}, \qquad \text{and} \qquad U +{\rm i}.$$
We choose $a>0$ small enough  so that the rectangle indeed contains $\Omega(x)$.  To this end, note that  since $I$ is a compact subset of $(-\sqrt 2,\sqrt 2)$ there exists a constant $c>0$ such that \begin{align}\label{eq:lowerb} \Im \Omega(x) \geq   (1-q^2)/b, \end{align} 
  for $x\in I$. 

We then estimate 
\begin{equation*}
|\Delta F_n'(\Omega(x))|\leq \frac{\sup_{w\in C} |\Delta F_n(w)|}{2\pi} \oint_C \frac{1}{|w-\Omega(x)|^2}\ |{\rm d} w|.
\end{equation*}
By definition of $\mathcal C_n(U,n^\delta)$ we have 
\begin{equation*}
\sup_{w\in C} |\Delta F_n(w)|\leq A n^\delta  \sup_{w\in C} \sqrt{\frac{n}{\Im w}}\leq An^{\delta+1/2+\gamma/2}/\sqrt a
\end{equation*}
for $\xi^{(n)}\in \mathcal C_n(U,n^\delta)$.  Moreover,   there exists a constant $c>0$ such that 
$$
\oint_C \frac{1}{|w-\Omega(x)|^2}\ |{\rm d} w|\leq \frac{c}{\Im \Omega(x)},$$
for $x\in I$.  Hence, there exists a constant $d$ such that 
\begin{equation}\nonumber
\mathcal E_2(x,\xi^{(n)})\leq \frac{1-q^2}{n} |\Delta F_n'(\Omega(x))|\leq d n^{\delta-(1-\gamma)/2}.
\end{equation}
for $x\in I$ and $\xi^{(n)}\in \mathcal C_n(U, n^{\delta})$. Since $ \delta-(1-\gamma)/2<\gamma(1-\gamma)/2$, this proves the statement for $\mathcal E_2$. 

Finally, the estimate for $\mathcal E_3$ can be obtained as follows. By \eqref{eq:lowerb}  there exists a constant $b>0$ such that  $$
\left(\frac{1-q^2}{n}\right)^{3/2}\sum _{j=1}^n\frac{1}{| \Omega(x)-\xi_j^{(n)}|^{3}} \leq b\frac{(1-q^2)^{1/2}}{n^{3/2}}\sum _{j=1}^n\frac{1}{| \Omega(x)-\xi_j^{(n)}|^{2}}.
$$ 

By using $|w-x|^{-2}= -(\Im w)^{-1} \Im (w-z)^{-1}$ for any $x\in \R$ and $w\in \C \setminus \R$, and again using \eqref{eq:lowerb} we have 
\begin{multline*}\left(\frac{1-q^2}{n}\right)^{3/2}\sum _{j=1}^n\frac{1}{| \Omega(x)-\xi_j^{(n)}|^{3}}\\=-\frac{b (1-q^2)^{1/2}}{n^{3/2}\Im \Omega(x)} \sum _{j=1}^n\Im \frac{1}{ \Omega(x)-\xi_j^{(n)}}\leq -\frac{b^2}{n^{3/2} (1-q^2)^{1/2}} \sum _{j=1}^n\Im \frac{1}{ \Omega(x)-\xi_j^{(n)}}
\end{multline*}
By invoking the definition of $\mathcal E_1$ this leads to
\begin{multline*}
\left(\frac{1-q^2}{n}\right)^{3/2}\sum _{j=1}^n\frac{1}{| \Omega(x)-\xi_j^{(n)}|^{3}}\\ \leq -\frac{b^2}{n^{1/2} (1-q^2)^{1/2}}\int \Im \frac{1}{ \Omega(x)-\xi}\ {\rm d}\mu_{s.c.}(\xi)-\frac{b^2}{n (1-q^2)} \Im \mathcal E_1(x;\xi^{(n)}),
\end{multline*}
 By combining this with the result for $\mathcal E_1(x)$, the fact that $n(1-q^2)\sim n^{1-\gamma} $ and the fact that the Cauchy transform of the semi-circle is bounded on $\C\setminus [-\sqrt 2,\sqrt 2]$, the statement also follows for $\mathcal E_3$. 
\end{proof}

\subsection{Saddle points}\label{subsec:saddle}
We now come to the critical point of $F_n$ and prove Lemma \ref{lem:saddlepre}.
\begin{lemma}\label{lem:saddle} (cf. Lemma \ref{lem:saddlepre})  Let $I \subset U$ be compact. Then for sufficiently large $n$,  $x\in I$  and $\xi^{(n)} \in \mathcal C_n(U,n^\delta)$ we have that there exists a unique $\Omega_n(x)$  in the ball 
\begin{align}\label{eq:approxball}\Omega(x)+ \left(\tfrac{1-q^2}{n}\right)^{\tfrac12-\delta} B_{0,1}\end{align}
such that  $F_n'(\Omega_n(x))=0$.
\end{lemma}

\begin{proof}
It is important to observe that the ball \eqref{eq:approxball} is entirely  contained in the upper half plane for $n$ sufficiently large.  Indeed, for any $w$ in the ball \eqref{eq:approxball} we have 
\begin{align}\label{eq:ballinhalfplane}
\Im w\geq \Im \Omega(x)  - \left(\tfrac{1-q^2}{n}\right)^{\tfrac12-\delta},
\end{align}
as $n\to \infty$, and by \eqref{eq:lowerb}, \eqref{eq:qto1} and  since we chose $\delta$ such that $0<\delta<\frac{1}{2}\frac{1-\gamma}{1+\gamma}$, the right-hand side is positive for $n$ sufficiently large.

The proof is based on  Rouch\'e's Theorem, that tells us that  if for all $w$ with  
\begin{align}\label{eq:rouchecirc}
|\Omega(x)-w|=\left(\frac{1-q^2}{n}\right)^{1/2-\delta}, 
\end{align}
we have the inequality
\begin{align}\label{eq:rouche}
|F_n'(w)-F'(w)|<  |F'(w)|,\end{align}
then there is precisely  one solution to $F_n'(w)=0$ in the given set \eqref{eq:approxball}.

Let $w\in \C$ satisfy \eqref{eq:rouchecirc}. By the triangular inequality we have 
\begin{multline}\label{eq:towardsrouche}
|F_n'(w)-F'(w)| \leq \left|F_n'(w)-F_n'(\Omega(x))-F_n''(\Omega(x))(w-\Omega(x))\right|\\
+\left|F_n'(\Omega(x))-F'(\Omega(x))\right|+\left|F_n''(\Omega(x))-F''(\Omega(x))\right|\left|w-\Omega(x)\right|\\+\left|F'(w)-F'(\Omega(x))-F''(\Omega(x))(w-\Omega(x))\right|
\end{multline}
We will estimate the four terms at the right-hand side separately.
By the definition of $F_n$ in \eqref{eq:defFn} we can rewrite the first term at the right-hand side by 
$$
\left|F_n'(w)-F_n'(\Omega(x))-F_n''(\Omega(x))(w-\Omega(x))\right|
=\frac{1-q^2}{n}\left|\sum_{j=1}^n \frac{(w-\Omega(x))^2}{(w-\xi_j^{(n)})(\Omega(x)-\xi_j^{(n)})^2}\right|.
$$
Since $w$ satisfies \eqref{eq:rouchecirc}, we have for sufficiently large $n$ that   $\Im w\geq |w-\Omega(x)|$ for $x\in I$ (see also \eqref{eq:ballinhalfplane}). Hence  also  
\begin{align*}
2 |w-\xi_j^{(n)}|\geq  |w-\xi_j^{(n)}| +\Im w \geq |w-\xi_j^{(n)}| +|w-\Omega(x)| \geq |\Omega(x)-\xi_j^{(n)}|,
\end{align*}
which implies
\begin{multline}\label{eq:roucheestimate1}
\left|F_n'(w)-F_n'(\Omega(x))-F_n''(\Omega(x))(w-\Omega(x))\right|\\
=2\left(\frac{1-q^2}{n}\right)^{2-2\delta}\sum_{j=1}^n \frac{1}{|\Omega(x)-\xi_j^{(n)}|^3}=2\left(\frac{1-q^2}{n}\right)^{1/2-2\delta}\mathcal E_3(x).
\end{multline}
The second and third at the right-hand side of \eqref{eq:towardsrouche} can be rewritten in terms of $\mathcal E_1$ and $\mathcal E_2$, since we have by definition that 
\begin{equation}\label{eq:roucheestimate2}
\begin{split}
\left|F_n'(\Omega(x))-F'(\Omega(x))\right|=\left(\frac{1-q^2}{n}\right)^{1/2}\mathcal E_1(x)\\
\left|F_n''(\Omega(x))-F''(\Omega(x))\right|\left|w-\Omega(x)\right|=\mathcal E_2(x)\left|w-\Omega(x)\right|.
\end{split}
\end{equation} By inserting \eqref{eq:roucheestimate1}  and \eqref{eq:roucheestimate2} into \eqref{eq:towardsrouche}, we have for sufficiently large $n$ that 
\begin{multline*}
|F_n'(w)-F'(w)| \\ \leq 2 \left(\frac{1-q^2}{n}\right)^{1/2-2\delta}\mathcal E_3(x)+\left(\frac{1-q^2}{n}\right)^{1/2}\mathcal E_1(x)+\mathcal E_2(x)\left|w-\Omega(x)\right|\\+\left|F'(w)-F'(\Omega(x))-F''(\Omega(x))(w-\Omega(x))\right|,
\end{multline*}
for $x\in I$.  The fourth term at the right-hand side  consists of the difference of $F'(w)$ and its first degree Taylor polynomial around $w=\Omega(x)$ and hence we can further estimate this too
\begin{multline}\label{eq:towardsrouche1}
|F_n'(w)-F'(w)| \\ \leq 2 \left(\frac{1-q^2}{n}\right)^{1/2-\delta}
\left(\left(\frac{1-q^2}{n}\right)^{\delta}\mathcal E_1(x)+\mathcal E_2(x)+\left( \frac{1-q^2}{n}\right)^{-\delta}
\mathcal E_3(x)\right)\\+\mathcal O\left( \left(\frac{1-q^2}{n}\right)^{1-2\delta}\right),
\end{multline}
Moreover, by a Taylor  series expansion and using $F'(\Omega(x))=0$ we have for any $w$ satisfying \eqref{eq:rouchecirc}
\begin{align}\label{eq:towardsrouche2}
|F'(w)|=\left|F''(\Omega(x))\right|  \left(\frac{1-q^2}{n}\right)^{1/2-\delta}+\mathcal O\left( \left(\frac{1-q^2}{n}\right)^{1-2\delta}\right),\end{align}
and hence by comparing \eqref{eq:towardsrouche2} with  \eqref{eq:towardsrouche1} and using \eqref{eq:condone} we  easily  verify that \eqref{eq:rouche} holds and hence we proved the statement. \end{proof}
\begin{corollary}
Let $I \subset U$ be compact. For $n$ sufficiently large, let $\Omega_n(x)$ be the saddle point of $F_n$ in the ball \eqref{eq:approxball}. Then there exists a constant $c>0$ such that for $n$ sufficiently large we have 
 \begin{align}\label{eq:lowerboundonImOm}
 \Im \Omega(x),\Im \Omega_n(x) \geq (1-q^2)/c,
 \end{align}
 for $x\in I$ and $\xi \in \mathcal C(U, A,\delta)$. 
\end{corollary}
\begin{proof}
For $\Omega(x)$, the statement follows by the definition of $\Omega(x)$ in \eqref{eq:defOmegax} and the compactness of $I$. The statement for $\Omega_n(x)$ follows then by applying Lemma \ref{lem:saddle} (and \eqref{eq:ballinhalfplane}). \end{proof} 
\begin{lemma}
 Let $I \subset U$ be compact. For $n$ sufficiently large, let  $\Omega_n(x)$ be the saddle point of $F_n$ in the ball \eqref{eq:approxball}. Then \begin{align}\label{eq:fdouble}
F_n''(\Omega_n(x))=2+ o(1),\end{align}
as $n \to \infty$, uniformly for $x\in I$ and $\xi\in \mathcal C(U,A,\delta)$ 
\end{lemma}
\begin{proof}
By the triangular inequality we have
\begin{equation}
\label{eq:fdouble1}
|F''_n(\Omega_n(x))-F''(\Omega(x))|\\
\leq |F''_n(\Omega_n(x))-F_n''(\Omega(x))|+|F''_n(\Omega(x))-F''(\Omega(x))|,
\end{equation}
and by definition $F_n$ in \eqref{eq:defFn} and  of $\mathcal E_2$ in \eqref{eq:defE2} it follows that
\begin{align}\label{eq:fdouble2}
|F''_n(\Omega_n(x))-F''(\Omega(x))|\leq |F''_n(\Omega_n(x))-F_n''(\Omega(x))|+\mathcal E_2(x).
\end{align}
Moreover,
\begin{multline} \label{eq:fdouble82}
 F''_n(\Omega_n(x))-F_n''(\Omega(x))
 = \frac{1-q^2}{n}Ê\sum_{j=1}^n\left( \frac{1}{(\Omega_n(x)-\xi_j^{(n)})^2}-\frac{1}{(\Omega(x)-\xi_j^{(n)})^2}\right)\\
= \frac{1-q^2}{n}Ê\sum_{j=1}^n\left( \frac{\Omega(x)-\Omega_n(x)}{(\Omega_n(x)-\xi_j^{(n)})(\Omega(x)-\xi_j^{(n)})}\left(\frac{1}{\Omega_n(x)-\xi_j^{(n)}}+\frac{1}{\Omega(x)-\xi_j^{(n)}}\right)\right)
\end{multline}
Note that because of \eqref{eq:approxball} and \eqref{eq:lowerboundonImOm}  we have for $n$ sufficiently large that
$\Im \Omega_n(x)\geq |\Omega_n(x)-\Omega(x)|$  for $x\in I$  and hence also
\begin{align*}
2 |\Omega_n(x)-\xi_j^{(n)}|\geq |\Omega_n(x)-\xi_j^{(n)}| +|\Omega_n(x)-\Omega(x)| \geq |\Omega(x)-\xi_j^{(n)}|,
\end{align*}
for $x\in I$. By using this inequality in \eqref{eq:fdouble82} we obtain
\begin{equation*}
 \left|F''_n(\Omega_n(x))-F_n''(\Omega(x))\right|
 \leq \frac{3(1-q^2)|\Omega(x)-\Omega_n(x)|}{n}Ê\sum_{j=1}^n\frac{1}{|\Omega(x)-\xi_j^{(n)}|^3},
 \end{equation*}
 for $x\in I$.  Combining this with \eqref{eq:defE3} and \eqref{eq:approxball} we obtain, for sufficiently large $n$, that 
 \begin{equation}\label{eq:fdouble3}
 \left|F''_n(\Omega_n(x))-F_n''(\Omega(x))\right|
 \leq 3\left(\frac{n}{1-q^2}\right)^{\delta} \mathcal E_3(x),\end{equation}
 for $x\in I$. 
Therefore by \eqref{eq:fdouble2} and \eqref{eq:fdouble3} we have 
\begin{align*}
|F''_n(\Omega_n(x))-F''(\Omega(x))|\leq \mathcal E_2(x)+3 \left(\frac{n}{1-q^2}\right)^{\delta} \mathcal E_3(x),
\end{align*}
and hence  by  \eqref{eq:condone}  and by the fact that $0<\delta<\frac{1}{2}{\frac{1-\gamma}{1+\gamma}}$,  we have 
\begin{align*}
F_n''(\Omega_n(x))=F''(\Omega(x))+o(1),
\end{align*}
uniformly for $x\in I $ and $\xi\in \mathcal C(U,A,\delta)$  as $n\to \infty$.  Finally, by a simple computation and \eqref{eq:qto1} we have
\begin{align*}
F''(\Omega(x))=\frac{4 q^2}{1+q^2-{\rm i} \frac{x}{\sqrt{2-x^2}}(1-q^2)}=2+o(1),
\end{align*}
uniformly for $x\in I $ and $\xi^{(n)}\in \mathcal C_n(U,n^\delta)$  as $n\to \infty$. , and we have proved the statement. \end{proof} 

\begin{lemma}\label{lem:Omegalipschitz}
 Let $I \subset U$ be compact. For $n$ sufficiently large, let $\Omega_n(x)$ be the saddle point of $F_n$ in the ball \eqref{eq:approxball}. Then \begin{align}\label{eq:Omegalipschitz}
\Omega_n(x)-\Omega_n(y)=(x-y) (1+ o(1)),\end{align}
uniformly for $x,y\in I$ and $\xi \in \mathcal C(U,A,\delta)$ as $n \to \infty$.
\end{lemma}
\begin{proof}
By differentiating $F_n'(\Omega_n(x);x)=0$ with respect to  $x$ we obtain 
\begin{align*}
\frac{{\rm d}}{{\rm d}x}\Omega_n(x)=\frac{2q^2}{F_n''(\Omega_n(x);x)}.
\end{align*}
Hence  by substituting \eqref{eq:fdouble} into
$$\Omega_n(x)-\Omega_n(y)=\int_y^x  \frac{{\rm d}}{{\rm d}x}\Omega_n(x) {\rm d}x=\int_y^x  \frac{2q^2}{F_n''(\Omega_n(x);x)} {\rm d}x,$$
and by   $q\to 1$ as $n \to \infty$, the statement follows.
\end{proof}

\subsection{Deforming the contours}\label{subsec:contours}
Now that we have established the existence and the approximate location of the saddle points $\Omega_n(x)$ and $\overline{\Omega}_n(x)$, we continue by deforming the contours $\Sigma$ and $\Gamma$. The deformation of the contours that we will discuss has also been used before in \cite{J5} in a similar context for longer times scales. We show here that the choice of contours still works for the shorter time scales that we are interested in. 

 To start with $\Gamma$ we, deform the contour to the contour $\Gamma^*$ 
\begin{align}\label{eq:contoursteepdesc}
\Gamma^*=\Re \Omega_n(x) + {\rm i} \R.\end{align} The contour $\Gamma$ is oriented from $-{\rm i} \infty$ to  $+{\rm i} \infty$. Then $\Gamma$ clearly passes through the saddle point $\Omega_n(x)$ and its conjugate. Moreover, it consists of paths of steep descent for $\Re F_n(w)$ leaving from the saddle points as the following lemma shows.
\begin{lemma} Let $I \subset U$ be compact.  For $n$ sufficient large and $x\in I$, let  $\Omega_n(x)$ be the critical point of $F_n$ in \eqref{eq:approxball}.  Then
 \begin{align}\label{eq:steepdesc1}\frac{{\rm d}}{{\rm d} t} \Re F_n(\Re \Omega_n(x)+{\rm i} t)>0
 \end{align} for $-\infty< t<-\Im \Omega_n(x)$ and  $0<t<\Im \Omega_n(x)$.  Similarly, 
 \begin{align}\label{eq:steepdesc2}\frac{{\rm d}}{{\rm d} t}\Re  F_n(\Re \Omega_n(x)+{\rm i} t)<0\end{align} for $-\Im \Omega_n(x)<t<0$ and $\Im \Omega_n(x)<t<\infty.$

Moreover, for $n$ sufficiently large we have \begin{align}\label{eq:steepdesc3}
 \Re F_n(\Re \Omega_n(x)+{\rm i} t)-\Re F(\Omega_n(x)) <\frac{9-t^2}{4}.\end{align}
 for $|t|>3$ and $x\in I$.  
\end{lemma}
\begin{proof}
A simple computation reveals \begin{equation*}
\frac{{\rm d}}{{\rm d} t}\Re F_n(\Re \Omega_n(x)+ {\rm i} t)
= -2q^2 t+\frac{t(1-q^2)}{n}Ê\sum_{j=1}^n \frac{1}{(\Re \Omega_n(x)-\xi_j^{(n)})^2+t^2} 
\end{equation*}
By taking the imaginary part of $F'_n(\Omega_n(x))=0$ and dividing by $\Im \Omega_n(x)$ we have
\begin{align}\label{eq:steepdescA}
2q^2=\frac{1-q^2}{n} \sum_{j=1}^n \frac{1}{(\Re \Omega_n(x)-\xi_j^{(n)})^2+(\Im \Omega_n(x))^2}
\end{align}
and hence \begin{multline} \label{eq:steepdescB}
\frac{{\rm d}}{{\rm d} t}\Re F_n(\Re \Omega_n(x)+ {\rm i} t)\\
=\frac{1-q^2}{n}Ê\sum_{j=1}^n \frac{t((\Im \Omega_n(x))^2-t^2)}{((\Re \Omega_n(x)-\xi_j^{(n)})^2+t^2)((\Re \Omega_n(x)-\xi_j^{(n)})^2+(\Im \Omega_n(x))^2)}. 
\end{multline}
From here \eqref{eq:steepdesc1} and \eqref{eq:steepdesc2} easily follow. 

To get \eqref{eq:steepdesc3}   we argue as follows. We will prove the case $t\geq 3$ only, the case $t\leq -3$ follows then by symmetry (note that $\Re F_n(\Omega_n(x))=\Re F_n(\overline{\Omega_n(x)})$).  Observe that  for sufficiently large $n$ we have by \eqref{eq:defOmegax} and \eqref{eq:approxball} we have  $|\Re \Omega_n(x)-\xi^{(n)}_j| \leq 2 \sqrt 2$ and $(\Im \Omega_n(x))^2\leq \frac{1}{2}$ and hence 
\begin{align}\label{eq:steepdesc8}
 \frac{t^2-(\Im \Omega_n(x))^2}{t^2+\left(\Re \Omega_n(x)-\xi^{(n)}_j \right)^2}\geq \frac{t^2-(\Im \Omega_n(x))^2}{t^2+8}\geq \frac{1}{2},
 \end{align}
 for $t \geq 3$, and hence 
 \begin{align*}
\frac{{\rm d}}{{\rm d} t}\Re F_n(\Re \Omega_n(x)+ {\rm i} t)\leq -q t \leq -t/2
 \end{align*} 
 for $t\geq 3$. Moreover, by \eqref{eq:steepdesc2} and $\Im \Omega_n(x) \leq 3 $ for sufficiently large $n$, we have 
\begin{multline*}
 \Re F_n(\Re \Omega_n(x)+{\rm i} t)-\Re F_n(\Omega_n(x)) \\=\int_{\Im \Omega_n(x)}^t \frac{{\rm d}}{{\rm d} s} \Re F_n(\Re \Omega_n(x)+{\rm i} s) {\rm d} s\\
 \leq \int_3^t \frac{{\rm d}}{{\rm d} s} \Re F_n(\Re \Omega_n(x)+{\rm i} s) {\rm d} s,
\end{multline*}
for sufficiently large. By substituting \eqref{eq:steepdesc8} into the right-hand side we obtain \eqref{eq:steepdesc3}.
 \end{proof}
Now we come to the contour $\Sigma$. We deform $\Sigma$ to the contour  $\Sigma^*=\Sigma_+^*\cup \Sigma_-^*$ where 
\begin{align}\label{eq:contoursteepasc}
\begin{split}
\Sigma _+^*&= \{\Re \Omega_n(x) -t +{\rm i} \sqrt{t^2/3+(\Im \Omega_n(x))^2} \ :\ -\infty< t< \infty\},\\
\Sigma_- ^*&= \{\Re \Omega_n(x) +t -{\rm i} \sqrt{t^2/3+(\Im \Omega_n(x))^2} \ :\ -\infty< t< \infty\}.
\end{split}
\end{align}
Here the parametrization chosen in the definition also defines the orientation. 
\begin{lemma}
Let $I \subset U$ be compact. For $n$ sufficiently large and $x \in I$, let $\Omega_n(x)$ be the saddle point of $F_n$ in the ball \eqref{eq:approxball}. Then
\begin{align}\label{eq:steepasc1}
\frac{{\rm d}}{{\rm d} t}\Re F_n\left(\Re \Omega_n(x) \mp t \pm {\rm i} \sqrt{t^2/3+(\Im \Omega_n(x))^2}\right)<0, 
\end{align}
for $t<0$ and 
\begin{align}\label{eq:steepasc2}
\frac{{\rm d}}{{\rm d} t}\Re F_n\left(\Re \Omega_n(x) \mp t \pm {\rm i} \sqrt{t^2/3+(\Im \Omega_n(x))^2}\right)>0, 
\end{align}
for $t>0$. Moreover, for $n$ sufficiently large
\begin{equation}
\label{eq:steepasc3}
 \Re F_n(\Re \Omega_n(x)+ t \pm {\rm i} \sqrt{t^2/3+(\Im \Omega_n(x))^2}) -\Re F(\Omega_n(x)) >\frac{t^2-49}{6}.
 \end{equation}
 for $|t|>7$.  
\end{lemma}

\begin{proof}
By a simple computation one can show that 
\begin{multline}\label{eq:steepascA}
\frac{{\rm d}}{{\rm d} t}\Re F_n\left(\Re \Omega_n(x) +t -{\rm i} \sqrt{t^2/3+(\Im \Omega_n(x))^2} \right)\\=2q(q (\Re \Omega_n(x)+t)-x)-\tfrac23q^2t \\+\frac{1-q^2}{n} \sum_{j=1}^n \frac{\Re \Omega_n+\tfrac43 t-x}{(\Re \Omega_n(x)+t-\xi_j^{(n)})^2+\tfrac13 t^2+(\Im \Omega_n(x))^2}.
\end{multline}
Moreover, by taking the real and imaginary part of $F'_n(\Omega_n(x))=0$ we obtain the following two equations
\begin{align}\label{eq:steepascB}
2q(q \Re \Omega_n(x)-x) &=-\frac{1-q^2}{n} \sum_{j=1}^n \frac{\Re \Omega_n-\xi_j^{(n)}}{(\Re \Omega_n(x)-\xi_j^{(n)})^2+(\Im \Omega_n(x))^2}\\
2q^2&=\frac{1-q^2}{n} \sum_{j=1}^n \frac{1}{(\Re \Omega_n(x)-\xi_j^{(n)})^2+(\Im \Omega_n(x))^2},\label{eq:steepascC}
\end{align}
where we also divided by $\Im \Omega_n(x)$ in the last equality.
By plugging \eqref{eq:steepascB} and \eqref{eq:steepascC} into \eqref{eq:steepascA} we obtain
\begin{multline*}
\frac{{\rm d}}{{\rm d} t}\Re F_n\left(\Re \Omega_n(x) +t +{\rm i} \sqrt{t^2/3+(\Im \Omega_n(x))^2} \right)\\=\frac{1-q^2}{n} \sum_{j=1}^n \left(\frac{\Re \Omega_n+\tfrac43 t-\xi_j^{(n)}}{(\Re \Omega_n(x)+t-\xi_j^{(n)})^2+\tfrac13 t^2+(\Im \Omega_n(x))^2}- \frac{\Re \Omega_n-\xi_j^{(n)}-2t/3}{(\Re \Omega_n(x)-\xi_j^{(n)})^2+(\Im \Omega_n(x))^2}\right).
\end{multline*}
By writing the difference of the fraction as one this can be written as
\begin{multline*}
\frac{{\rm d}}{{\rm d} t}\Re F_n\left(\Re \Omega_n(x) +t -{\rm i} \sqrt{t^2/3+(\Im \Omega_n(x))^2} \right)\\=t \frac{1-q^2}{n} \sum_{j=1}^n \left(\frac{2(\Im \Omega_n(x))^2+\tfrac89 t^2}{\left((\Re \Omega_n(x)+t-\xi_j^{(n)})^2+\tfrac13 t^2+(\Im \Omega_n(x))^2 \right)\left((\Re \Omega_n(x)-\xi_j^{(n)})^2+(\Im \Omega_n(x))^2\right)}\right),
\end{multline*}
and from here (and symmetry) \eqref{eq:steepasc1} and \eqref{eq:steepasc2} easily follow. 

It remains to prove  \eqref{eq:steepasc3}. For sufficiently large $n$ we have by \eqref{eq:defOmegax} and \eqref{eq:approxball} we have  $|\Re \Omega_n(x)-\xi^{(n)}_j| \leq 2 \sqrt 2$ and hence 
\begin{align*}
\frac{2 (\Im \Omega_n(x))^2+\tfrac{8}{9}t^2}{2 (\Re \Omega_n(x)-\xi^{(n)}_j)^2+\frac{7}{3} t^2 +(\Im \Omega_n(x))^2}\geq \frac{2 (\Im \Omega_n(x))^2+\tfrac{8}{9}t^2}{16+\frac{7}{3} t^2 +(\Im \Omega_n(x))^2}\geq \frac{1}{3},
\end{align*}
for $|t| \geq 7$, and hence 
\begin{align*}
\frac{{\rm d}}{{\rm d} t}\Re F_n\left(\Re \Omega_n(x) +t -{\rm i} \sqrt{t^2/3+(\Im \Omega_n(x))^2} \right)\geq 2 qt/3\geq t/3, \end{align*}
for $t \geq 7$. By inserting this inequality into 
\begin{multline*}
 \Re F_n(\Re \Omega_n(x)+t -{\rm i} \sqrt{t^2/3+(\Im \Omega_n(x))^2} )-\Re F(\Omega_n(x)) \\=\int_{0}^t \frac{{\rm d}}{{\rm d} s} \Re F(\Re \Omega(x)+s+{\rm i} \sqrt{s^2/3+(\Im \Omega_n(x))^2} ) {\rm d} s\\
 \geq \int_0^t \frac{{\rm d}}{{\rm d} s} \Re F(\Re \Omega(x)+{\rm i} s) {\rm d} s
\end{multline*}
we obtain \eqref{eq:steepasc3} at  one of the four parts of $\Sigma^*$. At the other parts on can use similar estimates or use symmetry. \end{proof}

\subsection{Asymptotics for $\phi_j$ and $\psi_j$}\label{subsec:asymphi}
We first derive the leading asymptotic term for the function $\phi_0$ and $\psi_0$. To this end, we split $\phi_0=\phi_0^++\phi_0^-$ and $ \psi_0= \psi_0^++\psi_0^-$ where \begin{align}\label{eq:splitphi0}
\begin{split}
 \phi_0^\pm(x)&=\sqrt{\tfrac{2  n q^2}{1-q^2}}\frac{1}{2 \pi {\rm i}} \int_{\Gamma^*\cap \mathbb H_\pm} {\rm e}^{\frac{n}{1-q^2}F_n(w;x)} \ {\rm d}w\\
 \psi_0^\pm(y)&=\sqrt{\tfrac{2  n q^2}{1-q^2}} \frac{1}{2 \pi {\rm i}}\int_{\Sigma^*\cap \mathbb H_\pm} {\rm e}^{-\frac{n}{1-q^2} F_n(z;y)} \ {\rm d}z.
\end{split}
\end{align}
Here $\mathbb H_+$ and $\mathbb H_-$ stand for the  upper and lower half plane respectively, and $\Gamma^*$ and $\Sigma^*$ are the contours of steep descent \eqref{eq:contoursteepdesc} and ascent \eqref{eq:contoursteepasc}. 

\begin{lemma}\label{lem:boundonphi0}
Let $I \subset U$ be compact. Then for any $0<\eps<(1-\gamma)/2$ we have 
\begin{align}\label{eq:estimateintegral1}
\begin{split} 
\phi_0^+(x)
&=\sqrt{\frac{q}{\pi |F''_n(\Omega_n)|}} \exp\left(-\frac12 {\rm i}\eps_n \right) {\rm e}^{\frac{n}{1-q^2}  F_n(\Omega_n(x))}\left(1+\mathcal O\left(n^{-\eps}\right)\right),
\\ 
\phi_0^-(x)
&=\sqrt{\frac{q}{\pi |F''_n(\Omega_n)|}} \exp\left(-\frac12 {\rm i}\eps_n \right) {\rm e}^{\frac{n}{1-q^2}  F_n(\overline{\Omega_n(x)})}\left(1+\mathcal O\left(n^{-\eps}\right)\right),\\
\psi_0^+(x)
&=-\frac{1}{{\rm i}} \sqrt{\frac{q}{\pi |F''_n(\Omega_n)|}} \exp\left(-\frac12 {\rm i}\eps_n \right)  {\rm e}^{-\frac{n}{1-q^2}  F_n(\Omega_n(x))} \left(1+\mathcal O\left(n^{-\eps}\right)\right),\\
\psi_0^-(x)&=\frac{1}{{\rm i}} \sqrt{\frac{q}{\pi |F''_n(\Omega_n)|}} \exp\left(-\frac12 {\rm i}\eps_n \right)  {\rm e}^{-\frac{n}{1-q^2}  F_n(\overline{\Omega_n(x)})}\left(1+\mathcal O\left(n^{-\eps}\right)\right),
\end{split}
\end{align}
 as $n\to \infty$, uniformly for $x\in I$ and $\xi \in \mathcal C(U,A,\delta)$. Here $\eps_n= \arg F_n''(\Omega_n(x))$.
\end{lemma}
\begin{proof}
We will only prove the asymptotics for $\phi_0^+$. The other cases follow from similar arguments and are left to the reader.

  Let $B$ be the ball 
$$B=\left\{w \ :\ |w-\Omega_n(x)|=\left(\frac{1-q^2}{n}\right)^{1/2-\delta'/3}\right\},$$
for some $0<\delta'<\delta$. By an argument similar to   the first line in the proof of Lemma \ref{lem:saddle}, it is not hard to show  that the ball $B$ is entirely  contained in the upper half plane. 

Let us split the integral as follows
\begin{align}\nonumber
\int_{\Gamma^* \cap \mathbb H^+}= \int_{\Gamma ^* \cap B}+  \int_{(\Gamma^*  \cap \mathbb H^+) \setminus  B}.\end{align}
We will estimate the two integrals at the right-hand side separately. Let us for the moment focus on the first integral, i.e. the integral  over $\Gamma^*\cap B$. 
We claim that inside $B$ the function $F_n(w)$ is a approximately quadratic. More precisely, we will show that 
\begin{equation}\label{eq:estimateintegralA}
\frac{n}{1-q^2} \left(F_n\left(\Omega_n(x)+\sqrt\frac{1-q^2}{n} s\right)-F_n(\Omega_n(x))\right)-\frac{1}{2}s^2 F''_n(\Omega_n(x)),
\end{equation}
is small. By subtracting $s F_n'(\Omega_n(x))=0$ and using the definition of $F_n$ in \eqref{eq:defFn} we can rewrite \eqref{eq:estimateintegralA}  to 
\begin{multline} \label{eq:estimateintegralB} 
\sum_{j=1}^n \left(\log\left(1+\sqrt{\frac{1-q^2}{n}}\frac{s}{\Omega_n(x)-\xi_j^{(n)}}\right)
\right.
\\
\left.-\sqrt{\frac{1-q^2}{n}}\frac{s}{\Omega_n(x)-\xi_j^{(n)}}+\frac{1-q^2}{n}\frac{s^2}{2(\Omega_n(x)-\xi_j^{(n)})^2}\right).
\end{multline}
By using \begin{align*}
\left|\log (1+w)-w+\tfrac12 w^2\right| \leq 2 |w|^3, \qquadÊ\text{for } |w|<1/2,
\end{align*}
we can estimate \eqref{eq:estimateintegralB} by 
\begin{multline*}
\left(\frac{2(1-q^2)}{n}\right)^{3/2}  \sum_{j=1}^n\frac{|s|^3}{ |\Omega_n(x)-\xi^{(n)}_j|^3}\leq 
\left(\frac{2(1-q^2)}{n}\right)^{3/2-\delta'} \sum_{j=1}^n\frac{1}{ |\Omega_n(x)-\xi^{(n)}_j|^3}\\
=2^{3/2-\delta'} \left(\frac{(1-q^2)}{n}\right)^{-\delta'} \mathcal E_3(x).
\end{multline*}
By combining this with \eqref{eq:condone} we have
\begin{multline}\label{eq:estimateintegralC}
\left|\frac{n}{1-q^2} \left(F_n\left(\Omega_n(x)+\sqrt\frac{1-q^2}{n} s\right)-F_n(\Omega_n(x))\right)-\frac{1}{2}s^2 F''_n(\Omega_n(x)) \right| \\\leq 2^{3/2-\delta'} \left(\frac{(1-q^2)}{n}\right)^{-\delta'} \mathcal E_3(x)=\mathcal O\left(n^{\delta' (1+\gamma)-(1-\gamma)/2} \right),
\end{multline}
uniformly for $x\in I$ and $\xi \in \mathcal C(U,A,\delta)$ as $n\to \infty$. 

Let us return to estimating the integral $\int_{\Gamma \cap B}$. After the change of variables
 \begin{align}\label{eq:localvariable}
w=\Omega_n(x)+{\rm i}\sqrt{\frac{1-q^2}{n}}\tilde s.
\end{align} 
and using \eqref{eq:estimateintegralC}  we obtain
\begin{multline*}
\sqrt{\tfrac{ 2q^2 n}{(1-q^2)}}\frac{1}{2 \pi {\rm i}}\int_{\Gamma^*\cap B} {\rm e}^{\frac{n}{1-q^2}F_n(w)} {\rm d} w\\
=  \frac{\sqrt{2q^2}}{2 \pi } {\rm e}^{\tfrac{n}{1-q^2} F_n(\Omega_n(x)}  \int_{-(n/(1-q^2))^{\delta/3}}^{+(n/(1-q^2))^{\delta/3}} {\rm e}^{- F_n''(\Omega_n(x)) \frac{\tilde s^2}{2}+\mathcal O\left(n^{\delta'( 1+\gamma)-(1-\gamma)/2} \right)} {\rm d} \tilde s,
\end{multline*}
and therefore (where we recall \eqref{eq:fdouble}),
\begin{multline*}
\sqrt{\tfrac{2q^2 n}{(1-q^2)}}\frac{1}{2 \pi {\rm i}}\int_{\Gamma^*\cap B} {\rm e}^{\frac{n}{1-q^2}F_n(w)} {\rm d} w\\= \frac{q}{ \sqrt{\pi|F_n''(\Omega_n(x))|}}  \exp\left(-\eps_n {\rm i}/2+\tfrac{n}{1-q^2}F_n(\Omega_n(x))\right)\left(1+\mathcal O\left(n^{\delta' (1+\gamma)-(1-\gamma)/2} \right)\right),
\end{multline*}
uniformly for $x\in I$ and $\xi \in \mathcal C(U,A,\delta)$ as $n\to \infty$. By setting $\eps = \delta'(1+\gamma) <\delta (1+\gamma)=(1-\gamma)/2$ we see that we get one part of the  right-hand side of the expression for $\phi_0^+$ in  \eqref{eq:estimateintegral1}.  To get the full expression for $\phi_0^+$ in \eqref{eq:estimateintegral1}  it sufficient to prove that the integral over the remaining part of the contour is exponentially small. By \eqref{eq:steepdesc1}, \eqref{eq:steepdesc2} and \eqref{eq:estimateintegralC}  we see that 
\begin{multline*}
\frac{n}{1-q^2}\left(\Re F_n(w)-\Re F_n(\Omega_n(x)\right)\\\leq -\Re F_n''(\Omega_n(x)) \left(\frac{n}{1-q^2}\right)^{2 \delta'/3} +
2^{3/2-\delta} \left(\frac{(1-q^2)}{n}\right)^{-\delta'} \mathcal E_3(x), \end{multline*}
for $w\in (\Gamma^* \cap \mathbb H^+)\setminus B$. Hence by \eqref{eq:fdouble} and \eqref{eq:condone} we have for sufficiently large $n$ that 
\begin{equation*}
\frac{n}{1-q^2}\left(\Re F_n(w)-\Re F_n(\Omega_n(x)\right)\\\leq -\left(\frac{n}{1-q^2}\right)^{2 \delta'/3},
\end{equation*}
uniformly for $w\in  (\Gamma^* \cap \mathbb H^+)\setminus B $ and $x\in I$. By using the latter and \eqref{eq:steepdesc3} it is not difficult to show that the integral over $ (\Gamma^* \cap \mathbb H^+)\setminus B $ only gives an exponentially small contribution and we arrive at the asymptotics for $\phi_0^+$ given in \eqref{eq:estimateintegral1}.
 \end{proof}
 
Note that as a particular consequence of the asymptotic formulas in \eqref{eq:estimateintegral1} we have the following bound. Under the same conditions as in the Lemma, there exists a constant $A>0$ such that for sufficiently large $n$ we have
\begin{equation}\label{eq:boundsonphipm}
\begin{split}|\phi_0^\pm (x)& {\rm e}^{-\frac{n}{1-q^2} \Re F_n(\Omega_n(x))}|\leq A, \\
|\psi_0^\pm (y)& {\rm e}^{\frac{n}{1-q^2} \Re F_n(\Omega_n(y))}|\leq A. 
\end{split}\end{equation}
for $x,y\in I$.

  In the following lemma, we express the asymptotic behavior of the functions $\phi_j$ and $\psi_j$ in terms of $\phi_0^\pm$ and $\psi_0^\pm$. 
 \begin{lemma}\label{lem:asymph}
Let $I \subset U$ be compact. Then for any $\eps>0$  and $j=1,\ldots,n$ we have
\begin{align}\label{eq:estimatephij}
\phi_j(x) &= \sqrt{\tfrac{1-q^2}{2  n q^2}} \left(\frac{\phi_0^+(x)}{\Omega_n(x)-\xi_j^{(n)}}+ \frac{ \phi_0^-(x) }{\overline{\Omega_n(x)}-\xi_j^{(n)}}\right)\\
 &\qquad \qquad \times \left(1+\mathcal O\left(\left(\tfrac{1-q^2}{n}\right)^{1/2-\eps}\frac{1}{|\Omega_n(x)-\xi_j^{(n)}|}\right)\right),\\
 \psi_j(y)&= -\sqrt{\tfrac{1-q^2}{2  n q^2}} \left(\frac{ \psi_0^+(y)}{\Omega_n(y)-\xi_j^{(n)}} +\frac{\psi_0^-(y)}{\overline{\Omega_n(y)}-\xi_j^{(n)}}\right)\\
  &\qquad \qquad \times\left(1+\mathcal O\left(\left(\tfrac{1-q^2}{n}\right)^{1/2-\eps}\frac{1}{|\Omega_n(x)-\xi_j^{(n)}|}\right)\right),\label{eq:estimatepsij}
\end{align}
 as $n\to \infty$, uniformly for $x,y\in I$, $\xi \in \mathcal C_n(U, A,\delta)$
\end{lemma}
\begin{proof} The proof goes by a steepest descent analysis for $\phi_j$ and $\psi_j$, similar to the proof of \eqref{eq:estimateintegral1}. The difference is that there is an extra term  in the integrand of $\phi_j$ and $\psi_j$ that we need to deal with.

We deform the contours $\Sigma $ and $\Gamma$ in the definition of $\phi_j$ and $\psi_j$  to the contours $\Sigma_1^*$ and $\Gamma^*$ and split  the contour into two parts, one in the upper half plane and one in the lower half plane. Hence we have $ \phi_j=\phi_j^++ \phi_j^-$ and $\psi_j=\psi_j^++\psi_j^-$ in the same way as \eqref{eq:splitphi0}.  In the integrand for $\phi_j^+$ we write
\begin{align*}
\frac{1}{w-\xi^{(n)}_j}=\frac{1}{\Omega_n(x)-\xi_j^{(n)}} \left(1+\frac{\Omega_n(x)-w}{w-\xi^{(n)}_j} \right)\end{align*}
and when we introduce the local variables \eqref{eq:localvariable} we estimate
$$
\left|\frac{ \Omega_n(x)-w  }{w-\xi_j^{(n)}}\right|= \frac{\left( \tfrac{1-q^2}{n}\right)^{1/2-\delta'/3} }{\left|\Omega_n(x)-\xi^{(n)}_j+\sqrt{\frac{1-q^2}{n}}{\rm i}s\right|}\leq \left(\tfrac{1-q^2}{n}\right)^{1/2-\delta'/3} \frac{ 2}{\left|\Omega_n(x)-\xi^{(n)}_j\right|},
$$
for $n$ sufficiently large. Analogous estimates hold for the integrals $\phi^-$, $\psi^+$ and $\psi^-$. After this estimate the arguments are the same as in the proof of \eqref{eq:estimateintegral1} and we obtain the statement after setting $\eps=\delta'/3$.
\end{proof}
By combining \eqref{eq:estimatephij} and \eqref{eq:estimatepsij} with \eqref{eq:boundsonphipm} 
we see that, under the same conditions of the last  lemma,  there exists a constant $A>0$ such that for sufficiently large $n$ we have
\begin{equation}\label{eq:boundsonphipm2}
\begin{split}|\phi_j^\pm (x)& {\rm e}^{-\frac{n}{1-q^2} \Re F_n(\Omega_n(x))}|\leq \frac{1-q^2}{2n q^2} \frac{A}{| \Omega_n(x)-\xi_j^{(n)}|}, \\
|\psi_0^\pm (y)& {\rm e}^{\frac{n}{1-q^2} \Re F_n(\Omega_n(y))}|\leq \frac{1-q^2}{2n q^2} \frac{A}{| \Omega_n(x)-\xi_j^{(n)}|}. 
\end{split}\end{equation}
for $x,y\in I$.
 
 \subsection{Proof of Lemma \ref{lem:asymKnpre}}\label{subsec:asymKn}
 
We are now ready to  compute the asymptotic behavior of $K_n^I$.
\begin{lemma}\label{lem:asymKn}
Let $I \subset U$ be compact. Then 
 \begin{multline}\label{eq:asymptoticsKnmain}
 (x-y) K_n(x,y)=\frac{1}{q}\left(\frac{ \phi_0^+(x)\psi_0^+(y)(x-y)}{\Omega_n(x)-\Omega_n(y)}+\frac{\phi_0^-(x) \psi_0^+(y)(x-y)}{\overline{\Omega_n(x)}-\Omega_n(y)}\right.\\
\left.+\frac{ \phi_0^+(x) \psi_0^-(y)(x-y)}{\Omega_n(x)-\overline{\Omega_n(y)}}
+\frac{ \phi_0^-(x) \psi_0^-(y)(x-y)}{\overline{\Omega_n(x)}-\overline{\Omega_n(y)}}
\right)\left(1+o(1)\right),
 \end{multline}
uniformly for $x,y\in I$, $\xi \in \mathcal C(U,A,\delta)$  as $n\to \infty$. 
\end{lemma}

\begin{proof}
We recall that by \eqref{eq:kernelintegrable} we have that 
$$(x-y)K_n(x,y)=\sum_{j=0}^n\phi_j(x)\psi_j(y).$$
To prove the statement we insert the asymptotic espressions \eqref{eq:estimatephij} andÊ \eqref{eq:estimatepsij} into the latter expression. Let us first replace $\phi_j$ and $\psi_j$ with the leading asymptotic terms in \eqref{eq:estimatephij} and \eqref{eq:estimatepsij}  and ignore the correction term. This leads to the following sum
\begin{equation}\label{eq:hoopgedoe}
\phi_0(x)\psi_0(y) - {\tfrac{1-q^2}{2  n q^2}}  \sum_{j=1}^n\left(\frac{\phi_0^+(x)}{\Omega_n(x)-\xi_j^{(n)}}+ \frac{ \phi_0^-(x) }{\overline{\Omega_n(x)}-\xi_j^{(n)}}\right)
\left(\frac{ \psi_0^+(y)}{\Omega_n(y)-\xi_j^{(n)}} +\frac{\psi_0^-(y)}{\overline{\Omega_n(y)}-\xi_j^{(n)}}\right).
\end{equation}
To compute this sum we first note that 
since  $F_n'(\Omega_n(x))=0$  we have 
\begin{equation*}
\frac{1-q^2}{n} \sum_{j=1}^n \frac{1}{\Omega_n(x)-\xi^{(n)}_j} =2q(x-q\Omega_n(x)),
\end{equation*}
and hence
\begin{multline*}
\frac{1-q^2}{n}\sum_{j=1}^n \frac{1}{\Omega_n(x)-\xi_j^{(n)}} \frac{1}{\Omega_n(y)-\xi_j^{(n)}}\\=\frac{1-q^2}{n}\frac{1}{\Omega_n(x)-\Omega_n(y)}\sum_{j=1}^n \left(\frac{1}{\Omega_n(y)-\xi_j^{(n)}} -\frac{1}{\Omega_n(x)-\xi_j^{(n)}}\right)
\\=\frac{2 q^2 (\Omega_n(x) -\Omega_n(y))-2q  (x-y) }{\Omega_n(x)-\Omega_n(y)}=2q^2 -\frac{2q  (x-y)}{\Omega_n(x)-\Omega_n(y)}.
\end{multline*}
Therefore we also have \begin{equation}\nonumber
\phi_0^+(x) \psi^+_0(y) -\frac{1-q^2}{2q^2 n}\sum_{j=1}^n \frac{\phi_0^+(x)}{\Omega_n(x)-\xi_j^{(n)}} \frac{\psi_0^+(y)}{\Omega_n(y)-\xi_j^{(n)}}=
\frac{\phi_0^+(x)\psi_0^+(y)  (x-y)}{q(\Omega_n(x)-\Omega_n(y))}.
\end{equation}
The same expressions hold with $\Omega_n(x)$ and/or   $\Omega_n(y)$ replaced with their complex conjugates. By substituting these expression into \eqref{eq:hoopgedoe} (and using $\phi_0=\phi_0^++\phi_0^-$ and $\psi_0=\psi_0^++\psi_0^-$ ) we obtain \begin{multline}\nonumber
\phi_0(x)\psi_0(y) - {\tfrac{1-q^2}{2  n q^2}}  \sum_{j=1}^n\left(\frac{\phi_0^+(x)}{\Omega_n(x)-\xi_j^{(n)}}+ \frac{ \phi_0^-(x) }{\overline{\Omega_n(x)}-\xi_j^{(n)}}\right)
\left(\frac{ \psi_0^+(y)}{\Omega_n(y)-\xi_j^{(n)}} +\frac{\psi_0^-(y)}{\overline{\Omega_n(y)}-\xi_j^{(n)}}\right)\\=\frac{1}{q}\left(\frac{ \phi_0^+(x)\psi_0^+(y)(x-y)}{\Omega_n(x)-\Omega_n(y)}+\frac{\phi_0^-(x) \psi_0^+(y)(x-y)}{\overline{\Omega_n(x)}-\Omega_n(y)}\right.\\
\left.+\frac{ \phi_0^+(x) \psi_0^-(y)(x-y)}{\Omega_n(x)-\overline{\Omega_n(y)}}
+\frac{ \phi_0^-(x) \psi_0^-(y)(x-y)}{\overline{\Omega_n(x)}-\overline{\Omega_n(y)}}
\right),
\end{multline}
which is the leading term in \eqref{eq:asymptoticsKnmain}.

It only remains to prove that the correction terms in \eqref{eq:estimatephij} and \eqref{eq:estimatepsij} do not contribute to the leading asymptotics in \eqref{eq:asymptoticsKnmain}. But this follows by collecting all the error terms and using \eqref{eq:boundsonphipm2} and \eqref{eq:condone} for $\mathcal E_3$.  
\end{proof}

We are now ready to prove Lemma \ref{lem:asymKnpre}. 
\begin{proof}[Proof of Lemma \ref{lem:asymKnpre}]
This follows directly from substituting \eqref{eq:estimateintegral1} in Lemma \ref{lem:asymKn} and using \eqref{eq:fdouble} and the fact that $q\to 1$ as $n\to \infty$. 
\end{proof}

\subsection{Asymptotics for $K_n(x,x)$}

In the next step we compute that asymptotic behavior of the diagonal of  the kernel $K_n(x,x)$ as $n \to \infty$. The proof is again based on standard principles of the steepest descent techniques. 
\begin{lemma}\label{lem:asymptoticsdiagonal}
 Let $I \subset U$ be compact. Then for any $0<\eps<(1-\gamma)/2$, we have
\begin{equation}\label{eq:asymptoticsdiagonal}
K_n(x,x)=\frac{2 n q}{(1-q^2)\pi } \Im \Omega_n(x) +\mathcal O\left(\frac{n^\eps}{1-q^2}\right) , 
 \end{equation}
uniformly for $x\in I$, $\xi^{(n)} \in \mathcal C_n(U, n^\delta)$  as $n\to \infty$.
\end{lemma}  
\begin{proof}
The proof goes again by steepest descent arguments, but now we start from the double integral formula \eqref{eq:defKn2} directly. We deform the contours $\Sigma$ and $\Gamma$ to the contours $\Sigma^*$ and $\Gamma^*$ in \eqref{eq:contoursteepasc} and \eqref{eq:contoursteepdesc}, consisting of path of steep descent/ascent leaving from the saddle points $\Omega_n(x)$ and $\overline{\Omega_n(x)}$.  When doing so, we pick up a residue due to the term $1/(w-z)$ and we obtain an additional single integral
\begin{multline} \label{eq:doublediag}
K_n(x,x)=\frac{2 q n}{(1-q^2)2 \pi {\rm i}  } \int_{\overline \Omega_n(x)}^{\Omega_n(x)}  {\rm d} z
\\+ \frac{2 q n}{(1-q^2)(2 \pi {\rm i})^2} \oint_{\Sigma^*} {\rm d}z \int_{\Gamma^*} {\rm d}w  \frac{ {\rm e}^{\frac{n} {1-q^2}F_n(w;x)}}{ {\rm e}^{\frac{n} {1-q^2}F_n(z;x)}} \frac{1}{w-z}. 
\end{multline} 
Note that in the double integral $\Sigma^*$ and $\Gamma^*$ intersect at $\Omega_n(x)$ and $\overline{\Omega_n(x)}$ and  the integrand has a singularity at that point. As an iterated integral the first integral as a principal value integral (since the contours are perpendicular at the intersection points, the singularity is integrable with respect to the double integral).

The single integral in \eqref{eq:doublediag} is trivial to compute and gives 
\begin{equation}\nonumber
\frac{2 q n}{(1-q^2)2 \pi {\rm i}  } \int_{\overline \Omega_n(x)}^{\Omega_n(x)}  {\rm d} z=\frac{2 q n }{(1-q^2) \pi } \Im \Omega_n(x),
\end{equation}
and this gives the first term in \eqref{eq:asymptoticsdiagonal}.  We split the double integral in \eqref{eq:doublediag}  into three parts, two double integrals where $z,w$ are both  in the neighborhood of $\Omega_n(x)$ and $\overline{\Omega_n(x)}$ and a third integral over the remaining parts. By introducing local variables and reasoning as in the proof of Lemma \ref{lem:boundonphi0} (in particular \eqref{eq:estimateintegralC})  we find  for any $0<\delta'<\delta$ and some $c>0$ that
\begin{multline*}
\frac{2 q n}{(1-q^2)(2 \pi {\rm i})^2} \oint_{\Sigma^*} {\rm d}z \int_{\Gamma^*} {\rm d}w  \frac{ {\rm e}^{\frac{n} {1-q^2}F_n(w;x)}}{ {\rm e}^{\frac{n} {1-q^2}F_n(z;x)}} \frac{1}{w-z}\\
=-\frac{2 q  {\rm i}}{(2 \pi {\rm i})^2}  \sqrt{\frac{n}{1-q^2}} \int_{-\infty}^{\infty} {\rm d}s \int_{-\infty}^{\infty} {\rm d}t   \frac{ {\rm e}^{-F_n''(\Omega_n(x)) \frac{1}{2}(s^2+t^2) +\mathcal O\left(\frac{n^{\delta'(1+\gamma)}}{\sqrt{(1-q^2)n}}\right)}}{s+ {\rm i}t}\\
   +\frac{2 q  {\rm i}}{(2 \pi {\rm i})^2} \sqrt{\frac{n}{1-q^2}} \int_{-\infty}^{\infty} {\rm d}s \int_{-\infty}^{\infty} {\rm d}t   \frac{ {\rm e}^{-F_n''(\overline{\Omega_n(x)}) \frac{1}{2}(s^2+t^2) +\mathcal O\left(\frac{n^{\delta'(1+\gamma)}}{\sqrt{(1-q^2)n}}\right)}}{s+ {\rm i}t}  +\mathcal O\left(\exp(-n^{c\delta'})\right),
\end{multline*}
as $n\to \infty$ uniformly for $x\to I$.  After changing to polar coordinates, it is not hard to see that 
$$
 \int_{-\infty}^{\infty} {\rm d}s \int_{-\infty}^{\infty} {\rm d}t   \frac{ {\rm e}^{-F_n''(\overline{\Omega_n(x)}) \frac{1}{2}(s^2+t^2)}}{s+ {\rm i}t} =0,$$
 and hence, for any $0<\delta'<\delta$, we have  
 \begin{multline*}
 \frac{2 q n}{(1-q^2)(2 \pi {\rm i})^2} \oint_{\Sigma^*} {\rm d}z \int_{\Gamma^*} {\rm d}w  \frac{ {\rm e}^{\frac{n} {1-q^2}F_n(w;x)}}{ {\rm e}^{\frac{n} {1-q^2}F_n(z;x)}} \frac{1}{w-z}\\
 =\mathcal O\left(\sqrt{\frac{n}{1-q^2}}\frac{n^{\delta'(1+\gamma) }}{\sqrt{(1-q^2)n}}\right)= \mathcal O\left(\frac{n^{\delta'(1+\gamma)}}{1-q^2}\right),
 \end{multline*}
 as $ n \to \infty$. After setting $\eps= \delta'(1+\gamma)< \delta(1+\gamma)=(1-\gamma)/2$  we arrive at the statement.
\end{proof}
Combining the latter with Lemma \ref{lem:saddle} we obtain the following.
\begin{corollary}\label{cor:asymptoticsdiagonal}
 Let $I \subset U$ be compact. Then
\begin{equation}\label{eq:asymptoticsdiagonal2}
\frac{1}{n}K_n(x,x)=\frac{\sqrt{2-x^2} }{\pi } +o(1), 
 \end{equation}
 as $n\to \infty$, uniformly for $x\in I$ and $\xi \in \mathcal C(U,A,\delta)$. \end{corollary}
\begin{proof}
Since for $\delta>0$ small enough we have $n (1-q^2)/n^\delta\to \infty$ as $n\to \infty$,  we obtain
\begin{equation*}
\frac{1}{n}K_n(x,x)=\frac{2  q}{(1-q^2)\pi } \Im \Omega_n(x) +o(1), 
 \end{equation*}
 as $n\to \infty$. By  Lemma \ref{lem:saddle} we can replace $\Omega_n(x)$ by $\Omega(x)$ and, finally, since $\Im \Omega(x)= ((1-q^2)/2q) \sqrt{2-x^2}$ we obtain the statement.  
\end{proof}

\subsection{Asymptotics for $R^I_n$}
Besides he asymptotic behavior of $K_n$ we will also need the asymptotic behavior of $R_n^I$ defined by 
\begin{equation}\label{eq:defRn}
R_n^I(x,y)=\int_I K_n(x,z) K_n(z,y) {\rm d} z -K_n(x,y), 
\end{equation}
for $x,y\in I$. Here $I$ is an interval in $(-\sqrt 2,\sqrt 2)$.   Before we come to its asymptotic behavior we will first give a quadruple integral expression for $R^I_n$.
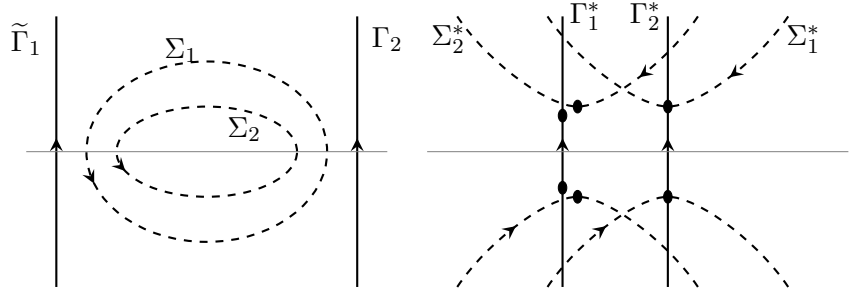
\begin{figure}
\begin{center}
\begin{tikzpicture}[xscale=.4,yscale=.6,decoration={%
   markings,%
   mark=at position .55 with {\arrow[scale=1.2,black]{stealth};}}]
\draw[ thick,-,postaction=decorate] (-5,-3)--(-5,3);
\draw[thick,-,postaction=decorate] (5,-3)--(5,3);
\draw[thick, dashed,postaction=decorate] (0,0) ellipse (4 and 2);
\draw[thick, dashed,postaction=decorate] (0,0) ellipse (3 and 1);
\draw[help lines] (-6,0)--(6,0);
\draw (6,2.5) node {$\Gamma_2$};
\draw (-6,2.5) node {$\widetilde \Gamma_1$};
\draw (-.85,2.25) node {$\Sigma_1$};
\draw (1.25,.5) node {$\Sigma_2$};
\end{tikzpicture}
\begin{tikzpicture}[xscale=.4,yscale=.6,decoration={%
   markings,%
   mark=at position .55 with {\arrow[scale=1.2,black]{stealth};}}]
\draw[thick,-,postaction=decorate] (-1.5,-3)--(-1.5,3);
\draw[thick,-,postaction=decorate] (2,-3)--(2,3);
\draw[thick, dashed] (-5,3) .. controls (-4,2) and (-2,1) .. (-1,1);
\draw[thick, dashed,postaction=decorate] (3,3) .. controls (2,2) and (0,1) .. (-1,1);
\draw[thick, dashed,postaction=decorate] (-5,-3) .. controls (-4,-2) and (-2,-1) .. (-1,-1);
\draw[thick, dashed] (3,-3) .. controls (2,-2) and (0,-1) .. (-1,-1);
\draw[thick, dashed] (-2,3) .. controls (-1,2) and (1,1) .. (2,1);
\draw[thick, dashed,postaction=decorate] (6,3) .. controls (5,2) and (3,1) .. (2,1);
\draw[thick, dashed,postaction=decorate] (-2,-3) .. controls (-1,-2) and (1,-1) .. (2,-1);
\draw[thick, dashed] (6,-3) .. controls (5,-2) and (3,-1) .. (2,-1);
\filldraw (-1.5,0.8) circle(4pt);
\filldraw (-1.5,-0.8) circle(4pt);
\filldraw (-1,1) circle(4pt);
\filldraw (-1,-1) circle(4pt);
\filldraw (2,1) circle(4pt);
\draw (6.5,2.5) node {$\Sigma_1^*$};
\draw (-5.3,2.5) node {$\Sigma_2^*$};
\draw (-.75,3) node {$\Gamma_1^*$};
\draw (1.25,3) node {$\Gamma_2^*$};
\filldraw (2,-1) circle(4pt);
\draw[help lines] (-6,0)--(8,0);
\end{tikzpicture}
\caption{The left picture shows the contours in the first  quadruple integral at the right hand side of \eqref{eq:Rnquadruple}, before deforming the contours. In the right picture the corresponding contours of steep descent and ascent are shown.}
\label{fig:step0}
\end{center}
\end{figure}

\begin{lemma} Let $I=[E_1,E_2] \subset \R$.   Then for $x,y\in I$ we have
\begin{multline}\label{eq:Rnquadruple}
R^I_n(x,y)\\=\frac{2 q n}{(1-q^2) (2\pi {\rm i})^4} \oint_{\Sigma_1}  \int_{\widetilde \Gamma_1}    \oint_{\Sigma_2}    \int_{\Gamma_2} 
\frac{{\rm e}^{\frac{n}{1-q^2}\left( F_n(w_1;x) + F_n(w_2;E_2)\right) }}{{\rm e}^{\frac{n}{1-q^2}\left( F_n(z_1;E_2) + F_n(z_2;y) \right)}}
\frac{  {\rm d}w_2{\rm d}z_2{\rm d}w_1 {\rm d}z_1}{(w_1-z_1)(w_2-z_2)(z_1-w_2)}
\\
-\frac{2q n}{(1-q^2) (2\pi {\rm i})^4} \oint_{\Sigma_1} \int_{\Gamma_1} \oint_{\Sigma_2}  \int_{\widetilde \Gamma_2}  \frac{{\rm e}^{\frac{n}{1-q^2}\left( F_n(w_1;x) + F_n(w_2;E_1)\right) }}{{\rm e}^{\frac{n}{1-q^2}\left( F_n(z_1;E_1) + F_n(z_2;y) \right)}}\frac{ {\rm d}w_2
{\rm d}z_2   {\rm d}w_1 {\rm d}z_1 }{(w_1-z_1)(w_2-z_2)(z_1-w_2)}
\end{multline} 
Here $\Sigma_1$, $\Sigma_2$,  $\Gamma_1$ and $\Gamma_2$ are as before with the additional restriction that $\Gamma_2$ is also to the right of $\Sigma_1$. The contours $\widetilde \Gamma_1$ and   $\widetilde \Gamma_2$ connect $-{\rm i} \infty$ to $+{\rm i} \infty$ but are at the left of both $\Sigma_1$ and $\Sigma_2$.
\end{lemma}
\begin{proof}
We define $K_n^I$ as the restriction of $K_n$ to $I\times I$, hence $ K_n^I(x,y)=K_n(x,y)$ if $x,y\in I$ and $K_n^I(x,y)=0$ otherwise. Then by inserting the double integral formula $\eqref{eq:defKn2}$ into 
\begin{align}\nonumber
(K^I_n)^2(x,y)=\int_\R K_n(x,u)^I K_n^I (u,y)Ê{\rm d}u=\int_I K_n(x,u)K_n (u,y)Ê{\rm d}u,
\end{align}
and switching the order of integration we obtain 
\begin{multline}\label{eq:RIA}
(K^I_n)^2(x,y)\\=\frac{2 q n}{(1-q^2) (2\pi {\rm i})^4} \oint_{\Sigma_1}  \int_{\Gamma_1}    \oint_{\Sigma_2}    \int_{\Gamma_2} 
\frac{{\rm e}^{\frac{n}{1-q^2} F_n(w_1;x) }{\rm e}^{\frac{n}{1-q^2} F_n(w_2;E_2) }}{{\rm e}^{\frac{n}{1-q^2} F_n(z_1;E_2) } {\rm e}^{\frac{n}{1-q^2} F_n(z_2;y) }}
\frac{  {\rm d}w_2{\rm d}z_2{\rm d}w_1 {\rm d}z_1}{(w_1-z_1)(w_2-z_2)(z_1-w_2)}
\\
-\frac{2 q n}{(1-q^2) (2\pi {\rm i})^4} \oint_{\Sigma_1} \int_{\Gamma_1} \oint_{\Sigma_2}  \int_{ \Gamma_2}  \frac{{\rm e}^{\frac{n}{1-q^2} F_n(w_1;x) }{\rm e}^{\frac{n}{1-q^2} F_n(w_2;E_1) }}{{\rm e}^{\frac{n}{1-q^2} F_n(z_1;E_1) } {\rm e}^{\frac{n}{1-q^2} F_n(z_2;y) }}
\frac{ {\rm d}w_2
{\rm d}z_2   {\rm d}w_1 {\rm d}z_1 }{(w_1-z_1)(w_2-z_2)(z_1-w_2)}.
\end{multline} 
We take $\Sigma_1$  around $\Sigma_2$. In the second quadruple integral in \eqref{eq:RIA} we deform $\Gamma_2$ to $\widetilde \Gamma_2$ that lies at the left of $\Sigma_1$ and $\Sigma_2$. By doing so, we pick up a residue at $w_2=z_2$ and $w_2=z_1$ which gives two additional triple integrals. However, in the triple integral arising by computing the residue at $w_2=z_2$ the integral over $\Sigma_2$ vanishes as there are no poles inside $\Sigma_2$ (since we assumed that $\Sigma_1$ goes around $\Sigma_2)$. In the same way,  in the triple integral arising after computing the residue at $w_2=z_1$ the integral over $\Sigma_1$ has a simple pole at $z_1=z_2$ and hence the triple integral reduces to a double integral.  Concluding we have 
\begin{multline*}
(K^I_n)^2(x,y)\\=\frac{2 q n}{(1-q^2) (2\pi {\rm i})^4} \oint_{\Sigma_1}  \int_{\Gamma_1}    \oint_{\Sigma_2}    \int_{\Gamma_2} 
\frac{{\rm e}^{\frac{n}{1-q^2} F_n(w_1;x) }{\rm e}^{\frac{n}{1-q^2} F_n(w_2;E_2) }}{{\rm e}^{\frac{n}{1-q^2} F_n(z_1;E_2) } {\rm e}^{\frac{n}{1-q^2} F_n(z_2;y) }}
\frac{  {\rm d}w_2{\rm d}z_2{\rm d}w_1 {\rm d}z_1}{(w_1-z_1)(w_2-z_2)(z_1-w_2)}
\\
-\frac{2  q n}{(1-q^2) (2\pi {\rm i})^4} \oint_{\Sigma_1} \int_{\Gamma_1} \oint_{\Sigma_2}  \int_{\widetilde \Gamma_2}  \frac{{\rm e}^{\frac{n}{1-q^2} F_n(w_1;x) }{\rm e}^{\frac{n}{1-q^2} F_n(w_2;E_1) }}{{\rm e}^{\frac{n}{1-q^2} F_n(z_1;E_1) } {\rm e}^{\frac{n}{1-q^2} F_n(z_2;y) }}
\frac{ {\rm d}w_2
{\rm d}z_2   {\rm d}w_1 {\rm d}z_1 }{(w_1-z_1)(w_2-z_2)(z_1-w_2)}
\\+\frac{2q n}{(1-q^2)(2 \pi {\rm i})^2} \oint_{\Sigma_2} {\rm d}z_2 \int_{\Gamma_1} {\rm d}w_1  \frac{ {\rm e}^{\frac{n} {1-q^2}F_n(w_1;x)}}{ {\rm e}^{\frac{n} {1-q^2}F_n(z_2;y)}} \frac{1}{w_1-z_2}.
\end{multline*}
In the latter, by \eqref{eq:defKn2} the double integral equals $K_n(x,y)=K^I_n(x,y)$. Moreover, in the first quadruple integral we can deform $\Gamma_1$ to $\widetilde \Gamma_1$. By doing so, we pick up a residue at $w_1=z_1$ which gives rise to a triple integral. However,  as the integral over $\Sigma_1$ does not encircle any pole it vanishes and we proved the statement.
\end{proof}

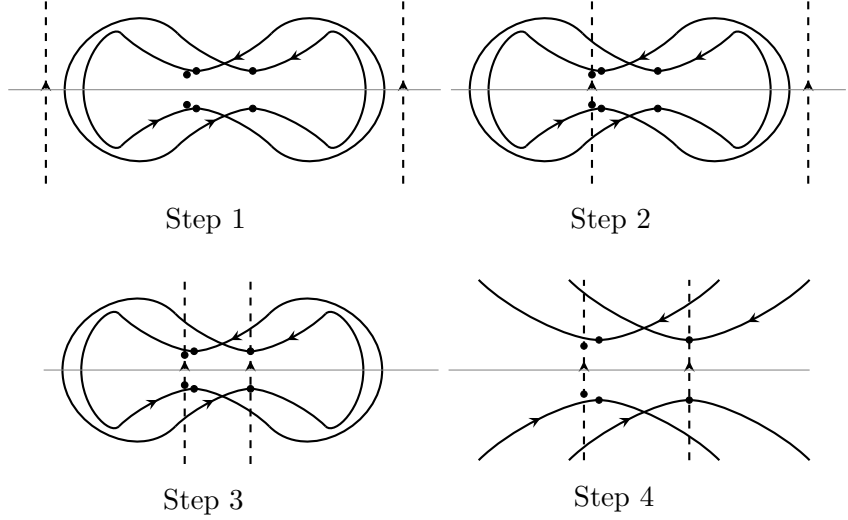
\begin{figure}
\begin{center}
\begin{tikzpicture}[scale=.25,decoration={%
   markings,%
   mark=at position .55 with {\arrow[scale=1,black]{stealth};}}]
\draw[thick,-,postaction=decorate, dashed] (-9,-5)--(-9,5);
\draw[thick,-,postaction=decorate, dashed] (10,-5)--(10,5);
\draw[thick] (-5,3) .. controls (-4,2) and (-2,1) .. (-1,1);
\draw[thick,postaction=decorate] (3,3) .. controls (2,2) and (0,1) .. (-1,1);
\draw[thick,postaction=decorate] (-5,-3) .. controls (-4,-2) and (-2,-1) .. (-1,-1);
\draw[thick] (3,-3) .. controls (2,-2) and (0,-1) .. (-1,-1);
\draw[thick] (-2,3) .. controls (-1,2) and (1,1) .. (2,1);
\draw[thick,postaction=decorate] (6,3) .. controls (5,2) and (3,1) .. (2,1);
\draw[thick,postaction=decorate] (-2,-3) .. controls (-1,-2) and (1,-1) .. (2,-1);
\draw[thick] (6,-3) .. controls (5,-2) and (3,-1) .. (2,-1);
\filldraw (-1.5,0.8) circle(5pt);
\filldraw (-1.5,-0.8) circle(5pt);
\filldraw (-1,1) circle(5pt);
\filldraw (-1,-1) circle(5pt);
\filldraw (2,1) circle(5pt);
\filldraw (2,-1) circle(5pt);
\draw[thick] (-2,3) .. controls (-4,5) and (-8,3) .. (-8,0);
\draw[thick] (-2,-3) .. controls (-4,-5) and (-8,-3) .. (-8,0);
\draw[thick] (3,3) .. controls (5,5) and (9,3) .. (9,0);
\draw[thick] (3,-3) .. controls (5,-5) and (9,-3) .. (9,0);
\draw[thick] (-5,3) .. controls (-5.5,3.5) and (-7,2) .. (-7,0);
\draw[thick] (-5,-3) .. controls (-5.5,-3.5) and (-7,-2) .. (-7,0);
\draw[thick] (6,3) .. controls (6.5,3.5) and (8,2) .. (8,0);
\draw[thick] (6,-3) .. controls (6.5,-3.5) and (8,-2) .. (8,0);
\draw[help lines] (-11,0)--(12,0);
\draw (-0.5,-7) node {Step 1};
\end{tikzpicture}
\begin{tikzpicture}[scale=.25,decoration={%
   markings,%
   mark=at position .55 with {\arrow[scale=1,black]{stealth};}}]
\draw[thick,-,postaction=decorate, dashed, dashed] (-1.5,-5)--(-1.5,5);
\draw[thick,-,postaction=decorate, dashed, dashed] (10,-5)--(10,5);
\draw[thick] (-5,3) .. controls (-4,2) and (-2,1) .. (-1,1);
\draw[thick,postaction=decorate] (3,3) .. controls (2,2) and (0,1) .. (-1,1);
\draw[thick,postaction=decorate] (-5,-3) .. controls (-4,-2) and (-2,-1) .. (-1,-1);
\draw[thick] (3,-3) .. controls (2,-2) and (0,-1) .. (-1,-1);
\draw[thick] (-2,3) .. controls (-1,2) and (1,1) .. (2,1);
\draw[thick,postaction=decorate] (6,3) .. controls (5,2) and (3,1) .. (2,1);
\draw[thick,postaction=decorate] (-2,-3) .. controls (-1,-2) and (1,-1) .. (2,-1);
\draw[thick] (6,-3) .. controls (5,-2) and (3,-1) .. (2,-1);
\filldraw (-1.5,0.8) circle(5pt);
\filldraw (-1.5,-0.8) circle(5pt);
\filldraw (-1,1) circle(5pt);
\filldraw (-1,-1) circle(5pt);
\filldraw (2,1) circle(5pt);
\filldraw (2,-1) circle(5pt);
\draw[thick] (-2,3) .. controls (-4,5) and (-8,3) .. (-8,0);
\draw[thick] (-2,-3) .. controls (-4,-5) and (-8,-3) .. (-8,0);
\draw[thick] (3,3) .. controls (5,5) and (9,3) .. (9,0);
\draw[thick] (3,-3) .. controls (5,-5) and (9,-3) .. (9,0);
\draw[thick] (-5,3) .. controls (-5.5,3.5) and (-7,2) .. (-7,0);
\draw[thick] (-5,-3) .. controls (-5.5,-3.5) and (-7,-2) .. (-7,0);
\draw[thick] (6,3) .. controls (6.5,3.5) and (8,2) .. (8,0);
\draw[thick] (6,-3) .. controls (6.5,-3.5) and (8,-2) .. (8,0);
\draw[help lines] (-9,0)--(12,0);
\draw (-0.5,-7) node {Step 2};
\end{tikzpicture}
\begin{tikzpicture}[scale=.25,decoration={%
   markings,%
   mark=at position .55 with {\arrow[scale=1,black]{stealth};}}]
\draw[thick,-,postaction=decorate, dashed] (-1.5,-5)--(-1.5,5);
\draw[thick,-,postaction=decorate, dashed] (2,-5)--(2,5);
\draw[thick] (-5,3) .. controls (-4,2) and (-2,1) .. (-1,1);
\draw[thick,postaction=decorate] (3,3) .. controls (2,2) and (0,1) .. (-1,1);
\draw[thick,postaction=decorate] (-5,-3) .. controls (-4,-2) and (-2,-1) .. (-1,-1);
\draw[thick] (3,-3) .. controls (2,-2) and (0,-1) .. (-1,-1);
\draw[thick] (-2,3) .. controls (-1,2) and (1,1) .. (2,1);
\draw[thick,postaction=decorate] (6,3) .. controls (5,2) and (3,1) .. (2,1);
\draw[thick,postaction=decorate] (-2,-3) .. controls (-1,-2) and (1,-1) .. (2,-1);
\draw[thick] (6,-3) .. controls (5,-2) and (3,-1) .. (2,-1);
\filldraw (-1.5,0.8) circle(5pt);
\filldraw (-1.5,-0.8) circle(5pt);
\filldraw (-1,1) circle(5pt);
\filldraw (-1,-1) circle(5pt);
\filldraw (2,1) circle(5pt);
\filldraw (2,-1) circle(5pt);
\draw[thick] (-2,3) .. controls (-4,5) and (-8,3) .. (-8,0);
\draw[thick] (-2,-3) .. controls (-4,-5) and (-8,-3) .. (-8,0);
\draw[thick] (3,3) .. controls (5,5) and (9,3) .. (9,0);
\draw[thick] (3,-3) .. controls (5,-5) and (9,-3) .. (9,0);
\draw[thick] (-5,3) .. controls (-5.5,3.5) and (-7,2) .. (-7,0);
\draw[thick] (-5,-3) .. controls (-5.5,-3.5) and (-7,-2) .. (-7,0);
\draw[thick] (6,3) .. controls (6.5,3.5) and (8,2) .. (8,0);
\draw[thick] (6,-3) .. controls (6.5,-3.5) and (8,-2) .. (8,0);
\draw[help lines] (-9,0)--(12,0);
\draw (-0.5,-7) node {Step 3};
\draw[white] (0.5,6) node {.};
\end{tikzpicture}
\begin{tikzpicture}[scale=.4,decoration={%
   markings,%
   mark=at position .55 with {\arrow[scale=1,black]{stealth};}}]
\draw[thick,-,postaction=decorate, dashed] (-1.5,-3)--(-1.5,3);
\draw[thick,-,postaction=decorate, dashed] (2,-3)--(2,3);
\draw[thick] (-5,3) .. controls (-4,2) and (-2,1) .. (-1,1);
\draw[thick,postaction=decorate] (3,3) .. controls (2,2) and (0,1) .. (-1,1);
\draw[thick,postaction=decorate] (-5,-3) .. controls (-4,-2) and (-2,-1) .. (-1,-1);
\draw[thick] (3,-3) .. controls (2,-2) and (0,-1) .. (-1,-1);
\draw[thick] (-2,3) .. controls (-1,2) and (1,1) .. (2,1);
\draw[thick,postaction=decorate] (6,3) .. controls (5,2) and (3,1) .. (2,1);
\draw[thick,postaction=decorate] (-2,-3) .. controls (-1,-2) and (1,-1) .. (2,-1);
\draw[thick] (6,-3) .. controls (5,-2) and (3,-1) .. (2,-1);
\filldraw (-1.5,0.8) circle(3pt);
\filldraw (-1.5,-0.8) circle(3pt);
\filldraw (-1,1) circle(3pt);
\filldraw (-1,-1) circle(3pt);
\filldraw (2,1) circle(3pt);
\filldraw (2,-1) circle(3pt);
\draw[help lines] (-6,0)--(6,0);
\draw (-0.5,-4.3) node {Step 4};
\draw[white] (0.5,-4.75) node {.};
\end{tikzpicture}
\end{center}
\caption{Overview in the various contour deformations in the proof of Lemma \ref{lem:boundR}.} \label{fig:contourdefqua}\end{figure}

As $n\to \infty$ the kernel $R_n^I$ has the following asymptotic behavior. 

\begin{lemma}\label{lem:boundR}
 Let $I=[E_1,E_2] \subset U$. Then 
\begin{align}\label{eq:boundR}
\tfrac{\exp\left(\frac{n}{1-q^2}\Re F_n(\Omega_n(y);y) \right)}{\exp\left( \frac{n}{1-q^2}\Re F_n(\Omega_n(x);x)\right)} R^I_n(x,y)=A_n(x,y) +\mathcal O\left( \left(\frac{1-q^2}{n}\right)^{1/2}\right),
\end{align} 
as $n\to \infty$, uniformly for $x,y$ in compact subsets of the interval $(E_1,E_2)$ and $\xi\in \mathcal C(U,A,\delta)$. Here $A_n$ is a kernel that has the form
\begin{align}\label{eq:defAn}
A_n(x,y)=\sum_{j=1}^4 f_j(x) g_j(y) A_{n,j}(x,y),\end{align}
where $f_j$, $g_j$ and $A_{n,j}$ are functions that are uniformly bounded in $n\in \N$ and 1-Lipschitz in both variables with constants that do not depend on $n\in \N$. More precisely, for each compact set $I'$ of the open interval $(E_1,E_2)$,  there exists a constant $a>0$ such that  for $n\in \N$ we have
\begin{enumerate}
\item $|A_n(x,y)| \leq a$ for $x,y\in I' $.
\item $|A_n(x_1,y)-A_n(x_2,y)|/|x_1-x_2|\leq a$ for $x_1,x_2,y \in I' $.
\item $|A_n(x,y_1)-A_n(x,y_2)|/|y_1-y_2|\leq a$ for $ x,y_1,y_2 \in I' $.
\end{enumerate}
\end{lemma}

\begin{proof}
The proof goes by  a steepest descent analysis on \eqref{eq:Rnquadruple}. We deform the contours to paths of steep descent and ascent. By doing so we will pick up residues. To keep an overview over the various terms we get, we will carry the analysis out in several steps. We will only deal with the first quadruple integral on the right-hand side of \eqref{eq:Rnquadruple}. The analysis for the second quadruple integral can be dealt with using similar arguments and we leave that case to the reader. 

Step 0. Let us start with describing the initial situation. Our starting point is the first quadruple integral on the right-hand side of \eqref{eq:Rnquadruple} and hence the configuration of the contours is as in the left picture of Figure \ref{fig:step0}. Our goal is to deform the contours to paths of steep descent/ascent so that we arrive at the configuration of contours in the right  picture of Figure \ref{fig:step0}. 

It important to observe that $\Omega_n(x)$ and $\Omega_n(y)$ can not be close to $\Omega_n(E_2)$, which will be necessary to have uniform bounds in \eqref{eq:boundR}. 

Step 1. We deform the contours $\Sigma_1$ to  contours $\Sigma_1^*$ that largely consists of paths of steep ascent for $\Re F_n(w;E_2)$ leaving from the saddle points $\Omega_n(E_2)$ and $\overline{\Omega_n(E_2)}$. However, close to $\infty$ we bend the contour so that it remains a closed contour between $\widetilde{\Gamma_1}$ and $\Gamma_2$. We do the same for $\Sigma_2$. See the top left picture in Figure \ref{fig:contourdefqua}.  

Step 2. In the next step, we deform the contour $\widetilde \Gamma_1$ to the contour of steep descent $\Gamma_1^*$. By doing so we pick up a reside due to the term $w_1=z_1$. This results in
\begin{multline}\label{eq:boundRA}
 \tfrac{2 q n}{(1-q^2) (2\pi {\rm i})^4} \oint_{\Sigma_1}  \int_{\widetilde \Gamma_1}    \oint_{\Sigma_2}    \int_{\Gamma_2}
 \frac{{\rm e}^{\frac{n}{1-q^2}\left( F_n(w_1;x) + F_n(w_2;E_2)\right) }}{{\rm e}^{\frac{n}{1-q^2}\left( F_n(z_1;E_2) + F_n(z_2;y) \right)}} \frac{ {\rm d}w_2
{\rm d}z_2   {\rm d}w_1 {\rm d}z_1 }{(w_1-z_1)(w_2-z_2)(z_1-w_2)}
 \\= \tfrac{2 q n}{(1-q^2) (2\pi {\rm i})^4} \oint_{\Sigma_1^*}  \int_{\Gamma_1^*}    \oint_{\Sigma_2^*}    \int_{\Gamma_2} 
 \frac{{\rm e}^{\frac{n}{1-q^2}\left( F_n(w_1;x) + F_n(w_2;E_2)\right) }}{{\rm e}^{\frac{n}{1-q^2}\left( F_n(z_1;E_2) + F_n(z_2;y) \right)}}
 \frac{ {\rm d}w_2
{\rm d}z_2   {\rm d}w_1 {\rm d}z_1 }{(w_1-z_1)(w_2-z_2)(z_1-w_2)}
\\+\tfrac{2 q n}{(1-q^2) (2\pi {\rm i})^3} \int_{\mathcal C_{\eta_1}}   \oint_{\Sigma_2^*}    \int_{\Gamma_2} 
\frac{{\rm e}^{\frac{n}{1-q^2}\left( F_n(z_1;x) + F_n(w_2;E_2)\right) }}{{\rm e}^{\frac{n}{1-q^2}\left( F_n(z_1;E_2) + F_n(z_2;y) \right)}}\frac{ {\rm d}w_2
{\rm d}z_2    {\rm d}z_1 }{(w_2-z_2)(z_1-w_2)}
\end{multline}
where $\eta_1$ is the point of intersection of $\Gamma_1^*$ and $\Sigma_1^*$ in the upper half plane and  $\mathcal C_{\eta_1}$ is over path that is directed from ${\overline \eta_1}$ to ${\eta_1}$. We also deform the path  $ \mathcal C_{\eta_1}$ to two rays, one leaving from $\overline{\eta_1}$ connecting to $\infty {\rm e}^{-(\pi+\delta) {\rm i}/2}$  and the other  leaving from  $\infty {\rm e}^{(\pi+\delta) {\rm i}/2}$ connecting to $\eta_1$  for some small $\delta >0$. The reason for deforming of $\mathcal C_{\eta_1}$ in this way,  is that by \eqref{eq:defFn} we have 
$$F_n(z_1;x)-F_n(z_1,E_2)=2q z_1(E_2-x) +x^2-E_2^2,$$
and hence, with this definition, $\mathcal C_{\eta_1}$ consist of two paths of steep descent for the real part of $F_n(z_1;x)-F_n(z_1,E_2)$ leaving from $\eta_1$ and $\overline \eta_1$. 

In the quadruple integral, the integral with respect to the variable $w_1$  is a principle value integral (note that the integrand has a singularity $1/(w_1-z_1)$.

Step 3. In the next step, we also deform $\Gamma_2$ to be a contour of steep descent in both the quadruple and triple integral at the right-hand side of  \eqref{eq:boundRA}. By deforming the contour we pick up residues at $w_2=z_1$ and $w_2$. This leads to the following five multiple integrals 
\begin{multline}\label{eq:boundRB}
 \tfrac{2 q n}{(1-q^2) (2\pi {\rm i})^4} \oint_{\Sigma_1}  \int_{\widetilde \Gamma_1}    \oint_{\Sigma_2}    \int_{\Gamma_2}
 \frac{{\rm e}^{\frac{n}{1-q^2}\left( F_n(w_1;x) + F_n(w_2;E_2)\right) }}{{\rm e}^{\frac{n}{1-q^2}\left( F_n(z_1;E_2) + F_n(z_2;y) \right)}} \frac{ {\rm d}w_2
{\rm d}z_2   {\rm d}w_1 {\rm d}z_1 }{(w_1-z_1)(w_2-z_2)(z_1-w_2)}
 \\= \tfrac{2 q n}{(1-q^2) (2\pi {\rm i})^4} \oint_{\Sigma_1^*}  \int_{\Gamma_1^*}    \oint_{\Sigma_2^*}    \int_{\Gamma_2^*} 
 \frac{{\rm e}^{\frac{n}{1-q^2}\left( F_n(w_1;x) + F_n(w_2;E_2)\right) }}{{\rm e}^{\frac{n}{1-q^2}\left( F_n(z_1;E_2) + F_n(z_2;y) \right)}}
 \frac{ {\rm d}w_2
{\rm d}z_2   {\rm d}w_1 {\rm d}z_1 }{(w_1-z_1)(w_2-z_2)(z_1-w_2)}
\\
+\tfrac{2 q n}{(1-q^2) (2\pi {\rm i})^3} \oint_{\Sigma_1^*}  \int_{\Gamma_1^*}    \int_{\mathcal C_{\eta_2}}  
 \frac{{\rm e}^{\frac{n}{1-q^2}\left( F_n(w_1;x) + F_n(z_2;E_2)\right) }}{{\rm e}^{\frac{n}{1-q^2}\left( F_n(z_1;E_2) + F_n(z_2;y) \right)}}
 \frac{ 
{\rm d}z_2   {\rm d}w_1 {\rm d}z_1 }{(w_1-z_1)(z_1-z_2)}
\\
- \tfrac{2 q n}{(1-q^2) (2\pi {\rm i})^3}   \int_{\overline{\Omega(E_2)}}^{\Omega(E_2)}  \int_{\Gamma_1^*}   \oint_{\Sigma_2^*}    
 \frac{{\rm e}^{\frac{n}{1-q^2} F_n(w_1;x)  }}{{\rm e}^{\frac{n}{1-q^2} F_n(z_2;y) }}
 \frac{ 
{\rm d}z_2   {\rm d}w_1 {\rm d}z_1 }{(w_1-z_1)(z_1-z_2)}
\\
+\tfrac{2 q n}{(1-q^2) (2\pi {\rm i})^3} \int_{\mathcal C_{\eta_1}}    \oint_{\Sigma_2^*}    \int_{\Gamma_2^*} 
\frac{{\rm e}^{\frac{n}{1-q^2}\left( F_n(z_1;x) + F_n(w_2;E_2)\right) }}{{\rm e}^{\frac{n}{1-q^2}\left( F_n(z_1;E_2) + F_n(z_2;y) \right)}}\frac{ {\rm d}w_2
{\rm d}z_2    {\rm d}z_1 }{(w_2-z_2)(z_1-w_2)}
\\
-\tfrac{2 q n}{(1-q^2) (2\pi {\rm i})^2} \int_{\mathcal C_{\eta_1}}  \int_{\mathcal C_{\eta_2}}    
\frac{{\rm e}^{\frac{n}{1-q^2}\left( F_n(z_1;x) + F_n(z_2;E_2)\right) }}{{\rm e}^{\frac{n}{1-q^2}\left( F_n(z_1;E_2) + F_n(z_2;y) \right)}}\frac{
{\rm d}z_2    {\rm d}z_1 }{z_1-z_2}.
\end{multline}
Here $\eta_2$ is the intersection point of $\Sigma_2^*$ and $\Gamma_2^*$ in the upper half plane and $\mathcal C_{\eta_2}$ is over path that is directed from ${\overline \eta_2}$ to ${\eta_2}$. We also deform the path  $ \mathcal C_{\eta_2}$ to two rays, one leaving from $\overline{\eta_2}$ connecting to $\infty {\rm e}^{-(\pi-\delta) {\rm i}/2}$  and the other  leaving from  $\infty {\rm e}^{(\pi-\delta) {\rm i}/2}$ connecting to $\eta_2$  for some small $\delta >0$. By using the same arguments for $\mathcal C_{\eta_1}$ in Step 2, one easily verifies that by this definition  $\mathcal C_{\eta_2}$ consists of two paths of steep descent for the real part of $F_n(z_2,E_2)-F_n(z_2,y)$.  

In all the integrals, in case of a singularly, the integrals are to be understood as principal value integrals.

Step 4. Finally, we deform the contours $\Sigma_{1,2}^*$ to be the paths of steep ascent also far away from the saddle points.  Because we  deformed the paths $\mathcal C_{\eta_1}$ and $\mathcal C_{\eta_2}$ to the rays, we do not pick up any residue anymore. This is this  final deformation that we need. 

Step 5. Now we conjugate all integrals with the exponential function as in the left-hand side of \eqref{eq:boundR}.  This is equivalent  to replacing the functions $F_n$ by $\widetilde F_n$ in \eqref{eq:boundRB} with  
\begin{align}\label{eq:fntotildefn}
\widetilde F_n(w_1;x) =F_n(w_1;x) -\Re F_n(\Omega_n(x_1);x_1),
\end{align}
and similarly for $F_n(w_2,E_2)$, $F_n(z_1,E_2)$ and $F_n(z_2,y)$.

Step 6. We are  now left with estimating the five multiple integrals at the right-hand side of \eqref{eq:boundRB} taking \eqref{eq:fntotildefn} into account. Every integral can be dealt with as as before in Lemma \ref{lem:boundonphi0} by introducing local variables around the saddle points and arguing that the that other parts of the integrals are exponentially small. As the arguments are much the same, we allow ourselves to be brief and discuss the important features only. 
 
 Let us start by estimating the quadruple integral 
$$\tfrac{2 q n}{(1-q^2) (2\pi {\rm i})^4} \oint_{\Sigma_1^*}  \int_{\Gamma_1^*}    \oint_{\Sigma_2^*}    \int_{\Gamma_2^*} 
 \frac{{\rm e}^{\frac{n}{1-q^2}\left( \widetilde F_n(w_1;x) +  \widetilde  F_n(w_2;E_2)\right) }}{{\rm e}^{\frac{n}{1-q^2}\left(  \widetilde F_n(z_1;E_2) + \widetilde  F_n(z_2;y) \right)}}
 \frac{ {\rm d}w_2
{\rm d}z_2   {\rm d}w_1 {\rm d}z_1 }{(w_1-z_1)(w_2-z_2)(z_1-w_2)}
$$
Since each of the contours  consists of paths  of steep descent/ascent leaving from the saddle point and by \eqref{eq:fntotildefn},  all the exponentials are bounded in absolute value by $1$. In fact, away from the saddle points the exponentials are exponentially decaying as $n\to \infty$. Therefore, we split the path of integration in each of the four integrals into three pieces: two pieces consist of the intersection of the path with  a small ball around  a saddle point, one in the upper and one in the lower half  plane, and the third is the remaining part of the contour. We end up with sixteen quadruple integrals where each integral is over an arc centered at a saddle point and a number of quadruple integrals where at least one integral is over a remaining part. Integrals over the remaining parts are exponentially small as $n\to \infty$ and hence they can be discarded. The sixteen other quadruple integrals are also small, but we need to introduce local variable to see that.  Indeed, by passing to local variables as in \eqref{eq:localvariable} (or the conjugates)  we gain a factor $((1-q^2)/n)^{2}$ from the scaling of the variables, but we also obtain a factor $(n/(1-q^2))^{1/2}$ due to the term $1/(z_1-w_2)$. Since the saddle points $\Omega_n(x)$ and $\Omega_n(y)$ are macroscopically far away from $\Omega_n(E_2)$ we can choose the balls around the saddle points to be small enough so that we can ignore the terms $1/(w_1-z_1)(w_2-z_2)$ and estimate them from above by a constant. Moreover, due to the prefactor $n/(1-q^2)$ it follows that the quadruple integral is of order $\mathcal O(\sqrt{(1-q^2)/n})$ as $n\to \infty$, uniformly for $x,y\in I$. Hence the quadruple integral does not contribute to the leading order term in \eqref{eq:boundR} but only to the error term.

Next we deal with the three integrals containing integrals over $C_{\eta_j}$. We claim that these integral are all exponentially small.   The reason for this is that we deformed $\mathcal C_{\eta_1}$ and $\mathcal C_{\eta_2}$ such that they are contours of steep descent for the real parts of $F_n(z_1,x)-F_n(z_1,E_2)$ and $F_n(z_2,E_2) -F_n(z_2,y)$ respectively.  Let us first focus on the integrals over  $\mathcal C_{\eta_1}$.  Since $\eta_1$ is on $\Sigma_1^*$ and $\Gamma_1^*$ we also have 
 $$\Re\left( F_n(\eta_1;x)- F_n(\eta_1;E_2)\right)<\Re\left(  F_n(\Omega_n(x);x)- F_n(\Omega_n(E_2);E_2)\right)=0.$$
 and therefore 
  $$\Re\left( \widetilde F_n(\eta_1;x)- \widetilde F_n(\eta_1;E_2)\right)<0,$$ and similarly for $\overline \eta_1$. An analogous statement holds for  $F_n(z_2,E_2) -F_n(z_2,y)$.
Moreover, as $x$ and $E$ are far away, it is not hard to see that the term left of the inequality is uniformly bounded away from $0$ for $\in \N$ and $x\in I_n$.  By combining this with the fact that $\mathcal C_{\eta_1}$ is a path of steep descent, we easily see that the integral over $C_{\eta_1}$ is of order $\mathcal O(\exp(-Dn/ (1-q^2)))$ as $n\to \infty$ for some constant $D>0$ that is independent of $n$ and $x\in I_n$. Similar arguments can be applied to the integral over $C_{\eta_2}$. Therefore we have that all the three integrals in \eqref{eq:boundRB} involving $C_{\eta_1}$ or $C_{\eta_2}$ are exponentially small and do not contribute the leading term in \eqref{eq:boundR}. 

The only contribution to the leading term in \eqref{eq:boundR} from the first quadruple integral in \eqref{eq:Rnquadruple} comes from the only remaining integral
\begin{equation}\label{eq:zomoe}- \tfrac{2 q n}{(1-q^2) (2\pi {\rm i})^3}   \int_{\overline{\Omega(E_2)}}^{\Omega(E_2)}  \int_{\Gamma_1^*}   \oint_{\Sigma_2^*}    
 \frac{{\rm e}^{\frac{n}{1-q^2} \widetilde  F_n(w_1;x)  }}{{\rm e}^{\frac{n}{1-q^2} \widetilde F_n(z_2;y) }}
 \frac{ 
{\rm d}z_2   {\rm d}w_1 {\rm d}z_1 }{(w_1-z_1)(z_1-z_2)}.
\end{equation}
Note that $w_1,z_2$ are  far away from $z_1$ and hence we have no singularities in the integral.  The main contribution comes again from neighborhoods around the saddle points. Let us denote the involution on $\C$ by taking complex conjugation by  $\iota:\C\to \C :z\mapsto \overline z$. Then by introducing local variables near $\iota^{j}( \Omega_n(x))$ and $\iota^k \Omega_n(y)$, we obtain for that \eqref{eq:zomoe} can be written as
\begin{align}\label{eq:moe}
\sum_{j,k=1}^2f_j (x)g_i(y)
  \int_{\overline{\Omega_n(E_2)}}^{\Omega_n(E_2)}   \frac{ 
{\rm d}z_1 }{(\iota^j(\Omega_n(x))-z_1)(z_1-\iota^k(\Omega_n(y)))}+\mathcal O\left(\left(\frac{1-q^2}{n}\right)^{1/2}\right),
\end{align}
where $f_j$ and $g_j$ are bounded function (with bounds that are uniform in $x,y$).  

Concluding, \eqref{eq:moe} is leading term of the first quadruple integral in \eqref{eq:Rnquadruple}, with the lower order terms being $\mathcal O(\sqrt{(1-q^2)/n)})$.  For the second quadruple integral in \eqref{eq:Rnquadruple} we can argue in the same way and we obtain \eqref{eq:moe} with $E_2$ replaced by $E_1$ but with the same functions $f_j$ and $g_k$.  Summarizing, we have \eqref{eq:boundR} with 
\begin{multline*}
A_{n,2(j-1)+k}(x,y)=  \int_{\overline{\Omega_n(E_2)}}^{\Omega_n(E_2)}   \frac{ 
{\rm d}z_1 }{(\iota^j(\Omega_n(x))-z_1)(z_1-\iota^k(\Omega_n(y)))}\\-  \int_{\overline{\Omega_n(E_1)}}^{\Omega_n(E_1)}   \frac{ 
{\rm d}z_1 }{(\iota^j(\Omega_n(x))-z_1)(z_1-\iota^k(\Omega_n(y)))}.
\end{multline*}
for $j,k=1,2$. 
The fact that these function are 1-Lipschitz with a constant that is independent of $n$ easily follows from \eqref{eq:Omegalipschitz}.
 \end{proof}

\section{Proof of Proposition \ref{prop:localconcenpre} } \label{sec:concen}
In this section we will prove  Proposition  \ref{prop:localconcenpre}. 
Throughout this section we will always assume \eqref{eq:para1}--\eqref{eq:para3}. 
\subsection{Preliminaries}
An important ingredient for the proof of Proposition \ref{prop:localconcenpre} is  a Fredholm determinant identity for the moment generating function of  a linear statistic. In the proof we estimate traces of various operators involving the kernels $K_n(x,y)$ and $R_n^I(x,y)$ defined in \eqref{eq:defKn} and \eqref{eq:defRn}  viewed as operators on $\mathbb L_2(I)$. In these estimates  we will frequently use standard estimates from operator theory that we will  briefly recall below. 

Let $I\subset \R$. For a compact operator $A$ on $\mathbb L_2(I)$, we denote the singular values by $\mathcal \sigma_j(A)$ (i.e. $\sigma_j(A)^2$ are the eigenvalues of the compact self-adjoint operator $A^*A$). The operator $A$ is said to be of trace class with   trace norm  $\|A\|_1$ defined by
\begin{align}\nonumber
\|A\|_1=\sum \sigma_j(A),
\end{align}
if and only if the latter is finite. For any trace class operator $A$, the trace $\Tr A$ is defined as the sum of the eigenvalues (Lidskii's Theorem). The operator is said to be Hilbert-Schmidt when $A^*A$ is trace class and, in that case, the Hilbert-Schmidt-norm $\|A\|_2$ is defined by 
\begin{align}\nonumber
\|A\|_2^2=\Tr A^*A=\sum \sigma_j(A)^2.
\end{align}
Finally, for any compact operator $A$ we denote the operator norm by $\|A\|_\infty=\sup_j \sigma_j(A)$. We denote the space of all trace class operator on $\mathbb L_2(I)$ by $\mathcal B_1$, all Hilbert-Schmidt operators by $\mathcal B_2$  and all bounded operators by $\mathcal B_\infty$.

In most situations, the operator $A$ is an integral operator with kernel $A(x,y)$ (like $K_n(x,y)$ or $R^I_n(x,y)$ or a combination thereof). Moreover, the rank of $A$ is typically finite (but often growing with $n$) and for such operators   we have in particular that 
\begin{align}\nonumber
\begin{split}
\Tr A&=\int_I A(x,x) {\rm d}x\\
\|A\|_2^2& =\iint_{I \times I} |A(x,y)|^2 {\rm d}x {\rm d} y.
\end{split}\end{align}

In the upcoming we will frequently use a number of identities that we listed in the following lemma.

\begin{lemma} \label{lem:tracestuff}
The following identities hold\begin{enumerate}
\item  $|\Tr A|\leq \|A\|_1$
\item If $A\in \mathcal B_1$,  $B\in \mathcal B_\infty$, then $\|AB \|_1\leq \|A\|_1 \|B\|_\infty$
\item If $A,B\in \mathcal B_2$, then $AB\in \mathcal B_1$  and $\Tr AB \leq \|A B\|_1\leq\|A\|_2\|B\|_2$
\item If $A$ has rank $n$, then $\|A\|_1 \leq \sqrt n \|A_2\|$ and $\|A\|_2\leq \sqrt n \|A\|_\infty$.
\end{enumerate}
\end{lemma}

\subsection{A first concentration inequality} 
Although it is difficult to derive a CLT directly from the determinantal structure, it is possible to use an idea that was also applied in \cite{BD} (see also \cite{BD2})   to find a bound on the moment-generating function at the left-hand side of \eqref{eq:CLTsuff}.  

\begin{proposition} \label{th:boundonlinstat}
Let $I \subset U$ be compact. Then there exists constants $d_1,d_2$ such that  for $h\in C^1_c (\R)$ with support in $I$ we have that the linear statistic $Y_h=\sum_{j=1}^n h(x_j(t))$ satisfies\footnote{We note that throughout this section we have  $\EE=\EE^0$ in the notation of Section~4}
\begin{equation}\label{eq:boundonlinstat}
\left|\log \EE[\exp \lambda \left( Y_h -\EE Y_h\right)]\right| \leq d_1 \left(\iint  \left(\frac{h(u) -h(v)}{u-v}\right)^2 {\rm d} u {\rm dÊ} v+\|h\|_\infty^2\right).
\end{equation}
for $\lambda \leq  1/(d_2\|h\|_\infty)$, $\xi \in C(U,A, \delta)$,  and $n$ sufficiently large.
\end{proposition}

To prepare the proof of  Proposition \ref{th:boundonlinstat} we first derive some lemmas.

Let us first introduce some notation. We recall the definition of $\Omega_n(x)$ in Lemma \ref{lem:saddle} and of $F_n$ in \eqref{eq:defFn}. Then we define $\widetilde K_n^I$ and $\widetilde R_n ^I$
\begin{align}\label{eq:deftildekr}
\begin{split}
\widetilde K_n^I(x,y)&=\tfrac{\exp\left(\frac{n}{1-q^2}\Re F_n(\Omega_n(y);y) \right)}{\exp\left( \frac{n}{1-q^2}\Re F_n(\Omega_n(x);x)\right)} K_n(x,y)\\
\widetilde R_n^I(x,y)&=\tfrac{\exp\left(\frac{n}{1-q^2}\Re F_n(\Omega_n(y);y) \right)}{\exp\left( \frac{n}{1-q^2}\Re F_n(\Omega_n(x);x)\right)} R^I_n(x,y)
\end{split},\end{align}
for $x,y\in I$.  We also define
\begin{align}
\begin{split}
\widetilde \phi_j(x)=\exp\left(-\frac{n}{1-q^2}\Re F_n(\Omega_n(x);x)\right) \phi_j(x)  \\
\widetilde \psi_j(y)=\exp\left(\frac{n}{1-q^2}\Re F_n(\Omega_n(y);y)\right) \psi_j(y),
\end{split}
\end{align}
for $x,y\in I$, where we recall the definition of $\phi_j$ and $\psi_j$ in \eqref{eq:defphi0}.

The point is that with normalization, the functions $\widetilde \phi_j$ and $\widetilde \psi_j$ are bounded on $I$. Moreover, in the next lemma we show that $\widetilde K_n^I$ is a bounded sequence of operators. 
\begin{lemma}   Let $I \subset U$ be compact and consider $\widetilde K_n^I$ as an operator on $\mathbb L_2(I)$. Then there exists a constant $a>0$ such that  $\| \widetilde K_n^I\|_\infty \leq a$ uniformly for  $\xi\in C(U, A ,\delta)$ and $n \in \N$.
\end{lemma}
\begin{proof}
We start by recalling the well-known fact  that the Hilbert transform $H$ defined by  
$$Hf(x)=\int \frac{f(y)}{x-y} {\rm d} y,$$
for $f\in \mathbb L_2(\R)$, where the integral is a principal value integral, is a projection operator on $\mathbb L_2(\R)$.  Hence the restriction $H^I$ of $H$ on $\mathbbÊL_2(I)$ is a bounded operator with norm $\leq 1$. Denote the multiplication operator on $\mathbb L_2(I)$ with multiplier $f\in \mathbb L_{\infty}(I)$ by $M_f$. Note that with this notation we can rewrite 
\begin{align*}
\widetilde K_n^I =M_{\widetilde \phi_0} H^I M_{\widetilde \psi_0} -\sum_{j=1}^n M_{\widetilde \phi_j} H^I M_{\widetilde \psi_j}.
\end{align*} 
Now we bound the operator norm of $\widetilde K_n^I$ in the following way. Let $f\in \mathbb L_2(I)$. Then by using Cauchy-Schwarz on the sum we obtain
$$\left|\widetilde K_n^I f(x)\right|^2=\left|\sum_{j=0}^n \widetilde \phi_j(x)\left( H^I \widetilde \psi_j f\right)(x)\right|^2
\leq \sum_{j=0}^n \left|\widetilde \phi_j(x)\right|^2 \sum_{j=0}^n \left|\left( H^I \widetilde \psi_j f\right)(x)\right|^2,
$$
and by combining this with the fact that the operator  norm of $H^I$ is bounded by one we get
$$\int_I \left|\widetilde K_n^I f(x)\right|^2\ {\rm d} x
\leq  \left( \sup_{x\in I} \sum_{j=0}^n \left|\widetilde \phi_j(x)\right|^2 \right) \int_I  \sum_{j=0}^n \left|\widetilde \psi_j(y) f(y)\right|^2 \ {\rm d} y.
$$
Concluding, we have
\begin{align}\label{eq:normKa}
\left\|\widetilde K^I_n\right\|^2 \leq\left( \sup_{x\in I} \sum_{j=0}^n \left|\widetilde \phi_j(x)\right|^2 \right)\left(\sup_{y\in I} \sum_{j=0}^n \left|\widetilde \psi_j(y)\right|^2  \right). \end{align}
Now note that by \eqref{eq:estimateintegral1} and Lemma \ref{lem:asymph} we have that $\widetilde \phi_j$, $\widetilde \phi_0^+$ and $\widetilde \phi_0^-$ are uniformly bounded for $x\in I$ and $n\in \mathbb N$. Hence there exists a constant $A>0$ such that 
\begin{multline*}
\sum_{j=1}^n |\widetilde \phi_j(x)|^2Ê\leq A \frac{1-q^2}{n} \sum_{j=1}^n \frac{1}{|\Omega_n(x)-\xi_j^{(n)}|^2}
=- A \frac{1-q^2}{n \Im \Omega_n(x)} \sum_{j=1}^n\Im\left( \frac{1}{\Omega_n(x)-\xi_j^{(n)}}\right),
\end{multline*}
for $x\in I$ and Ê$n\in \N$.  Since $\xi\in \mathcal C(U,A ,\delta)$,  we see that  $$
\sup_n  \sup_{x\in I} \sum_{j=0}^n \left|\widetilde \phi_j(x)\right|^2 <\infty.
$$
 Note that, by the same arguments, we also have the same estimate with $\widetilde \phi_j$ replaced by $\widetilde \psi_j$. Moreover, by plugging these estimates into \eqref{eq:normKa} and the fact that $\widetilde \phi_0$ and $\widetilde \psi_0$ are uniformly bounded,  we obtain the statement.\end{proof}
 
 In the latter proof we have used the notation $M_h$ for the multiplication  operator on $\mathbb L_2(\R)$  with multiplier $h\in \mathbb L_\infty(\R)$, i.e. $M_h:f\to hf$. We will frequently use multiplication operators and to avoid cumbersome notation, we will identify the multiplier with the operator and simply write $h$ instead of $M_h$. Moreover, since $\|M_h\|_\infty=\|h\|_{\mathbb L_\infty}$, we trust there is no confusion when writing $\|h\|_\infty$.
 
\begin{lemma}
 Let $I \subset U$ be compact, $m\geq 2$ and $h_1,\ldots,h_m$ bounded functions with support  $I_h\subset I$. Then   there is a constant $c>0$ such that 
\begin{align} \label{eq:combitrace}
\left|\Tr  \prod _{j=1}^{m-1}\left(h_{j} K_n \right)  h_{m}  R ^I_n \right| \leq c |I_h| \|\widetilde K_n^I\|^{m-2}_\infty \prod_{j=1}^{m} \|h_j\|_\infty , \end{align}
for all $\xi  \in \mathcal C(U,A ,\delta)$ and  $n$ sufficiently large. Here $|I_h|$ is the length of $I_h$ and  we identified $h_j$ with  the multiplication operator with multiplier $h_j$.\end{lemma}
\begin{proof}
Note that since all $h_j$ have support  $I_h$, we can view the operators $K_n$ and $R_n ^I$ as operators on $\mathbb L_2(I_h)$. Moreover, the trace is invariant under conjugation and therefore 
\begin{align}\nonumber
\Tr  \prod _{j=1}^{m-1}\left(h_{j} K_n \right)  h_{m}  R ^I_n =\Tr  \prod _{j=1}^{m-1}\left(h_{j} \widetilde K_n^{I_h} \right)  h_{m} \widetilde   R ^I_n.
\end{align}
with $\widetilde K_n^{I_h} $ and $\widetilde   R ^I_n$ as in \eqref{eq:deftildekr}
We start by splitting $\tilde R^I_n$ and write
\begin{equation}\label{eq:combitraceA}
\Tr  \prod _{j=1}^{m-1}\left(h_{j} \widetilde K_n^{I_h} \right)  h_{m} \widetilde  R ^I_n  =\Tr  \prod _{j=1}^{m-1}\left(h_{j} \widetilde K_n^{I_h} \right)  h_{m}  (\widetilde R ^I_n-A_n) +\Tr  \prod _{j=1}^{m-1}\left(h_{j} \widetilde K_n^{I_h} \right)  h_{m}  A_n,
\end{equation}
with $A_n$ as in \eqref{eq:defAn}.
It follows from \eqref{eq:boundR} that there exists a constant $c_1>0$ such that the Hilbert-Schmidt norm of $\widetilde R_n^I -A_n$ viewed as an operator on $\mathbb L_2(I_h)$ can be estimated by
$$\|\widetilde R ^I_n-A_n\|_2^2=\iint_{I_h^2} |\widetilde R^I_n(x,y)-A_n(x,y)|^2 {\rm d} x {\rm d} y\leq c_1 ^2|I_h|^2 (1-q^2)/n,$$
for $n\in \N$.  Moreover, it is important to note that $\widetilde K_n$ is an operator of rank $n$ (we are dealing with  a determinantal point process with $n$ points). Hence we can estimate the trace norm by  a product of $\sqrt n$ and the Hilbert-Schmidt norm. Therefore, we have by using the identities in Lemma \ref{lem:tracestuff} that
\begin{multline}\label{eq:combitraceG}
\left|\Tr  \prod _{j=1}^{m-1}\left(h_{j} \widetilde K_n ^{I_h}\right)  h_{m}  (\widetilde R ^I_n-A_n) \right | \leq \left\|  \prod _{j=1}^{m-1}\left(h_{j} \widetilde  K_n^{I_h} \right)  h_{m}  (\widetilde R ^I_n-A_n) \right \|_1 \\
\leq \sqrt n \left\|  \prod _{j=1}^{m-1}\left(h_{j} \widetilde  K_n^{I_h} \right)  h_{m}  ( \widetilde R ^I_n-A_n) \right \|_2\\ \leq \prod _{j=1}^{m}\|h_{j}\|_\infty \| \widetilde  K_n^{I_h}\|_\infty^{m-1} \sqrt n  \left\|  (\widetilde R ^I_n-A_n) \right \|_2 = c_1 \ |I_h| \sqrt{1-q^2} \prod _{j=1}^{m}\|h_{j}\|_\infty \|\widetilde K_n^{I_h}\|_\infty^{m-1} 
\end{multline}
for $n \in \mathbb N$. This estimates the first term on the right-hand side of \eqref{eq:combitraceA}. For the second term we use \eqref{eq:defAn} and write 
$$
\Tr  \prod _{j=1}^{m-1}\left(h_j \widetilde K_n^{I_h} \right) h_m  A_n= \sum_{k=1}^4 \Tr  \prod _{j=1}^{m-1}\left(h_j\widetilde K_n ^{I_h}\right)  h_m  f_k A_{n,k} g_k .
$$
Without loss of generality, we can ignore the uniformly bounded function  $f_k$ and $g_k$ as we can merge $f_k$ into $h_m$ and $g_k$ into $h_1$ which does not make a difference to the statement. For  clarity reasons,  we will thus restrict ourselves to analyzing
\begin{equation}\nonumber
 \Tr  \prod _{j=1}^{m-1}\left(h_j \widetilde K_n ^{I_h}\right)  h_m  A_{n,k}  .
\end{equation}
for $k=1,2,3,4$.  The main idea is now to rewrite each term as follows
\begin{multline}\label{eq:langlem}
\int \cdots \int  h_1( x_1) \cdots h_{m}( x_m) \widetilde K_n^{I_h} (x_1,x_2)  \cdots \widetilde K_n^{I_h}(x_{m-1},x_m) 
A_{n,k}(x_m,x_1){\rm d}x_1 \cdots {\rm d}x_m\\
=\int \cdots \int  h_1( x_1) \cdots h_{m}( x_m) \widetilde K_n^{I_h} (x_1,x_2)  \cdots \widetilde K_n^{I_h}(x_{m-1},x_m) \\ \hfill  \times
\left(\sum_{j=2}^m \left(A_{n,k}(x_j,x_1)-A_{n,k}(x_{j-1},x_1)\right) \right){\rm d}x_1 \cdots {\rm d}x_m
\\
+\int \cdots \int  h_1(x_1) \cdots h_{m}(x_m) \widetilde K_n ^{I_h}(x_1,x_2)  \cdots \widetilde K_n^{I_h}(x_{m-1},x_m) 
A_{n,k}(x_1,x_1){\rm d}x_1 \cdots {\rm d}x_m
\end{multline}
We start by bounding the second term on the right-hand side. By applying Cauchy-Schwarz on the integral with respect to $x_1$ this term can be estimated by 
\begin{multline}\label{eq:combitraceB}
\left|\int \cdots \int  h_1( x_1) \cdots h_{m}( x_m) \widetilde K_n^{I_h} (x_1,x_2)  \cdots \widetilde K_n^{I_h}(x_{m-1},x_m) 
A_{n,k}(x_1,x_1){\rm d}x_1 \cdots {\rm d}x_m\right|\\
\leq  \|h_1\|_{\mathbb L_2(I_h)} \left\|\widetilde K_n^{I_h} \left(h_2 \widetilde K_n^{I_h}\left(h_3 \widetilde K_n^{I_h}\left( \cdots h_m \right)\right)\right)\right\|_{\mathbb  L_2(I_h)} \sup_{x_1\in I_h} |A_{n,k}(x_1,x_1)|
\end{multline}
The  $\mathbb L_2(I_h)$-norms can be  estimated by 
$$
\left\|\widetilde K_n^{I_h} \left(h_2 \widetilde  K_n^{I_h} \left(h_3 \widetilde  K_n^{I_h}\left( \cdots h_m \right)\right)\right)\right\|_{\mathbb L_2(I_h)}\leq \prod_{j=2}^{m-1} \|h_j \|_\infty \|\widetilde K_n^{I_h}\|_\infty \|h_m \|_{\mathbb  L_2(I_h)} 
$$ 
and $\|h_j\|_{\mathbb  L_2(I_h)} \leq \|h_j\|_\infty |I_h|$. By substituting these estimates back into \eqref{eq:combitraceB} we get 
\begin{multline}\label{eq:combitraceF}
\left|\int \cdots \int  h_1(x_1) \cdots h_{m}( x_m) \widetilde K_n^{I_h} (x_1,x_2)  \cdots \widetilde K_n^{I_h}(x_{m-1},x_m) 
A_{n,k}(x_1,x_1){\rm d}x_1 \cdots {\rm d}x_m\right|\\\leq |I_h|^2 \prod_{j=1}^m\|h_j\|_{\infty} \|\widetilde K_n ^{I_h}Ê\|_\infty^{m-1} \sup_{x_1\in I_h} |A_{n,k}(x_1,x_1)|
\end{multline}
for $n\in \N$. This bounds the second term on the right-hand side of \eqref{eq:langlem}. Each term in the sum in the first term at the right-hand side of \eqref{eq:langlem} can be rewritten to 
\begin{multline}\label{eq:combitraceC}
\iint  \widetilde K_n^{I_h} \left(h_1 \widetilde K_n^I\left(h_2 \widetilde K_n^{I_h}I\left( \cdots h_{j-1}\right)\right)\right)(x_{j})\\
\times   \widetilde  K_n^{I_h} \left(h_{m} \widetilde K_n^{I_h}\left(h_{m-1} \widetilde K_n^{I_h}\left( \cdots h_{j+1}\right)\right)\right)(x_{j+1})  \\
\times \left(A_{n,k} (x_{j+1},x_j)-A_{n,k} (x_{j},x_1)\right) h_{j}( x_j)h_{j+1}( x_{j+1}) \widetilde K_n^{I_h}(x_{j},x_{j+1}) {\rm d}x_j {\rm d}x_{j+1}. \end{multline}
 By \eqref{eq:asymptoticsKnmain} and the fact that $A_{k,n}$ is 1-Lipschitz with a constant that is uniform in $n$, we see that there exists a constant $c_2>0$  
\begin{align}\label{eq:combitraceD}
\left|\widetilde K_n^{I_h}(x_{j},x_{j+1})(x_{j+1}-x_{j})\right| \left|\frac{A_{n,k} (x_{j+1},x_1)-A_{n,k} (x_{j},x_1)}{x_{j+1}-x_{j}}\right| \leq c_2
\end{align}
as $n\to \infty$ uniformly for $x_j,x_{j+1}\in I_n$. By inserting \eqref{eq:combitraceD} into  \eqref{eq:combitraceC} and estimating the $\mathbb L_2(I_h)$-norms as before we get 
\begin{multline} \label{eq:combitraceE}
\int \cdots \int  h_1( x_1) \cdots h_{m}(x_m) \widetilde K_n^{I_h} (x_1,x_2)  \cdots \widetilde K_n^{I_h}(x_{m-1},x_m) \\ \hfill  \times
\left(A_{n,k}(x_j,x_1)-A_{n,k}(x_{j-1},x_1) \right){\rm d}x_1 \cdots {\rm d}x_m
\\
\leq c_2 |I_h|^2 \prod_{j=1}^m\|h_j\|_{\infty} \|\widetilde K_n ^{I_h}\|_\infty^{m-2}
\end{multline}
Hence by inserting \eqref{eq:combitraceE} and  \eqref{eq:combitraceF}  into \eqref{eq:langlem}  we see that there exists a constant $c_3>0$ such that 
\begin{equation}
 \Tr  \prod _{j=1}^{m-1}\left(h_j K_n ^{I_h}\right)  h_m  A_{n,k}  \leq c_3 |I_h|^2 \prod_{j=1}^m\|h_j\|_{\infty} \|\widetilde K_n ^{I_h}\|_\infty^{m-2} ,
\end{equation}
for $n\in \N$. 
By combining this with \eqref{eq:combitraceG} and \eqref{eq:combitraceA} we  proved the statement. 
\end{proof}
\subsection{Proof of Poposition \ref{th:boundonlinstat}}
\begin{proof}[Proof of Proposition \ref{th:boundonlinstat}]
The starting point is that for  linear statistics for determinantal point processes we have the following identity (see for example \cite{J4})
\begin{align}\nonumber
\EE [\exp  \lambda \Tr h (M)]=\det\left(1+({\rm e}^{\lambda h}-1) K_n\right),
\end{align}
where the right-hand side is the Fredholm determinant of $({\rm e}^{\lambda h}-1) K_n$ viewed as  an operator on $\mathbb L_2(\R)$.  For bounded $h$ we can rewrite the determinant in terms of a sum of traces 
\begin{align}\nonumber
\log \EE [\exp  \lambda \Tr h(M)]=\sum_{j=1}^\infty \frac{(-1)^{j+1}}{j} \Tr \left( ({\rm e}^{\lambda h}-1) K_n\right)^j,
\end{align}
for sufficiently small $\lambda$. By expanding the exponential in a Taylor series we obtain
\begin{align}\nonumber
\log \EE [\exp  \lambda \Tr h(M)]=\sum_{j=1}^\infty \frac{(-1)^{j+1}}{j}  \sum_{l_1,\ldots, l_j =1}^{\infty} \lambda^{l_1+\cdots +l_j}\frac{\Tr h^{l_1} K_n \cdots h^{l_j} K_n}{l_1!\cdots l_j!}.
\end{align}
By an extra reorganization this can be turned into
\begin{align}\nonumber
\log \EE [\exp  \lambda \Tr h(M)]=\sum_{m=1}^\infty \lambda^m\sum_{j=1}^m \frac{(-1)^{j+1}}{j}  \sum_{\overset{l_1+\cdots +l_j=m}{l_i\geq 1}} \frac{\Tr h^{l_1} K_n \cdots h^{l_j} K_n}{l_1!\cdots l_j!}.
\end{align}
This standard identity is our starting point for proving the statement.  Note that because $h$ has support in $I$ we can replace all $K_n$ by $K_n^I$ which is the restriction of $K_n$ to $\mathbb L_2(I)$. Moreover, since the traces are invariant under conjugation we also are free to replace $K_n^I$ to $\widetilde K_n^I$ which is a bounded operator with a bound that is uniform in $n$.  For clarity reasons, we will simplify the cumbersome notation omit the tilde in $\widetilde K_n^I$  and simply write $K_n^I$. Hence we obtain
\begin{align}\label{eq:uitgang}
\log \EE [\exp  \lambda \Tr h(M)]=\sum_{m=1}^\infty \lambda^m\sum_{j=1}^m \frac{(-1)^{j+1}}{j}  \sum_{\overset{l_1+\cdots +l_j=m}{l_i\geq 1}} \frac{\Tr h^{l_1} K_n^I \cdots h^{l_j} K_n^I}{l_1!\cdots l_j!}.
\end{align}

The key to the proof is the following inequality: there exists a $c>0$ such that  \begin{align}\label{eq:claimdifftraces}
\left|\Tr h^{l_1} K_n^I \cdots h^{l_j} K_n^I-\Tr h^m K_n^I \right| \leq j \|h\|_\infty^{m-2}  \|K_n^I\|_\infty^{j-2} \left( m^2 \|[h,K_n^I]\|_2^2+c\|h\|_\infty^2\right),
\end{align}
for any $l_1,\ldots,l_j$ and $m\geq jÊ\geq 2 $ such that $l_1+\cdots+l_j=m$ and $l_i\geq 1$. Here $[h,K_n]$ stands for the commutator of $K_n$ and $h$ (viewed as a multiplication operator) and $\| \cdotÊ\|_2$ stands for the Hilbert-Schmidt norm. We proceed by  proving this claim first.

First assume that $j=2$. Then a straightforward computation using properties of the trace and the fact that $R^I_n=\left(K_n^I\right)^2-K_n^I$ shows that 
\begin{align} \label{eq:j2prooflemma}
\Tr h^{l_1} K_n^I h^{l_2}K_n^I =\Tr h^m K_n^I +\Tr h^m R_n^I+\frac{1}{2} \Tr [h^{l_1},K_n^I] [h^{l_2},K_n^I].
\end{align}
By \eqref{eq:boundR} there eixts a constant $c_1>0$ (independent of $n,m$) such that  
\begin{equation}\label{eq:boundhri}
\Tr h^m R_n^I \leq c_1 \|h\|^m_\infty 
\end{equation}
for $n\in \N$.
Moreover, 
\begin{align}\label{eq:estimatedifftraces}
\left| \Tr [h^{l_1},K_n^I] [h^{l_2},K_n^I]\right|Ê\leq \|  [h^{l_1},K_n^I]\|_2\| [h^{l_2},K_n^I]\|_2.
\end{align}
Now note that 
\begin{align}\nonumber
 \|  [h^{l_1},K_n^I]\|_2^2=\int \int \left((h(x))^{l_1}-(h(y))^{l_1}\right)^2 K_n^I(x,y) ^2 \ {\rm d}xÊ{\rm d} y,
\end{align}
and by using $a^l-b^l=(a-b) \sum_{j=0}^{l-1} a^jb^{l-1-j}$ we can estimate the right-hand side and obtain
\begin{multline}\label{eq:estimateonHScom}
 \|  [h^{l_1},K_n^I]\|_2^2\leq  l_1^2 \|h\|_\infty^{2(l_1-1)} \int \int (h(x))-h(y))^2 K_n^I(x,y)^2 \ {\rm d}xÊ{\rm d} y\\
 =l_1^2 \|h\|_\infty^{2(l_1-1)}  \|  [h,K_n^I]\|_2^2.
\end{multline}
By using this inequality twice in \eqref{eq:estimatedifftraces} and substituting this together with \eqref{eq:boundhri} into \eqref{eq:j2prooflemma} we obtain 
\begin{align}\label{eq:estimatedifftraces2}
\left|\Tr h^{l_1} K_n^I h^{l_2}K_n^I -\Tr h^m K_n^I \right|Ê\leq l_1 l_2 \|h\|_\infty^{m-2} \|  [h,K_n^I]\|_2^2 +c_1 \|h\|^m _\infty.
\end{align}
Since $l_1,l_2\leq m$ we see that we indeed have Ê\eqref{eq:claimdifftraces} 
for $j=2$.

Now let us  prove \eqref{eq:claimdifftraces} for $j\geq 3$. In this case, we use the following identity that follows from $R_n^I=(K_n^I)^2-K_n^I$
\begin{multline}\nonumber
h^{l_1} K_n^I \cdots h^{l_j} K_n^I =
h^{l_1} K_n^I \cdots h^{l_{j-2}}  K_n^I [h^{l_{j-1}},K_n^I][h^{l_j},K_n^I]\\
+h^{l_1} K_n^I \cdots h^{l_{j-1}+l_j} K_n^I+h^{l_1}K_n^I \cdots h^{l_{j-2}} R_n^I h^{l_{j-1}+l_j}K_n^I\\
+h^{l_1}K_n^I \cdots h^{l_{j-2}} K_n^I h^{l_{j-1}}R_n^I h^{l_j}-h^{l_1}K_n^I \cdots h^{l_{j-2}}R_n^I h^{l_{j-1}} K_n^Ih^{l_j}
\end{multline}
The terms including  $R_n^I$ can be estimated by using \eqref{eq:combitrace}.  Moreover, we have 
\begin{multline}\nonumber
\left|\Tr h^{l_1} K_n^I \cdots h^{l_{j-2}}  K_n^I [h^{l_{j-1}},K_n^I][h^{l_j},K_n^I]\right|\\
\leq \|h\|_\infty^{l_1+\cdots +l_{j-2}} \|K_n^I\|^{j-2}_\infty \|[h^{l_{j-1}},K_n^I]\|_2\|[h^{l_j},K_n^I]\|_2.
\end{multline}
By using \eqref{eq:estimateonHScom} and the fact that $l_{j-1},l_j\leq m^2$ we see that there exists a constant $c_2\geq 0$ such that 
\begin{multline}\nonumber
\left|\Tr h^{l_1} K_n^I \cdots h^{l_j} K_n^I -\Tr h^{l_1} K_n^I \cdots h^{l_{j-1}+l_j} K_n^I
\right|\\
\leq m^2\|h\|_\infty^{m-2} \|K_n^I\|^{j-2}_\infty \|[h,K_n^I]\|_2^2+ 3 c_2 \|h\|_\infty^{m} \|K^I_n\|_\infty^{j-2}
\\
= \|h\|_\infty^{m-2} \|K_n^I\|^{j-2}_\infty\left( m^2 \|[h,K_n^I]\|_2^2+ 3 c_2 \|h\|_\infty^{2}\right).
\end{multline}
By iterating this inequality we arrive at the claim \eqref{eq:claimdifftraces} and hence we proved the claim for all situations.
\bigskip

Now that we have \eqref{eq:claimdifftraces} we come back to the proof of \eqref{eq:boundonlinstat}.  Let us start by rewriting \eqref{eq:uitgang} to
\begin{multline}\label{eq:boundonlinstatstep2}
\log \EE [\exp  \lambda \Tr h(M)]=\sum_{m=1}^\infty \lambda^m\sum_{j=1}^m \frac{(-1)^{j+1}}{j}  \sum_{\overset{l_1+\cdots +l_j=m}{l_i\geq 1}} \frac{\Tr h^{m} K_n^I}{l_1!\cdots l_j!}\\
+\sum_{m=1}^\infty \lambda^m\sum_{j=1}^m \frac{(-1)^{j+1}}{j}  \sum_{\overset{l_1+\cdots +l_j=m}{l_i\geq 1}} \frac{\Tr h^{l_1} K_n^I \cdots h^{l_j} K_n^I -\Tr h^m K_n^I}{l_1!\cdots l_j!}.
\end{multline}
By expanding the right-hand side of $x=\log(1+({\rm e}^x-1))$ into a Taylor series we obtain
\[\sum_{j=1}^m \frac{(-1)^{j+1}}{j}  \sum_{\overset{l_1+\cdots +l_j=m}{l_i\geq 1}} \frac{1}{l_1!\cdots l_j!}=0 , \qquad m\geq 2.\]
By inserting this into \eqref{eq:boundonlinstatstep2} and also using the fact that the last term at the right-hand side always vanishes for $j=1$, we can simplify \eqref{eq:boundonlinstatstep2} to
\begin{multline}\nonumber
\log \EE [\exp  \lambda \Tr h(M)]= \Tr h K_n^I \\
+\sum_{m=2}^\infty \lambda^m\sum_{j=2}^m \frac{(-1)^{j+1}}{j}  \sum_{\overset{l_1+\cdots +l_j=m}{l_i\geq 1}} \frac{\Tr h^{l_1} K_n^I \cdots h^{l_j} K_n^I -\Tr h^m K_n^I}{l_1!\cdots l_j!}.
\end{multline}
By invoking \eqref{eq:claimdifftraces} we can estimate
\begin{multline} \label{eq:boundonlinstatprefinalstep}
\left| \log \EE [\exp  \lambda \Tr h(M)]-\Tr h K_n^I\right| \\
\leq  \sum_{m=2}^\infty \lambda^m \|h\|_\infty^{m-2}  \left(m^2  \| [h,K_n^I]\|_2^2+c\|h\|_\infty^2 \right) \sum_{j=2}^m \sum_{\overset{l_1+\cdots +l_j=m}{l_i\geq 1}} \frac{\|K_n^I\|_\infty^{j-2}}{l_1!\cdots l_j!}.
\end{multline}
Now, by expanding $m^m=(1+\ldots+1)^m$ and using the fact that we obtain 
\begin{align}\nonumber
\sum_{j=2}^m \sum_{\overset{l_1+\cdots +l_j=m}{l_i\geq 1}} \frac{\|K_n^I\|_\infty^{j-2}}{l_1!\cdots l_j!}
\leq \max( \|K_n^I\|_\infty^{m-2},1) \sum_{j=2}^m \sum_{\overset{l_1+\cdots +l_j=m}{l_i\geq 1}} \frac{1}{l_1!\cdots l_j!}\\<\max(\|K_n^I\|_\infty^{m-2},1) \frac{m^m}{m!} \leq \max(\|K_n^I\|_\infty^{m-2},1)  \frac{{\rm e}^m}{\sqrt{2 \pi m}}
\end{align}
By inserting this into \eqref{eq:boundonlinstatprefinalstep} we see that there exists  a constant $c_4$ (independent of $n$) such that
\begin{align} \label{eq:boundonlinstatfinalstep} 
\left| \log \EE [\exp  \lambda \Tr h(M)]-\Tr h K_n^I\right|
\leq  c_4 \left(\| [h,K_n^I]\|_2^2+\|h\|_\infty^2\right)  
\end{align}
uniformly for $|\lambda |\leq (3 \|h\|_\infty \max(\|K_n^I\|,1)^{-1}$.

To finish the proof we note that by \eqref{eq:asymptoticsKnmain} there exists a constant $c_5>0$ such that  
\begin{multline}\nonumber
\|[h,K_n^I]\|_2^2= \iint (h(x)-h(y))^2 K_n(x,y)^2 {\rm d} x {\rm d} y\\= \iint \left(\frac{h(x)-h(y)}{x-y}\right)^2 (x-y)^2K_n(x,y)^2 {\rm d} x {\rm d} y
\\ \leq c_5 \iint \left(\frac{h(x)-h(y)}{x-y}\right)^2{\rm d} x {\rm d} y,
\end{multline}
for $n\in \N$. By substituting this into \eqref{eq:boundonlinstatfinalstep} we obtain the statement.\end{proof}
\subsection{A concentration inequality using the logaritmic Sobolev inequality}
Proposition \ref{th:boundonlinstat} only holds for function that are continuously differentiable. To extend it to the class for which we have Proposition \ref{prop:localconcenpre} we also need the following concentration inequality due to Herbst. We refer to \cite[Lem. 2.3.3]{AGZ} for more details and background.

\begin{proposition}\label{prop:herbst}
Let $G$ be a complex valued functional on the space of Hermitian matrices such that
\begin{align}\label{eq:lipschitz1}
|G|_{\mathcal L}:= \sup_{M_1,M_2} \frac{|G(M_1)-G(M_2)|}{\|M_1-M_2\|_2}<\infty.
\end{align}
Here $\|\cdot\|_2$ stands for the Hilbert-Schmidt norm. Then 
\begin{align}\label{eq:ineqherbst}
\left|\EE^0\left[\exp \lambda \left(G-\EE^0 [G]\right)\right]\right|\leq \exp { \frac{(1-q^2) |G|_{\mathcal L}^2|\lambda|^2}{n}},
\end{align}
for $\lambda \in \C$. 
\end{proposition}

\begin{proof}
The idea behind the concentration inequality is the fact that the normal distribution satisfies the logarithmic Sobolev inequality. Since the entries are normal and independent, the random matrix $M_n(t)$  also satisfies matrix version of the logarithmic Sobolev equation (with constant proportional to $\frac{1-q^2}{n}$). For probability measures satisfying such a logarithmic Sobolev inequality there is a general concentration inequality, from which the statement follows.

For $\lambda \in \R$ and $G$ real valued,  Lemma 2.3.3 in \cite{AGZ} tells us 
\begin{align}\label{eq:ineqherbstreal}
\EE^0\left[\exp \lambda \left(G-\EE^0 [G]\right)\right]\leq \exp { \frac{(1-q^2) |G|_{\mathcal L}^2|\lambda|^2}{2n}}.
\end{align}
For the complex case we  note that 
\begin{multline*}
\left|\EE^0\left[\exp \lambda \left(G-\EE^0 [G]\right)\right]\right|\leq
\EE^0\left[\exp  \Re \left(\lambda \left(G-\EE^0 [G]\right)\right) \right]\\=\EE^0\left[\exp \left(\Re \lambda \Re  \left(G-\EE^0 [G]\right)-\Im \lambda \Im  \left(G-\EE^0 [G]\right)\right) \right]
\end{multline*}
By applying Cauchy-Schwarz and using  \eqref{eq:ineqherbstreal}   we obtain the statement for the complex case.
\end{proof}

As an example, consider the functional $G$ defined by  $G(M)=\Tr \frac{1}{z-M} J$ for some fixed matrix $J$. By invoking  the standard identities in Lemma \ref{lem:tracestuff}, we obtain 
\begin{multline*}
\left|G(M_1)-G(M_2)\right|=\left|\Tr \frac{1}{z-M_1} (M_1-M_2) \frac{1}{z-M_2} J\right|\\
\leq\left \|\frac{1}{z-M_1}\right\|_\infty  \left\|\frac{1}{z-M_2}\right\|_\infty \|J\|_2 \|M_1-M_2\|_2
\end{multline*}
and by using $\|(z-M_j)^{-1}\|_\infty\leq (\Im z)^{-1}$ and $\|J\|_2\leq \sqrt n \|J\|_\infty$ we then obtain 
\begin{align} \label{eq:ineqfirstconsequenceherbst}
|G|_{\mathcal L} \leq \frac{\sqrt n\|J\|_\infty}{(\Im z)^2}.
\end{align}
By inserting the latter into \eqref{eq:ineqherbst} and replacing $z$ by $z/n^\alpha$ we obtain a bound on the moment-generating  function that is especially strong for small values of $\alpha$. However, when $\alpha$ is close to $1$ the bound is not strong enough for our purposes. For linear statistics however, this bound can be improved, which is the content of Proposition \ref{prop:localconcenpre} which we will now prove.

\subsection{Proof of Proposition \ref{prop:localconcenpre}}
\begin{proof}
Let us start  with some remarks on $|\cdot |_{\mathcal L_w}$. By the triangular inequality and $|x-y|\leq \sqrt{1+x^2} \sqrt{1+y^2}$,  we see that for any function $h\in  \mathcal L_w$  we have
$$|h(x)|\leq |h(x)-h(y)| +|h(y)| \leq \|h\|_{\mathcal L_w}+|h(y)|.$$
By taking $y \to \infty$ and the supremum over $x$, it follows that   any function $h\in  \mathcal L_w$ is necessarily bounded with $$\|h\|_\infty \leq|h|_{\mathcal L_w}.$$ 
Moreover,  for any two functions $g,h \in \mathcal L_w$ we have that  \begin{equation}\label{eq:estimatedifnorm}
 | g h|_{\mathcal L_w}\leq \|h\|_\infty |  g|_{\mathcal L_w}+ \|g\|_\infty |  h|_{\mathcal L_w}\leq 2|g|_{\mathcal L_w} |h|_{\mathcal L_w}.\end{equation}
It is also  straightforward to check that any function  $h\in \mathcal L_w$ is a Sobolev function and that the  Sobolev norm $\|h\|_{\mathcal S}$ can be estimated using 
\begin{equation}\label{eq:estimatedifnorma}
\|h\|_{\mathcal S}^ 2:= \iint \left(\frac{h(x)-h(y)}{x-y}\right)^2 {\rm d} x {\rm d}y \leq \pi^2 |h|_{\mathcal L_w}^2.
\end{equation} Hence, if $h$ is compactly supported the statement follows straightforwardly from Proposition \ref{th:boundonlinstat} and it remains to check the Proposition for functions that have unbounded support. 
 
 Without loss of generality we will assume that $x_*=0$. We will use the notation $h_\alpha(x)= h(n^\alpha x).$
 The key to the proof is to split $h_{\alpha}$ into two parts $h_\alpha=h_{1,\alpha}+h_{2,\alpha}$ such that $h_{1,\alpha}$ has support in a sufficiently small closed interval $[-a-\eps,a+\eps]\subset I$ around $x_*=0$ and $h_{2,\alpha}$  vanishes on $[-a,a]\ni x_*$ for some sufficiently small $\eps >0$. To this end, we choose any  function $g\in \mathcal L_w$ such that $g=1$ on $[-a,a]$ and $g=0$ on $\R \setminus [-a-\eps,a+\eps]$ and $0\leq g \leq 1$ on $\R$. Then we define $h_{1,\alpha}= h_{\alpha} g$ and $h_{2,\alpha}=h_{\alpha} (1-g)$.

 Since $h_{1,\alpha}$ has compact support, we have by Proposition \ref{th:boundonlinstat} that there exists  universal constants $c_1,c_2>0$ such that 
\begin{align}\label{eq:aap10}
\left|\EE^0\left[\exp \lambda \left(\Tr h_{1,\alpha}(M)-\EE^0[\Tr h_{1,\alpha}(M)\right)\right]\right|\leq \exp c_2 \left( \|h_{1,\alpha}\|_{\mathcal S}^2+\|h_{1,\alpha}\|_\infty^2 \right) |\lambda |^2 ,
\end{align}
for $|\lambda |\leq \frac{1}{c_1 \|h_{1,\alpha}\|_\infty}$. Since $\|h_{1,\alpha}\|_\infty \leq \|h\|_\infty\leq |h|_{\mathcal L_w}$, this holds in particular for all $|\lambda |\leq \frac{1}{c_1 |h|_{\mathcal L_w}}$, which is sufficient for our purposes.  Now  note that 
$$
\|h_{1,\alpha}\|_{\mathcal S}^2\leq 2 \|g\|_\infty^2 \|h_{\alpha}\|_{\mathcal S}^2+2 \|h_\alpha\|_\infty^2 \|g\|_{\mathcal S}^2
$$
Since $\|h_{\alpha}\|_{\mathcal S}=\|h\|_{\mathcal S}\leq \
|h|_{\mathcal L_w}$ and $\|h_{\alpha}\|_{\infty}=\|h\|_{\infty}\leq |h|_{\mathcal L_w}$  we have 
$$
\|h_{1,\alpha}\|_{\mathcal S}^2\leq 2 |h|_{\mathcal L_w}\left( \|g\|_\infty^2+ \|g\|_{\mathcal S}^2\right).
$$
By substituting this into \eqref{eq:aap10} and also using  $\|h_{1,\alpha}\|_{\infty}=\|h\|_{\infty}\leq |h|_{\mathcal L_w}$ we obtain 
\begin{align}\label{eq:aap1}
\left|\EE^0\left[\exp \lambda \left(\Tr h_{1,\alpha}(M)-\EE^0[\Tr h_{1,\alpha}(M)\right)\right]\right|\leq \exp \tilde c_2 |h|_{\mathcal L_w} ^2|\lambda |^2 ,
\end{align}
for some new constant $\tilde c_2$ (which is still independent of $h$).

For $h_{2,\alpha}$ we use the result in \eqref{eq:ineqherbst}. To this end, we first compute the Lipschitz norm of $h_{2,\alpha}$. Let $M_1$ and $M_2$ be two $n\times n$  Hermitian matrices and order the eigenvalues as follows $\lambda_1 ^{(1)} \leq \ldots \leq \lambda_n ^{(1)}$ and $\lambda_1 ^{(2)} \leq \ldots \leq  \lambda_n ^{(2)}$. Then 
\begin{multline}\label{eq:splittingh2}
\Tr h_{2,\alpha}(M_1)-\Tr h_{2,\alpha}(M_2)= \sum_{j=1}^n  \left( h_2(n^\alpha \lambda_j ^{(1)})-h_2(n^\alpha \lambda_j ^{(2)})\right)\\
= \sum_{j=1}^n  \left( h(n^\alpha \lambda_j ^{(1)})(1-g(\lambda_j^{(1)}))-h(n^\alpha \lambda_j ^{(2)})(1-g(\lambda_j^{(2)})\right)\\
= \sum_{j=1}^n  \left( h(n^\alpha \lambda_j ^{(1)})-h(n^\alpha \lambda_j ^{(2)})\right)(1-g(\lambda_j^{(2)}))- \sum_{j=1}^n  \left(g(\lambda_j^{(1)})-g(\lambda_j^{(2)})\right) h(n^\alpha \lambda_j ^{(1)}).
\end{multline}
We estimate the two terms on the right-hand side separately. To start with the first term,   by the definition of $|\cdot|_{\mathcal L_w}$  we have
\begin{multline*}
\left|\sum_{j =1}^n   \left( h(n^\alpha \lambda_j ^{(1)}) - h(n^\alpha \lambda_j ^{(2)}) \right) (1-g(\lambda_j^{(2)}))\right|
\leq\sum_{j=1}^n   (1-\chi_{[-a,a]})( \lambda ^{(2)}_j)  \left| h(n^\alpha \lambda_j ^{(1)}) - h(n^\alpha \lambda_j ^{(2)})\right|
\\
\leq |h|_{\mathcal L_w}\sum_{j=1}^n (1-\chi_{[-a,a]})( \lambda ^{(2)}_j) \frac{n^{\alpha} }{\sqrt{(1+n^{2\alpha} (\lambda_j^{(1)})^2)(1+n^{2\alpha} (\lambda_j^{(2)})^2)}}  \left|\lambda_j ^{(1)}- \lambda_j ^{(2)}\right|
\end{multline*}
where $\chi_{[-a,a]}(x)=1$ if $x\in [-a,a]$ and $\chi_{[-a,a]}(x)=0$ otherwise. Therefore,
\begin{multline*}
\left|\sum_{j =1}^n   \left( h(n^\alpha \lambda_j ^{(1)}) - h(n^\alpha \lambda_j ^{(2)}) \right) (1-g(\lambda_j^{(2)}))\right|
\leq 
|h|_{\mathcal L_w} \frac{1}{a} \sum_{j=1}^n  \left|\lambda_j ^{(1)}- \lambda_j ^{(2)}\right|
\\ \leq  \frac{|h|_{\mathcal L_w}  \sqrt n}{a} \sqrt {\sum_{j=1}^n  \left|\lambda_j ^{(1)}- \lambda_j ^{(2)}\right|^2},
\end{multline*}
where we used the Cauchy-Schwarz inequality in the last step. The Wielandt-Hoffman inequality (with $p=2$) tells us that  
$$\sqrt {\sum_{j=1}^n  \left|\lambda_j ^{(1)}- \lambda_j ^{(2)}\right|^2}\leq \|M_1-M_2\|_2.$$  
Therefore
\begin{align}\label{eq:splittingh2a}
\left|\sum_{j =1}^n   \left( h(n^\alpha \lambda_j ^{(1)}) - h(n^\alpha \lambda_j ^{(2)}) \right) (1-g(\lambda_j^{(2)}))\right|\leq \frac{|h|_{\mathcal L_w}  \sqrt n}{a} \|M_1-M_2\|_2.
\end{align}
This estimates the first term on the right-hand side of \eqref{eq:splittingh2}. For the second term, we note that 
$$  \left|\sum_{j=1}^n  \left(g(\lambda_j^{(1)})-g(\lambda_j^{(2)})\right) h(n^\alpha \lambda_j ^{(1)})\right|
\leq  \|h\|_\infty \left|\sum_{j=1}^n  \left(g(\lambda_j^{(1)})-g(\lambda_j^{(2)})\right) \right|$$
and by using the Lipschitz property of $g$, Cauchy-Schwarz and the Wielandt-Hoffman inequality, we obtain 
\begin{align}\label{eq:splittingh2b}
\left|\sum_{j=1}^n  \left(g(\lambda_j^{(1)})-g(\lambda_j^{(2)})\right) h(n^\alpha \lambda_j ^{(1)})\right|\leq {\|h\|_\infty |g|_{\mathcal L}  \sqrt n}\|M_1-M_2\|_2.
\end{align}
By substituting \eqref{eq:splittingh2a} and \eqref{eq:splittingh2b} into \eqref{eq:splittingh2} and using $\|h\|_\infty\leq |h|_{\mathcal L_w}$ we obtain 
$$|h_{2,\alpha}|_{\mathcal L}=\sup_{M_1,M_2} \frac{\left|\sum_{j=1}^n h_2(n^\alpha \lambda_j ^{(1)}) - h_2(n^\alpha \lambda_j ^{(2)})\right|}{\|M_1-M_2\|_2} \leq \sqrt n |h|_{\mathcal L_w}\left(\frac{1}{a}+ |g|_{\mathcal L}\right) .$$
By \eqref{eq:ineqherbst} we conclude 
\begin{align}\label{eq:aap2}
\left|\EE^0\left[\exp \lambda \left(\Tr h_{2,\alpha}(M)-\EE^0[\Tr h_{2,\alpha}(M)\right)\right]\right|\leq \exp |h|_{\mathcal L_w}^2 |\lambda |^2 B,
\end{align}
for some constant $B>0$ and $\lambda \in \C$. 

By using the  Cauchy-Schwarz inequality and \eqref{eq:aap1} and \eqref{eq:aap2}, we see that there exists universal constants $\tilde d_1,\tilde d_2\geq 0$ such that
\begin{align}\label{eq:ineqanalytic2a}
\left|\EE^0\left[\exp \lambda \left(\Tr h_{\alpha}(M)-\EE^0[\Tr h_{\alpha}(M)\right]\right)\right|\leq \exp \tilde d_2 |h|_{\mathcal L_w}^2 |\lambda |^2 ,
\end{align}
for $|\lambda| \leq \frac{1}{\tilde d_1 |h|_{\mathcal L_w}}$. Moreover, by increasing $\tilde d_1$  (if necessary) we can make sure that we always have 
\begin{align}\label{eq:ineqanalytic2alta}
\left|\EE^0\left[\exp \lambda \left(\Tr h_{\alpha}(M)-\EE^0[\Tr h_{\alpha}(M)\right] \right)-1\right|\leq \frac{1}{2},
\end{align}
for $|\lambda| \leq \frac{1}{d_1 |h|_{\mathcal L_w}}$. To finish the proof, we note that 
$$ \sup_{|s|\leq 1/2} \frac{|\log(1+s)|}{\log|1+s|}>\infty$$
and by combining this with  \eqref{eq:ineqanalytic2alta} and \eqref{eq:ineqanalytic2a}, the statement follows.
\end{proof}

\subsection{One more concentration inequality}

By combining the Propositions \ref{prop:localconcenpre} and \ref{prop:herbst} we obtain the following corollary.
\begin{corollary}\label{cor:boundmixed}
Let $I \subset U$ be compact,  let $m\in \N$,  $h_j \in \mathcal L_w$  with support in $I$ for $j=1,\ldots,m$,  and $G$ a functional satisfying \eqref{eq:lipschitz1}. Then there exists constants $c_1,c_2>0$ such that  for sufficiently large $n$  we have
	\begin{multline} \label{eq:ineqanalytic3}
		\left|  \log \EE \left[\exp  \left( \sum_{j=1}^m \lambda _j \left (\Tr h_j(M)-\EE[\Tr h_j(M)]\right) +\mu  \left(G-\EE [G]\right)\right) \right]\right|\\
	\leq  c_2 \left( \sum_{j=1}^m |\lambda_j|^2 |h_j|_{\mathcal L_w}^2 + \frac{1-q^2}{n}| \mu|^2 |G|_{\mathcal L}^2\right)
\end{multline}
for  $\mu\in \C$, $\lambda_1,\ldots,\lambda_m $ such that $|\lambda_j|\leq c_1 (\sup _j \|h_j\|_{\infty})^{-1}$ and $\xi\in \mathcal C(U,A,\delta)$. 
\end{corollary}
\begin{proof}
To bound the logarithm  we first bound $|\EE^0[\cdot]|$, for which we apply the H\"older inequality in general form to \eqref{eq:ineqanalytic3} and use \eqref{eq:localconcent} and \eqref{eq:ineqherbst}. The statement then follows by  following the arguments similar to the ones   below \eqref{eq:ineqanalytic2a}.
\end{proof}
\section{Proof of Lemma \ref{lem:estimateA}}\label{sec:estimatesloopeq}

 In this section we prove Lemma \ref{lem:estimateA}. We will work with the same assumptions \eqref{eq:para1}--\eqref{eq:para3}.  We also assume $f\in C_c^1(\R)$ and define $f^\eps = P^\eps *f$ as in \eqref{eq:deffeps}. Finally, we will assume without loss of generality that 
 $$x_*=0.$$
 The general case goes completely similar, but is notationally more cumbersome.
 \subsection{Preliminaries}

We first discuss some properties of the smoothened functions $f^\eps$ and derive some first consequences of Proposition \ref{prop:localconcenpre} that we will use later on.

\begin{lemma}
For $z\in \C$ we define the function $\phi_z:\R \to \C$ by $\phi_z(x)=1/(z-x)$. Then for any compact $S\subset \C\setminus \R$ we have 
\begin{align}\label{eq:supnormphiz}
\sup _{z\in S} |\phi_z|_{\mathcal L_w} <\infty.\end{align}
Moreover, for any $f\in C^1_c(\R)$  we have $f^\eps \in \mathcal L_w$ for $\eps> 0$. 
\end{lemma}
\begin{proof}
The proof of \eqref{eq:supnormphiz} follows easily from
\begin{equation}\label{eq:supnormphiza}
\sqrt{1+x^2}\sqrt{1+y^2} \frac{\phi_z(x)-\phi_z(y)}{x-y}=\frac{\sqrt{1+x^2}\sqrt{1+y^2}}{(z-x)(z-y)}.
\end{equation}
This implies that for each $z\in  \C\setminus \R$ we have $|\phi_z|_{\mathcal L_w}<\infty$ and since the norm depends continuously on $z$ we obtain a uniform bound for $z$ in compacta.

Moreover, for $\eps>0$ we have
\begin{multline}\nonumber
 \sqrt{1+x^2}\sqrt{1+y^2} \left|\frac{f^\eps(x)-f^\eps(y)}{x-y}\right|\\\leq \frac{1}{\pi} \int | f(s)|  \sqrt{1+x^2}\sqrt{1+y^2} \left|\frac{\phi_{s+{\rm i}\eps} (x)-\phi_{s+{\rm i}\eps} (y)}{x-y}\right|{\rm d} s.
\end{multline}
Combining this with  \eqref{eq:supnormphiz} and the fact that $f$ is of compact support we see that $f ^\eps\in \mathcal L_w$.
\end{proof}

The following results are already proved for compactly supported test functions, but we show that they extend to functions $f^\eps$.
\begin{lemma} Let $0<\alpha<1$ and define $\phi_{z,\alpha}(M)$
as
\begin{align} \label{eq:defphialpha}
\phi_{z,\alpha}(M)=\frac{1}{z-n^{\alpha} M}=\frac{1}{n^\alpha} \frac{1}{z/n^{\alpha}-M}.
\end{align}
Then we have
\begin{align} \label{eq:boundvariancephifunctions}
\begin{split}
 \EE^0[\Im \Tr \phi_{z,\alpha}(M)]&=\mathcal O(n^{1-\alpha}), \\
 \Var^0[\Im \Tr \phi_{z,\alpha}(M)]&=\mathcal O(1), 
\end{split}
\end{align}
as $n\to \infty$, where the constant in the order terms is uniform for $z$ in compact subsets of $\C\setminus \R$. Moreover, for $f\in C_c^1(\R)$ we have 
\begin{align} \label{eq:boundvarianceTfunctions}
\begin{split}
\EE^0[\Tr f_{\alpha}^\eps(M)]&=\mathcal O(n^{1-\alpha}), \\
 \Var^0[\Tr f_{\alpha}^\eps(M)]&=\mathcal O(1), 
\end{split}
\end{align}
 as $n\to \infty$, uniformly for $0<\eps<1$.
\end{lemma}
\begin{proof}  The statements on the  variance in both \eqref{eq:boundvariancephifunctions} and \eqref{eq:boundvarianceTfunctions} are a direct consequence of the fact that $f^\eps\in \mathcal L_w$ and Proposition \ref{prop:localconcenpre}.

The statement on the expectation in \eqref{eq:boundvarianceTfunctions} goes as follows
\begin{align*}
\EE^0[\Tr \phi_{z,\alpha}(M)]&= \frac{1}{n^\alpha} \int K_n(x,x)  \Im \frac{1}{z/n^\alpha-x} {\rm d} x\\
& =\frac{1}{n^\alpha} \int_I K_n(x,x)  \Im \frac{1}{z/n^\alpha-x} {\rm d} x+ \frac{1}{n^\alpha} \int_{\R\setminus I} K_n(x,x)  \Im \frac{1}{z/n^\alpha-x} {\rm d} x,
\end{align*}
where we take $I\subset U$ to be a compact interval around $x_*=0$.

By Corollary \ref{cor:asymptoticsdiagonal} we have $K_n(x,x) = \mathcal O(n)$ as $n\to \infty$ uniformly for $x\in I$. Moreover, since 
$$\int_\R   \Im \frac{1}{z/n^\alpha-x} {\rm d} x = \pm \pi,
 $$
 for $\pm \Im z >0$, we obtain 
 \begin{equation}\label{eq:I}
 \frac{1}{n^\alpha} \int_I K_n(x,x)  \Im \frac{1}{z/n^\alpha-x} {\rm d} x=\mathcal O(n^{1-\alpha}), 
 \end{equation}\\
 as $n \to \infty$.
 On the other hand, $\Im 1/(z/n^\alpha-x)=\mathcal O(1)$ as $n\to \infty$ uniformly for $x\in \R \setminus I$. Combining this with $K_n(x,x)\geq 0$ and $\int K_n(x,x){\rm d}x=n$ we find 
\begin{equation}\label{eq:IwegR}
 \frac{1}{n^\alpha} \int_{\R \setminus I} K_n(x,x)  \Im \frac{1}{z/n^\alpha-x} {\rm d} x=\mathcal O(n^{1-\alpha})). 
 \end{equation}
 By combining \eqref{eq:I} and \eqref{eq:IwegR}, we obtain the estimate on the expectation in \eqref{eq:boundvariancephifunctions}. The bound on the expectation in \eqref{eq:boundvarianceTfunctions} follows by \eqref{eq:boundvariancephifunctions} and  Fubini's Lemma.
\end{proof}

We also mention some standard identities for  resolvents that we will use frequently. We also recall the identities in Lemma \ref{lem:tracestuff}. 

 \begin{lemma}  \label{lem:resolventstuff} If $A$ is self-adjoint ($A^*=A$) matrix  and $w\in \C \setminus \R$, then 
 \begin{enumerate}
 \item   $\|(A-w)^{-1}\|_\infty\leq 1/|\Im w|$
 \item  \begin{multline}
\left\|\frac{1}{A-w} \right\|_2^2=\Tr  \frac{1}{A- \overline w }  \frac{1}{A-w} =\frac{1}{w-\overline w} \left(\Tr \frac{1}{A-w} -\Tr  \frac{1}{A-\overline w}  \right)\\=\frac{1}{\Im w} \Im \Tr  \frac{1}{A-w}.
\end{multline}
\end{enumerate}
 \end{lemma}
 
 We end this paragraph with an illustration of one of the main ideas for the upcoming proofs. Let $0\leq \alpha <1$.  Note that we can write
\begin{multline}
D_n^{\lambda h_\alpha} (z)= n^{\alpha} \frac{\partial }{\partial \mu} \log \EE^0 \left[\exp \lambda \left(\Tr h_\alpha(M)-\EE^0 [\Tr h_\alpha(M)\right)\right.\\+\left. \left.\mu \left(\Tr \phi_{z,\alpha}(M)-\EE^0 [\Tr \phi_{z,\alpha}(M)]\right)\right]\right|_{\mu=0},
\end{multline}
where $\phi_{z,\alpha}$ is defined as in \eqref{eq:defphialpha}.

Now note that for any function $F$ that is analytic in a neighborhood of the origin, we have by estimating  Cauchy's integral formula that $|F'(0)| \leq r^{-1}  \sup_{|z|=r_1} |F(z)| $ for any $r$ such that the ball $B_{0,r}$ is inside the neighborhood of analyticity.  Hence, by Corollary \ref{cor:boundmixed}, there exists a $a_1>0$ such that  for any compact $S\subset \C\setminus \R$  we have that there exists a constant  $a_2$ such that 
\begin{equation}\label{eq:firstestimateDn}
|D_n^{\lambda h_\alpha} (z)|\leq a_2 n^\alpha,
\end{equation}
for $z\in S$, $n\in \N$, $|\lambda |\leq a_1 /\|h\|_\infty $ and $h \in \mathcal L_w$.

 This idea of  writing the quantity that we are interested in (in this case  $D_n^{\lambda h}$) as a logarithmic derivative will be frequently used in the upcoming estimates.

\subsection{Estimating $A_n^{\lambda f_\alpha^\eps}$}
In this paragraph we will estimate $A_n^{\lambda f_\alpha^\eps}$. We start by a result on the imaginary part of $\zeta^{\lambda f_\alpha}$. 
\begin{lemma} \label{lem:anf}  For any compact set $S\subset \C\setminus \R$ and $f \in C_c^1(\R)$ there exists a constant $a>0$ such that for $n$ sufficiently large we have 
\begin{align} \label{eq:boundzeta2}
|\Im \zeta^{\lambda f_\alpha^\eps}(z/n^\alpha)|\geq a(1-q^2)\end{align}
for  $z\in S$, $\lambda $ in a sufficiently small neighborhood of the origin,  $0 <\eps<1$ and $\xi \in \mathcal C(U,A,\delta)$. 
\end{lemma}
\begin{proof}
We will only consider the case $\Im z>0$. The proof for $\Im z <0$ goes similarly.

Let us first deal with $\lambda =0$. In that case note that
\begin{equation} \label{eq:lemanf1}
\Im \zeta^0(z/n^{\alpha})=\frac{\Im z}{n^\alpha} -\frac{1-q^2}{n} \EE^0\left[\Im \Tr \frac{1}{z/n^\alpha-M}\right].
\end{equation}
By definition we have that  
\begin{equation}\label{eq:lemanf2}
-\frac{1}{n}  \EE^0\left[ \Im \Tr \frac{1}{z/n^{\alpha}-M}\right]= \frac{1}{n} \int K_n(x,x) \frac{\Im z/n^\alpha}{(\Re z/n^\alpha-x)^2+\Im z^2/n^{2\alpha}} {\rm d}x.
\end{equation}
And since, by Corollary \ref{cor:asymptoticsdiagonal},  there exists a  small interval $I$ around  $x_*=0$ such that $$K_n(x,x)\geq \tilde c n$$ uniformly for $x\in I $ and $n\in \N$ we see that
\begin{multline*}\frac{1}{n} \int K_n(x,x)\frac{\Im z/n^\alpha}{(\Re z/n^\alpha-x)^2+\Im z^2/n^{2\alpha}} {\rm d}x
\\
\geq \tilde c \int_I  \frac{\Im z/n^\alpha}{(\Re z/n^\alpha-x)^2+\Im z^2/n^{2\alpha}}{\rm d}x  
=\tilde c \int_{n^\alpha I}  \frac{\Im z}{(\Re z-x)^2+\Im z^2}{\rm d}x .
\end{multline*}
And hence the statement follows for $\zeta^0(z/n^\alpha)$. 

To prove the statement for general $\lambda$ we note that 
$$
\zeta^{\lambda f_\alpha^\eps}(z/n^{\alpha}) -\zeta^0(z/n^\alpha)= \frac{1-q^2}{n} D_n^{\lambda f_\alpha^\eps}(z/n^{\alpha}). $$
And therefore the result for general $\lambda$  follows by combining the result for $\lambda =0$ and \eqref{eq:firstestimateDn}. 
\end{proof}
Now we are ready to estimate $ A^{\lambda f_\alpha^\eps}_n$. 
\begin{proof}[Proof of \eqref{eq:lemestimateA} in Lemma \ref{lem:estimateA}]
By definition
\begin{equation}\nonumber
n^{\alpha} |A^{\lambda f_\alpha^\eps}_n(z/n^{\alpha})| =
{n^\alpha} \left(\frac{1-q^2}{2 n}\right)^2 \left|\Tr \frac{1}{\zeta^{\lambda \alpha}(z/n^\alpha) -q \Xi_n}\frac{1}{(\zeta^{0}(z/n^\alpha) -q \Xi_n)^2}\right|.
\end{equation}
We estimate the latter by using the standard identities in Lemma \ref{lem:tracestuff}, which leads to 
$$
n^{\alpha} |A^{\lambda f_\alpha^\eps}_n(z/n^{\alpha})| \leq {n^\alpha} \left(\frac{1-q^2}{2n}\right)^2 \left\|\frac{1}{\zeta^{\lambda f_\alpha^\eps}(z/n^\alpha) -q \Xi_n}\right\|_\infty\left\|\frac{1}{\zeta^{0}(z/n^\alpha) -q \Xi_n}\right\|_2^2.
$$
By using \eqref{eq:boundzeta2} we can see that for any compact set $S\subset \C \setminus \R$ there exists a constant $a>0$ such that
$$
n^{\alpha} |A^{\lambda f_\alpha^\eps}_n(z/n^{\alpha})| \leq 
 \frac{(1-q^2)a_2}{2n^{2-\alpha}} \left\|\frac{1}{\zeta^{0}(z/n^\alpha) -q \Xi_n}\right\|_2^2.
$$
for $z\in S$, $\lambda$ sufficiently small neighborhood of the origin, $0 <\eps<1$ and $\xi \in \mathcal C(U,A,\delta)$. 

 Finally, we note that by the second identity in Lemma \ref{lem:resolventstuff}
$$\left\|\frac{1}{\zeta^{0}(z/n^\alpha) -q \Xi_n}\right\|_2^2=-\frac{1}{\Im \zeta^0(z/n^\alpha)} \Im  \Tr \frac{1}{\zeta^{0}(z/n^\alpha) -q \Xi_n}.
$$ By using  \eqref{eq:boundzeta2} once  more and the fact that $\xi \in \mathcal C(U,A,\delta)$   we have 
$$
n^{\alpha} |A^{\lambda f_\alpha^\eps}_n(z/n^{\alpha})| \leq 
 \mathcal O (n^{\alpha-1}), \qquad n\to \infty,
$$
where the order is uniform for $z$ in compact subset of $\C\setminus\R$, $\lambda$ in a sufficiently small neighborhood of the origin, $0<\eps<1$ and $\xi \in \mathcal C(U,A,\delta)$. 
\end{proof}

\subsection{Estimating   $B_n^{\lambda f_\alpha^\eps}$}
We first prove the following lemma. 
\begin{lemma}\label{lem:estimateBhelp}
Let $f\in C_c^1(\R)$. Then 
\begin{equation} \label{eq:estimateBn}
B^{\lambda f_\alpha^\eps}_n(z/n^{\alpha})=1+\frac{1-q^2}{n} \Tr \frac{1}{(\zeta^0(z/n^\alpha)-q \Xi_n)^2}+\mathcal O(n^{\alpha-1}).
\end{equation}
as $n\to \infty$, uniformly for $\lambda$ in a sufficiently small nieghborhood of the origin, $0<\eps<1$, $z$ in compact subsets of $\C\setminus \R$ and $\xi \in \mathcal C(U, A,\delta)$. 
\end{lemma}
\begin{proof}
 From the definition of $B^{\lambda h_\alpha}_n$ in \eqref{eq:defABC} we see that to prove \eqref{eq:estimateBn} we need to show that 
\begin{align}\label{eq:estimateBaa}
\left(\frac{1-q^2}{2n}\right)^2 \EE^{\lambda f_\alpha^\eps} \left[\Tr \frac{{f_\alpha^\eps}'(M)}{z/n^{\alpha} -M} \frac{1}{\zeta^{\lambda f_\alpha^\eps}(z/n^{\alpha})-q\Xi_n}
\frac{1}{\zeta^0(z/n^{\alpha})-q\Xi_n} \right]=\mathcal O(n^{\alpha-1})
\end{align} 
as $n\to \infty$ uniformly for  $\lambda$ in a neighborhood of the origin, $0<\eps<1$, $z$ in compact subsets of $\C\setminus \R$ and $\xi \in \mathcal C(U,A, \delta)$. 
To this end, we note that 
\begin{equation}\label{eq:estimateBbb}
\Tr {f_\alpha^\eps}'(M)=\frac{1}{\pi  n^{\alpha}} \int \ f(x)  \Im \Tr \left(\frac{1}{(x-{\rm i}\eps)/n^\alpha-M}\right)^2 {\rm d} x.
\end{equation}
 To prove \eqref{eq:estimateBn}  it is therefore sufficient to prove 
\begin{align}\label{eq:estimateBcc}
\frac{(1-q^2)^2}{n^{2+\alpha}}\left|\EE^{\lambda f_\alpha^\eps}\left[\Tr \frac{1}{z/n^\alpha-M} \left( \frac{1}{w/n^\alpha-M}\right)^2\frac{1}{\zeta^0(z/n^\alpha)-q Y}\frac{1}{\zeta^{\lambda f_\alpha}(z/n^\alpha)-q Y} \right]\right|=\mathcal O(n^{\alpha-1})
\end{align}
as $n\to \infty$ uniformly for $\lambda$ in a neighborhood of the origin $z,w$ in compact subsets of $\C\setminus \R$,  $0<\eps<1$ and $\xi \in \mathcal C(U, A,\delta)$. 

  By \eqref{eq:defEh}, \eqref{eq:localconcent}, \eqref{eq:supnormphiz} and the Cauchy-Schwarz inequality, we see that  there exists a constant $\hat c$ such that 
\begin{multline*}
\left|\EE^{\lambda f_\alpha^\eps}\left[\Tr \frac{1}{z/n^\alpha-M} \left( \frac{1}{w/n^\alpha-M}\right)^2\frac{1}{\zeta^0(z/n^\alpha)-q Y}\frac{1}{\zeta^{\lambda f_\alpha^\eps}(z/n^\alpha)-q Y} \right]\right|
\\\leq  \hat c \left(\EE^0 \left[\left|\Tr \frac{1}{z/n^\alpha-M} \left( \frac{1}{w/n^\alpha-M}\right)^2\frac{1}{\zeta^0(z/n^\alpha)-q Y}\frac{1}{\zeta^{\lambda f_\alpha^\eps}(z/n^\alpha)-q Y} \right|^2\right]\right)^{1/2}.
\end{multline*} 
 By   invoking  the standard  identities in Lemma \ref{lem:tracestuff}, we can rewrite this inequality as
\begin{multline*}
\left|\EE^{\lambda f_\alpha^\eps}\left[\Tr \frac{1}{z/n^\alpha-M} \left( \frac{1}{w/n^\alpha-M}\right)^2\frac{1}{\zeta^0(z/n^\alpha)-q Y}\frac{1}{\zeta^{\lambda f_\alpha}(z/n^\alpha)-q Y} \right]\right|
\\
\leq  \hat c  \left\|\frac{1}{\zeta^0(z/n^\alpha)-q Y}\right\|_\infty\left\|\frac{1}{\zeta^{\lambda f_\alpha^\eps}(z/n^\alpha)-q Y} \right\|_\infty \left(\EE^0 \left[\left\|\frac{1}{z/n^\alpha-M}\right\|_\infty^2 \left\| \frac{1}{w/n^\alpha-M}\right\|_2^4\right]\right)^{1/2}.
\end{multline*}
Moreover, we also use $\|(z-A)^{-1}\|_\infty \leq (\Im z)^{-1}$ and \eqref{eq:boundzeta2} to deduce that  for every  compact $S\subset \C\setminus \R$ there exists a constant $\tilde c$ such that 
\begin{multline*}
\left|\EE^{\lambda f_\alpha^\eps}\left[\Tr \frac{1}{z/n^\alpha-M} \left( \frac{1}{w/n^\alpha-M}\right)^2\frac{1}{\zeta^0(z/n^\alpha)-q Y}\frac{1}{\zeta^{\lambda f_\alpha^\eps}(z/n^\alpha)-q Y} \right]\right|
\\ \leq \frac{ \tilde c n^{\alpha}}{(1-q^2)^2|\Im z|}  \left(\EE^0 \left[\left\|\frac{1}{w/n^\alpha-M}\right\|_2^4\right]\right)^{1/2}.
\end{multline*}
for $\lambda$ in a neighborhood of the origin, $0<\eps<1$, $z\in S$ and $\xi \in \mathcal C(U, A,\delta)$.  By the second identity in Lemma \ref{lem:resolventstuff} we can further simplify the right-hand side to 
\begin{multline*}
\left|\EE^{\lambda f_\alpha^\eps}\left[\Tr  \frac{1}{z/n^\alpha-M} \left( \frac{1}{w/n^\alpha-M}\right)^2\frac{1}{\zeta^0(z/n^\alpha)-q Y}\frac{1}{\zeta^{\lambda f_\alpha}(z/n^\alpha)-q Y} \right]\right|
\\\leq \frac{ \tilde c n^{2\alpha}}{|\Im z|| |\Im w| (1-q^2)^2}  \left(\EE^0 \left[\left(\Im \Tr \frac{1}{w/n^\alpha-M}\right)^2\right]\right)^{1/2}
\\=\frac{ \tilde c n^{2\alpha}}{|\Im z| |\Im w| (1-q^2)^2}  \left(\Var \left[\Im \Tr \frac{1}{w/n^\alpha-M}\right]+ \left(\EE^0 \left[\Im \Tr \frac{1}{w/n^\alpha-M}\right]\right)^2\right)^{1/2}.
\end{multline*}
By combining this with the estimates in \eqref{eq:boundvariancephifunctions} we have \eqref{eq:estimateBcc}. By substituting \eqref{eq:estimateBcc} into \eqref{eq:estimateBbb} we obtain \eqref{eq:estimateBaa} and this proves the statement. 
\end{proof}
We are now ready to prove \eqref{eq:lemestimateB} in Lemma \ref{lem:estimateA}.
\begin{proof}[Proof of \eqref{eq:lemestimateB} in Lemma \ref{lem:estimateA}]
The proof follows by \eqref{eq:estimateBn}.  Since $\xi^{(n)}\in \mathcal C_n(U,An^\delta)$ we can argue as in the proof for the estimate for $\mathcal E_2$ in  \eqref{eq:condone} with $\Omega(x)$ replaced by $\zeta^0$ (where we also recall \eqref{eq:boundzeta2}) and  the fact that $q\to 1$ to obtain 
$$\frac{1-q^2}{n} \Tr \frac{1}{(\zeta^0(z/n^\alpha)-q \Xi_n)^2}=o(1),$$
as $n\to \infty$, uniformly under the conditions in the statement. This proves the statement.
\end{proof} 

\subsection{Estimating $C^{\lambda f_\alpha^\eps}_n$ for $0<\alpha<1/2$}
Now that we have estimated $A^{\lambda f_\alpha^\eps}_n$ and $B^{\lambda f_\alpha^\eps}_n$ it remains to estimate $C^{\lambda f_\alpha^\eps}_n$. We will prove that the leading asymptotic part is linear in $\lambda$. In case $0<\alpha<1/2$ this can be directly proved by rather straightforward estimates based on representing the involved quantities in terms of logarithmic derivatives and using \eqref{eq:ineqherbst} and \eqref{eq:ineqanalytic3}. In the general situation the proof follows by using the same (but now insufficient) estimates but  together with extra iterations of  the loop equations. This will be done in the next paragraph. 
\begin{lemma}\label{lem:Lemmaalphahalf}
Let $f\in C_c^1(\R)$ and assume that $0<\alpha<1$.  Then 
\begin{equation}\label{eq:DNlinearalphahalf}
\frac{1}{n^\alpha} C_n^{\lambda f_\alpha^\eps}(z/n^\alpha) =-\lambda \frac{1-q^2}{n^{1+ \alpha}}\EE^0\left[\frac{{f_\alpha^\eps}'(M)}{z/n^\alpha-M} \frac{1}{\zeta^0(z/n^\alpha)-q\Xi_n}\right] +\mathcal O(n^{2 \alpha-1}),
\end{equation}
as $n\to \infty$, uniformly for $\lambda$ in a neighborhood of the origin, $0<\eps<1$, $z$ in compact subsets of $\C\setminus \R$ and $\xi \in \mathcal C(U, A,\delta)$.
\end{lemma}

\begin{proof}
In the proof we use the following standard principle for analytic functions. Let $F$ be an analytic function in a  neighborhood of a disk $D_r$. Then by Cauchy's integral formula we have
\begin{equation}
\label{eq:ineqanalytic}
|F'(w)|\leq \frac{r  \|F\|_{\mathbb{L}_\infty(\partial D_r)}}{ {\dist}(w,\partial D_r)^{2}}.
\end{equation}
Similarly, if $F$ is an analytic function of two variables
\begin{equation}
\label{eq:ineqanalytic2}
\left|\frac{\partial^2}{\partial w_1 \partial w_2}F(w_1,w_2)\right|\leq \frac{r_1r_2  \|F\|_{\mathbb{L}_\infty(\partial D_{r_1} \times \partial D_{r_2})}}{ {\dist}(w_1,\partial D_{r_1})^{2}{\dist}(w_2,\partial D_{r_2})^{2}}.
\end{equation}
Let us recall  and slightly rewrite the definition of $C_n^{\lambda f_\alpha^\eps} $ in  \eqref{eq:defABC}
\begin{multline}\label{eq:defC2}
C_n^{\lambda f_\alpha^\eps}(z)=-\lambda  \frac{1-q^2}{2n}  \EE^{0}\left[\Tr \frac{{f_\alpha^\eps}'(M)}{z-M} \frac{1}{\zeta^0(z)-q \Xi_n}\right]\\
-\lambda  \frac{1-q^2}{2n} \left( \EE^{\lambda f_\alpha^\eps}\left[\Tr \frac{{f_\alpha^\eps}'(M)}{z-M} \frac{1}{\zeta^0(z)-q \Xi_n}\right]-\lambda  \frac{1-q^2}{n}  \EE^{0}\left[\Tr \frac{{f_\alpha^\eps}'(M)}{z-M} \frac{1}{\zeta^0(z)-q \Xi_n}\right]\right)\\
 -\tfrac{1-q^2}{2n}\left(K_n^{\lambda {f_\alpha^\eps,I}}(z)-K^{0,I}(z) \right).
\end{multline}
 To prove \eqref{eq:DNlinearalphahalf} we need to estimate the last two terms of \eqref{eq:defC2}.
Let us start with the first of those two. We note that
\begin{multline} \label{eq:chalfirsta}
\frac{1-q^2}{2n^{1+\alpha}} \EE^{\lambda f_\alpha^\eps} \left[\Tr \frac{{f_\alpha^\eps}'(M)}{z/n^\alpha-M} \frac{1}{\zeta^{0}(z/n^\alpha)-q\Xi_n}\right]
\\
=\frac{1-q^2}{2n^{1+2\alpha}} \int_\R \EE^{\lambda{ f_\alpha^\eps}} \left[\Tr \frac{f(s)}{z/n^\alpha-M}\left(\frac{1}{(s-{\rm i} \eps)/n^\alpha-M}\right)^2 \frac{1}{\zeta^{0}(z/n^\alpha)-q\Xi_n}\right] {\rm d}s.
\end{multline}
We remark that for any $0<\alpha<1$, the function ${f_\alpha^\eps}'(M)$ here is to be interpreted as 
 \begin{align}\label{eq:halphaprime}
{ f_\alpha^\eps}'(M)=\frac{1}{\pi n^\alpha} \int_\R  f(s) \Im \Tr  \left(\frac{1}{(s-{\rm i} \eps)/n^\alpha-M}\right)^2 {\rm d} s,
 \end{align}
 which follows by a rescaling of the integration variable in \eqref{eq:interphprimeM}.  
 
Therefore it suffices to analyze 
\begin{equation*}
\frac{1-q^2}{2n^{1+2\alpha}} \left(H^{\lambda f_\alpha^\eps}(z/n^\alpha,w/n^\alpha)-H^0(z/n^\alpha,w/n^\alpha)\right)
\end{equation*}
asymptotically as $n\to \infty$, where 
\begin{equation} \nonumber
H^{\lambda f_\alpha^\eps}(z/n^\alpha,w/n^\alpha)= \EE^{\lambda f_\alpha^\eps} \left[\Tr \frac{1}{z/n^\alpha-M}\left( \frac{1}{w/n^\alpha-M} \right)^2\frac{1}{\zeta^{0}(z/n^\alpha)-q\Xi_n}\right].
\end{equation}
As before, we write this difference in terms of a logarithmic derivative
\begin{multline}\label{eq:chalfirstaa}
\frac{1-q^2}{2n^{1+2\alpha}} \left(H^{\lambda f_\alpha^\eps}(z/n^\alpha,w/n^\alpha)-H^0(z/n^\alpha,w/n^\alpha)\right)\\
=\frac{1-q^2}{2n} \left. \frac{\partial }{\partial \mu }\log \EE^0\left[\exp \left(\lambda \left(\Tr f_\alpha^\eps -\EE^0[\Tr f_\alpha^\eps]\right)+\mu \left( g_\alpha^{(1)} -\EE^0[ g_\alpha^{(1)}]\right)\right)\right]\right|_{\mu=0}.
\end{multline}
where 
$$g^{(1)}(M)=\frac{1}{n^{2\alpha}}\Tr \frac{1}{z/n^{\alpha} -M}\left(\frac{1}{w/n^{\alpha} -M}\right)^2\frac{1}{\zeta^{\lambda f_\alpha^\eps}(z/n ^{\alpha})-q\Xi_n}$$
It is not hard to check using \eqref{eq:boundzeta2} and the ideas leading to \eqref{eq:ineqfirstconsequenceherbst} , that $$|g^{(1)}|_{\mathcal L} =\mathcal O\left(n^{2 \alpha+1/2}/(1-q^2)\right),$$ as $n\to \infty$ uniformly under the conditions given in the statement. Hence by \eqref{eq:chalfirstaa}, \eqref{eq:ineqanalytic}  and \eqref{eq:ineqherbst} we have  
\begin{equation}\label{eq:DNlinearalphaA}
\frac{1-q^2}{n^{1+2\alpha}} \left(H^{\lambda f_\alpha^\eps}(z/n^\alpha,w/n^\alpha)-H^0(z/n^\alpha,w/n^\alpha)\right)=\mathcal O(n^{2 \alpha-1}), 
\end{equation}
as $n\to \infty$, uniformly under the conditions given in the statement.  By substituting this into \eqref{eq:chalfirsta} we see that the first of the last two terms  in  \eqref{eq:defC2} is of order $\mathcal O(n^{2\alpha-1})$ as $n\to \infty$. 

Now we come to $K^{\lambda f_\alpha^\eps,I}$, which is in the last term in the expression for $C^{\lambda f_\alpha^\eps}_n$ in  \eqref{eq:defC2}. With $\phi_{z,\alpha}(M)=\frac{1}{n^\alpha}\Tr (z/n^\alpha-M)^{-1},$ this term can be written in terms of a logarithmic derivative as 
\begin{multline*}
\frac{1-q^2}{2n^\alpha} K^{\lambda f_\alpha^\eps,I}(z/n^\alpha)=\frac{1}{2}\frac{\partial^2 }{\partial \mu_1 \partial \mu_2 }\log \EE^0\left[\exp \left(\lambda \left(\Tr f^\eps_{\alpha} -\EE^0[\Tr f^\eps_{\alpha}]\right)\right.\right.\\ \left. \left. \left.+\mu_1\left(\Tr \phi_{z,\alpha} -\EE^0[\Tr \phi_{z,\alpha}]\right)+\mu_2 \left(g^{(2)} -\EE^0[g^{(2)}]\right)\right)\right]\right|_{\mu_1=\mu_2=0}.
\end{multline*}
where 
$$g^{(2)}(M)=\Tr \frac{1}{z/n^{\alpha} -M}\frac{1}{\zeta^{\lambda f_\alpha^\eps}(z/n ^{\alpha})-q\Xi_n}$$
Also in this case we have  $|g^{(2)}|_{\mathcal L} =\mathcal O\left(n^{2 \alpha+1/2}/(1-q^2)\right)$ so that by \eqref{eq:ineqanalytic2} and \eqref{eq:ineqherbst}  we have  
\begin{equation}\label{eq:DNlinearalphaB}
\frac{1-q^2}{n^\alpha} K^{\lambda f_\alpha^\eps,I}(z/n^\alpha)=\mathcal O(n^{2 \alpha-1}), 
\end{equation}
as $n\to \infty$ uniformly under the conditions given in the statement.  In particular, it holds for $\lambda =0$ with which we bound the term with $K^{0,I}$.  

Concluding, by combining \eqref{eq:DNlinearalphaA}, \eqref{eq:DNlinearalphaB} and \eqref{eq:defC2} we obtain \eqref{eq:DNlinearalphahalf} for $0<\alpha<1/2$. 
\end{proof}

\subsection{Estimating $C^{\lambda f_\alpha^\eps}_n$ for $0<\alpha<1$}

We now extend Proposition \ref{lem:Lemmaalphahalf} to the case of general $0<\alpha<1$. The key idea is to iterate the loop equations sufficiently often and use the self-improving mechanism. This is done in the proof of the following lemma.

\begin{lemma} let $\alpha, \gamma,\delta$ and $U$ be as in Lemma \ref{lem:estimateA}.  Let $\phi_{z,\alpha}$ as in \eqref{eq:defphialpha}. For any $n\times n$ matrix $J$ we have
\begin{equation}\label{eq:linear1}\EE^{\lambda f_\alpha^\eps +\mu \phi_{\alpha,z}}\left[\frac{1}{z-M} J \right]= \Tr \frac{1}{\zeta^{0}(z/n^{\alpha})-q \Xi_n} J+\mathcal O(n^\alpha),
\end{equation}
as $n\to \infty$ uniformly for $\lambda, \mu$ in sufficiently small neighboorhoods of the origin, $0<\eps<1$ and for  $z$ in compact subsets  of $\C \setminus \R$. Moreover, for any $r>0$ both the constant in the order term and the neighborhoods  for $\lambda, \mu$ can be chosen   uniformly for $J$ such that $ \|J\|_\infty<r$.
\end{lemma}
\begin{proof}
The core of the proof is to write the loop equation \eqref{eq:loopeqn} in such a way that it defines a recurrence.

 Let $m\in \N$. Then for $\ell=0,\ldots,m-1$  we consider the functions 
\begin{equation}\nonumber
G^{J_\ell,\alpha}(z, \lambda,\mu)=
\EE^{\lambda f_\alpha^\eps +\mu \phi_{\alpha,z}}\left[\Tr \frac{1}{z-M} J_{\ell} \right]- \Tr J_{\ell} \frac{1}{\zeta^0 (z/n^{\alpha})-q \Xi_n} 
\end{equation}
and 
\begin{multline*}
g^{J_\ell,\alpha}(z,\mu, \lambda)=\frac{1}{n}\EE^{\lambda f_\alpha^\eps +\mu \phi_{z,\alpha} } \left[\Tr \frac{\lambda {f_\alpha^\eps}'(M) + \mu \phi_{z,\alpha}'(M)}{z/n^\alpha-M}   J_\ell \frac{(1-q^2)/2}{\zeta^{\lambda {f_\alpha^\eps} +\mu \phi_z}(z/n^\alpha)-M}\right]\\
+\left(\zeta^0(z/n^\alpha)-\zeta^{\lambda {f_\alpha^\eps} + \mu \phi_{z,\alpha}}(z/n^\alpha)\right) \Tr \frac{1}{\zeta^0(z/n^\alpha)-q \Xi_n} \frac{1}{\zeta^{\lambda {f_\alpha^\eps}+\mu \phi_{z,\alpha}}(z/n^\alpha)-q \Xi_n} J_\ell
\end{multline*}
where the matrix   $J_\ell$ is defined as  
$$J_\ell=J_\ell(\mu_1,\ldots,\mu_\ell)=J \frac{(1-q^2)/2}{\zeta^{\lambda f_\alpha^\eps+\mu_1 \phi}(z/n^\alpha)-q\Xi_n} \cdots \frac{(1-q^2)/2}{\zeta^{\lambda f_\alpha^\eps+\mu_{\ell} \phi}(z/n^\alpha) -q\Xi_n}, $$ and $J_0=J$. Note that from \eqref{eq:boundzeta2} it follows that for compact set $S\subset \C\setminus \R$ there exists a neighborhood $V$ and a constant $c>0$ such that 
\begin{align} \label{eq:boundJell}
\|J_{\ell}\|\leq c \|J\|,\end{align} 
uniformly for $z\in S$, $0<\eps<1$ and $\lambda, \mu, \mu_j \in V$ for $j=1,\ldots, \ell$ and $\ell=1,\ldots,m$.

Then from the loop equation \eqref{eq:loopeqn} we learn that $G^{J_\ell,\alpha}$ can be expressed in terms of $G^{J_{\ell+1},\alpha}$ as follows
\begin{equation}\label{eq:iterationA}
G^{J_{\ell},\alpha}(z,\mu, \lambda)=g^{J_{\ell},\alpha}(z,\mu, \lambda)+\left.\frac{n^\alpha}{n} \frac{\partial }{\partial \mu} G^{J_{\ell+1},\alpha}(z,\mu, \lambda)\right|_{\mu_{\ell+1}=\mu}
\end{equation}
where the last term is a different way of writing $K^{\lambda f_\alpha^\eps,J_\ell}$. 

Now \eqref{eq:linear1} is proved as follows. Fix a value of $m$.  We claim that   for each compact set $S\subset \C\setminus \R$ there exists a constant $A$ and a neighborhood $V$ of the origin such that 
\begin{align}\label{eq:iterationB}
|g^{J_m,\alpha}(z,\mu, \lambda)|\leq A n^\alpha, 
\end{align}
and 
\begin{align}\label{eq:iterationC}
\left|\left. \frac{\partial }{\partial \mu} G^{J_{m},\alpha}(z,\mu, \lambda)\right|_{\mu_{m}=\mu}\right|\leq c n^{2\alpha}, 
\end{align}
 for $z\in S$ and $\lambda, \mu, \mu_j \in V$ for $j=1,\ldots, m$ and $0<\eps<1$.  Indeed, \eqref{eq:iterationB} follows by the same arguments as in the proofs of \eqref{eq:lemestimateA}  and \eqref{eq:lemestimateB}.  Moreover, \eqref{eq:iterationC} follows by the same arguments used to prove  \eqref{eq:DNlinearalphaB}.  The main difference in the arguments is that we have a matrix $J_m$. However, in the estimates this term gives only contributes to irrelevant constants by \eqref{eq:boundJell}. 

 Then by using \eqref{eq:iterationA}, \eqref{eq:iterationB} and  \eqref{eq:iterationC} we have
\begin{align}\label{eq:iterationD}
|G^{J_{m-1},\alpha}(z,\mu, \lambda)|\leq c n^\alpha +c \frac{n^{3\alpha}}{n}
\end{align}
for $z \in S$ and $\lambda, \mu,\mu_j$ in a sufficiently small neighborhood of the origin. By using \eqref{eq:iterationA},\eqref{eq:iterationD} and \eqref{eq:ineqanalytic} again, we obtain an estimate for $G^{J_{m-2},\alpha}(z,\mu, \lambda)$ in a  slightly smaller neighborhood. By using the same argument repeatedly we see that there exists a constant $\tilde c>0$ such that 
\begin{align}\nonumber
|G^{J,\alpha}(z,\mu, \lambda)|=|G^{J_0,\alpha}(z,\mu, \lambda)|\leq \sum_{j=0}^{m-1}\tilde c  \left(\frac{n^\alpha}{n}\right)^j n^\alpha+\tilde c n^{2\alpha}  \left(\frac{n^\alpha}{n}\right)^m,
\end{align}
for $z\in S$, $0<\eps<1$ and $\lambda ,\mu$ in a sufficiently small neighborhood of the origin. By taking $m$ such that $m(\alpha-1)<-1$ we obtain \eqref{eq:linear1}. 
\end{proof}
We now come to the proof of \eqref{eq:DNlinearalphahalf2} in Lemma \ref{lem:estimateA}.
\begin{proof}[Proof of \eqref{eq:DNlinearalphahalf2} in Lemma \ref{lem:estimateA}]
Let us introduce one more notation. We define $D^{\lambda f_\alpha^\eps +\mu \phi_{z,\alpha},J}$ by 
\begin{equation}\nonumber
D^{\lambda f_\alpha^\eps+\mu \phi_{z,\alpha},J}_n(z/n^\alpha)=\frac{1}{n^\alpha} \left(\EE^{\lambda f_\alpha^\eps +\mu \phi_{z,\alpha}}\left[\Tr \frac{1} {z/n^\alpha-M}J\right]-\EE^{0}\left[\Tr \frac{1} {z/n^\alpha-M}J\right]\right).
\end{equation}
By taking the difference of \eqref{eq:linear1} for general $\lambda$ and \eqref{eq:linear1} with $\lambda =0$, we have the following estimate
\begin{align}\label{eq:linear1a}
D^{\lambda f_\alpha¨\eps+\mu \phi_{z,\alpha},J}_n(z/n^\alpha)=\mathcal O(1), 
\end{align}
as $n\to \infty$ uniformly for sufficiently small $\lambda$ and $\mu$, $0,\eps<1$ and $z$ in compact subsets of $\C\setminus \R$. 

As in the proof of Proposition \ref{lem:Lemmaalphahalf} we need to estimate the last two terms in \eqref{eq:defC2} and show that we have better bounds than \eqref{eq:DNlinearalphaA} and \eqref{eq:DNlinearalphaB}.

Let us start with $K^{\lambda f_\alpha^\eps,0}$.  Note that 
\begin{align}\nonumber
\frac{1-q^2}{2n^{1+\alpha}} K^{\lambda f_\alpha^\eps,0}_n (z/ n^{\alpha})= \frac{n^\alpha}{2n} \left.\frac{\partial }{\partial \mu } D^{\lambda f_\alpha^\eps +\mu \phi_{z,\alpha},J}_n(z/n^\alpha)\right|_{\mu=0}
\end{align}
where 
$$J=\frac{1-q^2}{\zeta^{\lambda f_\alpha^\eps}(z/n^\alpha)-q \Xi_n}.$$ 
Hence, by \eqref{eq:linear1a} we have 
\begin{equation} \label{eq:generalA}
\frac{1-q^2}{2n^{1+\alpha}} K^{\lambda f_\alpha^\eps}_n (z/ n^{\alpha})=\mathcal O(n^{\alpha-1}), 
\end{equation}
as $n\to \infty$ uniformly for $\lambda$ in a sufficiently small neighborhood of the origin, $0<\eps<1$ and $z$ in compact subset of $\C\setminus \R$.  

For the middle term in \eqref{eq:defC2} we note that 
\begin{multline}\label{eq:generalB}
\frac{1-q^2}{2n^{1+\alpha}} \EE^{\lambda f_\alpha^\eps} \left[\Tr \frac{{f_\alpha^\eps}'(M)}{z/n^\alpha-M} \frac{1}{\zeta^{0}(z/n^\alpha)-q\Xi_n}\right]
\\
=\frac{1-q^2}{2n^{1+\alpha}} \int \EE^{\lambda {f_\alpha^\eps}} \left[\Tr \frac{ f'(s)}{z/n^\alpha-M}\frac{1}{s-M-{\rm i} \eps} \frac{1}{\zeta^{0}(z/n^\alpha)-q\Xi_n}\right] {\rm d}s
\end{multline}
Therefore it suffices to analyze 
\begin{multline}\nonumber
\frac{1-q^2}{2n^{1+\alpha}} \EE^{\lambda f_\alpha^\eps} \left[\Tr \frac{1}{z/n^\alpha-M} \frac{1}{w/n^\alpha-M} \frac{1}{\zeta^{0}(z/n^\alpha)-q\Xi_n}\right]\\
-\frac{1-q^2}{2n^{1+\alpha}} \EE^0\left[\Tr \frac{1}{z/n^\alpha-M} \frac{1}{w/n^\alpha-M} \frac{1}{\zeta^{0}(z/n^\alpha)-q\Xi_n}\right]
\end{multline}
asymptotically as $n\to \infty$. To this end, note that 
\begin{multline}\nonumber
\frac{1-q^2}{2n^{1+\alpha}} \EE^{\lambda f_\alpha^\eps} \left[\Tr \frac{1}{z/n^\alpha-M} \frac{1}{w/n^\alpha-M} \frac{1}{\zeta^{0}(z/n^\alpha)-q\Xi_n}\right]
\\
=\frac{1-q^2}{2(z-w)n}  \int_w^z \frac{\partial}{\partial \eta} \EE^{\lambda f_\alpha^\eps}  \left[\Tr \frac{1}{\eta/n^\alpha-M}\frac{1}{\zeta^{0}(z/n^\alpha)-q\Xi_n}\right] {\rm d} \eta
\end{multline}
and therefore 
\begin{multline}\label{eq:generalC}
\frac{1-q^2}{2n^{1+\alpha}} \EE^{\lambda f_\alpha^\eps} \left[\Tr \frac{1}{z/n^\alpha-M} \frac{1}{w/n^\alpha-M} \frac{1}{\zeta^{0}(z/n^\alpha)-q\Xi_n}\right]\\
-\frac{1-q^2}{2n^{1+\alpha}} \EE^0\left[\Tr \frac{1}{z/n^\alpha-M} \frac{1}{w/n^\alpha-M} \frac{1}{\zeta^{0}(z/n^\alpha)-q\Xi_n}\right]
\\
= \frac{n^\alpha}{(z-w)n }  \int_w^z \frac{\partial}{\partial \eta}\left( D^{\lambda f_\alpha^\eps,J}(\eta/n^\alpha)\right)  {\rm d} \eta
\end{multline}
with $J=\frac{(1-q^2)/2}{\zeta^0(z/n^\alpha)-q \Xi_n}$ (which does not depend on the integration variable $\eta$).  By \eqref{eq:linear1a} and analyticity we obtain a bound for the left-hand side of \eqref{eq:generalC}, which after using \eqref{eq:generalB} leads to 
\begin{multline} \label{eq:generalD}
\frac{1-q^2}{2n^{1+\alpha}} \EE^{\lambda f_\alpha} \left[\Tr \frac{{f^\eps_\alpha}'(M)}{z/n^\alpha-M}\frac{1}{\zeta^{0}(z/n^\alpha)-q\Xi_n}\right]\\
-\frac{1-q^2}{2n^{1+\alpha}} \EE^0\left[\Tr \frac{{f_\alpha^\eps}'(M)}{z/n^\alpha-M} \frac{1}{\zeta^{0}(z/n^\alpha)-q\Xi_n}\right]
=\mathcal O(n^{\alpha-1}), 
\end{multline}
as $n\to \infty$ uniformly for $\lambda$ in a sufficiently small neighborhood of the origin, $0<\eps<1$ and $z$ in compact subset of $\C\setminus \R$.  

The statement now follows by inserting \eqref{eq:generalA} and \eqref{eq:generalD} into the representation of $C_n^{\lambda f_\alpha^\eps}$ given in  \eqref{eq:defC2} and this finishes the proof.\end{proof}
\section{Random initial points}

In this section  we prove Theorem \ref{th:random1} and part of Theorem \ref{th:random2}. We will only prove the last theorem under the assumption $\mu_0(f)\neq 0$. The general case is  significantly more cumbersome from a technical point of view and will be treated in a separate section. Again, we will assume \eqref{eq:para1}--\eqref{eq:para3} and, as in the previous section, we will assume without loss of generality that 
$$x_*=0.$$

The idea behind the proof  is the following. We recall that $\EE _{K_n}$ denotes the expectation with respect to the determinantal point process with kernel $K_{n}$ after fixing $\xi^{(n)}$ and that $\EE _{\xi}$ denote the expectation with respect to the random initial point $\xi^{(n)}$ which are taken independently from the semi-circle law. Then we split 
\begin{equation}\nonumber
Y_n(f)-\EE_{\xi} \EE _{K_n}  Y_n(f)=Y_n(f)-\EE _{K_n}  Y_n(f)+\EE_{K_n} Y_n(f)-\EE_{\xi} \EE _{K_{n}}  Y_n(f),
\end{equation}
and analyze \begin{equation}\label{eq:randomseparate}
Y_n(f)-\EE _{K_n}  Y_n(f), \quad \text{and} \quad\EE_{K_n} Y_n(f)-\EE_{\xi} \EE _{K_n}  Y_n(f),
\end{equation}
separately. The point is that for  the first term  on the right-hand side of \eqref{eq:randomseparate} we have the Central Limit Theorem in Theorem \ref{th:CLTfixed} for fixed $\xi^{(n)}\in C_n(\R,A n^\delta)$ for any $A,\delta>0$. Moreover, this Central Limit Theorem is independent of  $\xi^{(n)}\in C_n(\R,A n^\delta)$.  We first prove that $\xi^{(n)}\in C_n(\R,A n^\delta)$ with high probability and use Theorem \ref{th:CLTfixed} for the first term in \eqref{eq:randomseparate}.  We then analyze the second term and prove a Central Limit Theorem for that term. The proofs of Theorem \ref{th:random1} and \ref{th:random2} then follows by comparing which of the two terms is dominant.
\subsection{Preliminary lemmas}
We start with some  simple statements that we will use.
\begin{lemma}  \label{lem:simple1} Let $X,Y$ be two random variables such that $\EE X=0$ and $|Y| \leq c$ almost surely for some constant $c>0$. Then $\Var XY \leq c^2 \Var X$.
\end{lemma}
\begin{proof}
By  the conditions in the statement we have 
\begin{multline*}
\Var XY = \Var (X-\EE X)Y = \EE \left[ (X-\EE X)^2Y^2\right]-\left(\EE \left[ (X-\EE X)Y\right]\right)^2\\
\leq c^2 \EE \left[ (X-\EE X)^2\right]=c^2 \Var X,
\end{multline*}
and thus the statement follows.
\end{proof}
\begin{lemma} \label{lem:simple2}
Let $X,Y$ be two real-valued random variables such that $\EE X= \EE Y$. Then 
$$\left|\EE \left[\exp {\rm i } X\right]-\EE \left[\exp {\rm i } Y\right]\right|\leq \sqrt{\Var (X-Y)}.$$
\end{lemma}
\begin{proof}
This follows easily by
\begin{multline*}
\left|\EE \left[\exp {\rm i } X\right]-\EE \left[\exp {\rm i } Y\right]\right|
\leq\left|\EE \left[\left(\exp {\rm i } (X-Y)-1\right)\exp {\rm i } Y\right]\right|
\\
\leq \EE \left[\left|\exp {\rm i } (X-Y)-1\right|\right]\leq \EE \left[|X-Y|\right] \leq\sqrt{\EE \left[(X-Y)^2\right]}= \sqrt{\Var(X-Y)},
\end{multline*}
where, in the last step, we used the assumption $\EE (X-Y)=0$.
\end{proof}
\begin{lemma} \label{lem:simple3} If $\xi$ is a random variable with the semi-circle law as distribution, then 
\begin{align}\label{eq:varianceImC}
\Var \Im \frac{1}{w-\xi} \leq \frac{\frac12\sqrt 2}{ |\Im w|},\\
\Var \Re \frac{1}{w-\xi} \leq  \frac{\frac12\sqrt 2}{ |\Im w|},\label{eq:varianceReC}
\end{align}
for $w\in \C\setminus \R$. 
\end{lemma}
\begin{proof}
The variance can be estimated as follows 
\begin{multline*}\Var  \Im \frac{1 }{w-\xi} \leq  \EE \left[ \left(\Im \frac{1 }{w-\xi}\right)^2 \right]=\frac{1}{\pi} \int_{-\sqrt{2}}^{\sqrt 2} \frac{(\Im w)^2}{((\Im w)^2+(\Re w-\xi)^2)^2}\sqrt{2-\xi^2} {\rm d}\xi\\
\\ \leq \frac{\sqrt 2}{\pi}\int_\R \frac{(\Im w)^2}{((\Im w)^2+(\Re w-\xi)^2)^2} {\rm d}\xi 
= \frac{\sqrt 2}{\pi |\Im w|} \int_\R \frac{1}{(1+\xi^2)^2} {\rm d}\xi=\frac{\frac12\sqrt 2}{ |\Im w|} .
\end{multline*}
The proof of \eqref{eq:varianceReC} goes completely analogous.

\end{proof}
\begin{lemma} \label{lem:simple4}
If $\xi$ is a random variable with the semi-circle law as distribution, then 
\begin{multline}\nonumber
\Var \Im \left(\frac{1}{w_1-\xi} -\frac{1}{w_2-\xi}\right)\\
	\leq 
\frac12 \left(\frac{|w_1-w_2|^2}{|\Im w_1| | \Im w_2|}\right)\left(\frac{\Im (w_1-\sqrt{w_1^2-2})}{\Im w_1}+\frac{\Im (w_2-\sqrt{w_2^2-2})}{\Im w_2}\right),	
\end{multline}
and
\begin{multline}\nonumber
\Var \Re \left(\frac{1}{w_1-\xi} -\frac{1}{w_2-\xi}\right)\\
	\leq 
\frac12 \left(\frac{|w_1-w_2|^2}{|\Im w_1| | \Im w_2|}\right)\left(\frac{\Im (w_1-\sqrt{w_1^2-2})}{\Im w_1}+\frac{\Im (w_2-\sqrt{w_2^2-2})}{\Im w_2}\right)	
\end{multline}
for $w_1,w_2\in \C\setminus \R$.  
\end{lemma}
\begin{proof}
We start by writing
\begin{multline*}\Var \Im  \left(\frac{1}{w_1-\xi} -\frac{1}{w_2-\xi}\right)= \Var \Im  \left(\frac{w_2-w_1}{(w_1-\xi)(w_2-\xi)}\right) \\\leq   \EE \left( \Im  \left(\frac{w_2-w_1}{(w_1-\xi)(w_2-\xi)}\right) \right)^2 \leq  |w_1-w_2|^2  \EE \left[ \frac{1}{|w_1-\xi|^2|w_2-\xi|^2}\right]
\end{multline*}
Now 
\begin{multline*}
  \EE \left[\frac{1}{|w-\xi_1|^2|w-\xi|^2}\right]
\leq \frac{1}{\Im w_1 \Im w_2} \EE \left[\frac{1}{|w_1-\xi||w_2-\xi|}\right]\\\leq  \frac{1}{2\Im w_1 \Im w_2} \left(\EE \left[\frac{1}{|w_1-\xi|^2}+\frac{1}{|w_2-\xi|^2}\right]\right)
\end{multline*}
and combining this  with 
$$\EE \left[\frac{1}{|w_j-\xi|^2}\right]=\frac{1}{\Im w_j} \Im \EE \left[\frac{1}{w_j-\xi}\right]=\frac{\Im (w_j-\sqrt{w_j^2-2})}{\Im w_j},$$
gives the statement the imaginary part. The proof for the real part is analogous.
\end{proof}

\begin{lemma}\label{lem:stupidCLT}
Let $\xi_{j}^{(n)}$ for $j=1,\ldots,n$ and $n=1,2,\ldots$ be a sequence of independent random variables that have the semi-circle law as distribution. Let $f$ and $\{f_n\}_{n\in \N}$ be bounded functions such that $\|f-f_n\|_{\mathbb L_2(\R)}\to 0$ as $n \to \infty$ and the sequences $\|f_n\|_{\mathbb L_1(\R)}$ and $\|f_n\|_{\mathbb L_\infty(\R)}$  are bounded in $n$. Then for any $0 <\alpha<1$ we have
\begin{equation}\label{eq:standardCLT}
 \lim_{n\to \infty} \EE\left[\exp \frac{{\rm i}t }{n^{(1-\alpha)/2} }\left(\sum_{j=1}^n f_n(n^\alpha \xi_j^{(n)})-\EE \left[\sum_{j=1}^n f_n(n^\alpha \xi_j^{(n)})\right]\right) \right]= {\rm e}^{-\frac{\sqrt 2}{2 \pi}t^2\|f\|_2^2},
 \end{equation}
for $t\in \R$.
\end{lemma}
\begin{proof} The proof follows by standard arguments for independent random variables. We only note that if $\xi$ is a random with respect to the semi-circle law, then for any $g$ we have
\begin{multline*}
\Var g(n^\alpha \xi)= \frac{1}{\pi} \int_{-\sqrt 2}^{\sqrt{2}}   g(n^\alpha \xi )^2  \sqrt{2-\xi^2}  \ {\rm d} \xi - 
\left(\frac{1}{\pi} \int_{-\sqrt 2}^{\sqrt{2}}  g(n^\alpha \xi )   \sqrt{2-\xi^2}  \ {\rm d} \xi \right)^2
\\
=\frac{1}{\pi n^\alpha} \int_{-n^\alpha \sqrt 2}^{n^\alpha \sqrt 2} g(\xi)^2  \sqrt{2-\xi^2/n^{2\alpha}} {\rm d} \xi+ \frac{1}{ \pi  n^{2\alpha}} \left( \frac{1}{\pi} \int_{-n^\alpha \sqrt 2}^{n^\alpha \sqrt 2}  |g(\xi)| \sqrt{2-\xi^2/n^\alpha} {\rm d} \xi\right)^2\\
=\frac{\sqrt{2}}{\pi n^\alpha} \|g\|_{\mathbb L_2(\R)}^2  + o(n^{-\alpha}), 
\end{multline*}
as $n\to \infty$. Here the order is uniform for $g$ in bounded subset of the space $\mathbb L_\infty(\R)\cap \mathbb L_1(\R)$ equipped with the norm $\|g\|_{\mathbb L_1(\R)} + \|g\|_{\mathbb L_\infty(\R)}$.  Hence, under the conditions of the lemma we have
$$
\lim_{n\to \infty}  \Var n^{\alpha/2} f_n(n^\alpha \xi)=\frac{\sqrt{2}}{\pi} \|f\|_{{\mathbb L}_2(\R)}^2.$$ 
The statement now follows after a standard rewriting the left-hand side \eqref{eq:standardCLT}  and 
estimating the higher moments.  \end{proof}
\subsection{Regularity of the initial points}
We continue by proving that $\xi^{(n)}\in \mathcal C_n(\R,A n^\delta)$ with high probability
\begin{lemma}\label{lem:highprob}
Let $A,\delta>0$. For $n\in \N$ let $\mathcal C_n(\R,A n^\delta)$ be as in \eqref{eq:defCn}. Then 
\begin{equation*}
P(\xi^{(n)} \notin \mathcal C_n(\R,A n^\delta) ) =\mathcal O\left(n^{4-2\delta} \exp\left(-\frac{3 A}{4}n^\delta\right)\right),
\end{equation*}
as $n\to \infty$.
\end{lemma}
\begin{proof}
The proof the lemma is based on Bernstein's inequality for independent random variables. The main technical difficulty  is that the defining inequality for $\mathcal C_n(\R,A n^\delta)$ is with respect to the supremum over all $w$ such that $\Im w\geq 1/n$.  By a \lq net-argument\rq \ we show that it is sufficient to check the inequality for only for a finite number of values of $w$. 

We first note that when $w$ is large the  defining inequality for $C_n(\R,A n^\delta)$ is automatic and hence we can restrict $w$ to a bounded domain without changing $C_n(\R,A n^\delta)$. Hence in the first step,  we define $L_n$ as the set 
$$L_n= \{w\in \C \mid \Im w \geq \frac{n}{4A^2n^{2\delta} } \text{  or  } |\Re w | \geq \sqrt 2+\frac{n}{A^2n^{2 \delta}}\}.$$
The bounds defining the set  $L_n$ are chosen such that  we have
\begin{align*}
\sqrt{\frac{\Im w}{n}}\left|\frac{1}{w-\xi}\right|\leq \sqrt{\frac{\Im w}{n}}\frac{1}{|\Im w|+ |\Re w-\xi|}\leq \frac{A n^\delta }{2 n}, 
\end{align*}
for all $w\in L_n$ and $\xi \in (-\sqrt 2, \sqrt 2)$. From here it follows by some straightforward estimates that 
\begin{align*}
\sup_{w \in L_n} \sqrt{\frac{\Im w}{n}} \left|\sum_{j=1}^n\left( \frac{1}{w-\xi_j^{(n)}}-\int\frac{1}{w-\xi} {\rm d} \mu(\xi)\right)\right| \leq A n^\delta.
\end{align*}
This means that we can rewrite the definition of $\mathcal C_n(\R, A n^{\delta})$ to 
\begin{multline*}
\mathcal C_n(\R,A n^\delta)=\left\{\xi^{(n)} \in \R^n \mid\right.\\ \left.\sup_{\Im w \geq 1/n, w \notin L_n} \sqrt{\frac{\Im w}{n}} \left|\sum_{j=1}^n\left( \frac{1}{w-\xi_j^{(n)}}-\int\frac{1}{w-\xi} {\rm d} \mu(\xi)\right)\right| \leq A n^\delta\right\}.
\end{multline*}
Now that we have restricted the inequality to a bounded set, the next step is to replace it by the finite set  $G_{N}$ defined by
\begin{align*}
G_N=\{ w \mid \Im w\geq 1/n, N w\in \mathbb Z \times {\rm i}  \N,  \text {  and  } w \notin L_n\}. 
\end{align*}
To see that we can indeed restrict ourselves to $G_N$ for an appropriate value of $N$,   let $w_0\in G_N$. Then for any $w\in B_{w_0,1/N}$ (and $\Im w\geq 1/n$), we have
\begin{multline*}
\left|\frac{\sqrt{\Im w}}{w-\xi}-\frac{\sqrt{\Im w_0}}{w_0-\xi}\right|= \left|\frac{\sqrt{\Im w}-\sqrt{\Im w_0}}{ w-\xi}+\frac{\sqrt{\Im w_0}(w_0-w)}{( w-\xi) ( w_0-\xi)}\right|\\
\leq \frac{|\sqrt{\Im w}-\sqrt{\Im w_0}|}{\Im w}+\frac{\sqrt{\Im w_0}}{\Im w \Im w_0}|w-w_0|\leq 2n^{3/2} |w-w_0|\leq \frac{2n^{3/2}}{N}. 
\end{multline*}
Hence if we set $N=8 n^{2}/A n^\delta$ then 
\begin{align*}
\begin{split}
\sup_{w \in B_{w_0,1/N} \cap \{\Im w\geq 1/n\} } \sqrt{\frac{\Im w}{n}} \left|\sum_{j=1}^n\left( \frac{1}{w-\xi_j^{(n)}}-\int\frac{1}{w-\xi} {\rm d} \mu(\xi)\right)\right| \\
\leq  \sqrt{\frac{\Im w_0}{n}} \left|\sum_{j=1}^n\left( \frac{1}{w_0-\xi_j^{(n)}}-\int\frac{1}{w_0-\xi} {\rm d} \mu(\xi)\right)\right|+\frac{A n^\delta}{2}.
\end{split}
\end{align*} 
Moreover, the union of all balls $B_{w_0,1/N}$ with $w_0 \in G_N$ covers the set of all $w$ with $\Im w\geq 1/n$ and $w\notin L_n$. Therefore, after defining for $w_0\in G_N$, 
$$\mathcal  C_n^{(w_0)}(\R,A n^\delta )=\left\{\xi^{(n)} \in \R^n \mid  \sqrt{\frac{\Im w_0}{n}} \left|\sum_{j=1}^n\left( \frac{1}{w_0-\xi_j^{(n)}}-\int\frac{1}{w_0-\xi} {\rm d} \mu(\xi)\right)\right|\leq A n^\delta/2\right\},$$
we see that we have
$$\mathcal C_n(\R,A n^\delta) \supset \bigcap_{w_0 \in G_{8n^2/An^\delta}} C_n^{(w_0)}(\R,A n^\delta).$$
It follows that 
$$P(\xi^{(n)} \notin \mathcal C_n(\R,A n^\delta) ) \leq    \sum_{w_0 \in G_{8n^2/An^\delta}} P (\xi^{(n)} \notin C_n^{(w_0)}(\R,A n^\delta )).$$
Summarizing, instead of dealing with the supremum over all $w$ such that $\Im w\geq 1/n$ in the definition of $\mathcal  C_n(\R,A n^\delta)$, we only need to check the inequality for the finite grid $G_N$ which we will do now.

Note that  $\xi^{(n)} \notin C_n^{(w_0)}(\R,A n^\delta ))$ iff 
$$\sqrt{\frac{\Im w_0}{n}} \left|\sum_{j=1}^n\left( \frac{1}{w_0-\xi_j^{(n)}}-\int\frac{1}{w_0-\xi} {\rm d} \mu(\xi)\right)\right|\geq A n^\delta /2,$$
and if this happens, then we must have
$$\sqrt{\frac{\Im w_0}{n}} \Re \left(\sum_{j=1}^n\left( \frac{1}{w_0-\xi_j^{(n)}}-\int\frac{1}{w_0-\xi} {\rm d} \mu(\xi)\right)\right)\geq A n^\delta /4,$$
or
$$\sqrt{\frac{\Im w_0}{n}} \Im \left(\sum_{j=1}^n\left( \frac{1}{w_0-\xi_j^{(n)}}-\int\frac{1}{w_0-\xi} {\rm d} \mu(\xi)\right)\right)\geq A n^\delta /4.$$
We now use Bernstein's inequality to prove that the probability  that one of those events happens, is small. There exist several version in the literature and the one we will use is the following: Let $X_j$ for $j=1,\ldots,n$ be independent random variables such  that $\EE X_j=0$ and  $|X_j| \leq M$ almost surely for some $M\geq 0$. Then  \cite[Th 2.7]{Bern} says
\begin{align}\nonumber
P\left(\sum_{j=1}^n X_j > t\right) \leq \exp \left(-\frac{\frac12 t^2}{\EE \sum_{j=1}^n X_j^2 + \frac{1}{3} t M}\right).
\end{align}

We want to apply this inequality with $X_j=\sqrt{\frac{\Im w_0}{n}} \left( \frac{1}{w_0-\xi^{(n)}_j} -\EE\frac{1}{w_0-\xi^{(n)}_j} \right)$. Note that
\begin{equation}\nonumber
\begin{split}
 \sqrt \frac{\Im w_0}{n}\Im \frac{1 }{w_0-\xi} \leq \sqrt \frac{\Im w_0}{n}\frac{1 }{|w_0-\xi|}  \leq \frac{1}{\sqrt {n \Im w_0}}\leq 1,
\\
\sqrt \frac{\Im w_0}{n}\Re \frac{1 }{w_0-\xi} \leq\sqrt \frac{\Im w_0}{n} \frac{1 }{|w_0-\xi|} \leq \frac{1}{\sqrt {n \Im w_0}}\leq 1. 
\end{split}
\end{equation}
and hence we take $M=2$. By  \eqref{eq:varianceImC} and  \eqref{eq:varianceReC}, we thus obtain
$$P\left(\sqrt{\frac{\Im w_0}{n}} \Im \left(\sum_{j=1}^n\left( \frac{1}{w_0-\xi_j^{(n)}}-\int\frac{1}{w_0-\xi} {\rm d} \mu(\xi)\right)\right)\geq A n^\delta/4\right)\leq \exp\left (- \frac{A^2 n^{2\delta}/32}{\sqrt 2+ \frac{An^\delta}{24} } \right)$$
and the same for the real part. Hence
$$P\left(\sqrt{\frac{\Im w_0}{n}} \left|\sum_{j=1}^n\left( \frac{1}{w_0-\xi_j^{(n)}}-\int\frac{1}{w_0-\xi} {\rm d} \mu(\xi)\right)\right|\geq An^\delta/2\right)\leq 2 \exp\left (- \frac{A^2n^{2\delta}/32}{\sqrt 2+ \frac{An^\delta}{24} } \right).$$
Hence we see that 
\begin{multline*}
P(\xi^{(n)} \notin \mathcal C_n(\R,A n^\delta) ) =\mathcal O\left(N^2 \exp\left(-\frac{3 A}{4}n^\delta\right) \right)\\=\mathcal O\left(n^{4-2\delta} \exp\left(-\frac{3 A}{4}n^\delta\right)\right),
\end{multline*}
as $n\to \infty$ and this proves the statement.
\end{proof}

At this point, we will introduce some new notation that will be useful in the rest of the paper. We write 
\begin{equation} \label{eq:conditioningonc}
\EE_{\xi| C} X=\EE_{\xi}  \chi_{ C} X, \textrm{  and  } \Var_{\xi|C} X=\Var_{\xi}\chi_{ C}  X 
\end{equation}
where $\chi_{C}$ is the indicator function for the set $ C$.  Typically, we will work with sets $ C= \mathcal C(\R, A, \delta)$ as defined in \eqref{eq:defC}.
\subsection{Smoothening the test function}

The proof of Theorem \ref{th:random1} will be based on the loop equation \eqref{eq:loopeqn}.  To apply this equation we want to use Cauchy's integral formula for  functions $f$ that are analytic. However,  $f$ in Theorem \ref{th:random1} is a $C^\infty$ function with compact support and hence not analytic. Therefore we use an approximation argument. 

Let $\eps>0$ be small. Then let $\psi_n \in C_c^{\infty}(\R)$ such that $0\leq |\psi|\leq 1$ and
\begin{equation}\label{eq:mollifier}
\begin{cases}
\psi_n(\omega)= 1, & |\omega|\leq n^{\eps}\\
\psi_n(\omega)=0, & |\omega|\geq n^{\eps}+1\\
 \|\psi^{(k)}_n \|_\infty\leq d_k,  & k\in \N,
\end{cases}
\end{equation}
for some constants $d_k>0$ that do not depend on $n$.

Denote the Fourier transform of $f$ by $\hat f$, i.e.
\begin{equation} \label{eq:fourier}\hat f(\omega)= \frac{1}{\sqrt{2\pi}} \int _{-\infty}^{\infty} f(x) {\rm e}^{-{\rm i} x \omega} {\rm d} x.
\end{equation}
Then we define the function $f_n(z)$ by 
\begin{align}\label{eq:entireapproxf}
f_n(z)= \frac{1}{\sqrt{2\pi}} \int_{- \infty}^\infty \psi_n(\omega)\hat f(\omega) {\rm e}^{{\rm i} \omega z} {\rm d} w.  
\end{align}
Clearly, since $\psi_n$ has compact support we have that  $f_n$ is an entire function. 
The idea behind the construction of $f_n$ as as approximation to $f$ is that we have the following lemma.
\begin{lemma} \label{lem:analyticapprox} Let $ f\in C_c^\infty(\R)$. For $\eps>0$ and $n \in \N$, let $f_n$ as in \eqref{eq:entireapproxf}.  Then for any $p\in \N$ there exists a constant $C_p(f)$ (independent of $n$)  such that 
\begin{equation}\label{eq:estimatefunctionwithentire}
|f_n(x)-f(x)|\leq C_p(f) n^{-p \eps},
\end{equation}
for $x\in \R$, and 
\begin{equation} \label{eq:estimateonentireinstrip}
|f_n(z)| \leq C_p(f) |z|^{-p},
\end{equation}
for $z$ in the strip $\{z \in \C \mid \Im z \leq 3 n^{-\eps}\}$.
\end{lemma}
\begin{proof}
It is a classical result that for $f\in C_c^{\infty}(\R)$ and $M\in \N$ there exists a constant $\tilde C_M(f)$ such that 
$$|\hat{f} (\omega) | \leq \frac{ \tilde C_M(f)}{(1+|\omega|^2)^{(M+1)/2}},$$
for  $\omega \in \R$. Hence 
\begin{multline*}
|f(x)-f_n(x)| = \frac{1}{\sqrt {2 \pi}} \left| \int_{\infty}^\infty (1-\psi_n(x)) \hat f(\omega) {\rm e} ^{{\rm i} \omega x}  {\rm d} \omega \right|
\\
 \leq  \frac{2\tilde C_M(f)}{\sqrt{2\pi}} \int_{n^\eps}^\infty \frac{1}{(1+|\omega|^2)^{(M+1)/2}} {\rm d} \omega \leq C_M(f) n^{-M \eps}.
\end{multline*}
for some constant $C_M(f)$. This proves  \eqref{eq:estimatefunctionwithentire}.

In order to prove \eqref{eq:estimateonentireinstrip} we argue as follows.  By integration by parts we have

\begin{equation}\label{eq:entirefourier} f_n(z)= \frac{1}{(-{\rm i} z)^M} \frac{1}{\sqrt{2\pi}} \int_{-\infty}^\infty \frac{{\rm d}^p}{{\rm d} \omega^M} \left(\psi_n(\omega) \hat f(\omega) \right) {\rm e} ^{{\rm i} \omega z} {\rm d} \omega.
\end{equation}
By definition of $\psi_n$ we have $\|\psi_n^{(m)} \|_\infty \leq d_m$ for $1\leq m \leq M$. Moreover, 
$$
\hat f^{(m)} (\omega)= \frac{1}{\sqrt{2 \pi}} \int_{-\infty}^\infty (-{\rm i} x)^m f(x) {\rm e}^{-{\rm i} x \omega} {\rm d}x,
$$
and since $f \in C_c^\infty(\R)$ there exists constant $D_{m,M}(f)$ such that 
$$|\hat f^{(m)} (\omega)|  \leq \frac{D_{m,M}(f)}{(1+|\omega|^2)^{(M+1)/2}}.$$
And hence, by plugging this into \eqref{eq:entirefourier}, we see that we  have 
$$|f_n(z) |  \leq |z|^{-M} {\rm e}^{(\Im z)(n^\eps+1)} \hat D_M(f), $$
for some constant $\hat D_M(f)$. Hence, we can choose $C_M(f)$ large enough such that  for $\Im z \leq 3 n^{-\eps}$ we have \eqref{eq:estimateonentireinstrip}.
\end{proof}
We have the following immediate corollary.
\begin{corollary} \label{cor:analyticapprox} Let $f_n, f$ and $\eps>0$ as in Lemma \ref{lem:analyticapprox}. Then for any $M$ there exists  a constant $C_M(f)$ such that have 
$$\left|Y_n(f)-Y_n(f_n)\right| \leq  C_M(f) n^{-M \eps+1}.$$
\end{corollary}
By choosing $p$ sufficiently large, this corollary shows that we can replace $f$ by its approximation $f_n$. For the approximation $f_n$ we can use Cauchy's intergal formula to write
\begin{equation*}
f_n(x)= \frac{1}{2\pi {\rm i} } \int_{\Gamma}  f_n(z) \frac{1}{z-x} {\rm d} z,
\end{equation*}
where $\Gamma=  ({\rm i} n^{-\eps}-\R) \cup (- {\rm i} n^{-\eps}+\R)$. 
Now
\begin{multline*} 
Y_n(f_n)= \sum_{j=1}^n f_n(n^\alpha x_j) = \frac{1}{2 \pi {\rm i}} \int_{\Gamma} f_n(z) \sum_{j=1}^n \frac{1}{z-n^\alpha x_j} {\rm d} z\\
= \frac{1}{2 \pi {\rm i}n^\alpha} \int_{\Gamma} f_n(z) \sum_{j=1}^n \frac{1}{z/n^\alpha- x_j} {\rm d} z.
\end{multline*}
Hence
$$\EE_{K_n} \left[Y_n(f_n)\right]= \frac{1}{2 \pi {\rm i}n^\alpha} \int_{\Gamma} f_n(z) \EE_{K_{n}} \left[\sum_{j=1}^n \frac{1}{z/n^\alpha- x_j} \right]{\rm d} z.$$
In the next paragraph we will therefore start with analyzing $U_n(z)$ defined by 
\begin{equation}\label{eq:defUn}
U_n(z)= \frac{1}{n} \EE _{K_n} \left[\Tr\frac{1}{z-M}\right]=\EE_{K_{n}} \left[\sum_{j=1}^n \frac{1}{z- x_j}\right]. 
\end{equation}

\subsection{Approximating $U_n(z)$}

For any $\eps>0$, let $S_n$ be the set
\begin{equation}\label{eq:defSn} 
S_n=\left\{ z \in \C \ \mid \ \frac12 n^{-\eps} \leq \Im z \leq 2 n^{-\eps}, \ |z|\leq n^\eps\right\}.
\end{equation}
Note that for $z\in S_n$ we have  \eqref{eq:estimateonentireinstrip}. 

Let $U(z)$  be the function defined by
\begin{equation}\label{eq:defU}
U(z)=z-\sqrt{z^2-2}. \end{equation}
We take the square root such that $U$ is analytic in $\C \setminus [-\sqrt 2, \sqrt 2]$ and takes positive values on $[\sqrt 2, \infty)$. 
An important property of this function, is that if $\xi$ is a random variable distributed according to the semi-circle  law then 
\begin{equation}\nonumber
U(z)=\frac{1}{q}\EE \left[\frac{1}{z/q-c_2 U(z)-\xi}\right]
\end{equation}
for $z\in \C \setminus \R$ and 
\begin{equation}\label{eq:defc2}
c_2=\frac{1/q-q}{2}= \sinh t.
\end{equation}
We also recall the notation \eqref{eq:conditioningonc}.
\begin{lemma}\label{lem:aftersimple} Let $0<\alpha, \gamma<1$. Let $0<\delta<\frac{1}{2}\frac{1-\gamma}{1+\gamma}$ , $0<\eps<\min(\alpha,(1-\alpha)/2, \gamma(1-\gamma))$ and  $A>0$. Then for $\pm \Im z>0$ we have
\begin{equation}\label{eq:rapproxUnz}
 U_n(z/n^{\alpha})=\mp {\rm i} \sqrt 2+o(1),
 \end{equation}
as $n\to \infty$,  uniformly for $z \in S_n$ and $\xi \in C(\R,A,\delta)$. 

Moreover,
\begin{multline}\nonumber
\Var_{\xi|\mathcal C(\R,A,\delta)} \Im  \left(U_n(z/n^\alpha) - \frac{1}{qn } \sum_{j=1}^n \frac{1}{z/qn^\alpha-c_2  U(z/n^\alpha) -\xi^{(n)}_j}\right) \\= \begin{cases}
o(n^{\alpha-1}),& \alpha\leq \gamma \\
o(n^{\gamma-1}), & \alpha> \gamma,
\end{cases}
\end{multline} 
as $n\to \infty$, uniformly for $z \in S_n $. The same holds when taking the real instead of the imaginary part.
\end{lemma}
\begin{proof}
The starting point of the proof is \eqref{eq:loopeqn} which in the new notation reads
\begin{equation}\label{eq:loopeqninU}
U_n(z/n^\alpha)= \frac{1}{qn } \sum_{j=1}^n \frac{1}{z/qn^\alpha-c_2  U_n(z/n^\alpha) -\xi^{(n)}_j} +R_n(z/n^\alpha). 
\end{equation}
(with $R_n(z)=-\frac{1-q^2}{n^2} K^{0,I}(z)$). By rewriting this equation and estimating  various terms,  we prove the  statement. For these estimates we first recall some bounds that he have derived before.

First of all, by \eqref{eq:generalA} we have 
\begin{equation}\label{eq:estimater1o1}
R_n(z/n^{\alpha})=\mathcal O(n^{2(\alpha+ \eps-1)}), 
\end{equation}
 uniformly or $z \in S_n$ and $\xi \in C(\R,A,\delta)$, where we have an extra $\eps$ in the exponent compared to \eqref{eq:loopeqn} since we have $z\in S_n$.  We will also need  that there exists a constant $B>0$ such that 
 \begin{equation}
\label{eq:approxUnz}
\mp \Im U_n(z/n^\alpha)  \geq B,
\end{equation}
for $z\in S_n$ and $\pm \Im z>0$, and  $\xi \in C(\R,A,\delta)$.  This follows by using the arguments in the proof of Lemma \ref{lem:anf} (in particular, \eqref{eq:lemanf1} and \eqref{eq:lemanf2}). We also need that
\begin{equation}\label{eq:approxUz}
 U(z/n^{\alpha})=\mp {\rm i} \sqrt 2+o(1),
 \end{equation}
as $n\to \infty$ and $\pm \Im z>0$, uniformly for $z \in S_n$ and  $\xi \in C(\R,A,\delta)$.

We now rewrite \eqref{eq:loopeqninU} as follows
\begin{multline}\nonumber
U_n(z/n^\alpha)-U(z/n^\alpha)
= U_n(z/n^\alpha)-\frac{1}{qn } \sum_{j=1}^n \frac{1}{z/qn^\alpha-c_2  U(z/n^\alpha) -\xi^{(n)}_j} \\+
\frac{1}{qn } \sum_{j=1}^n \frac{1}{z/qn^\alpha-c_2  U(z/n^\alpha) -\xi^{(n)}_j} -U(z/n^\alpha)\\
= \frac{1}{qn } \sum_{j=1}^n \frac{1}{z/qn^\alpha-c_2  U_n(z/n^\alpha) -\xi^{(n)}_j}-\frac{1}{qn } \sum_{j=1}^n \frac{1}{z/qn^\alpha-c_2  U(z/n^\alpha) -\xi^{(n)}_j}  +R_n(z/n^\alpha)\\
\\+\frac{1}{qn } \sum_{j=1}^n \frac{1}{z/qn^\alpha-c_2  U(z/n^\alpha) -\xi^{(n)}_j} -U(z/n^\alpha)\\
= \frac{1}{qn } \sum_{j=1}^n \frac{c_2( U_n(x/n^\alpha)-U(z/n^\alpha))}{(z/qn^\alpha-c_2  U_n(z/n^\alpha) -\xi^{(n)}_j)(z/qn^\alpha-c_2  U(z/n^\alpha) -\xi^{(n)}_j)}  +R_n(z/n^\alpha)\\+ \frac{1}{qn } \sum_{j=1}^n \frac{1}{z/qn^\alpha-c_2  U(z/n^\alpha) -\xi^{(n)}_j} -U(z/n^\alpha). 
\end{multline}
After a reorganization this gives
\begin{multline}\label{eq:firstapproxU}
U_n(z/n^\alpha)-U(z/n^\alpha)=\frac{1}{1-F} \left( \frac{1}{qn } \sum_{j=1}^n \frac{1}{z/qn^\alpha-c_2  U(z/n^\alpha) -\xi^{(n)}_j} -U(z/n^\alpha) +R_n(z/n^\alpha)\right),
\end{multline}
where $F$ is given by 
\begin{multline}\label{eq:ffffn1}
F=\frac{c_2}{qn } \sum_{j=1}^n \frac{1}{(z/qn^\alpha-c_2  U_n(z/n^\alpha) -\xi^{(n)}_j)(z/qn^\alpha-c_2  U(z/n^\alpha) -\xi^{(n)}_j)}\\
= -\frac{1}{c_2 (U(z/n^\alpha)-U_n(z/n^\alpha)} \int_{z/qn^\alpha-c_2  U_n(z/n^\alpha) }^{z/qn^\alpha-c_2  U(z/n^\alpha) } \frac{c_2}{q n}\sum_{j=1}^n \frac{1}{(w-\xi^{(n)}_j)^2} {\rm d} w
\end{multline}
We will now first show that $F$ is small. We start by rewriting the integrand by using
\begin{multline*}
 \frac{c_2}{q n}\sum_{j=1}^n\frac{1}{(w-\xi^{(n)}_j)^2} =
 \frac{c_2}{q n}\sum_{j=1}^n \left(\frac{1}{(w-\xi^{(n)}_j)^2} -\int_{-\sqrt 2}^{\sqrt 2}\frac{1}{(w-\xi)^2} \sqrt {2-\xi^2} {\rm d} \xi \right)\\+
 \frac{c_2}{q} \int_{-\sqrt 2}^{\sqrt 2}\frac{1}{(w-\xi)^2} \sqrt {2-\xi^2} {\rm d} \xi
 \end{multline*}
  By the same argument as in the proof of Lemma \ref{lem:condone} we have that there exists a constant $a_1$ such that 
$$
\frac{c_2}{q n} \left| \sum_{j=1}^n \left(\frac{1}{(w-\xi^{(n)}_j)^2}-\int_{-\sqrt 2}^{\sqrt 2}\frac{\sqrt {2-\xi^2} {\rm d} \xi }{(w-\xi)^2} \right)\right| \leq a_1 n^{-\gamma(1-\gamma)/2},
$$
for  $\Im w\geq n^{-\gamma}/10 B$, uniformly for $\xi\in \mathcal C(\R,A,\delta)$.  
Because of \eqref{eq:approxUnz} and \eqref{eq:approxUz} we therefore have that  (by possibly increasing  the value of $a_1$) 
\begin{multline}\label{eq:ffffn2}
\left| \frac{1}{q c_2  (U(z/n^\alpha)-U_n(z/n^\alpha) } \sum_{j=1}^n\frac{c_2}{n }  \int_{z/qn^\alpha-c_2  U_n(z/n^\alpha) }^{z/qn^\alpha-c_2  U(z/n^\alpha) }\left(\frac{1}{(w-\xi^{(n)}_j)^2}-\int_{-\sqrt 2}^{\sqrt 2}\frac{\sqrt {2-\xi^2} {\rm d} \xi }{(w-\xi)^2} \right) {\rm d} w\right|\\ \leq a_1 n^{-\gamma(1-\gamma)/2}.
\end{multline}
Moreover, we have the simple computation reveals that 
\begin{multline*}
\frac{1}{y_1-y_2} \int_{y_1}^{y_2} \frac{1}{\pi} \int_{-\sqrt 2} ^{\sqrt 2} \frac{\sqrt{2-\xi^2} }{(w-\xi)^2}  {\rm d} \xi {\rm d} w
\\
=\frac{y_1-y_2-\sqrt{y_1^2-2}+ \sqrt{y_2^2-2}}{y_1-y_2}= 1- \frac{y_1+y_2}{\sqrt{y_1^2-2}+\sqrt{y_2^2-2}}.
\end{multline*}  Hence by \eqref{eq:approxUnz}, \eqref{eq:approxUz} and the fact that $c_2\sim n^{-\gamma}$  there exists a constant $a_2$ such that  
\begin{multline}\label{eq:ffffn3}
\left| \frac{1}{q   (U(z/n^\alpha)-U_n(z/n^\alpha) }  \int_{z/qn^\alpha-c_2  U_n(z/n^\alpha) }^{z/qn^\alpha-c_2  U(z/n^\alpha) }\int_{-\sqrt 2}^{\sqrt 2}\frac{\sqrt {2-\xi^2} {\rm d} \xi }{(w-\xi)^2} {\rm d} w\right|\\ \leq a_2 n^{-\gamma/2}.
\end{multline}
By substituting \eqref{eq:ffffn2} and \eqref{eq:ffffn3} in \eqref{eq:ffffn1}
\begin{equation}\label{eq:estimateonFo1}
|F|=\mathcal O(n^{-\gamma(1-\gamma)/2}),
\end{equation}
as $n\to \infty$, uniformly for $\xi\in \mathcal C(U,A,\delta)$ and $z\in S_n$.   

Now note that by \eqref{eq:approxUnz}, \eqref{eq:approxUz}  and the definition of $\mathcal C_n(\R, An^\delta)$ we have that there exists a constant $\tilde A$ such that 
\begin{multline} \label{eq:estimateonFo2}
\frac{1}{qn } \left|\sum_{j=1}^n \frac{1}{z/qn^\alpha-c_2  U(z/n^\alpha) -\xi^{(n)}_j} -U(z/n^\alpha)\right|\\\leq \frac{A n^\delta }{\sqrt n \sqrt{\Im \left(z/n^\alpha-(1-q^2)U(z/n^\alpha)\right)}} \leq \tilde A n^{\delta-(1-\gamma)/2}=\tilde A n^{-\gamma \delta}
\end{multline}
By inserting \eqref{eq:estimateonFo1} and \eqref{eq:estimateonFo2} into \eqref{eq:firstapproxU} and using \eqref{eq:approxUz} we obtain \eqref{eq:rapproxUnz}. 

In  a similar way, we  can rewrite
\begin{multline}
U_n(z/n^\alpha)-\frac{1}{qn } \sum_{j=1}^n \frac{1}{z/qn^\alpha-c_2  U(z/n^\alpha) -\xi^{(n)}_j} 
\\
=U_n(z/n^\alpha)-U(z/n^\alpha)+U(z/n^\alpha) -\frac{1}{qn } \sum_{j=1}^n \frac{1}{z/qn^\alpha-c_2  U(z/n^\alpha) -\xi^{(n)}_j} 
\\=  \left(\frac{1}{qn } \sum_{j=1}^n \frac{1}{z/qn^\alpha-c_2  U(z/n^\alpha) -\xi^{(n)}_j}  -U(z/n^\alpha) \right)\frac{F}{1-F} +R_n(z/n^\alpha)\frac{1}{1-F}
\label{eq:unzminsumu}
\end{multline}
The contribution from the second term on the right-hand side  to the variance is $o(n^{\alpha-1})$ by \eqref{eq:estimater1o1} and the assumption $\eps<(1-\alpha)/2$. As for the first term,  note that by Lemma \ref{lem:simple3} and \eqref{eq:approxUnz} we have
\begin{multline} \label{eq:estimateonvarun}
\Var _{\xi^{(n)} } \Im 
\frac{1}{qn } \sum_{j=1}^n \frac{1}{z/qn^\alpha-c_2  U(z/n^\alpha) -\xi^{(n)}_j} \\=
\frac{1}{q^2n }\Var  \Im    \frac{1}{z/qn^\alpha-c_2  U(z/n^\alpha) -\xi}
\\= \mathcal O \left(\frac{1}{n \Im (z/qn^\alpha-c_2  U(z/n^\alpha) )}\right) =\begin{cases} \mathcal O(n^{\alpha+\eps-1}), & \alpha + \eps\leq \gamma\\
\mathcal O(n^{\gamma-1}), & \alpha +\eps >\gamma,
\end{cases}
\end{multline}
as $n\to \infty$, uniformly for $z\in S_n$. The same estimates hold when taking the real part. Moreover,
\begin{equation}\label{eq:meanUn}
\EE_{\xi} \left[\frac{1}{qn } \sum_{j=1}^n \frac{1}{z/qn^\alpha-c_2  U(z/n^\alpha) -\xi^{(n)}_j}  \right]=U(z/n^\alpha).
\end{equation}
Now the statement follows after estimating the real and imaginary parts of  first term on the right-hand side of \eqref{eq:unzminsumu} by applying Lemma  \ref{lem:simple1}, using   \eqref{eq:estimateonFo1} with \eqref{eq:estimateonvarun} and \eqref{eq:meanUn}. \end{proof}

\begin{corollary} \label{cor:aftersimple} Let $0<\alpha, \gamma<1$. Let $0<\delta<\frac{1}{2}\frac{1-\gamma}{1+\gamma}$, $0<\eps<\min(\alpha,\frac25\gamma(1-\gamma)$  and  $A>0$. Then, for $\pm \Im z>0$, we have
\begin{align*}
\Var_{\xi |\mathcal C(\R,A,\delta)} \Im \left(U_n(z/n^\alpha) - \frac{1}{n } \sum_{j=1}^n \frac{1}{z/n^\alpha\pm  {\rm i} \tau/n^\gamma-\xi^{(n)}_j}\right)  \\ =o(n^{\alpha-1}), \quad 0<\alpha\leq \gamma,\\
\Var_{\xi|\mathcal C(\R,A,\delta)}  \Im \left(U_n(z/n^\alpha) - \frac{1}{n } \sum_{j=1}^n \frac{1}{\pm  {\rm i} \tau/n^\gamma -\xi^{(n)}_j}\right) \\=o(n^{\gamma-1}), \quad \gamma<\alpha <1,
\end{align*} 
as $n\to \infty$,  uniformly for $z\in S_n$. The same statements also hold when taking the real instead of imaginary part. 
\end{corollary}
\begin{proof}
Let us first deal with the case $0<\alpha \leq \gamma <1$. Note that  we have
$$
U(z/n^\alpha)= \mp {\rm i} \sqrt 2+\mathcal O(n^{\eps-\alpha}),
$$ 
for $ \pm \Im z>0$ 
and 
$$c_2= \tau n^{-\gamma} / \sqrt{2}+\mathcal O(n^{-2 \gamma}), $$
as $n\to \infty$, where the order is uniform for $z\in S_n$.

Therefore, if we set 
\begin{align*}
w_1&=z/n^\alpha-c_2U(z/n^\alpha)\\
w_2 &= z/n^\alpha\pm  {\rm i} \tau/n^\gamma,
\end{align*}
for $\pm \Im z>0$, then 
$$\left|w_1-w_2\right|= \mathcal O(n^{-\min( 2\gamma, \gamma+\alpha-\eps)})$$
$$ \frac{1}{|\Im w_1|}\leq \frac{1}{|\Im z/n^\alpha|}=\mathcal O(n^{\alpha+\eps}),$$
$$ \frac{1}{|\Im w_2|} \leq  \frac{1}{|\Im z/n^\alpha|}=\mathcal O(n^{\alpha+\eps}),$$
and therefore,
\begin{multline}\label{eq:easycasehelp1}
\frac
{\left|w_1-w_2\right|^2}
{ |\Im w_1|| \Im w_2| }
\left( 
\frac{\Im (w_1-\sqrt{w_1^2-2})}{\Im w_1}+\frac{\Im (w_2-\sqrt{w_2^2-2})}{\Im w_2}
\right)\\
=  n^{\alpha} \mathcal O \left(n^{-\min(4\gamma-2 \alpha-3\eps,2 \gamma-5 \eps)}\right),
\end{multline}
as $n\to \infty$, where the order is uniform for $z\in S_n$.
Since $\alpha\leq \gamma$ and $\eps< \frac{2}{5} \gamma$ we have 
$$\min(4\gamma-2 \alpha-3\eps,2 \gamma-5 \eps)>\min(2\gamma-3\eps,2 \gamma-5 \eps)=2\gamma-5 \eps>0.$$
By combining the latter with $\eqref{eq:easycasehelp1}$ and applying  Lemmas \ref{lem:aftersimple} and \ref{lem:simple4}, the statement follows.

Finally, we deal with the case $0<\gamma<\alpha<1$. Then  we use
$$\left|w_1-w_2\right|= \mathcal O(n^{-\min( 2\gamma, \gamma+\alpha-\eps)})$$
$$\frac{1}{ |\Im w_1|}\leq \frac{1}{c_2|\Im U(z/n^\alpha)|}=\mathcal O(n^{\gamma}),$$
$$ \frac{1}{|\Im w_2|}\leq \frac{1}{c_2|\Im U(z/n^\alpha)|} =\mathcal O(n^{\gamma}),$$
and therefore,
\begin{multline}\label{eq:easycasehelp2}
\frac
{\left|w_1-w_2\right|^2}
{ |\Im w_1|| \Im w_2| }
\left( 
\frac{\Im (w_1-\sqrt{w_1^2-2})}{\Im w_1}+\frac{\Im (w_2-\sqrt{w_2^2-2})}{\Im w_2}
\right)\\
=  n^{\gamma} \mathcal O \left(n^{-\min(2\gamma, 2\alpha-2 \eps)}\right)=o(n^{\gamma}),
\end{multline}
as $n\to \infty$, where the order is uniform for $z\in S_n$. Again, by combining  $\eqref{eq:easycasehelp2}$ and applying  Lemmas \ref{lem:aftersimple} and \ref{lem:simple4}, the statement follows. 
\end{proof}
\begin{lemma}\label{lem:ptg}
Let $g$ be a function that is analytic in a strip $\{|\Im z| \leq  \eps \}$, for some $\eps>0$  and $g\in \mathbb L_2(\R)$. Let $\mathcal P_t$ as in \eqref{eq:defPeps} for $t>0$. Then 
\begin{align}\nonumber
\mathcal P_t  g(x)= \frac{1}{2 \pi {\rm i} } \int_{ {\rm i} \eps -\R} \frac{g(z)}{z-x+{\rm i} t} {\rm d} z - \frac{1}{2 \pi {\rm i} } \int_{ -{\rm i} \eps +\R} \frac{g(z)}{z-x-{\rm i} t} {\rm d} z,
\end{align}
for $x\in \R$.
\end{lemma}
\begin{proof} This is a standard exercise. By applying Fubini we find 
\begin{multline*}
\mathcal P_t  g(x)= \int_\R g(y) \frac{1}{\pi} \frac{t}{(x-y)^2+t^2 } {\rm d} y \\
= \frac{1}{2\pi {\rm i} }  \int_\R g(y)  \left(\frac{1}{x-y-{\rm i} t}-\frac{1}{x-y+{\rm i} t}\right) {\rm d} y   \\
  = \frac{1}{2\pi {\rm i}} \int_\R   \frac{1}{2 \pi {\rm i} }  \int_{ ( {\rm i} \eps -\R) \cup (-{\rm i} \eps+ \R)}  g(z) \frac{1}{z-y} \left(\frac{1}{x-y-{\rm i} t}-\frac{1}{x-y+{\rm i} t}\right) {\rm d} z    {\rm d} y  \\
    =  \frac{1}{2 \pi {\rm i} }  \int_{ ( {\rm i} \eps -\R) \cup (-{\rm i}\eps+ \R)}  g(z)  \frac{1}{2\pi {\rm i}} \int_\R  \frac{1}{z-y}  \left(\frac{1}{x-y-{\rm i} t}-\frac{1}{x-y+{\rm i} t}\right)    {\rm d} y {\rm d} z ,  
\end{multline*}
and the statement follows after a residue calculus.
\end{proof}
We now prove a Central Limit Theorem for the second term at the right-hand side of \eqref{eq:randomseparate}.
\begin{proposition} \label{prop:random}
Let $0<\alpha,\gamma<1$ and $f\in C^\infty_0(\R)$. Then,
\begin{enumerate}
\item for $0<\alpha<\gamma<1$, we have
\begin{equation}\nonumber
\lim_{n \to \infty} \EE_{\xi} \left[\exp\frac{ {\rm i} t }{n^{(1-\alpha)/2}}\left(\EE_ {K_n } Y_n(f) -\EE_{\xi} \EE_ {K_n} Y_n(f) \right)\right]= {\rm e}^{ -\frac{\sqrt 2t^2}{2 \pi}  \|f\|_2^2}
\end{equation}
\item for $0<\alpha=\gamma<1$, we have
\begin{equation} \nonumber
\lim_{n \to \infty} \EE_{\xi} \left[\exp\frac{ {\rm i} t }{n^{(1-\alpha)/2}}\left(\EE_ {K_n } Y_n(f) -\EE_{\xi} \EE_ {K_n} Y_n(f) \right)\right]= {\rm e}^{ -\frac{\sqrt 2 t^2}{2 \pi} \|\mathcal P_\tau f\|_2^2}
\end{equation}
\item  for $0<\gamma<\alpha<1$ and $\mu_0(f):=\int _\R f(x) {\rm d}x\neq 0$, we have
\begin{equation}\nonumber
\lim_{n \to \infty} \EE_{\xi} \left[\exp\frac{ {\rm i} t }{n^{(1-2\alpha+\gamma)/2}}\left(\EE_ {K_{n} } Y_n(f) -\EE_{\xi} \EE_ {K_n} Y_n(f) \right)\right]= {\rm e}^{ -\frac{\sqrt 2t^2}{2 \pi} \mu_0(f)^2}.
\end{equation}
\end{enumerate}
\end{proposition}
\begin{proof}  First, let us separate the different cases and  start with assuming that $0 <\alpha \leq \gamma<1$.
We will use the approximation $f_n$ defined in \eqref{eq:entireapproxf} for sufficiently small $\eps$.  The argument follows a number of steps. First, we write  
\begin{equation} \label{eq:proofrandom1}
\EE _{K_n} \left[Y_n(f_n) \right] = \frac{n^{1-\alpha}}{2 \pi {\rm i} }  \int _\Gamma f_n(z) U_n(z/n^\alpha) {\rm d} z,
\end{equation}
where $\Gamma=\left( {\rm i} n^{-\eps} -\R  \right) \cup \left(-{\rm i} \eps+ \R\right)$. By \eqref{eq:estimateonentireinstrip} we see that, at the cost of a (sufficiently) small error,  we can replace $\Gamma$ by $\Gamma\cup S_n$ and the right-hand side of \eqref{eq:proofrandom1} by 
$$
 \frac{n^{1-\alpha}}{2 \pi {\rm i} }  \int _{\Gamma\cap S_n} f_n(z) U_n(z/n^\alpha) {\rm d} z.
 $$
By applying Corollary \ref{cor:aftersimple} and \eqref{eq:estimateonentireinstrip} to the latter and the fact that $\xi^{(n)} \in \mathcal C_n(\R,An^\delta)$ we can approximate this by 
$$
 \frac{n^{-\alpha}}{2 \pi {\rm i} }  \int _{\Gamma\cap S_n} f_n(z) \sum_{j=1}^n \frac{1}{z/n^\alpha \pm {\rm i} \tau-\xi^{(n)}_j}{\rm d} z,
$$
Again by \eqref{eq:estimateonentireinstrip}, we can undo the cut-off and use Lemma \ref{lem:ptg} to write 
$$
n^{-\alpha} \sum_{j=1}^n \mathcal P_{\tau/n^{\gamma-\alpha}} f_n(n^\alpha \xi^{(n)}_j).
$$
 Concluding, by taking all the error terms in account, we have
\begin{equation}\label{eq:variancesmallenough1}
n^{\alpha-1}\Var_{\xi|C} \left(\EE_ {K_{n}} Y_n(f_n) - \sum_{j=1}^n \mathcal P_{\tau/n^{\gamma-\alpha}}f_n(n^{\alpha} \xi^{(n)}_j)\right)
=o(1)
\end{equation}
as $n\to \infty$.
Now by using respectively Corollary \ref{cor:analyticapprox} and Lemma \ref{lem:highprob} we find
\begin{multline*}
\lim_{n \to \infty} \EE_{\xi} \left[\exp\frac{ {\rm i} t }{n^{(1-\alpha)/2}}\left(\EE_ {K_n } Y_n(f) -\EE_{\xi} \EE_ {K_n} Y_n(f) \right)\right]\\
=\lim_{n \to \infty} \EE_{\xi} \left[\exp\frac{ {\rm i} t }{n^{(1-\alpha)/2}}\left(\EE_ {K_{n} } Y_n(f_n) -\EE_{\xi} \EE_ {K_{n}} Y_n(f_n) \right)\right]\\
=\lim_{n \to \infty} \EE_{\xi|C} \left[\exp\frac{ {\rm i} t }{n^{(1-\alpha)/2}}\left(\EE_ {K_n} Y_n(f_n) -\EE_{\xi|C} \EE_ {K_n} Y_n(f_n) \right)\right]
\end{multline*}
By Lemma \ref{lem:simple2} and \eqref{eq:variancesmallenough1}  we find 
\begin{multline*}
\lim_{n \to \infty} \EE_{\xi} \left[\exp\frac{ {\rm i} t }{n^{(1-\alpha)/2}}\left(\EE_ {K_n } Y_n(f) -\EE_\xi \EE_ {K_n} Y_n(f) \right)\right]\\
=\lim_{n \to \infty} \EE_{\xi|C} \left[\exp\frac{ {\rm i} t }{n^{(1-\alpha)/2}}\left(\sum_{j=1}^n \mathcal P_{\tau/n^{\gamma-\alpha}} f_n(n^{\alpha} \xi_j^{(n)})-\EE_{\xi|C}  \left[\sum_{j=1}^n P_{\tau/n^{\gamma-\alpha}}  f_n(n^{\alpha} \xi_j^{(n)})\right]\right)\right]\\
=\lim_{n \to \infty} \EE_{\xi} \left[\exp\frac{ {\rm i} t }{n^{(1-\alpha)/2}}\left(\sum_{j=1}^n \mathcal P_{\tau/n^{\gamma-\alpha}} f_n(n^{\alpha} \xi_j^{(n)})-\EE_{\xi}  \left[\sum_{j=1}^n P_{\tau/n^{\gamma-\alpha}}  f_n(n^{\alpha} \xi_j^{(n)})\right]\right)\right],
\end{multline*}
where we also used Lemma \ref{lem:highprob} in the last step.  Now the statement for $0<\alpha\leq \gamma<1$ follows by invoking Lemma \ref{lem:stupidCLT} and  using
$$\lim_{n\to \infty} P_{\tau/n^{\gamma-\alpha}} f_n= \begin{cases} f, & \text{if }  \alpha <\gamma\\
\mathcal P_\tau f,& \text{if }  \alpha =\gamma,
\end{cases}$$
where the convergence is in $\mathbb L_2(\R)$ and the fact that both the $\mathbb L_\infty$ and $\mathbb L_1$ norms of $P_{\tau/n^{\gamma-\alpha}} f_n$ are bounded.

Finally, we consider the case $0<\gamma<\alpha<1$. In this case,  we argue similarly and use Corollary \ref{cor:aftersimple}, but now obtain
\begin{equation} \nonumber n^{-1+2\alpha -\gamma}\Var_{\xi|C} \left(\EE_ {K_{n} } Y_n(f_n) -\frac{\mu_0(f_n)}{n^\alpha}  \sum_{j=1}^n \Im  \frac{1}{{\rm i} \tau/n^\gamma-\xi^{(n)}_j}\right)=o(1),
\end{equation}
as $n\to \infty$.  Moreover,
$$\mu_0(f_n) = \int_\R f_n(x) {\rm d}x = \hat f_n(0)=  \psi_n(0) \hat f(0)= \hat f(0)= \mu_0(f).$$ Hence, by arguing as in the case $\alpha<\gamma$, we find
\begin{multline*}
\lim_{n \to \infty} \EE_{\xi} \left[\exp\frac{ {\rm i} t }{n^{(1+\gamma-2\alpha)/2}}\left(\EE_ {K_n } Y_n(f) -\EE_{\xi} \EE_ {K_{n}} Y_n(f) \right)\right]\\
=\lim_{n \to \infty} \EE_{\xi} \left[\exp\frac{\mu_0(f)  {\rm i} t }{n^{(1+\gamma)/2}}\left( \sum_{j=1}^n \Im  \frac{1}{{\rm i} \tau/n^\gamma-\xi^{(n)}_j}-\EE_{\xi}  \sum_{j=1}^n \Im  \frac{1}{{\rm i} \tau/n^\gamma-\xi^{(n)}_j}\right)\right]
\end{multline*}
and it is not hard to see that the latter is a Gaussian with the required variance.
\end{proof}
\subsection{Proof of Theorem \ref{th:random1}, and Theorem \ref{th:random2} with the assumption $\mu_0(f) \neq 0$}  \label{subsec:random2firstcase}
\begin{proof}[Proof of Theorem \ref{th:random1}, and Theorem \ref{th:random2} with the assumption $\mu_0(f) \neq 0$]
To start with, set
\begin{align*}X_1&= Y_n(f)-\EE _{K_n}  Y_n(f)\\
 X_2&=\EE _{K_{n}} Y_n(f)-\EE_{\xi} \EE _{K_{n}}  Y_n(f),
 \end{align*}
and let $\beta$ defined as 
$$\beta= \begin{cases}
(1-\alpha)/2,& \text{ if } \alpha\leq \gamma,\\
(1+\gamma-2 \alpha)/2, & \text{ if } \gamma< \alpha< (1+\gamma)/2,\\
0,& \text{ if } \alpha\geq (1+\gamma)/2.
\end{cases}$$
Moreover, let $A>0$,  $0<\delta<\frac{1}{2}\frac{1-\gamma}{1+\gamma}$ and set $C=C_n(\R,A n^\delta)$. 

By Lemma \ref{lem:highprob} we have 
\begin{multline}\nonumber
\lim_{n\to \infty} \EE_{\xi} \left[\EE_{K_{n}} \left[{\rm e}^{ {\rm i} t n^{-\beta} (X_1+X_2)}\right]\right]=\lim_{n\to \infty} \EE_{\xi|C} \left[\EE_{K_{n}} \left[{\rm e}^{ {\rm i} t n^{-\beta}(X_1+X_2)}\right]\right]\\\
=\lim_{n\to \infty} \EE_{\xi|C} \left[{\rm e}^{{\rm i} t n^{-\beta} X_2}  \EE_{K_{n}} \left[{\rm e}^{ {\rm i} t n^{-\beta} X_1 }\right]\right]
\end{multline}
The proof of Theorems \ref{th:random1} and \ref{th:random2} now follow by applying Theorems \ref{th:variance}, \ref{th:CLTfixed} and Proposition \ref{prop:random} to the latter.
 
Let us first consider the  case that  $0<\alpha<(1+\gamma)/2<1$. Then by Theorem \ref{th:variance} we have 
$$n^{-2\beta}\Var _{K_n} X_1=\mathcal O(n^{-2 \beta}),$$
as $n\to \infty$ uniformly for $\xi^{(n)}\in C$. Hence by Lemma \ref{lem:simple2} and  Lemma \ref{lem:highprob} we have 
\begin{multline}\nonumber
\lim_{n\to \infty} \EE_{\xi^{(n)}} \left[\EE_{K_{n}} \left[{\rm e}^{ {\rm i} t n^{-\beta} (X_1+X_2)}\right]\right]=\lim_{n\to \infty} \EE_{\xi|C} \left[{\rm e}^{{\rm i} t n^{-\beta} X_2} \right]\\
=\lim_{n\to \infty} \EE_{\xi} \left[{\rm e}^{{\rm i} t n^{-\beta} X_2} \right],
\end{multline}
and for the latter we use Proposition \ref{prop:random} to obtain the statement of Theorem \ref{th:random1} and the part of Theorem \ref{th:random2} with $\gamma<\alpha<(1+\gamma)/2$. 

Now consider the case $0<(1+\gamma)/2<\alpha<1$ (and hence $\beta=0$). In that case,  the variance of $ X_2 =o(1)$ as $n\to \infty$, and hence we can argue similarly to see that  we can ignore $X_2$. That is,
\begin{multline}\nonumber
\lim_{n\to \infty} \EE_{\xi} \left[\EE_{K_n}\left[{\rm e}^{ {\rm i} t (X_1+X_2)}\right]\right]
=\lim_{n\to \infty} \EE_{\xi|C} \left[{\rm e}^{{\rm i} t X_2}  \EE_{K_{n}} \left[{\rm e}^{ {\rm i} t  X_1 }\right]\right]
\\
=\lim_{n\to \infty} \EE_{\xi|C} \left[\EE_{K_{n}} \left[{\rm e}^{ {\rm i} t  X_1 }\right]\right].
\end{multline}
By Theorem \ref{th:CLTfixed} we have that $\EE_{K_{n}} \left[{\rm e}^{ {\rm i} t X_1 }\right]$  converges uniformly to  a gaussian that does not depend on  $\xi^{(n)}_{n \in \N}$. Hence, the extra averaging has of $\xi^{(n)}$ has no effect and the statement directly follows from \ref{th:CLTfixed}.

Finally, let us consider the case $0<\alpha=(1+\gamma)/2<1$ (and again $\beta=0$). In that case, we also use the fact that $\EE_{K_{n}} \left[{\rm e}^{ {\rm i} tX_1 }\right]$ converges to a Gaussian that does not depend on  $\xi^{(n)}_{n \in \N}$. In this case this lead to 
$$
\lim_{n\to \infty} \EE_{\xi} \left[\EE_{K_{n}} \left[{\rm e}^{ {\rm i} t (X_1+X_2)}\right]\right]
=   {\rm e}^{-\frac{\sigma_\tau^2}{2} t^2} \lim_{n\to \infty}\EE_{\xi|C} \left[{\rm e}^{{\rm i} t X_2}  \right].
$$
The remaining part of the  statement of Theorem \ref{th:random2} now follows by applying Lemma \ref{lem:highprob}, which gives 
$$
\lim_{n\to \infty} \EE_{\xi} \left[\EE_{K_{n}} \left[{\rm e}^{ {\rm i} t  (X_1+X_2)}\right]\right]
= {\rm e}^{-\frac{\sigma_\tau^2}{2} t^2} \lim_{n\to \infty}\EE_{\xi} \left[{\rm e}^{{\rm i} t  X_2}  \right].
$$
and then Proposition \ref{prop:random}.  \end{proof}

\section{Proof of Theorem \ref{th:random2}: the general case}

In this section we prove Theorem \ref{th:random2} without the assumption that $\mu_0(f) \neq 0$, but first we mention some assumptions and quantities that we will use.

 Throughout this section, we will always assume that 
$$0 <\gamma <\alpha< \frac{1+\gamma}{2} <1.$$
We  will also use the following small parameter 
\begin{align}
0<\eps & < \min\left(\gamma,\alpha-\gamma,\frac{1+\gamma}{4} -\frac{\alpha}{2}, \frac{1-\gamma}{3}\right).\label{eq:finalass1} 
\end{align}  
and set
\begin{equation}\label{eq:defc2a}
c_1=\cosh t, \qquad \text{ and } \qquad c_2=\sinh t.
\end{equation}
In the proof of Theorem \ref{th:random2} we will derive various inequalities. These inequalities hold with high probability with respect to the distribution of the initial points $\xi^{(n)}_j$. We use Lemma \ref{lem:highprob} to formulate a straightforward notion of high probability. We say that an inequality holds with high probability, if it holds for all $(\xi^{(n)})_{n \in \N} \in\mathcal  C(\R,1,\eps)$. 

\subsection{Smoothening of the testfunction}

Let $f\in C_c^\infty(R)$. Instead of proving Theorem \ref{th:random2} directly for $f$, we will use  the smoothened function $f_n$
\begin{equation}\label{eq:approxf}
\frac{1}{\sqrt{2\pi}} \int_{-{\rm i} \infty}^\infty \psi_n(\omega)\hat f(\omega) {\rm e}^{{\rm i} \omega z} {\rm d} z,
\end{equation}
 with $\psi_n$ as in  \eqref{eq:mollifier} and  $\eps$ in \eqref{eq:finalass1}.    It is important to note that for every $k\in \N $ we have 
\begin{equation}\label{eq:preservingmoments} 
\mu_k(f_n)= \int_\R x^k f_n(x) {\rm d} x=\frac{1}{(-{\rm i})^k}\frac{{\rm d}^k \hat f_n}{ {\rm d} \omega^k}  (0)=\frac{1}{(-{\rm i})^k}\frac{{\rm d}^k \hat f}{ {\rm d} \omega^k} (0)=\mu_k(f).\end{equation}
Hence, in the approximation we do not change the moments.  Moreover, by Corollary \ref{cor:analyticapprox} 
we see that  $$Y_n(f) \approx Y_n(f_n)$$ with  a very small error. Therefore, we continue by analyzing $\EE_{K_{n}} Y_n(f_n)$, which by \eqref{eq:defUn} can be written as 
\begin{equation}\label{eq:cauchycauchy}
\EE_{ K_{n}} Y_n(f_n)= \frac{1}{2 \pi {\rm i} n^\alpha}  \int _{\Gamma} f_n(z) U_n(z) {\rm d} z,
\end{equation}
with $\Gamma= ({\rm i} n^{-\eps} -\R)\cup  (-{\rm i} n^{-\eps} +\R)$.

The following proposition is an important key to the proof. We also recall the definitions of $U$ in \eqref{eq:defU} 
\begin{proposition}\label{prop:final}
Let $f\in C_c^\infty(R)$  and let  $p \geq 0$ be such that  $\mu_s(f)=0$ for $0\leq s <p$ and assume that $\mu_p(f)\neq 0$. Then we can write
\begin{multline}\label{eq:prop:final}
\int_{-\infty}^\infty f_n(x+ {\rm i} n^{-\eps}) \frac{1}{n^\alpha} \left(U_n\left(\frac{x+{\rm i} n^{-\eps}}{n^\alpha}\right)-U\left(\frac{x+ {\rm i} n^{-\eps}}{n^\alpha}\right)\right) {\rm d} x\\
= \frac{1}{q n^{(p+1)\alpha} }(-1/q)^p X_p+Q_{n,p} + T_{n,p},
\end{multline}
where 
\begin{align}\label{eq:defXp}
X_p=\frac{1}{n} \sum_{j=1}^n \frac{1}{({\rm i} c_2 \sqrt 2-\xi_j^{(n)})^{p+1}}-
\frac{1}{\pi} \int_{-\sqrt 2} ^{\sqrt 2} \frac{\sqrt{2-\xi^2}}{({\rm i} c_2 \sqrt 2-\xi)^{p+1}} {\rm d} \xi,
\end{align}
and $Q_{n,p}$ and $T_{n,p}$ are such that, with high probability, we have
\begin{align}
|Q_{n,p} | &\leq  C n^{-\eps_1} \frac{1}{n^{(p+1)\alpha-p \gamma+\frac{1-\gamma}{2}}},\label{eq:prop:final1}\\
|T_{n,p}|& \leq  \frac{C}{n^{1+\eps_2}} \label{eq:prop:final2}
\end{align}
for some constants $C,\eps_1,\eps_2>0$ that do not depend on $n$. 
\end{proposition}
The proof of this proposition will be done in several steps.
\subsection{Change of variables}
The first step towards the proof of Proposition \eqref{prop:final} is the following lemma. Let us first define $H_n(\omega)$ as 
\begin{equation}\label{eq:defHn}
H_n(\omega)= \frac{c_2}{n} \sum_{j=1}^n \frac{1}{\frac{\omega}{qn^\alpha}+{\rm i} c_2  -\xi_j^{(n)}}-\frac{ c_2}{\pi} \int_{\sqrt 2}^{\sqrt 2}\frac{\sqrt{2-\xi^2}}{\frac{\omega}{q n^\alpha}+{\rm i} c_2-\xi} {\rm d} \xi.\\
\end{equation}
Hence, by expanding  $H_n(\omega)$  in powers of $\omega$, we see that $H_n(\omega)$ is a generating function for the random numbers $X_p$ as defined in \eqref{eq:defXp}.  We also need the function $\omega(\zeta)$ defined by 
$$\omega (\zeta)  = q c_1 \zeta+q c_2 n^\alpha \left(\sqrt{\zeta^2/n^{2\alpha}-2}-{\rm i } \sqrt 2 \right).$$
Then we have the following result.
\begin{lemma}\label{lemma:final1}
Let $f \in C_c^
\infty(\R)$ and $f_n$ as in \eqref{eq:approxf}. Then 
\begin{multline}\label{eq:lem:final1}
\int_{-\infty}^\infty f_n(x+{\rm i} n^{-\eps}) \frac{1}{n^\alpha} \left(U_n(x+ {\rm i} n^{-\eps} ) -U(x+ {\rm i} n^{-\eps} )\right) {\rm r} x \\
= \frac{1}{q c_2n^{2 \alpha} } \int_{\Gamma_+} f_n(\zeta) H_n(\omega(\zeta)) {\rm d} \zeta\\+  \frac{1}{q c_2n^{2 \alpha} } \sum_{m=2}^\infty \int_{\Gamma_+} f_n^{(m-1)} (\zeta) \left(H_n(\omega(\zeta)) \right)^{m} \omega'(\zeta) {\rm d} \zeta+ \sum_{j=1}^6 T_n^{(j)}
\end{multline}
where $\Gamma_+$ is the set  $z= x+ {\rm i} n^{-\eps}$ with $|x|\leq \frac{1}{2} n^\eps$, and $T_n^{(j)}$ are defined in \eqref{eq:defT1n}, \eqref{eq:defT2n}, \eqref{eq:defT3n}, \eqref{eq:defT4n}, \eqref{eq:defT5n} and \eqref{eq:defT6n}.
\end{lemma}
The proof of this lemma, that will take the rest of this subsection, is based on a sequential change of variables in the integral at the right-hand side of \eqref{eq:cauchycauchy}. To start with,  let us first define functions $G_n$ and $g_n$ by 
\begin{align}\nonumber
G_n(w)&= \frac{1}{n^{1+\alpha} q} \sum_{j=1}^n \frac{1}{\frac{w}{qn^\alpha}+{\rm i} c_2  -\xi_j^{(n)}}+\frac{{\rm i} \sqrt 2}{n^\alpha},\\
g_n(w)&=\EE_{\xi^{(n)}}  G_n(w)= \frac{1}{n^\alpha  q \pi} \int_{\sqrt 2}^{\sqrt 2}\frac{\sqrt{2-\xi^2}}{\frac{w}{q n^\alpha}+{\rm i} c_2-\xi} {\rm d} \xi +\frac{{\rm i} \sqrt 2}{n^\alpha}\\
\nonumber &=  \frac{1}{n^\alpha q}  U\left(\frac{w}{q n^\alpha}+ {\rm i} c_2 \right)+\frac{{\rm i} \sqrt 2}{n^\alpha}.
\end{align}
for $\Im w >0$.  Note that with these definitions, we have that $H_n(w)$ as defined in \eqref{eq:defHn} can be written as
\begin{equation}\nonumber
H_n(w)= q c_2 n^\alpha (G_n(w)-g_n(w))
\end{equation}
We perform a sequential change of variables 
\begin{align}
w&=z-q c_2n^\alpha R_n\left(z/n^\alpha\right) \label{eq:coordchange1}\\
w&= \omega+ q c_2 n^{2\alpha} G_n(\omega)\label{eq:coordchange2} \\
\omega &  = q c_1 \zeta+q c_2 n^\alpha \left(\sqrt{\zeta^2/n^{2\alpha}-2}-{\rm i } \sqrt 2 \right).\label{eq:coordchange3}
\end{align}
 Then the maps $\zeta\mapsto \omega$, $\omega \mapsto w$ and $w \mapsto z$, maps the set $\Gamma_+$ successively to $\Gamma_+^{(1)}$, $\Gamma_+^{(2)}$ and $\Gamma_+^{(3)}$.
\begin{lemma} 
The change of variables \eqref{eq:coordchange1}--\eqref{eq:coordchange3} are well-defined, i.e. all maps are invertible. Moreover, the contours $\Gamma^{(j)}_+$ are all close to $\Gamma_+$ in the sense that all maps converge uniformly to the identity as $n\to \infty$.
\end{lemma}
\begin{proof}

Let us start with the  transform $\zeta \mapsto \omega$. For this transform the inverse can be computed explicitly. Moreover, by an expansion of the-right hand side for $\zeta \in \Gamma_+$ (where we use that   the assumption of $ \eps$ in \eqref{eq:finalass1} implies $\eps<\alpha$). This expansion  shows that $\zeta$ is the leading term at the right-hand side of \eqref{eq:coordchange3} and hence $\omega-\zeta$ is small. Hence $\Gamma_+^{(1)}$ and $\Gamma_+$ are close.

For the second transformation $\omega\mapsto w$ we rewrite \eqref{eq:coordchange2} to 
$$w= \omega+ q c_2n^{2\alpha} g_n(\omega)+q c_2n^{2\alpha} (G_n(\omega)-g_n(\omega)).
$$
Again by expanding  $g_n(\omega)$ in powers of $\omega$ we see that $qc_2 n^{2\alpha } g_n(\omega)$ is small, uniformly for $\omega\in \Gamma_+^{(1)}$. Moreover, 
\begin{multline} \label{eq:estimateonGg}
q c_2n^{2\alpha} (G_n(\omega)-g_n(\omega))\\
= \frac{c_2 n^{\alpha}}{n} \left(\sum_{j=1}^n \frac{1}{\omega/n^\alpha q+ {\rm i} c_2 \sqrt 2-\xi^{(n)}_j} -\frac{1}{\pi} \int _{-\sqrt 2}^{\sqrt 2}\frac{\sqrt{2-\xi^2}}{\omega/n^\alpha q + {\rm i} c_2 \sqrt 2- \xi} {\rm d} \xi \right),
\end{multline}
By the fact that $\gamma <\alpha$ we have 
$$\Im\left( \omega/n^\alpha q + {\rm i} c_2 \sqrt 2\right) \geq  c_2 n^{-\gamma},$$
and hence, by definition of $\mathcal C_n(\R,n^\eps)$,  we see that 
$$
|q c_2n^{2\alpha} (G_n(\omega)-g_n(\omega))| \leq   n^{\alpha+ \eps} \sqrt \frac{c_2}{n} \leq  n^{\alpha+\eps - (1+\gamma)/2},
$$
with high probability.  Combining this with the bound in \eqref{eq:finalass1}, shows that indeed $w-\omega$ is small.  The bounds extend to a small neighborhood around $\Gamma_+^{(1)}$ and the invertibility can proved by Rouche's Theorem.  

Finally,we deal with the transformation $z\mapsto w$.   By \eqref{eq:estimater1o1} we also have that $z-w(z)$ is small with high probability for $z\in S_n$ and again the invertibility follows from an application of Rouch\'e's Theorem.
\end{proof}

\subsubsection{Initial cut-off}
We define $T_n^{(1)}$ by the equation
\begin{equation}\label{eq:defT1n}
\frac{1}{n^\alpha}\int_{-\infty}^\infty f_n(x+ {\rm i} n^{-\eps} )  U_n(x+ {\rm i} n^{-\eps}) {\rm d} x
=\frac{1}{n^\alpha}\int_{\Gamma_+^{(3)}}  f_n(z )  U_n(z) {\rm d} z+T^{(1)}_n.
\end{equation} 
For the proof of Lemma \ref{lemma:final1} we  use the change of variables \eqref{eq:coordchange1}--\eqref{eq:coordchange3} in the integral at the right-hand side of \eqref{eq:defT1n}. After each transformation we keep the asymptotic relevant parts and estimate the correction terms. 

\subsubsection{Transform $z\mapsto w$}
The purpose of the first transformation is to eliminate $R_n(z/n^\alpha)$ in \eqref{eq:loopeqninU}. We start with a lemma. 
\begin{lemma}\label{lem:defVn}
Let $0 <\lambda <1$. Then with high probability, we have that for every $z\in S_n$ the equation 
\begin{equation} \label{eq:equationforVn}
w=\frac{1}{n q} \sum_{j=1}^n \frac{1}{z/q n^\alpha-c_2 w-\xi_j^{(n)}},
\end{equation}
has a unique solution in $B(- {\rm i} \sqrt 2, \lambda)$ which we denote by $V_n(z/n^\alpha)$. Moreover,
\begin{equation} \label{eq:asymptoticsforVn}
\sup_{z\in S_n} |V_n(z/n^\alpha)+{\rm i} \sqrt 2|\to 0,
\end{equation}
as $n \to \infty$. Moreover,\begin{equation}\label{eq:fromUntoVn}
U_n(z/n^\alpha) = V_n\left( z/n^\alpha - q c_2 R_n(z/n^\alpha) \right)+ R_n(z/n^\alpha), 
\end{equation}
for $z\in S_n$ .
\end{lemma}
\begin{proof}
The first part of the statement, the existence of $V_n$ and \eqref{eq:asymptoticsforVn}, can be proved using similar arguments for the proof of Lemma \ref{lem:aftersimple}. Just as for $U_n(z/n^\alpha)$, one can prove that there exists a (unique) solution $V_n(z/n^\alpha)$ to \eqref{eq:equationforVn} that is close $U(z/n^\alpha)$.  The identity \eqref{eq:fromUntoVn} then follows by rewriting \eqref{eq:loopeqninU}. 
\end{proof}

If we write $z_n(w)$ for the inverse map of \eqref{eq:coordchange1} then we can write
\begin{equation}
\frac{1}{n^\alpha} \int_{\Gamma_+^{(3)}} f_n(z) U_n(z/n^\alpha)  {\rm d} z= \frac{1}{n^\alpha} \int _{\Gamma^{(2)}_+} f_n(z_n(w)) \left(V_n(w/n^\alpha) +R_n(z_n(w)/n^\alpha) \right) z'_n(w) {\rm d} w.\nonumber
\end{equation} 
Define
\begin{align}\label{eq:defT2n}
T^{(2)}_n= \frac{1}{n^\alpha} \int_{\Gamma_+^{(3)}} f_n(z) U_n(z/n^\alpha) {\rm d} z-\frac{1}{n^\alpha} \int_{\Gamma_+^{(2)}} f_n(w) V_n(w/n^\alpha) {\rm d} w.
\end{align}
Then we have 
\begin{equation}\nonumber
\frac{1}{n^\alpha} \int_{\Gamma_+^{(3)}} f_n(z) U_n(z/n^\alpha)  {\rm d} z= \frac{1}{n^\alpha} \int_{\Gamma_+^{(2)}} f_n(w) V_n(w/n^\alpha) {\rm d} w+T^{(2)}_n.
\end{equation} 
We will show later that the second term at the right-hand side is small. Therefore we continue to rewrite the first term at the right-hand side, using the second change of variable \eqref{eq:coordchange2}.

\subsubsection{Transform $w\mapsto \omega$}

For the second transformation let us introduce the auxiliary function 
\begin{equation}\nonumber
d_n(w)= \frac{1}{n^\alpha} (V_n(w/n^\alpha) + {\rm i} \sqrt 2), 
\end{equation}
for $w \in S_n$. The following lemma provides an expression for the inverse of \eqref{eq:coordchange2} in terms of $d_n(w)$. 
\begin{lemma}
With high probability, we have that the equation \eqref{eq:coordchange2}, $w= \omega+ qc_2 n^{2\alpha} G_n(\omega)$, with $w\in S_n$ has a  unique solution $\omega \in B(0,q c_2n^\alpha)$ given by 
\begin{equation}\nonumber
\omega= w- q c_2 n^{2 \alpha} d_n(w). 
\end{equation}
\end{lemma}
\begin{proof} A straightforward computation shows that if $\tilde \omega(w)$ is a solution of 
\begin{equation}\nonumber
\label{eq:fromomegatow} w= \omega+ qc_2 n^{2\alpha} G_n(\omega),
\end{equation}
then 
$$\tilde V_n(w/n^\alpha)= \frac{w-\tilde \omega(w)}{q c_2 n^{\alpha}} - {\rm i} \sqrt 2,$$
is a solution of \eqref{eq:equationforVn} and, vice versa, for every solution $\tilde V_n(w/n^\alpha)$ of \eqref{eq:equationforVn} we have that 
\begin{equation}\label{eq:fromwtoomega}
\tilde \omega(w)= w-q c_2 n^{\alpha}\left(\tilde V_n(w/n^\alpha)+ {\rm i} \sqrt 2\right),\end{equation}
is a solution  of \eqref{eq:fromomegatow}. Hence the statement follows from Lemma \ref{lem:defVn} and \eqref{eq:fromwtoomega}.
\end{proof}
By definition of $V_n(z/n^\alpha)$ we see that $d_n(w)$ satisfies
 \begin{multline}\nonumber
 d_n(w)= \frac{1}{n^{1+\alpha} } \sum_{j=1}^n \frac{1}{w/qn^\alpha +{\rm i} c_2 \sqrt 2-c_2 n^\alpha d_n(w)-\xi^{(n)}_j} + {\rm i}  \sqrt 2/n^\alpha\\
 = G_n\left(w- q c_2 n^{2 \alpha} d_n(w) \right).
 \end{multline}
By combining  the latter with the second transformation $w \mapsto \omega$ we obtain the following
\begin{multline}\label{eq:resultafter2ndtransform}
\frac{1}{n^\alpha} \int_{\Gamma_+^{(2)}} f_n(w) V_n(w/n^\alpha) {{\rm d} w} +\frac{{\rm i} \sqrt 2}{n^\alpha} \int_{\Gamma_+^{(2)}} f_n(w)  {{\rm d} w} \\
= \frac{1}{n^\alpha} \int_{\Gamma_+^{(2)}} f_n(w) d_n(w) {{\rm d} w}  = \frac{1}{n^\alpha} \int_{\Gamma_+^{(2)}} f_n(w) G_n\left(w- q c_2 n^{2 \alpha} d_n(w) \right)  {{\rm d} w}  \\
= \frac{1}{ n^\alpha } \int_{\Gamma_+^{(1)}} f_n\left(\omega+q c_2 n^{2 \alpha} G_n(\omega) \right)G_n(\omega) \left(1+ q c_2 n^{2 \alpha } G_n'(\omega) \right) {\rm d} \omega.
\end{multline}

\subsubsection{Transform $\omega\mapsto \zeta$}
We now come to the last change of variables in \eqref{eq:coordchange3}. It is not hard to show that the inverse of \eqref{eq:coordchange3} is given by 
\begin{equation}\label{eq:coordchange3a}
\zeta= \omega+ q c_2 n^{2\alpha} g_n(\omega).
\end{equation}
By writing 
$$ q c_2 n^{2 \alpha} G_n(\omega)=q c_2 n^{2 \alpha} g_n(\omega)+ H_n(\omega)$$
and using \eqref{eq:coordchange3a}, we can rewrite \eqref{eq:resultafter2ndtransform} to  \begin{multline}\label{eq:coordchange3result}
\frac{1}{n^\alpha} \int_{\Gamma_+^{(2)}} f_n(w) V_n(w/n^\alpha) {{\rm d} w} +\frac{{\rm i} \sqrt 2}{n^\alpha} \int_{\Gamma_+^{(2)}} f_n(w)  {{\rm d} w} \\
= \frac{1}{ qc_2 n^{2\alpha}} \int_{\Gamma_+} f_n\left(\zeta+H_n(\omega(\zeta))\right)(\zeta-\omega(\zeta)+H_n(\omega(\zeta)))  \left(H_n'(\omega(\zeta)) \omega'(\zeta) +1 \right) {\rm d} \zeta.
\end{multline}

\subsubsection{Taylor expansion}
The final step towards the proof of Lemma \ref{lemma:final1} is a Taylor expansion of $f_n$ around the $\zeta$ in the integrand at the right-hand side of \eqref{eq:coordchange3result}. In this way we get 
\begin{multline}\label{eq:TaylorSigmas}
\frac{1}{ qc_2 n^{2\alpha} } \int_{\Gamma_+} f_n\left(\zeta+H_n(\omega(\zeta))\right)(\zeta-\omega(\zeta)+H_n(\omega(\zeta)))  \left(H_n'(\omega(\zeta)) \omega'(\zeta) +1 \right) {\rm d} \zeta\\
=\Sigma_1+\Sigma_2 + \Sigma_3,
\end{multline}
where 
\begin{align}
\Sigma_1&= \frac{1}{ qc_2 n^{2\alpha} }  \sum_{m=0}^\infty \frac{1}{m!} \int _{\Gamma_+} f^{(m)}_n( \zeta) H_n(\omega(\zeta))^m \left(\zeta-\omega(\zeta) + H_n(\omega(\zeta) )\right) {\rm d} \zeta,\label{eq:defSigma1}\\
\Sigma_2&= \frac{1}{ qc_2 n^{2\alpha} }  \sum_{m=0}^\infty \int _{\Gamma_+} f^{(m)}_n (\zeta) H_n(\omega(\zeta))^{m+1} (\zeta-\omega(\zeta)) H_n'(\omega(\zeta)) \omega'(\zeta)  \  {\rm d} \zeta,\label{eq:defSigma2}\\
\Sigma_3 &= \frac{1}{ qc_2 n^{2\alpha} } \sum_{m=0}^\infty \frac{1}{m!} \int_{\Gamma_+} f^{(m)}_n (\zeta) H_n(\omega(\zeta))^{m+1} H_n'(\omega(\zeta)) \omega'(\zeta) \ {\rm d} \zeta\label{eq:defSigma3}
\end{align}
In both $\Sigma_2$ and $\Sigma_3$ we use integration by parts. For $\Sigma_2$ we thus write
\begin{multline}\label{eq:defSigma2a}
\Sigma_2= - \frac{1}{ qc_2 n^{2\alpha} } \sum_{m=0}^\infty \frac{1}{(m+1)!} \int_{\Gamma_+} \left[f^{(m+1)}_n(\zeta) ( \zeta-\omega(\zeta))+f^{(m)}_n (\zeta)(1-\omega'(\zeta))\right]H_n(\omega(\zeta))^{m+1} {\rm d} \zeta\\
+ T^{(3)}_n \\
= - \frac{1}{ qc_2 n^{2\alpha} } \sum_{m=1}^\infty \frac{1}{m!} \int_{\Gamma_+} \left[f^{(m)}_n(\zeta) ( \zeta-\omega(\zeta))\right]H_n(\omega(\zeta))^{m} {\rm d} \zeta\\+ \frac{1}{ qc_2 n^{2\alpha} }  \sum_{m=0}^\infty \frac{1}{(m+1)!} \int_{\Gamma_+}f^{(m)}_n (\zeta)(\omega'(\zeta)-1) H_n(\omega(\zeta))^{m+1} {\rm d} \zeta
+ T^{(3)}_n 
\end{multline}
where 
\begin{equation}\label{eq:defT3n}
T_n^{(3)}= \frac{1}{ qc_2 n^{2\alpha} }\sum_{m=0}^{\infty} \frac{1}{(m+1)!} \left[f_n^{(m)}(\zeta) (\zeta-\omega(\zeta) ) H_n(\omega(\zeta))^{m+1} \right]_{-n^\eps +{\rm i} a}^{n^\eps +{\rm i} a}.
\end{equation}
Similarly, for $\Sigma_3$ as defined in \eqref{eq:defSigma3}, we have
\begin{multline}\label{eq:defSigma3a}
\Sigma_3= - \frac{1}{ qc_2 n^{2\alpha} }   \sum_{m=1}^\infty \frac{m}{m!(m+1)} \int_{\Gamma_+} f_n^{(m)}(\zeta) H_n(\omega(\zeta))^{m+1} {\rm d} \zeta + T^{(4)}_n,
\end{multline}
where
\begin{equation}\label{eq:defT4n}
T^{(4)}_n =  \frac{1}{ qc_2 n^{2\alpha} }  \sum_{m=0}^\infty \frac{1}{m!(m+2)}\left[f^{(m)}(\zeta) H_n(\omega(\zeta))^{m+1}\right]_{-n^\eps+ {\rm i} a}^{n^\eps+ {\rm i} a}.
\end{equation}
Hence we see that in \eqref{eq:TaylorSigmas},  that part of $\Sigma_1$ is canceled by the first sum on the right-hand side of \eqref{eq:defSigma2a}. Similarly, the first term at the right-hand side of \eqref{eq:defSigma3a} can be combined with the remaining part of $\Sigma_1$. The result of putting everything together is
\begin{multline} \label{eq:TaylorSigmasResult}
\frac{1}{ qc_2 n^{2\alpha} } \int_{\Gamma_+} f_n\left(\zeta+H_n(\omega(\zeta))\right)(\zeta-\omega(\zeta)+H_n(\omega(\zeta)))  \left(H_n'(\omega(\zeta)) \omega'(\zeta) +1 \right) {\rm d} \zeta\\
=\frac{1}{ qc_2 n^{2\alpha} } \int_{\Gamma_+} f_n\left(\zeta)\right)(\zeta-\omega(\zeta))  {\rm d} \zeta
+\frac{1}{ qc_2 n^{2\alpha} } \int_{\Gamma_+} f_n(\zeta)H_n(\omega(\zeta)) {\rm d} \zeta\\
+\frac{1}{ qc_2 n^{2\alpha} } \sum_{m=2}^\infty \frac{1}{m!} \int_{\Gamma_+} f_n^{(m-1)}(\zeta)\left(H_n(\omega(\zeta))\right)^m w'(\zeta) {\rm d} \zeta +T_n^{(3)}+ T_n^{(4)}.
\end{multline}
Finally, we put all steps together and prove Lemma \ref{lemma:final1}.
\subsubsection{Proof of Lemma \ref{lemma:final1}}
\begin{proof}
By \eqref{eq:defT1n} and \eqref{eq:defT2n} we have
\begin{multline}\nonumber
\frac{1}{n^\alpha}\int_{-\infty}^\infty f_n(x+ {\rm i} n^{-\eps} )  U_n(x+ {\rm i} n^{-\eps}) {\rm d} x
= \frac{1}{n^\alpha}\int_{\Gamma_+^{(3)}}f_n(z )  U_n(z) {\rm d} x+T^{(1)}_n\\
= \frac{1}{n^\alpha}\int_{\Gamma_+^{(3)}}f_n(z )  V_n(z) {\rm d} x+T^{(1)}_n+T^{(2)}_n
\end{multline}
By also inserting  \eqref{eq:coordchange3result} and  \eqref{eq:TaylorSigmasResult}
\begin{multline}\label{eq:beforebeforedefT5n}
\frac{1}{n^\alpha}\int_{-\infty}^\infty f_n(x+ {\rm i} n^{-\eps} )  U_n(x+ {\rm i} n^{-\eps}) {\rm d} x=\frac{1}{n^\alpha} \int_{\Gamma_+^{(2)}} f_n(w) {\rm d} w\\+
\frac{1}{ qc_2 n^{2\alpha} } \int_{\Gamma_+} f_n\left(\zeta)\right)(\zeta-\omega(\zeta))  {\rm d} \zeta
+\frac{1}{ qc_2 n^{2\alpha} } \int_{\Gamma_+} f_n(\zeta)H_n(\omega(\zeta)) {\rm d} \zeta\\
+\frac{1}{ qc_2 n^{2\alpha} } \sum_{m=2}^\infty \frac{1}{m!} \int_{\Gamma_+} f_n^{(m-1)}(\zeta)\left(H_n(\omega(\zeta))\right)^m w'(\zeta) {\rm d} \zeta\\+T_n^{(1)}+ T_n^{(2)} +T_n^{(3)}+ T_n^{(4)}.
\end{multline}
By  rewriting \eqref{eq:coordchange3} using $ 1-q c_1=q c_2$ we get
\begin{equation}
\label{eq:beforedefT5n}
\zeta-\omega(\zeta)=q c_2  U(\zeta/n^\alpha) + q c_2 {\rm i} n^\alpha \sqrt 2.
\end{equation}
Hence, after defining 
\begin{equation}\label{eq:defT5n}
T_n^{(5)}= \frac{{\rm i} \sqrt 2}{n^\alpha} \left(\int_{\Gamma_+} f_n(\zeta) {\rm d} \zeta-\int_{\Gamma^{(2)}_+} f_n(\zeta)  \  {\rm d} \zeta\right).
\end{equation}
and 
\begin{equation}\label{eq:defT6n}
T^{(6)}_n= \frac{1}{n^\alpha} \int_{\Gamma_+} f_n(\zeta)U (\zeta/n^\alpha) {\rm d} \zeta- \frac{1}{n^\alpha} \int_{-\infty}^\infty f_n(x+{\rm i } a) U((x+ {\rm i} a)/n^\alpha) {\rm d}x. 
\end{equation}
and inserting  \eqref{eq:beforedefT5n}, \eqref{eq:defT5n} and \eqref{eq:defT6n} into \eqref{eq:beforebeforedefT5n} we obtain the statement. 
\end{proof}
\subsection{Expansion into moments}
We now use the fact that $H_n$ as defined in \eqref{eq:defHn} is a generating function for $X_p$ as defined in \eqref{eq:defXp}, and expand $H_n$ in the right-hand side of  \eqref{eq:lem:final1}. That is we use
\begin{align}\label{eq:messexpHn}
H_n(\omega(\zeta)) = c_2 n^\alpha \sum_{r=0}^\infty \left(-\frac{1}{qn^\alpha} \right)^r X_r \omega(\zeta)^r, 
\end{align}
and then we expand $\omega(\zeta)$ as defined  in \eqref{eq:coordchange3}, which gives
\begin{equation}\label{eq:expandingomega}
\omega(\zeta)= \zeta- {\rm i} q c_2 \zeta \sum_{k=0}^\infty \frac{b_k}{n^{\alpha k}} \zeta^k
\end{equation}
Note that we have $b_0=-{\rm i}$, $b_{2k}=0$ and $b_{2k-1}= - \sqrt 2 \begin{pmatrix} 1/2\\ k \end{pmatrix} (-\frac12 )^k.$

We state the end result in a lemma, but before that we first introduce some notation 
\begin{equation}\label{eq:messdefws}
W_s= \sum_{r=0}^\infty \left(-\frac{1}{q}\right)^r X_r D_{s-r,r},
\end{equation}
where we successively define
\begin{align}
D_{t,r} &= \sum_{k=1}^r \begin{pmatrix} r \\  k \end{pmatrix} (- {\rm i} q c_2)^k B_{t,k} \label{eq:messdefDtr}\\
B_{t,k}& = \sum_{s_1+ \cdots + s_k=t, \ s_i \geq 0} b_{s_1} \cdots b_{s_k}.\label{eq:messdefBtk}
\end{align}
We also need 
\begin{align}
Z_s&= \frac{1}{c_2} \sum_{m=2}^\infty \frac{(-1)^m }{n^{\alpha } } \frac{(s+m-1)!}{m! s!} Z_{s+m-1,m}^{(2)}\label{eq:messdefz}\\
Z^{(2)}_{t,m}&= Z^{(1)}_{t,m}- {\rm i} q c_2 \sum_{s=0}^t Z^{(1)}_{s,m} (t-s+1) b_{t-s}\label{eq:messdefz2} \\
Z^{(1)}_{s,m} &= (c_2 n^\alpha)^m \left(\left(-\frac{1}{q} \right)^s Y_{s,m} + \sum_{r=0}^s \left(-\frac{1}{q} \right)^r Y_{r,m} D_{s-r,r}\right) \label{eq:messdefz1}\\
Y_{r,m}& = \sum_{k_1 + \ldots + k_m=r, k_i \geq 0 }  X_{k_1} \cdots X_{k_m}. \label{eq:messdefYrm}
\end{align}
and the error terms
\begin{align}\label{eq:messdefT7ntm}
T^{(7)}_{n,t,m}&=  \int_{\Gamma_+} f^{(m-1)} (\zeta) \zeta^t {\rm d} \zeta - \frac{t! (-1)^{m-1}}{(t-m+1)!}\int_{\Gamma_+} f_n(\zeta) \zeta^{t-m+1} {\rm d} \zeta\\
T^{(7)}_n&= \frac{1}{q c_2 n^{2 \alpha} } \sum_{m=2}^\infty \sum_{t=0}^\infty \frac{Z^{(2)}_{t,m}}{m! n^{\alpha t} } T_{n,t,m}^{(7)}.\label{eq:messdefT7n}
\end{align}
All these quantities arise naturally when we expand $H_n(\omega(\zeta))^m$ and $\omega'(\zeta)$ in powers of $\zeta$ as we will see in the proof of the next lemma. As a consequence of these expansions we can express the relevant quantities in terms of $\int_{\Gamma_+} \zeta^s f_n(\zeta) {\rm d} \zeta$.
\begin{lemma}\label{lem:messres}
We have that
\begin{multline}\label{eq:messres1}
\frac{1}{q c_2 n^\alpha} \int_{\Gamma_+} f_n(\zeta) H_n(\omega(\zeta)) \ {\rm d} \zeta= \frac{1}{q n^\alpha} \sum_{s=0}^\infty \left(-\frac{1}{q n^\alpha} \right)^s X_s \left(\int_{\Gamma_+} \zeta^s f_n(\zeta) \ {\rm d} \zeta\right)\\
+ \frac{1}{qn^\alpha} \sum_{s=0}^\infty \frac{ W_s}{ n^{\alpha s} } \left(\int_{\Gamma_+}\zeta^s f_n(\zeta) \ {\rm d} \zeta\right)
\end{multline}
\begin{multline}\label{eq:messres2}
\frac{1}{q c_2 n^{2 \alpha}} \sum_{m=2}^\infty \frac{1}{ m!} \int_{\Gamma_+} f_n^{(m-1)}(\zeta) H_n(\omega(\zeta))^m \omega'(\zeta) \ {\rm d} \zeta\\
= \frac{1}{q c_2 n^{2\alpha}} \sum_{s=0}^\infty \frac{Z_s}{n^{\alpha s} }\left( \int_{\Gamma_+} \zeta^s f_n(\zeta) \  {\rm d} \zeta\right). 
\end{multline}
\end{lemma}
\begin{proof}
Let us start with expanding $H_n( \omega(\zeta))$ in powers of $\zeta$. First note that from \eqref{eq:expandingomega} it follows that  
\begin{multline}\nonumber
\left(\omega (\zeta) \right)^r= \zeta^r \left(1- c_2 q  \sum_{\ell=0}^\infty \frac{b_\ell}{n^{\alpha \ell}} \zeta^\ell \right)^r
= \zeta^r \sum_{k=0}^r \begin{pmatrix} r \\ k \end{pmatrix}  (-c_2 q)^ k  \left(\sum_{\ell=0}^\infty \frac{b_\ell}{n^{\alpha \ell}} \zeta^\ell  \right)^k
\\= \zeta^r \sum_{k=0}^r \begin{pmatrix} r \\ k \end{pmatrix}  (-c_2 q)^ k  \sum_{t=0}^\infty \frac{\zeta^t}{n^{\alpha t}} B_{t,k},
\end{multline}
where in the last step we used \eqref{eq:messdefBtk}.
By changing the order of summation and using \eqref{eq:messdefDtr} we therefore obtain
\begin{equation} \label{eq:messomegar}
\left(\omega (\zeta) \right)^r= \zeta^r+ \sum_{t=0}^\infty\frac{ D_{t,r}} {n^{\alpha t}} \zeta^{t+r}.
\end{equation}
Moreover, by \eqref{eq:messexpHn} and \eqref{eq:messdefws} we therefore have 
\begin{equation}\nonumber
H_n(\omega(\zeta))= c_2 n^\alpha \sum_{s=0}^\infty \left(-\frac{1}{q}\right)^s X_s \frac{\zeta^s}{n^{\alpha s}} + c_2 n^\alpha \sum_{s=0}^\infty \frac{W_s}{n^{\alpha s} } \zeta^s.
\end{equation}
By inserting this in the left-hand side of \eqref{eq:messres1} we obtain the right-hand side.

It remains to prove \eqref{eq:messres2}. To this end, we first write from \eqref{eq:messexpHn} and \eqref{eq:messdefYrm}
\begin{equation}\nonumber
H_n(\omega(\zeta))^m = (c_2 n^\alpha)^m \sum_{r=0}^\infty \left(-\frac{1}{q n^\alpha} \right)^r Y_{r,m} \omega(\zeta)^r.
\end{equation}
Now by inserting \eqref{eq:messomegar} in the latter and using \eqref{eq:messdefz1} gives
\begin{multline}\label{eq:messexpHnm}
H_n(\omega(\zeta))^m =  (c_2 n^\alpha)^m \sum_{r=0}^\infty \left(-\frac{1}{q n^\alpha} \right)^r Y_{r,m} \left( \zeta^r+ \sum_{t=0}^\infty\frac{ D_{t,r}} {n^{\alpha t}} \zeta^{t+r}\right)
\\= \sum_{s=0}^\infty Z^{(1)}_{s,m} \frac{\zeta^s}{n^{ \alpha s}}.
\end{multline}
Now also use 
$$
\omega'(\zeta)= 1- {\rm i} q c_2  \sum_{k=0}^\infty (k+1)  b_k \frac{\zeta^k }{n^{\alpha k } },$$
Hence by multiplying the latter with  \eqref{eq:messexpHnm} and using \eqref{eq:messdefz2} we 
get 
$$
H_n(\omega(\zeta))^m \omega'(\zeta) = \sum_{t=0}^\infty Z^{(2)}_{t,m} \frac{\zeta^t}{n^{\alpha t}},
$$
and hence we have 
\begin{multline*}
\frac{1}{q c_2 n^{2 \alpha}} \sum_{m=2}^\infty \frac{1}{ m!} \int_{\Gamma_+} f_n^{(m-1)}(\zeta) H_n(\omega(\zeta))^m \omega'(\zeta) \ {\rm d} \zeta\\
= \frac{1}{q c_2 n^{2\alpha}} \sum_{m=2}^\infty \sum_{t=0}^\infty \frac{Z^{(2)}_{t,m}}{m! n^{\alpha t}} \int_{\Gamma_+} \zeta^t f_n^{(m-1)} ( \zeta) \  {\rm d} \zeta
\end{multline*}
and by inserting \eqref{eq:messdefT7ntm} and \eqref{eq:messdefT7n} and using \eqref{eq:messdefz} we now obtain \eqref{eq:messres2}.
\end{proof}

\subsection{Proof of Proposition \ref{prop:final}}

In this subsection we proof Proposition \ref{prop:final}. Let us first recollect what we have achieved so far.  By  combining Lemma's \ref{lemma:final1} and \ref{lem:messres} we have
\begin{multline*}
\int_{-\infty}^{\infty} f_n(x+ {\rm i} n^{-\eps}) \frac{1}{n^{\alpha}} \left(U_n\left( \frac{x+ {\rm i }n^{-\eps}}{n^\alpha}\right)-U\left(\frac{x+ {\rm i }n^{-\eps}}{n^\alpha}\right)  \right) \ {\rm d} x\\
= \frac{1}{q} \sum_{s=0}^\infty \frac{1}{n^{\alpha(s+1)}}  \left[\left(-\frac{1}{q}\right)^s X_s+ W_s+Z_s\right] \int_{\Gamma_+} f_n(\zeta) \zeta^s {\rm d} \zeta+ \sum_{j=1}^7 T_n^{(j)}.
\end{multline*}
where $X_s$, $W_s$ and $Z_s$ are defined as in \eqref{eq:defXp}, \eqref{eq:messdefws} and \eqref{eq:messdefz}. The error terms can be found in \eqref{eq:defT1n}, \eqref{eq:defT2n}, \eqref{eq:defT3n}, \eqref{eq:defT4n}, \eqref{eq:defT5n}, \eqref{eq:defT6n} and \eqref{eq:messdefT7n}. 

Since we assume that the moment $\mu_j$ vanish for $j=1,\ldots,p-1$, we define a final error term by
\begin{equation}\label{eq:defT8n}
T_n^{(8)}= \frac{1}{q} \sum_{s=0}^{p-1} \frac{1}{n^{\alpha(s+1)}}  \left[\left(-\frac{1}{q}\right)^s X_s+ W_s+Z_s\right] \int_{\Gamma_+} f_n(\zeta) \zeta^s {\rm d} \zeta.
\end{equation} 
and write
 \begin{equation}\label{eq:defTnpfinal}
 T_{n,p}= \sum_{j=1}^8 T_n^{(8)},
 \end{equation}
where we indicated the dependence on $p$, since $T_n^{(8)}$ depends on $p$. 
If we further define
\begin{align} \label{eq:defqnp1}
Q_{n,p}^{(1)}&= \frac{1}{q n^{\alpha(p+1)}} (W_p+Z_p) \int_{\Gamma_+} f_n(\zeta) \zeta^p {\rm d} \zeta\\
Q_{n,p}^{(2)}&= \frac{1}{q} \sum_{s=p+1}^\infty \frac{1}{n^{\alpha(s+1)}}  \left[\left(-\frac{1}{q}\right)^s X_s+ W_s+Z_s\right] \int_{\Gamma_+} f_n(\zeta) \zeta^s {\rm d} \zeta,\label{eq:defqnp2}
\end{align} 
and set $Q_{n,p}= Q_{n,p}^{(1)}+Q_{n,p}^{(2)}$,
then we see that we indeed have \eqref{eq:prop:final}. To finish Proposition \ref{prop:final} it remains to verify \eqref{eq:prop:final1} and \eqref{eq:prop:final2} for which  we need to estimate the various terms.

We start with a Lemma containing estimates on the various auxillary quantities that we have introduced in this Section.

\begin{lemma}\label{lem:estimatesauxil}
There exists a constant $C>0$ such that for for $n$ sufficiently large and  all $k,r,s,t,m\in \N$ we have
\begin{enumerate}
\item $|b_k|\leq C/k,$
\item $|D_{t,r} | \leq C^{t+r+1} c_2,$
\item $|X_k|\leq  C n^{k \gamma+ (\gamma-1)/2+ \eps},$
\item $|Y_{r,m}| \leq (C n^{(\gamma-1)/2+\eps})^m n^{\gamma r},$
\item $|W_s| \leq  C^{s+1} n^{\gamma s+ (\gamma-1)/2+\eps} c_2 (s+1)$
\item $|Z_{s,m}^{(1)}|\leq C c_2 (s+1) C^s n^{s {\gamma}} (C n^{(\gamma-1)/2+\eps+\alpha-\gamma})^m,$ 
\item $|Z_{t,m}^{(2)}|\leq C c_2 (t+1)^2 C^t n^{t {\gamma}} (C n^{(\gamma-1)/2+\eps+\alpha-\gamma})^m,$ 
\item $|Z_s|\leq C^{s+1} n^{\gamma s-1+ 2 \eps}$.
\end{enumerate}
\end{lemma}
\begin{proof}
1. We will use the estimate
\begin{equation}
\label{eq:standardbinomialestimate} \begin{pmatrix} n\\ k\end{pmatrix} \leq \left(\frac{n {\rm e}}{k}\right)^k.
\end{equation}
Since we have $b_{2k}=0$ and $b_{2k-1}$ is given by 
$$b_{2k-1}= -\sqrt 2 \begin{pmatrix} 1/2 \\ k \end{pmatrix} (-\frac{1}{2})^k= \sqrt 2 \begin{pmatrix} 2k \\ k \end{pmatrix} \frac{1}{2^{3k}(2k-1)} \leq \frac{\sqrt {2}}{2k-1} ({\rm e}/4)^{k},$$
and this proves the estimate on $b_k$.

2. By the first estimate we also have that that the $b_k$ are bounded, and hence \eqref{eq:messdefBtk} give
$$|B_{t,k}|\leq C^t \sum_{s_1+\ldots+s_k=t, \ s_i \geq 0} 1 = C^t \begin{pmatrix} k+t-1\\ t \end{pmatrix} \leq C {\rm e}^{t+k},
$$
where we used \eqref{eq:standardbinomialestimate}. Then, by inserting this into \eqref{eq:messdefDtr}, we also have 
\begin{equation*}
|D_{t,r}|  \leq C^t \sum_{k=1}^r \begin{pmatrix} r \\ k \end{pmatrix} (q c_2 C)^k= qc_2 C^{t+1} \sum_{k=1}^r \begin{pmatrix} r \\ k \end{pmatrix} (q c_2 C)^{k-1}\leq 
C^{t+r+1} c_2,
\end{equation*}
since $q c_2 C \leq 1$ if $n$ is sufficiently large.

3.  We start by  recalling \eqref{eq:estimateonGg}, \eqref{eq:defHn} and \eqref{eq:messexpHnm}, from which we deduce that with high probability we have
\begin{equation} nonumber
|X_k|=  \frac{(qn^\alpha)^{k+1}}{2\pi} \left|\int_{|w|= q c_ n^{\alpha}} \frac{G_n(w)-g_n(w)}{w^{k+1}} {\rm d} w \right| \leq \frac{C' n^{(\gamma-1)/2+ \eps}}{(q c_2)^k} 
\end{equation}
for some constant $C'>0$. The statement now follows from \eqref{eq:defc2a}.

4. Follows  from \eqref{eq:messdefYrm} and  the estimate on $X_k$.

5. Follows from \eqref{eq:messdefws} and the estimates on  $X_k$ and $D_{t,r}$.

6. Follows from  \eqref{eq:messdefz1} and the estimates on  $Y_{r,m}$ and $D_{t,r}$.

7. Follows from   \eqref{eq:messdefz2}  and the estimates on  $Z^{(1)}_{s,m}$ and $b_k$.

8. First we mention that 
$$ \frac{(s+m-1)!}{m! s!}= \frac{1}{m} \begin{pmatrix} s+m-1\\ s\end{pmatrix}\leq {\rm e}^{s+m}.
$$
By using this and the estimate  on $Z^{(2)}_{s,m}$   in \eqref{eq:messdefz}, we obtain
\begin{multline*}
|Z_s| \leq \frac{1}{c_2} \sum_{m=2}^\infty\frac{1}{n^{\alpha m }} {\rm e}^{s+m} C c_2 (s+m-1)^2 C^s n^{\gamma(s+m-1)} \left(Cn^{(\gamma-1)/2+ \eps + \alpha-\gamma}\right)^m\\
\leq C^{s+1} n^{\gamma(s-1)} \sum_{m=2}^\infty \left(C n^{(\gamma-1)/2+ \eps}\right)^m (s+m-1)^2 \leq C^{s+1} n^{\gamma s -1 + 2 \eps},
\end{multline*}
for $n$ sufficiently large, and this finishes the proof.  \end{proof} 
 
 We are now ready to prove Proposition \ref{prop:final}.
 \begin{proof}[Proof of Proposition \ref{prop:final}]
It remains to show that $Q_{n,p}= Q^{(1)}_{n,p}+ Q^{(2)}_{n,p}$ as defined in \eqref{eq:defqnp1} and \eqref{eq:defqnp2} and $T_{n,p}$ as defined in \eqref{eq:defTnpfinal}. 

To start with, we note that from  \eqref{eq:estimateonentireinstrip} we find that for every $M\in \N$ there exists a constants $C_M(f),\tilde C_M(f)$ such that 
\begin{align}\label{eq:esimatesemimoment}
\int_{\Gamma_+} f_n(\zeta) \zeta^s {\rm d} \zeta \leq  C_M(f) \int_{-n^\eps}^{n^\eps} \frac{|x+ {\rm i} n^{-\eps}|^s}{(1+x^2)^{M/2}}  {\rm d} x
\leq  \tilde C_M(f)(n^{(s-M+1) \eps}+1).
\end{align}
By using this with  $M=p+2$ and the estimates for $W_p$ and $Z_p$ in Lemma \ref{lem:estimatesauxil} that there exists a constant $C$ such that 
\begin{multline}\nonumber
Q^{(1)}_{n,p} \leq  \frac{C }{qn^{\alpha(p+1)}} \left(n^{\gamma p + (\gamma-1)/2+\eps}n^{-\gamma}+ n^{\gamma p-1+ 2\eps}\right) \\
\leq  C n^{\eps}\left(n^{-\gamma} + n^{\eps-(1+\gamma)/2}\right)  \frac{1}{n^{(p+1)\alpha-p \gamma +(1-\gamma)/2}}.
\end{multline}
By \eqref{eq:finalass1} we also have 
$$
\eps-\frac{1+\gamma}{2}< \frac{1-\gamma}{2}  -\frac{1+\gamma}{2}= -\gamma,
$$
and hence we obtain the estimate \eqref{eq:prop:final1} for $Q^{(1)}_{n,p}$.

We prove that the estimate also holds for $Q^{(2)}_{n,p}$. To start with, from the estimates in Lemma \ref{lem:estimatesauxil} we have 
\begin{equation}\label{eq:messofthree}
\left|(-\frac1q)^s X_s+W_s+Z_s\right| \leq C^s n^{\gamma s+ \frac{\gamma-1}{2}+ \eps}.
\end{equation}
 Again, by using \eqref{eq:esimatesemimoment} with $M=p+2$ we find that there exists constants $C,C', C_M(f)>0$ such that 
 \begin{multline}\nonumber
|Q_{n,p}^{(2)}| \leq \frac{C_m(f)}{qn^\alpha} \sum_{s=p+1}^\infty \frac{1}{n^{\alpha s}}  C^s n^{\gamma s+ \frac{\gamma-1}{2}+ \eps} (n^{(s-M+1)\eps}+1)\\
\leq C' \frac{n^{(\gamma-1)/2+\eps}}{n^{(M-1)\eps+ \alpha}} \frac{1}{n^{(\alpha-\gamma)(p+1)-\eps(p+1)}}+ C' \frac{n^{\frac{\gamma-1}{2}+\eps}}{n^{\alpha+(\alpha-\gamma)(p+1)}},
 \end{multline}
 where we also used that $0<\eps<\alpha-\gamma$. Hence we have that the estimate in \eqref{eq:prop:final1} also holds for $Q_{n,p}^{(2)}$ and hence we proved that it holds for $Q_{n,p}$. 
 
 It remains to show \eqref{eq:prop:final2}. This boils down to estimate the eight terms $T_n^{(j)}$ for $j=1,\ldots,8$, which we will do in numerical order. We let $\eps_2>0$ be  sufficiently small.
 
 For $T^{(1)}_n$ in \eqref{eq:defT1n}, we note that we have $\frac{1}{n^\alpha} |U_n(x+ {\rm i}n^{-\eps})  | \leq n^\eps$ for  $x\in \R$. Then it follows from \eqref{eq:estimateonentireinstrip} that there exists a constant $C'>0$ such that 
 \begin{equation}\label{eq:essieonT1}
 |T^{(1)}_n| \leq C' n^{-(M-2)\eps}\leq C' n^{-1-\eps_2},
 \end{equation}
 where we  have chosen $M$ large enough so that $(M-2)\eps \geq 1+ \eps_2$.
 
 Next we estimate $T^{(2)}_n$ in \eqref{eq:defT2n}. We start by noting that by using the map $w\mapsto z$ we have
\begin{multline} T^{(2)}_n = \frac{1}{n^\alpha} \int_{\Gamma_+^{(2)}} f_n(z(w)) \left(V_n(w/n^\alpha)+R_n(z(w)/n^{\alpha})\right)z'(w){\rm d}w\\-\frac{1}{n^\alpha} \int_{\Gamma_+^{(2)}} f(w) V_n(w/n^\alpha) {\rm d}w
\end{multline}
  Now note that since $R_n$  satisfies \eqref{eq:estimater1o1} there exists a $C>0$ such that
 \begin{equation*}
 |z(w)-w| \leq q c_2 n^{\alpha} |R_n(z(w)/n^\alpha)| \leq C n^{3 \alpha -\gamma-2+ 2 \eps},
 \end{equation*}
 for $w\in \Gamma_+^{(2)}$.   The bound also holds in a  $\frac13n^{-\eps}$-neighborhood of $\Gamma_+^{(2)}$  and therefore we  can choose $C$ such that we also have  $$|z'(w)-1|\leq  C n^{  3 \alpha -\gamma-2+ 3 \eps}.$$ 
 By \eqref{eq:finalass1} we have $\eps< \frac{2}{3}((\gamma+1)/2-\alpha)$ and $\eps<(1-\alpha)/2$.  By using these inequalities we find that $T^{(2)}_n \leq C/n^{1+\eps_2}$ for some constant $C> 0$ and sufficiently small $\eps_2>0$.
 
 Before we come to the estimate for $T^{(3)}_n$, $T^{(4)}_n$ and $T^{(5)}_n$ we note from \eqref{eq:estimateonentireinstrip} and Cauchy's intergal transform we have, for some constant $C_M(f)$, 
 \begin{equation}\nonumber
 |f_n^{(m)}(x+ {\rm i} n^{-\eps}) | \leq \frac{C_M(f)  n^{m\eps}}{(1+x^2)^{M/2}},
 \end{equation}
 for $M\in \N$. By the fact that $\zeta-\omega(\zeta)$ and $H_n(\omega(\zeta))$ are small (as proved in Lemma \ref{lemma:final1}) we have, by using \eqref{eq:defT3n}, that
 \begin{equation}\nonumber
 |T^{(3)}_n| \leq \frac{C}{n^{2\alpha -\gamma+(M-1)\eps}} \leq \frac{C}{n^{1+\eps_2}}.
 \end{equation}
Here  we choose $M$ large enough (depending on $\alpha, \gamma$ and $\eps,\eps_2>0$ only).

 The same estimates can be done for $T^{(4)}_n$ as defined in \eqref{eq:defT4n}.
 
 For $T^{(5)}_n$ in \eqref{eq:defT5n} we remark that we can apply Cauchy's Theorem (by completing the path $\Gamma_+\cup(-\Gamma_+^{(2)})$ to a closed contour) and \eqref{eq:estimateonentireinstrip} to show that $T^{(5)}_n \leq C/n^{1+\eps_2}$ for some come constant $C>0$ and sufficiently  small $\eps_2>0$. 
 
 Since $|\frac{1}{n^{\alpha}} U(z/n^{\alpha}) | \leq \frac{1}{|\Im z|}$, the estimate \eqref{eq:estimateonentireinstrip} also proves that $T^{(6)}_n\leq C/n^{1+\eps_2}$.

 To deal with $T^{(7)}_n$ in \eqref{eq:messdefT7n}, we recall that $T^{(7)}_{n,t,m}$ in \eqref{eq:messdefT7ntm}  contains all the boundary terms from the integration by parts. These boundary terms van be estimated in the same way as $T^{(3)}_n$ and $T^{(4)}_n$. When we insert these back into \eqref{eq:messdefT7n} we obtain 
 \begin{equation}\nonumber
 |T^{(7)}_n| \leq \frac{1}{qn^\alpha} \sum_{m=2}^\infty \sum_{t=0}^\infty \frac{|Z^{(2)}_{t,m}|}{n^{\alpha t}} \frac{(t+m-1)!}{m! t!} \frac{(Cn^\eps)^t}{n},
 \end{equation}
 By using the estimate on $Z^{(2)}_{t,m}$  in Lemma \ref{lem:estimatesauxil} and the inequality in \eqref{eq:standardbinomialestimate} we obtain  the following 
  \begin{equation}\nonumber
 |T^{(7)}_n| \leq \frac{1}{qn^\alpha} \sum_{m=2}^\infty \sum_{t=0}^\infty \left(\frac{C}{n^{\alpha-\gamma-\delta}}\right)^t \left(C n^{\frac{\gamma-1}{2} + \eps+ \alpha-\gamma}\right)^m \leq \frac{C}{n^{1+\eps_2}},
 \end{equation}
 for some sufficiently small $\eps_2>0$, 
 since $\eps+ \alpha-\frac{1+\gamma}{2}<0$ by \eqref{eq:finalass1}.
 
 Finally we come to $T^{(8)}_n $ in \eqref{eq:defT8n}. For $0\leq s < p$, we have by \eqref{eq:estimateonentireinstrip} and \eqref{eq:estimateonentireinstrip}, the fact that $\mu_s(f)=0$ and \eqref{eq:preservingmoments}, we have
  $$\left|\int_{\Gamma_+} f_n(\zeta) \zeta^s {\rm d} \zeta\right| \leq C \int_{n^\eps}^\infty \frac{{\rm d}x}{ x^{M-p+1}} = C n^{-(M-p)\eps}.$$
 Together with \eqref{eq:messofthree} this implies that \eqref{eq:defT8n} can be estimated as 
$$T^{(8)}_n \leq C \sum_{s=0}^{p-1} \frac{C^s}{n^{\alpha(s+1)}} n^{\gamma s + (\gamma-1) /2 + \eps}n^{-(M-p)\eps} \leq \frac{C}{n^{1+\eps_2}},$$
where the last step follows after choosing $M$ such that $(M-p)\eps \geq 1+ \eps_2$. 

Concluding, there exists constants $C,\eps_2>0$ such that $T^{(j)}_n \leq C/n^{1+\eps_2}$ and hence we have \eqref{eq:prop:final2} and this concludes the proof. 
 \end{proof}

\subsection{Proof of Theorem \ref{th:random2}: the general case}

\begin{proof}[Proof of Theorem \ref{th:random2}]
 The proof is similar to the proof for the case $\mu_0(f)\neq 0$ in Subsection \ref{subsec:random2firstcase}. The only difference is now that instead of using Proposition \ref{prop:random}, we use the following claims that we will prove here 
\begin{equation}\label{eq:finalclaim}
\frac{1}{n^{(1-\alpha)/2+(p+1/2)(\gamma-\alpha)}}\left(\EE_{K_{n}} Y_n(f)- \EE_{\xi} \EE_{K_{n}} Y_n(f)\right) \to N(0,C_p \mu_p(f)^2), 
\end{equation}
as $n \to \infty$, if $1-\alpha+(2p+1)(\gamma-\alpha)\geq 0$, and 
\begin{equation}\label{eq:finalclaim2}
\Var\left( \EE_{K_n} Y_n(f)- \EE_{\xi} \EE_{K_n} Y_n(f) \right)\to 0,  
\end{equation}
as $n \to \infty$, in case $1-\alpha+(2p+1)(\gamma-\alpha)\leq 0$.

Let us start with proving \eqref{eq:finalclaim}.  To this end, first note that  because of Corollary \ref{cor:analyticapprox},\eqref{eq:cauchycauchy} and the fact $f$ is real-valued, we can write 
\begin{multline*}
Y_n(f)\approx Y_n(f_n)= -\frac{n}{n^\alpha 2\pi {\rm i}} \int_{-\infty} ^\infty f_n(x+{\rm i} n^{-\eps} )U_n(x+ {\rm i} n^{\eps}) {\rm d} x\\+ \frac{n}{n^\alpha 2\pi {\rm i}} \int_{-\infty} ^\infty f_n(x-{\rm i} n^{-\eps} )U_n(x- {\rm i} n^{\eps}) {\rm d} x\\
= -\frac{n}{n^\alpha \pi}  \Im \int_{-\infty} ^\infty f_n(x+{\rm i} n^{-\eps} )U_n(x+ {\rm i} n^{\eps}) {\rm d} x.
\end{multline*} 
The key to the claim in \eqref{eq:finalclaim} is that the dominant term in \eqref{eq:prop:final} is $X_p$.  Since $X_p$ is a sum of $n$ independent random variables it satisfies a Central Limit Theorem in a natural way. Of course, \eqref{eq:prop:final} only holds with high probability, but since we can argue as in Subsection \ref{subsec:random2firstcase} and restrict ourselves to $\mathcal C_n(\R,1,\eps)$ we will ignore this technical issue here.

We start by finding the dominant term in the asymptotic behavior of the variance of $\Im X_p$.  To start with, by definition we have 
\begin{multline}\nonumber
\Var \Im  X_p = \frac{1}{n \pi } \int_{-\sqrt 2}^{\sqrt 2} \Im \left(\frac{1}{{\rm i} \sqrt 2 c_2- \xi}\right)^{p+1}\Im \left(\frac{1}{{\rm i} \sqrt 2 c_2- \xi}\right)^{p+1}  \sqrt{2-\xi^2} {\rm d} \xi\\
-\frac{1}{n} \left(\frac{1}{ \pi } \int_{-\sqrt 2}^{\sqrt 2} \Im \left(\frac{1}{{\rm i} \sqrt 2 c_2- \xi}\right)^{p+1} \sqrt{2-\xi^2} {\rm d} \xi\right)^2.
 \end{multline}
 After rescaling variables,  we find that the dominant term comes from the first integral at the right-hand side and takes the form
 $$
\frac{\sqrt 2}{n \pi(c_2 \sqrt 2)^{2p+1} } \int_{-\infty}^\infty \Im \left(\frac{1}{{\rm i} - \xi}\right)^{p+1}\Im \left(\frac{1}{{\rm i} - \xi}\right)^{p+1}   {\rm d} \xi
$$
By expanding the imaginary part we see that  the integrand is a sum of four terms, two of which can be shown to vanish  using Cauchy's integral formula. Hence  this can be reduced to 
$$
\frac{\sqrt 2}{n 2 \pi(c_2 \sqrt 2)^{2p+1} } \int_{-\infty}^\infty \frac{1}{(1+\xi^2)^{p+1}}   {\rm d} \xi= \frac{(2p)!}{ n \sqrt{ 2}  \pi (p!)^2 4^p  ( c_2 \sqrt 2)^{2p+1} }.
$$
By using $c_2 \sqrt 2 = \tau(1+o(1))$, we find 
$$\Var \Im X_p= n^{(2p+1)\gamma-1}\frac{(2p)!} {\sqrt{2}(p!)^2   4^p\tau ^{2p+1} }(1+o(1)),$$
as $n\to \infty$. In particular, this implies that 
$$\frac{n}{n^{\alpha(p+1)}} \Im  X_p \sim n^{-(p+1)\alpha+\gamma (p+1/2)+ 1/2},$$
and by comparing this with \eqref{eq:prop:final1} and \eqref{eq:prop:final2} it shows that $X_p$ is indeed the dominant term at the right-hand side of \eqref{eq:prop:final}  in case $1-\alpha+(2p+1)(\gamma-\alpha)\geq 0$. The proof of \eqref{eq:finalclaim} now follows by applying the standard arguments for a Central Limit Theorem for independent random variables. 

If $1-\alpha+(2p+1)(\gamma-\alpha)< 0$, then \eqref{eq:finalclaim2} follows by the same estimate on $X_p$ and the bounds in \eqref{eq:prop:final1} and \eqref{eq:prop:final2}. 
\end{proof}
\appendix

\section{Appendix}
For completeness we provide a short discussion on the the fact that the points $x_j(t)$ for $j=1,\ldots,n$ and $t>0$  form a determinantal point process with kernel \eqref{eq:defKn}. We refer to \cite{J1} for more details. The main idea, is that the distribution for the eigenvalues can also be obtained as a model of non-colliding particles driven by the Ornstein-Uhlenbeck process.

Fix $T>0$ and let us consider $n$ particles that evolve as follows. At time $t=0$, they start in $\xi^{(n)}_j$ for $j=1,\ldots,n$ and at time $T$ they end in consecutive integers $j-1$. Each particle is driven by an Ornstein-Uhlenbeck process with parameters chosen such that we have the following transition function
\begin{align}\nonumber
f_{t}(y,x)=
\frac{\sqrt{n}}{\sqrt{\pi(1-e^{-2t})}} e^{-n\frac{(e^{-t}y -x)^2}{1-e^{-2t}}},
\end{align} 
where $t>0$ is the time and $y$ the starting point.  Finally, the particles are conditioned never to collide as for $0\leq t \leq T$.   

By the Karlin-McGregor formula  the distribution of the particles at time $t$ is given by the following probability measure
\begin{align}\label{eq:karlinmc}
\frac{1}{Z_n} \det\left(f_{t}(y_i,x_j)\right)_{i,j=1}^n \det\left(f_{T-t}(x_i,j-1)\right)_{i,j=1}^n {\rm d} x_1 \ldots {\rm d} x_n.
\end{align}
It was shown in \cite{J1} that in the limit $T\to \infty$, the distribution \eqref{eq:karlinmc} of the particles coincides with the distribution of the eigenvalues of $M_n$ in \eqref{eq:interpolatingmodel}, which basically follows by comparing the limit $T\to \infty$ in \eqref{eq:karlinmc} with the Harish-Chandra/Itzyskon-Zuber formula for $M_n$. 

Let us for the moment first fix $T>0$. Then, by the Eynard-Mehta Theorem (see for example \cite{Bor,J4}), the positions of the particles at time $t$ form a determinantal point process with kernel $K_{n,T}$ given by 
\begin{align}\label{eq:A3}
K_{n,T}(x,y)=\sum_{i=1}^n\sum_{j=1}^n \left( A^{-1}\right)_{ij} f_{T-t}(x,i-1)  f_{t} (\xi_j^{(n)},y).
\end{align}
where $A$ is the Gram matrix given by 
\begin{equation}\nonumber
A_{ij}= \int_{-\infty}^\infty f_{T-t}(x,i-1)  f_{t}(\xi^{(n)}_j,x) {\rm d} x
= \sqrt{\frac{n}{\pi (1-{\rm e}^{-T})}} \exp \left(-n \frac{({\rm e}^{-T}  \xi^{(n)}_j -(j-1))^2}{1-{\rm e} ^{-2T}}\right).
\end{equation}
Observe that, by view $K_n$ as an operators on $\mathbb L_2(\R)$, we can now easily verify the following two identities
\begin{eqnarray}
\mathrm{Rank} (K_n)=n,\\
K_n^2=K_n.
\end{eqnarray}
These identities were used in \eqref{eq:reproducing} and Section 6. 

In the next step, we rewrite the formula for the kernel further. It will be convenient to also have the following notation 
\begin{align}\nonumber
g_{t}(y,x)=
\frac{\sqrt{n}e^{t-T}}{{\rm i} \sqrt{\pi(1-e^{2(t-T)})}} \exp\left({n\frac{(y-e^{t-T}x)^2}{1-e^{2(t-T)}}}\right),
\end{align} 
for $t>0$ which may be thought of as the inverse transition function in view of the following
\begin{align}\nonumber
f_{T-t}(y,x)=\int_{{\rm i} \R} f_{T}(w,y) g_{t}(x,w) {\rm d}w,
\end{align}
for $0<t<T$.   We also note that by using the well-known Vandermonde determinant we have 
\begin{multline*}
\det A = \left(\frac{n}{\pi(1- {\rm e}^{-2T})}\right)^{n/2} \exp\left(-\frac{n}{1-{\rm e}^{-2T}} \sum_{j=1}^n\left( {\rm e}^{-2T}\left(\xi^{(n)_j}\right)^2+(j-1)^2\right)\right)  \\
\times \prod_{1\leq i < j \leq n}   \left({\rm e}^{ \frac{n  \xi^{(n)}_i}{\sinh(T)}}-{\rm e}^{ \frac{n  \xi^{(n)}_j}{\sinh(T)}}\right)^2.
\end{multline*}
By Kramer's rule we can rewrite \eqref{eq:A3} as 
\begin{align*}
K_{n,T}(x,y)&=\sum_{j=1}^n  \frac{\det \left(  A| \textrm{  column $j$ replaced by } f_{T-t}(x,i-1)\right) }{ \det A}f_{t}(\xi^{(n)}_j,y)\\
&=\int_{{\rm i} \R}
\sum_{j=1}^n  \frac{\det \left(  A| \textrm{  column $j$ replaced by } f_{T,w}(i-1)\right) }{ \det A}f_{t}(\xi_j^{(n)},y)  g_{t}(x,	w) {\rm d}w\\
&= \int_{{\rm i} \R}
\sum_{j=1}^n  
\frac{\prod_{s=1, s
\neq j}^n\left( {\rm e}^{ \frac{n w}{\sinh T}}-{\rm e}^{ \frac{n \xi_s^{(n)}}{\sinh T}}\right)}{\prod_{s=1, s
\neq j}^n\left( {\rm e}^{ \frac{n \xi_j^{(n)}}{\sinh T}}-{\rm e}^{ \frac{n \xi_s^{(n)}}{\sinh T}}\right)}\frac{{\rm e}^{-n{\frac{{\rm e}^{-2T}w^2}{1-{\rm e}^{-2T}}}}}{{\rm e}^{-n{\frac{{\rm e}^{-2T}(\xi^{(n)}_j)^2}{1-{\rm e}^{-2T}}}}}f_{t}(\xi^{(n)}_j,y)  g_{t}(x,w) {\rm d}w.
\end{align*}
 By taking the limit $T\to \infty$ we get 
\begin{align}\nonumber
\lim_{T\to \infty} K_{n,T}(x,y)= \int_{{\rm i} \R}
\sum_{j=1}^n  
\frac{\prod_{s=1, s\neq j}^n w-\xi^{(n)}_s}{\prod_{s=1, s
\neq j}^n \xi^{(n)}_j-\xi^{(n)}_s}f_{t}(\xi^{(n)}_j,y)  g_{t}(x,w) {\rm d}w,
\end{align}
which we can rewrite by Cauchy's theorem to 
\begin{align*}
\lim_{T\to \infty}&  K_{n,T}(x,y)= \frac{1}{2\pi {\rm  i}} \int_{{\rm i} \R}
\sum_{j=1}^n  
\frac{\prod_{s=1, s\neq j}^n w-\xi^{(n)}_s}{\prod_{s=1, s
\neq j}^n \xi^{(n)}_j-\xi^{(n)}_s}f_{t}(\xi^{(n)}_j,y)  g_{t}(x,w) {\rm d}w\\
&= \frac{1}{2\pi{\rm i}} \oint_\Sigma {\rm d} z \int_{\Gamma} {\rm d}w
\frac{\prod_{s=1}^nw -\xi^{(n)}_s}{\prod_{s=1}^n {z-\xi^{(n)}_s}}f_{t}(z,y)  g_{t}(x,w)\frac{1}{w-z},\\
&= \frac{n }{\sinh t(2\pi{\rm i})^2} \oint_\Sigma {\rm d} z \int_{\Gamma} {\rm d}w
\frac{\prod_{s=1}^nw -\xi^{(n)}_s}{\prod_{s=1}^n {z-\xi^{(n)}_s}} \frac{{\rm e}^{\frac{n}{1-{\rm e}^{-2t}} ({\rm e}^{-2t}w^2-2 {\rm e}^{-t}wx)}}{{\rm e}^{\frac{n}{1-{\rm e}^{-2t}} ({\rm e}^{-2t}z^2-2 {\rm e}^{-t}zy)}}\frac{1}{w-z},
\end{align*}
where $\Sigma$ is simple, closed and counter clockwise oriented contour around the points $\xi_j^{(n)}$ and $\Gamma$ is contour connecting $-{\rm i} \infty$ to $+ {\rm i} \infty$ to that lies to the right of $\Sigma$.   Hence we have proved that 
\begin{align*}
\lim_{T\to \infty}&  K_{n,T}(x,y)=K_n(x,y)\frac{{\rm e} ^{\frac{n y^2}{1-{\rm e}^{-2t}}}}{{\rm e} ^{\frac{n x^2}{1-{\rm e}^{-2t}}}}.
\end{align*}
Since a conjugation of a kernel of a determinantal point process with any non-vanishing function is also a kernel for the same process, we have proved that the eigenvalues of $M_n(t)$ in \eqref{eq:interpolatingmodel} indeed for a determinantal point process with kernel $K_n$ in \eqref{eq:defKn}.


\begin{thebibliography}{99}
\bibitem{AHM} Y. Ameur, H. Hedenmalm and N. Makarov, \emph{ Fluctuations of eigenvalues of random normal matrices},
Duke Math. J. 159 (2011), no. 1, 31--81.
\bibitem{AGZ} G.W.~Anderson, A.~Guionnet and O.~Zeitouni, 
\emph{An introduction to random matrices.}
Cambridge Studies in Advanced Mathematics, 118.  Cambridge University Press, Cambridge, 2010. 

 
 \bibitem{Bender}  M. Bender, \emph{Global fluctuations in general $\beta$ Dyson's Brownian motion}, Stochastic Process. Appl. 118 (2008), no. 6, 1022--1042.
 
 \bibitem{Bor} A. Borodin, \emph{Biorthogonal ensembles}, Nuclear Phys. B 536 (1999), no. 3, 704--732. 

 \bibitem{BorDet} A.~Borodin, \emph{Determinantal point processes},  In: Oxford Handbook on Random Matrix theory, edited by Akemann G.; Baik, J. ; Di Francesco P., Oxford University Press, 2011. (arXiv:0911.1153) 
 
 \bibitem{Ca} T.~Cabanal-Duvillard, \emph{Fluctuations de la loi empirique de grandes matrices aléatoires},  Ann. Inst. H. Poincaré Probab. Statist. 37 (2001), no. 3, 373--402.
 
 \bibitem{BdMK1}
   	A. Boutet de Monvel and A. Khorunzy,
	\emph{Asymptotic distribution of smoothed eigenvalue density I. Gaussian random matrices,}
	Random Oper. Stochastic Equations { 7} (1999), pp 1--22.
	
	
 \bibitem{BdMK2}
   	A. Boutet de Monvel and A. Khorunzy,
	\emph{Asymptotic distribution of smoothed eigenvalue density II. Wigner random matrices.} Random Oper. Stochastic Equations 7 (1999), no. 2, 149--168.
	

	
\bibitem{BD} J. Breuer and M. Duits, 
	\emph{The Nevai condition and a local law of large numbers for orthogonal polynomial ensembles,} arXiv:1301.2061
	
	\bibitem{BD2} J. Breuer and M. Duits, \emph{Central Limit Theorems for biorthogonal ensembles with a recurrence}, arXiv:1309.6224
	
	
\bibitem{Dyson} F.J. Dyson, 
	\emph{A Brownian-motion model for the eigenvalues of a random matrix},
	J. Math. Phys. 3 (1962), 1191--1198. 
	

\bibitem{EKN1} 
L.Erd\H{o}s and  A. Knowles, \emph{The Altshuler-Shklovskii Formulas for Random Band Matrices I: the Unimodular Case}, arXiv:1309.5106

\bibitem{EKN2} L.Erd\H{o}s and  A. Knowles,\emph{
The Altshuler-Shklovskii Formulas for Random Band Matrices II: the General Case}, arXiv:1309.5107

		
	
\bibitem{EY}
	L. Erd\H{o}s and H.-T.~Yau,
	\emph{Universality of local spectral statistics of random matrices}, Bull. Amer. Math. Soc. (N.S.) 49 (2012), no. 3, 377--414.
	
	\bibitem{EYY}
	L. Erd\H{o}s, H.-T.~Yau and J.~Yin, \emph{
Rigidity of eigenvalues of generalized Wigner matrices}, 
Adv. Math. 229 (2012), no. 3, 1435--1515. 


\bibitem{FN1}
	P.J. Forrester and T. Nagao, 
	\emph{Correlations for the circular Dyson Brownian motion model with Poisson initial conditions}
Nuclear Phys. B 532 (1998), no. 3, 733--752. 
\bibitem{FN2}
	P.J. Forrester and T. Nagao, 
	\emph{Multilevel dynamical correlation functions for Dyon's Brownian motion model of random matrics}, Phys. Lett. A 247 (1998), 42--46.
	
\bibitem{F}
	P.J.~Forrester, \emph{Log-gases and random matrices}, London Mathematical Society Monographs Series, 34. Princeton University Press, Princeton, NJ, 2010.
	
\bibitem{GG}
	T. Guhr and A. M\"uller-Groeling, \emph{Spectral correlations in the crossover between GUE and Poisson regularity: on the identification of scales.}
J. Math. Phys. 38 (1997), no. 4, 1870--1887. 

\bibitem{G}
	T. Guhr, \emph{
Transitions toward quantum chaos: with supersymmetry from Poisson to Gauss.}
Ann. Physics 250 (1996), no. 1, 145--192. 
	
\bibitem{HKPV} J.~B.~Hough, M.~Krishnapur, Y.~Peres and B.~Vir\'ag, \emph{Determinantal processes and independence}, Prob.\ Surv.\ {\bf 3} (2006), 206--229.	

\bibitem{Israelsson}
	S.~Israelsson,
\emph{Asymptotic fluctuations of a particle system with singular interaction},
Stochastic Process. Appl. 93 (2001), no. 1, 25--56. 

\bibitem{Jduke}
	K. Johansson, 
	\emph{On fluctuations of eigenvalues of random Hermitian matrices,} Duke Math. J. 91 (1998), no. 1, 151--204. 

\bibitem{J1} K.~Johansson, 
	\emph{Universality of the local spacing distribution in certain ensembles of Hermitian Wigner Matrices}, 
	Commun.\ Math.\ Phys.\ {\bf 215} (2001), 683--705.



\bibitem{J4} K.~Johansson, \emph{Random matrices and determinantal processes},  
		Mathematical Statistical Physics, Elsevier B.V.\ Amsterdam (2006) 1--55.
		
\bibitem{J5}	K.~Johansson,	\emph{Universality for certain Hermitian Wigner matrices under weak moment conditions}, 
Ann. Inst. Henri Poincaré Probab. Stat. 48 (2012), no. 1, 47--79. 


\bibitem{K} W.~K\"onig, \emph{Orthogonal polynomial ensembles in probability theory}, Probab.\ Surveys {\bf 2} (2005), 385--447.

\bibitem{L} R.~Lyons, \emph{Determinantal probability measures}, Publ.\ Math.\ Inst.\ Hautes Etudes Sci.\ {\bf 98}  (2003), 167--212.

\bibitem{Bern} C.~McDiarmid, \emph{Concentration.} In: Probabilistic methods for algorithmic discrete mathematics, pp 195?-248, Algorithms Combin., 16, Springer, Berlin, 1998.
\bibitem{Sosh} A.~Soshnikov, \emph{Determinantal random point fields},  Uspekhi Mat.\ Nauk {\bf 55} (2000), no.\ 5 (335), 107--160; translation in Russian Math.\ Surveys {\bf 55} (2000), no.\ 5, 923--975.



\bibitem{Sosh3} A. Soshnikov, \emph{The central limit theorem for local linear statistics in classical compact groups and related combinatorial identities,} Ann. Probab.
28 (2000), no. 3, 1353--1370.
\end{thebibliography}
\end{document}